\documentclass{statsoc} 
\usepackage[a4paper,margin=1in]{geometry}
   
\usepackage{amssymb}
\usepackage{amsmath}
\usepackage{graphicx,color}
\usepackage{xcolor}
\usepackage{amsfonts}
\usepackage{pifont}
\usepackage{natbib}
\usepackage{comment}
\usepackage{setspace}
\usepackage{float}
\floatstyle{plaintop}
\restylefloat{table} 

\usepackage[labelfont={sf,bf,small}, textfont={sf,small}]{caption}
\usepackage[labelformat=simple]{subcaption}

 \usepackage{array}
 
\usepackage[flushleft]{threeparttable}
\usepackage{tablefootnote}

\usepackage{tikz}
\usetikzlibrary{arrows,positioning,arrows.meta} %arrows.meta shows the arrows shown as triangular when using Latex
\usepackage{pgfplots}
\usetikzlibrary{shapes.geometric}
\usetikzlibrary{shapes}

\usepackage{xr}
\makeatletter
\newcommand*{\addFileDependency}[1]{% argument=file name and extension
  \typeout{(#1)}
  \@addtofilelist{#1}
  \IfFileExists{#1}{}{\typeout{No file #1.}}
}
\makeatother

\usepackage{url}
\usepackage{natbib}
\usepackage[hidelinks]{hyperref}

\graphicspath{ {figures/} }

\newtheorem{theorem}{Theorem}

\newtheorem{condition}{Condition}

\newtheorem{example}[theorem]{Example}

\newtheorem{lemma}[theorem]{Lemma}

\newtheorem{proposition}[theorem]{Proposition}
\newtheorem{remark}{Remark}

\newcommand{\Y}{Y}
\newcommand{\M}{M}
\renewcommand{\S}{S}
\newcommand{\X}{\boldsymbol{X}}
\newcommand{\DE}{\tau_{\S}}

\newcommand{\nM}{J}

\newcommand{\betaX}{\boldsymbol{\beta}_{\X}}
\newcommand{\betaM}{\beta_M}
\newcommand{\betaMn}{\beta_{M,n}}
\newcommand{\hatbetaM}{\hat{\beta}_M}
\newcommand{\hatbetaMn}{\hat{\beta}_{M,n}}
\newcommand{\hatbetaMgn}{\hat{\beta}_{M,\perp',n}}

\newcommand{\alphaS}{\alpha_{\S}}
\newcommand{\alphaSn}{\alpha_{\S,n}}
\newcommand{\hatalphaS}{\hat{\alpha}_{\S}}
\newcommand{\hatalphaSn}{\hat{\alpha}_{\S,n}}
\newcommand{\hatalphaSgn}{\hat{\alpha}_{\S,\perp,n}}
\newcommand{\alphaX}{\boldsymbol{\alpha}_{\X}}
\newcommand{\hatalphaXn}{\hat{\boldsymbol{\alpha}}_{\X,n}}

\newcommand{\eY}{\epsilon_Y}

\newcommand{\eM}{\epsilon_M}

\newcommand{\hatenY}{\hat{\epsilon}_{Y,n}}
\newcommand{\hatenM}{\hat{\epsilon}_{M,n}}

\newcommand{\error}{\epsilon}

\newcommand{\hatenYi}{\hat{\epsilon}_{Y,n,i}}
\newcommand{\hatenMi}{\hat{\epsilon}_{M,n,i}}

\newcommand{\hatenMperpboot}{\hat{\epsilon}_{M,\perp,n}}
\newcommand{\hatenYperpboot}{\hat{\epsilon}_{Y,\perp', n}}
\newcommand{\hatenMperpbooti}{\hat{\epsilon}_{M,\perp,n, i}}
\newcommand{\hatenYperpbooti}{\hat{\epsilon}_{Y,\perp', n, i}}
\newcommand{\hatsigmaalphaperpboot}{\hat{\sigma}_{\alphaS ,\perp,n}^*}
\newcommand{\hatsigmabetaperpboot}{\hat{\sigma}_{\betaM, \perp', n}^*}

\newcommand{\Sperp}{\S_{\perp}}
\newcommand{\Sperpi}{\S_{\perp,i}}
\newcommand{\hatSperp}{\hat{\S}_{\perp}}
\newcommand{\hatSperpi}{\hat{\S}_{\perp,i}}
\newcommand{\Sperpboot}{\S_{\perp}^*}
\newcommand{\Sperpbooti}{\S_{\perp,i}^*}

\newcommand{\Mperp}{\M_{\perp}}
\newcommand{\Mperpi}{\M_{\perp,i}}
\newcommand{\hatMperp}{\hat{\M}_{\perp}}
\newcommand{\hatMperpi}{\hat{\M}_{\perp,i}}
\newcommand{\Mperpboot}{\M_{\perp}^*}
\newcommand{\Mperpbooti}{\M_{\perp,i}^*}

\newcommand{\Mperpp}{\M_{\perp'}}
\newcommand{\Mperppi}{\M_{\perp',i}}
\newcommand{\hatMperpp}{\hat{\M}_{\perp'}}
\newcommand{\hatMperppi}{\hat{\M}_{\perp',i}}
\newcommand{\Mperppboot}{\M_{\perp'}^*}
\newcommand{\Mperppbooti}{\M_{\perp',i}^*}

\newcommand{\Yperppi}{\Y_{\perp',i}}
\newcommand{\hatYperpp}{\hat{\Y}_{\perp'}}
\newcommand{\hatYperppi}{\hat{\Y}_{\perp',i}}
\newcommand{\Yperppboot}{\Y_{\perp'}^*}
\newcommand{\Yperppbooti}{\Y_{\perp',i}^*}

\newcommand{\ZZSbootn}{\mathbb{Z}_{\S,n}^*}
\newcommand{\ZZMbootn}{\mathbb{Z}_{\M,n}^*}
\newcommand{\ZZSperpbootn}{\mathbb{Z}_{\S, \perp, n}^*}
\newcommand{\ZZMperpbootn}{\mathbb{Z}_{\M, \perp', n}^*}

\newcommand{\VVSn}{\mathbb{V}_{\S,n}}
\newcommand{\VVMn}{\mathbb{V}_{\M,n}}
\newcommand{\VVSbootn}{\mathbb{V}_{\S,n}^*}
\newcommand{\VVMbootn}{\mathbb{V}_{\M,n}^*}
\newcommand{\VVSperpbootn}{\mathbb{V}_{\S, \perp, n}^*}
\newcommand{\VVMperpbootn}{\mathbb{V}_{\M, \perp', n}^*}

\newcommand{\sigmaeY}{\sigma_{\eY}}
\newcommand{\sigmaeM}{\sigma_{\eM}}

\newcommand{\sigmaalpha}{\sigma_{\alphaS}}
\newcommand{\sigmabeta}{\sigma_{\betaM}}
\newcommand{\hatsigmaalpha}{\hat{\sigma}_{\alphaS}}
\newcommand{\hatsigmabeta}{\hat{\sigma}_{\betaM}}
\newcommand{\hatsigmaalphan}{\hat{\sigma}_{\alphaS,n}}
\newcommand{\hatsigmabetan}{\hat{\sigma}_{\betaM,n}}

\newcommand{\bbeta}{b_{\beta}} %the local parameter for beta
\newcommand{\balpha}{b_{\alpha}} %the local parameter for gamma

\newcommand{\hatmaxpn}{\hat{\theta}_n}
\newcommand{\maxpnull}{\theta_0}
\newcommand{\maxpnulln}{\theta_{0,n}}

\newcommand{\myindicator}{\mathrm{I}}
\newcommand{\realnumber}{\mathbb{R}}
\newcommand{\mytrans}{\top}
\newcommand{\myident}{I} %identity matrix

\newcommand{\RNum}[1]{\uppercase\expandafter{\romannumeral #1\relax}}

\title[]{Adaptive Bootstrap Tests for Composite Null Hypotheses in the Mediation Pathway Analysis}
\author[]{Yinqiu He}
\address{University of Wisconsin, Madison, USA.}
\author[]{Peter X.-K. Song, and Gongjun Xu\footnote{Correspondence.  Gongjun Xu, Department of Statistics, University of Michigan, Ann Arbor, USA. \href{gongjun@umich.edu}{Email: gongjun@umich.edu}}}
\address{University of Michigan, Ann Arbor, USA.}

\begin{document}

\begin{abstract}
Mediation analysis aims to assess if, and how, a certain exposure influences an outcome of interest through intermediate variables. This problem has recently gained a surge of attention due to the tremendous need for such analyses in scientific fields. Testing for the mediation effect is greatly challenged by the fact that the underlying null hypothesis (i.e. the absence of mediation effects) is composite. Most existing mediation tests are overly conservative and thus underpowered. 
To overcome this significant methodological hurdle, we develop an adaptive bootstrap testing framework that can accommodate different types of composite null hypotheses in the mediation pathway analysis. Applied to the product of coefficients (PoC) test and the joint significance (JS)  test, our adaptive testing procedures provide type \RNum{1} error control under the composite null, resulting in much improved statistical power compared to existing tests. Both theoretical properties and numerical examples of the proposed methodology are discussed. 	 
\end{abstract}

\keywords{Mediation analysis, Structural equation model, Composite hypothesis, Bootstrap}

\section{Introduction}
Mediation analysis plays a crucial role in investigating
the underlying mechanism or pathway between  an exposure and an outcome through an intermediate variable called a  mediator  \citep{mackinnon2008introduction,vanderweele2015explanation}. 
It decomposes the ``total effect'' of an exposure on an outcome into an indirect effect that is through a given mediator and a direct effect, not through the mediator.
The former holds the key to uncovering the exposure-outcome mechanism and is often known as the mediation effect. 
The mediation effect was initially studied under structural equation models (SEMs) in social sciences \citep{sobel1982asymptotic,baron1986moderator} and has  been given formal causal definitions  \citep{robins1992identifiability,peardirect2001,imai2010identification} within the counterfactual framework  \citep{imbens2015causal}. 
Examining 
the presence or absence of the mediation effect  
can facilitate a deeper understanding of the underlying causal pathway from the exposure to the outcome and can give essential  insights into 
intervention consequences, e.g.,   manipulating the mediator to change the exposure-outcome mechanism. 
As a result, it is of interest to apply  mediation analysis
in many scientific fields, 
such as psychology \citep{mackinnon2009current,valeri2013mediation}, 
genomics \citep{zhao2014more,huang2016hypothesis,huang2018joint,guo2022high}, and 
 epidemiology \citep{barfield2017testing,fulcher2019estimation}, among others.

To analyze the mediation effect, one classical setting  models the relationship between the exposure, the potential mediator, and the outcome as a directed acyclic graph; see Figure \ref{fig:model}. 
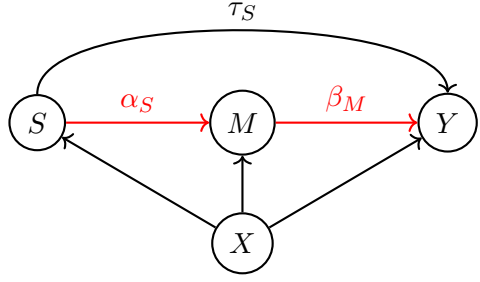
\begin{figure}[!tbhp]
\begin{center}
\vspace{0.5em}
\begin{tikzpicture}
\node[shape=circle,draw, thick] (S) at (-2.7,0) {$S$};   
\node[shape=circle,draw, thick] (M) at (0,0) {$M$}; 
\node[shape=circle,draw, thick] (Y) at (2.7,0) {$Y$}; 
\node[shape=circle,draw, thick] (X) at (0, -1.6) {$X$}; 

\draw[->, thick, red] (S.east) -- node[midway,above] {{$\alpha_S$}} (M.west);
\draw[->, thick, red] (M.east) -- node[midway,above] {{$\beta_M$}} (Y.west);
\draw[->, thick] (X.north) --  (M.south);
\draw[->, thick] (X) --  (Y);
\draw[->, thick] (X) --  (S);
\draw[->, thick] (S) .. controls +(up:15.5mm) and + (up:14.5mm)  .. node[midway,above] {{$\tau_S$}}  (Y);
\end{tikzpicture}
\vspace{-1em}
\end{center}
\caption{Directed Acyclic Graph for Mediation Analysis. The exposure is $\S$; the mediator is $\M$; the outcome is $\Y$; the potential confounders are $X$.}\label{fig:model}
\end{figure}
Specifically, 
let $\alphaS$ parametrize the causal effect of the exposure on the mediator, and $\betaM$ parametrize the causal effect of the mediator on the outcome conditioning on the exposure. 
Then in the classical linear SEM   \citep{sobel1982asymptotic,baron1986moderator}, 
the causal mediation effect is proportional to  $\alphaS \betaM$
under suitable identification assumptions  \citep{imai2010general}.  
More generally, 
this product expression $\alphaS \betaM$ may also appear in the  causal mediation effect  under many other models, 
such as generalized linear models and survival analysis models \citep{vanderweele2010odds,vanderweele2011causal,huang2016mediation}. 
Therefore, 
the important scientific question of whether or not the mediation effect is absent can be formulated as the hypothesis testing problem $H_0: \alphaS \betaM = 0$ against $H_A: \alphaS\neq 0$ and $\betaM \neq 0$ \citep{mackinnon2008introduction}. 
Note that $H_0: \alphaS\betaM=0$ holds if and only if $\alphaS=0$ or $\betaM = 0$,
corresponding to two parameter sets $\mathcal{P}_{\alpha}=\{(\alphaS,\betaM) : \alphaS = 0\}$
and $\mathcal{P}_{\beta}=\{(\alphaS,\betaM) : \betaM = 0\}$, respectively. 
It follows that the parameter set of $H_0:\alphaS\betaM=0$ is the union of two sets  $\mathcal{P}_{\alpha}$ and $\mathcal{P}_{\beta}$. 
We visualize  $\mathcal{P}_{\alpha}$, $\mathcal{P}_{\beta}$, and their \textit{union}  $\mathcal{P}_{\alpha}\cup \mathcal{P}_{\beta}$
in Figures \ref{fig:illusalphas0}--\ref{fig:visualofcompositenull}, respectively.

\begin{figure}[!htbp]
\begin{subfigure}{0.24\textwidth}
\begin{tikzpicture}
\pgfplotsset{width=5cm}
  \begin{axis}[grid=none,ymin=-1,ymax=1,xmax=1,xmin=-1,xticklabel=\empty,yticklabel=\empty,
               minor tick num=0, 
               axis line style = {line width=1pt, draw=gray!45}, 
               grid style={line width=.1pt, draw=gray!45},
               axis lines = middle, 
               major tick style={draw=none}, 
               xlabel=$\alphaS$,ylabel=$\betaM$,label style =
               {at={(ticklabel cs:1.18)}, font=\footnotesize}]
 \addplot[-stealth, green!60!black, line width=1.7pt] coordinates {(0,-1) (0,1)}; 
 \addplot[-stealth, line width=0.5pt, gray] coordinates { (0.24,0.44) (0.02,0.44)};
   \addplot[-, line width=.5pt, draw=gray!45] coordinates { (-1,-1) (-1,1)};
   \addplot[-, line width=.5pt, draw=gray!45] coordinates { (-1,-1) (1,-1)};
   \addplot[-, line width=.5pt, draw=gray!45] coordinates { (1,-1) (1,1)};
   \addplot[-, line width=.5pt, draw=gray!45] coordinates { (1,1) (-1,1)};

 \node[green!60!black, align=center, thick] at (axis cs: 0.61,0.46) {\footnotesize $\alphaS=0$};

  \end{axis} 
\end{tikzpicture}
\caption{$\alphaS=0$}\label{fig:illusalphas0}
\end{subfigure}
\hspace{1pt}
\begin{subfigure}{0.24\textwidth}
\begin{tikzpicture}
\pgfplotsset{width=5cm}
  \begin{axis}[grid=none,ymin=-1,ymax=1,xmax=1,xmin=-1,xticklabel=\empty,yticklabel=\empty,
               minor tick num=0,
               axis line style =  {line width=1pt, draw=gray!45}, 
               grid style={line width=.1pt, draw=gray!45},
               axis lines = middle, 
               minor tick length=0.1pt,
               major tick style={draw=none}, 
               xlabel=$\alphaS$,ylabel=$\betaM$,label style =
               {at={(ticklabel cs:1.18)}, font=\footnotesize}]
 \addplot[-stealth, blue!85!black, line width=1.6pt] coordinates
           {(-1,0) (1,0)}; 
 \addplot[-stealth, line width=0.5pt, gray] coordinates { (0.58,0.34) (0.58,0.05)};
   \addplot[-, line width=.5pt, draw=gray!45] coordinates { (-1,-1) (-1,1)};
   \addplot[-, line width=.5pt, draw=gray!45] coordinates { (-1,-1) (1,-1)};
   \addplot[-, line width=.5pt, draw=gray!45] coordinates { (1,-1) (1,1)};
   \addplot[-, line width=.5pt, draw=gray!45] coordinates { (1,1) (-1,1)};

 \node[blue!85!black, align=center, thick] at (axis cs:  0.58,0.46) {\footnotesize $\betaM=0$}; 
  \end{axis} 
\end{tikzpicture}
\caption{$\betaM=0$}\label{fig:illusbetam0}
\end{subfigure} 
\hspace{1pt}
\begin{subfigure}{0.24\textwidth}
\begin{tikzpicture}
\pgfplotsset{width=5cm}
  \begin{axis}[grid=none,ymin=-1,ymax=1,xmax=1,xmin=-1,xticklabel=\empty,yticklabel=\empty,
               minor tick num=0, 
                {line width=1pt, draw=gray!45}, 
               grid style={line width=.1pt, draw=gray!45},
               axis lines = middle, 
               major tick style={draw=none}, 
               xlabel=$\alphaS$,ylabel=$\betaM$,label style =
               {at={(ticklabel cs:1.18)}, font=\footnotesize}]
   \addplot[-, line width=.5pt, draw=gray!45] coordinates { (-1,-1) (-1,1)};
   \addplot[-, line width=.5pt, draw=gray!45] coordinates { (-1,-1) (1,-1)};
   \addplot[-, line width=.5pt, draw=gray!45] coordinates { (1,-1) (1,1)};
   \addplot[-, line width=.5pt, draw=gray!45] coordinates { (1,1) (-1,1)};

 \addplot[-stealth, red, line width=1.5pt] coordinates {(-1,0) (1,0)}; 
 \addplot[-stealth, line width=0.5pt, gray] coordinates { (0.6,0.45) (0.6,0.05)};
 \node[red, align=center, thick] at (axis cs:  0.52,0.7) {\footnotesize $\alphaS\betaM=0$};
  \addplot[-stealth, line width=0.5pt, gray] coordinates {(0.38,0.44) (0.02,0.3)};  
  \addplot[-stealth, red, line width=1.5pt] coordinates {(0,-1) (0,1)}; 
  
  \end{axis} 
\end{tikzpicture}
\caption{$\alphaS\betaM=0$} \label{fig:visualofcompositenull}
\end{subfigure}
\hspace{1pt}
\begin{subfigure}{0.24\textwidth}
\begin{tikzpicture}
\pgfplotsset{width=5cm}
  \begin{axis}[grid=none,ymin=-1,ymax=1,xmax=1,xmin=-1,xticklabel=\empty,yticklabel=\empty,
               minor tick num=0, 
               axis line style =  {line width=1pt, draw=gray!45}, 
               grid style={line width=.1pt, draw=gray!45},
               axis lines = middle, 
               major tick style={draw=none}, 
               xlabel=$\alphaS$,ylabel=$\betaM$,label style =
               {at={(ticklabel cs:1.18)}, font=\footnotesize}] 
   \addplot[-, line width=.5pt, draw=gray!45] coordinates { (-1,-1) (-1,1)};
   \addplot[-, line width=.5pt, draw=gray!45] coordinates { (-1,-1) (1,-1)};
   \addplot[-, line width=.5pt, draw=gray!45] coordinates { (1,-1) (1,1)};
   \addplot[-, line width=.5pt, draw=gray!45] coordinates { (1,1) (-1,1)};

 \addplot[-stealth, line width=0.5pt, gray] coordinates { (0.45,0.37) (0.11,0.11)};
 \node[orange, align=center, thick] at (axis cs: 0.5,0.79) {\footnotesize $\alphaS=0$};
  \node[orange, align=center, thick] at (axis cs: 0.5,0.53) {\footnotesize $\betaM=0$};
 
 \node[regular polygon,regular polygon sides=3, orange, draw, fill=orange, inner sep=1.7pt] at (axis cs: 0,0) {};

  \end{axis} 
\end{tikzpicture}
\caption{$\alphaS=\betaM=0$}\label{fig:visualofsingleton}
\end{subfigure}
\caption{Visualization of parameter spaces of $(\alphaS, \betaM)$ under different constraints}
\label{fig:visualizationoftypesofnull}
\end{figure}
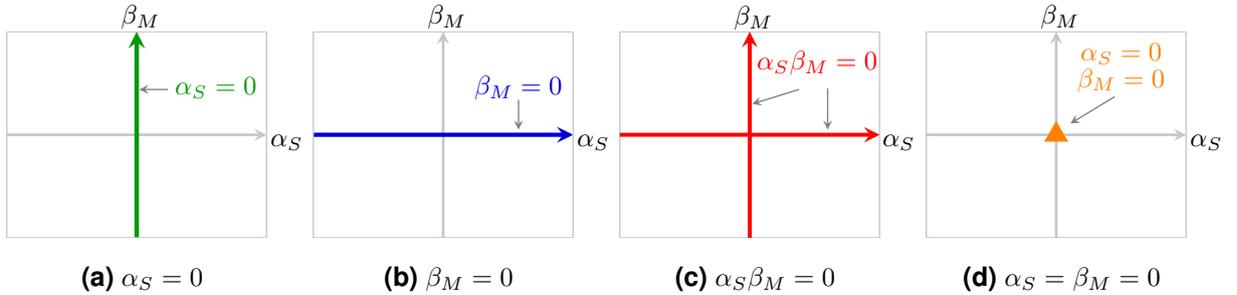

To test $H_0:\alphaS\betaM=0$,   
a broad class of methods 
 is  based on the  product of coefficients (PoC)
 $\hatalphaSn \hatbetaMn$, where $\hatalphaSn$ and $\hatbetaMn$ denote the sample estimates 
 of parameters $\alphaS$ and $\betaM$, respectively. 
One popular  PoC method 
is Sobel's test  \citep{sobel1982asymptotic}, which is a Wald-type test and approximates the variance of $\hatalphaSn \hatbetaMn$ by the first-order Delta method. 
In addition,  the joint significance (JS) test \citep{fritz2007required}, also known as the MaxP test, is another widely used test that rejects the $H_0$ of no mediation effect 
if both $\hatalphaSn$ and $\hatbetaMn$ pass a certain cutoff of statistical significance. \cite{liu2020large} pointed out that 
the MaxP test is a kind of likelihood ratio test under  normality assumptions.

Although there are 
various procedures available for testing mediation effects, 
properly controlling the type \RNum{1} error remains a challenge due to the  intrinsic structure of the null parameter space. 
In particular,  $H_0: \alphaS\betaM=0$ is  composed of three different parameter cases: 
(i) $\alphaS = 0$ and $ \betaM \neq 0$; 
(ii) $\alphaS \neq 0$ and $\betaM = 0$;
(iii) $\alphaS = 0$ and $\betaM = 0$.  
Case (iii), illustrated in Figure \ref{fig:visualofsingleton}, 
 is a singleton given by the \textit{intersection} set $\mathcal{P}_{\alpha}\cap \mathcal{P}_{\beta}$.  
 Under case (iii), both parameters $\alphaS$ and $\betaM$ are fixed at $0$, 
 whereas cases (i) and (ii) have one fixed parameter and the other parameter to be estimated.
This intrinsic difference leads to distinct asymptotic behaviors of test statistics.
Since the underlying truth is typically unknown in practice, it is difficult to obtain  proper $p$-values under the composite null hypothesis.  

Particularly, in the popular Sobel's test and the MaxP test, the asymptotic distributions of the test statistics under cases (i) and (ii) are known to be different from those under case (iii). These tests have been shown to be overly conservative in case (iii), because statistical inference is carried out according to the asymptotic distributions determined in cases (i) and (ii) \citep{mackinnon2002comparison,fritz2007required}. 
This issue has gained a surge of attention 
in recent genome-wide  epidemiological studies, 
where for the majority of omics markers, 
it holds that $\alphaS =\betaM = 0$, and 
% where for the majority of omics markers would have null hypotheses of $\alphaS =\betaM = 0$, and 
the classical tests are generally underpowered \citep{barfield2017testing}.  
 Some recent work    \citep{liu2020large,dai2020multiple,du2022methods}  utilized the relative proportions of the cases (i)--(iii) in the population, but they rely on accurate estimation of the true proportions.  
\cite{huang2019genome,huang2019variance}  adjusted the  composition of $H_0:\alphaS\betaM=0$ through the variances of test statistics but required that the non-zero 
coefficients are weak and sparse, which can be violated when the sample size is large. 
Another line of related research \citep{sampson2018fwer,djordjilovic2019global,djordjilovic2020optimal,derkach2020grouptesting} used a screening step
to control the family-wise error rate or the false discovery rate for a large group of hypotheses,
but they did not directly provide proper  $p$-values for each of the  composite null hypotheses. 
\cite{van2021nearly} proposed to  construct a critical region for testing that can nearly    control the type \RNum{1} error at one prespecified significance level.   \label{pageref:citevan} 
 \cite{miles2021optimal}  construct a rejection region that can achieve type \RNum{1} error control at significance level $\omega$ with $\omega^{-1}$ a positive integer.   \label{pageref:miles}
Despite these developments, the fundamental issue of correctly characterizing the distributions of test statistics  to obtain 
well-calibrated $p$-values under a composite null hypothesis remains an important challenging  problem in the current literature of mediation analyses.

In this paper, we develop a new hypothesis testing methodology to  address the challenge of obtaining \textit{uniformly distributed} $p$-values under the composite null hypothesis 
 of no mediation effect.  
Particularly, we propose an adaptive bootstrap procedure that can flexibly accommodate different types of null hypotheses. 
In the current literature, 
the nonparametric bootstrap is   directly applied to the PoC test statistic  $\hatalphaSn\hatbetaMn$,
which has been, unfortunately, found numerically to be
 overly conservative when $\alphaS=\betaM=0$  \citep{fritz2007required,barfield2017testing}. 
This paper unveils analytically the reason for the failure of the conventional  nonparametric bootstrap method, which stems from non-regular limiting behaviors 
of the PoC test statistic at the neighborhood of  $(\alphaS, \betaM)=(0,0)$. 
To overcome the non-regularity near $(\alphaS, \betaM)=(0,0)$,
we derive an explicit representation of the asymptotic distribution of the PoC test statistic through a local model, 
and perform a consistent bootstrap estimation by incorporating suitable thresholds.
In addition, for the JS test, 
we show that the conventional nonparametric bootstrap also fails to control type I error properly, which can be fixed by an adaptive bootstrap test similar to the procedure of the PoC test.
For both the PoC test and the JS test, 
the proposed methods can  circumvent the nonstandard  limiting behaviors of the test statistics 
and therefore uniformly adapt to different types of null cases  of no mediation effect.

The structure of this paper is as follows. In Section  \ref{sec:review}, 
we briefly review the basic problem setting and several popular testing methods in the literature. 
In Section \ref{sec:abt}, we introduce the adaptive bootstrap method that can be applied to the representative PoC and JS tests under classical linear SEMs. 
In Section \ref{sec:sim}, we conduct extensive simulation studies to compare the finite sample performances of the proposed tests with  popular counterparts. 
In Section \ref{sec:datanalysis}, we apply our adaptive bootstrap tests to investigate the mediation pathways of metabolites on the association of  environmental exposures with a health outcome. 
In Section \ref{sec:extendmodel}, we develop extensions of the  adaptive bootstrap, including joint testing of multivariate mediators and testing mediation effects under non-linear models.  
We conclude the paper and discuss interesting extensions in Section \ref{sec:discussion}. 
% All the technical proofs are deferred to the Supplementary Material. 

% \smallskip 
\textit{Notation}. For two sequences of real numbers $\{a_n\}$ and $\{b_n\}$,
we write $a_n=o(b_n)$ if $\lim_{n\to \infty} a_n/b_n=0$. We let $\xrightarrow{d}$ denote convergence in distribution.  
We let $\overset{d^*}{\leadsto} $ denote   bootstrap consistency relative to the Kolmogorov-Smirnov distance; 
see an introduction of this consistency notion in Section 23 of \cite{van2000asymptotic}. To ensure clarity, we also provide the definitions of all the convergence modes in Section \ref{sec:defnotationconv} of the Supplementary  Material.\label{page:notationconv}

\section{Hypothesis Tests of No Mediation Effect} \label{sec:review}
\noindent \textit{Mediation Analysis Model.} \  To examine the mediation effect of the exposure $\S$ on the outcome $\Y$ through the intermediate variable $\M$, the causal inference literature  utilizes the counterfactual framework  \citep{vanderweele2015explanation}. 
In particular, let $M(s)$ denote the potential value of the mediator under the exposure $\S=s$, and let $\Y(s, m)$ denote the potential outcome that would have been observed if $\S$ and $\M$ had been set to $s$ and $m$, respectively. 
Throughout the paper, we adopt the  Stable Unit Treatment Value Assumption \citep{rubin1980randomization}, so that $\M=\M (\S)$ and $Y=Y (\S,\M(S))$.
Then the mediation effect or the natural indirect effect of $\S=s$ versus $s^*$ \citep{imai2010identification} is defined as 
\begin{align}
	\mathrm{E}\big\{ \Y(s, \M(s)) - \Y(s, \M(s^*))  \big\}. \label{eq:MEpotential}
\end{align} 
For  ease of illustration, we consider the popular linear Structural Equation Model (SEM)  \citep{mackinnon2008introduction,vanderweele2015explanation}: 
\begin{eqnarray}
	\M &=& \alphaS \S +  \X^\mytrans \alphaX +\eM, \label{eq:fullmod1} \\  \Y   &=&  \betaM \M + \X^\mytrans  \betaX + \DE \S + \eY,   \notag  
\end{eqnarray}
where $\X$ denotes a vector of confounders with the first element being 1 for the intercept, 
and $\eY$ and $\eM$ are independent error terms with mean zero and finite variances $\sigmaeY^2$ and $\sigmaeM^2$, respectively.
We assume that there are $n$ independent and identically distributed (i.i.d.) observations, $\{(\S_i, \X_i, \M_i, \Y_i): i=1,\ldots, n\}$,  sampled from Model \eqref{eq:fullmod1}. 
The independence of $\eY$ and $\eM$  holds under the following no unmeasured confounding assumptions. 
In particular,  
 let $A\perp  B \mid C$ denote the independence of $A$ and $B$ conditional on $C$,
and we assume that for all levels of $s, s^*$ and $m$, 
(i) $\Y(s,m) \perp \S \mid \{\X =\boldsymbol{x}\}$, no confounder for the relation of $\Y$ and $\S$;  (ii) $\Y(s,m) \perp \M \mid \{\S=s, \X=\boldsymbol{x}\}$, no confounder for the relation of $\Y$ and $\M$ conditioning on $\S=s$; (iii) $\M(s)\perp \S \mid \{\X=\boldsymbol{x}\}$, no confounder for the relation of $\M$ and $\S$;  (iv) $\Y(s,m) \perp \M(s^*) \mid \{\X = \boldsymbol{x}\}$,   
 no confounder for the $\M$-$\Y$ relation that is affected by $\S$ 
 \citep{vanderweele2009conceptual}.   
Under these assumptions, 
 the model can be visualized by the directed acyclic graph in Figure \ref{fig:model}, and 
the mediation effect $\eqref{eq:MEpotential}$ equals $\alphaS\betaM (s-s^*)$.

Therefore,
 the scientific goal of detecting the presence of a mediation effect can be formulated as the following hypothesis testing problem:
\begin{align*}
	H_0: \alphaS  \betaM = 0 \quad \text{\ versus\ } \quad H_A: \alphaS \betaM  \neq 0.
\end{align*}
This null hypothesis is composite and 
can be decomposed into three disjoint cases:
\begin{equation}
	H_0: \begin{cases}
H_{0,1}:  & \alphaS = 0, ~\betaM \neq 0; \\ 
H_{0,2}: & \alphaS \neq 0,~ \betaM =0; \\
H_{0,3}: & \alphaS = \betaM =0, 
\end{cases} \label{eq:compositenull}
\end{equation}
and the alternative hypothesis is $H_A: \alphaS \neq 0 \text{ and } \betaM \neq 0.$

\begin{remark}
Composite null problems similar to  \eqref{eq:compositenull} can occur in  settings other than Model \eqref{eq:fullmod1}; the latter is considered to demonstrate the essential analytic details useful for possible generalizations.  
% In Section \ref{sec:extendmodel}, we develop extensions under   multivariate linear models and generalized linear models. 
Similar issues have also been observed in many other scenarios, including partially linear models \citep{hines2021robust}, 
survival analysis \citep{vanderweele2011causal}, and high-dimensional models \citep{zhou2020estimation}.   
The analytic details of the methodology development in this paper can pave the path for useful generalizations to other important statistical models and applications. 
\end{remark}

To test  the composite null \eqref{eq:compositenull}, various methods have been proposed, and a comprehensive survey can be found in  \cite{mackinnon2002comparison}.
There are two representative classes of tests: 
(I) the product of coefficients (PoC) test, which corresponds to a Wald-type test, 
 and (II) the joint significance (JS) test, which  is the likelihood ratio test under  normality of the error terms \citep{liu2020large}. 
\textit{(I)} 
The first class of methods examine the PoC: $\hatalphaSn \hatbetaMn$, where  $\hatalphaSn$ and $\hatbetaMn$  denote 
consistent estimates of $\alphaS$ and $\betaM$, respectively.  
One common practice is to apply a normal approximation 
to $\hatalphaSn \hatbetaMn$ divided by  its standard error, where  
\cite{sobel1982asymptotic} derives the standard error formula by   the  first-order Delta method. 
In addition to the large-sample approximation, 
the bootstrap has also been used to evaluate the distribution of  $\hatalphaSn\hatbetaMn$  \citep{mackinnon2004confidence,fritz2007required}.   
\textit{(II)}  The JS test, also known as the MaxP test, rejects $H_0: \alphaS \betaM = 0$ if  $\max\{p_{\alphaS}, p_{\betaM}\} < \omega$,
where $\omega$ is a prespecified significance level,
and $p_{\alphaS}$  and $p_{\betaM}$ denote the $p$-values for $H_0: \alphaS=0$ (the link $\S \rightarrow \M$)  and $H_0: \betaM=0$ (the link $\M \rightarrow \Y$), respectively. 
Despite their popularity, these methods have been found numerically to be  overly conservative  under $H_{0,3}$ in \eqref{eq:compositenull}   \citep{mackinnon2002comparison,barfield2017testing}. 
See a further discussion on the non-regular asymptotic behaviors of statistics underlying the conservatism in Section \ref{sec:abt}.

% We will show below that the conservatism is caused by the non-regular asymptotic behaviors of statistics under the singleton $H_{0,3}$. 
% In Section \ref{sec:abt}, we will delve into the non-regular asymptotic behaviors of statistics that lead to the conservatism  under the singleton $H_{0,3}$. 
% Prior to that, we point out that 
Similar issues have also been broadly recognized for Wald tests in various   statistical problems including three-way contingency table analysis and  factor analysis  
\citep{glonek1993behaviour,dufour2013wald,drton2016wald}. \label{literature} 
However, characterizing non-regular asymptotic behaviors under the singular null hypothesis $H_{0,3}$  is still insufficient to address intrinsic technical challenges in testing \eqref{eq:compositenull}. 
In particular, the composite null  \eqref{eq:compositenull} includes not only the singular case $H_{0,3}$ but also the other two non-singular cases $H_{0,1}$ and $H_{0,2}$.
Since a test statistic follows different distributions under different null cases, and the underlying true null case is unknown,  
it is difficult to obtain \textit{uniformly distributed} $p$-values through one simple asymptotic  distribution under \eqref{eq:compositenull}. 
To address this technical difficulty, 
we adopt, justify, and evaluate an adaptive bootstrap procedure. 
For both  Wald-type PoC test and non-Wald JS test, we will show that the proposed procedure can naturally adapt to the three types of null hypotheses in \eqref{eq:compositenull} and yield uniformly distributed $p$-values across all different null cases.

\section{Adaptive Bootstrap Tests} \label{sec:abt}
In this section, we propose a general  Adaptive Bootstrap (AB)  procedure for testing the composite null hypothesis \eqref{eq:compositenull}. 
For illustration, we apply the  adaptive bootstrap to the representative PoC test and show it can address the  non-regularity issue. 
We emphasize that this general strategy can be applied in a wide range of scenarios. We also derive adaptive bootstrap procedure  in  other examples, including the Joint Significance test (Section \ref{sec:jst} in the Supplementary Material), joint testing of multivariate mediators (Section \ref{sec:jointestmulti}) and testing mediation effect under the generalized linear models  (Sections \ref{sec:asymbinarymedout} and \ref{sec:modelbinmed}.).

To conduct hypothesis testing or estimate confidence intervals for statistics whose limiting distributions deviate  from the normal,
a simple and powerful approach is to apply the bootstrap resampling technique.  
	However, the classical bootstrap is not a panacea,  and on some occasions it can fail to work properly, including unfortunately the non-regular scenarios considered in this paper. 
In particular, it has been observed through simulation studies that the classical  bootstrap technique is overly conservative under $H_{0,3}: \alphaS=\betaM=0$  \citep{mackinnon2002comparison,barfield2017testing}. 
We next unveil the key insights  underlying the failure of the classical bootstrap, which motivates our use of the adaptive bootstrap. 
\label{page:motivationbootstrap}

\medskip
\noindent \textit{Non-Regularity of the PoC Test.} \quad 
When  $(\alphaS, \betaM)\neq (0,0)$, one of the first-order gradients $\frac{\partial \alphaS \betaM}{\partial \alphaS} = \betaM$ and $\frac{\partial \alphaS \betaM}{\partial \betaM} = \alphaS$ is  nonzero. 
Thus the Delta method  can be applied to support 
the use of Sobel's test (based on asymptotic normality) and classical bootstrap \citep{barfield2017testing}. 
However, under $H_{0,3}: (\alphaS, \betaM)=(0,0)$, 
the gradients $\frac{\partial \alphaS \betaM}{\partial \alphaS} =\frac{\partial \alphaS \betaM}{\partial \betaM} = 0$,  
and validity of Sobel's test and the classical bootstrap cannot be  obtained as above.  

We next illustrate the non-regular limiting behavior of the PoC $\hatalphaSn \hatbetaMn$ under $H_{0,3}$. For ease of exposition, consider a  special case of \eqref{eq:fullmod1}: 
$\M = \alphaS \S +\eM,$ and $\Y = \betaM \M +  \eY.$ 
Let $(\hatalphaSn,\hatbetaMn)$ denote the ordinary least squares estimators of $(\alphaS,\betaM)$, and let $(\hatalphaSn^*, \hatbetaMn^*)$ the corresponding nonparametric bootstrap estimators. 
Here and throughout this paper, 
we use the superscript ``$*$" to indicate estimators obtained from the nonparametric bootstrap, namely ``bootstrap in pairs" in the regression setting. 
By classical asymptotic theory \citep{van2000asymptotic}, under mild conditions,  
\begin{align}\label{eq:simplecltsep}
	\sqrt{n}(\hatalphaSn - \alphaS, \, \hatbetaMn - \betaM)^{\mytrans} \xrightarrow{d} (Z_{\S,0}, \, Z_{\M,0})^{\mytrans},
\end{align} 
 where  
 $(Z_{S,0},Z_{M,0})^{\top}$ denotes a mean-zero normal random vector with a covariance matrix given by that of the random vector 
$(\eM \S /V_{\S,0}, \eY \M/V_{\M,0})^{\mytrans}$,   $V_{\S,0}=\mathrm{E}(\S^2)$,  $V_{\M,0}=\mathrm{E}(\M^2)$. 

Moreover, 
\begin{align}\label{eq:simplecltsepboot}
\sqrt{n}(\hatalphaSn^* - \hatalphaSn, \hatbetaMn^* - \hatbetaMn)^{\mytrans} 
\overset{d^*}{\leadsto} 
(Z_{\S,0}', \, Z_{\M,0}')^{\mytrans}, 
\end{align} 
where $(Z_{\S,0}', Z_{\M,0}')$ is an independent copy of  $(Z_{\S,0}, Z_{\M,0})$ in \eqref{eq:simplecltsep} under the same distribution.
By \eqref{eq:simplecltsep}, 
$n( \hatalphaS \hatbetaM - \alphaS \betaM )\xrightarrow{d} Z_{\S,0}Z_{\M,0}$ under $H_{0,3}$, 
{with the convergence rate $n$ different from the standard parametric $\sqrt{n}$ rate.}
% This leads to a proper version of the bootstrap estimator of   $n( \hatalphaSn \hatbetaMn - \alphaS \betaM )$,   
% instead of $\sqrt{n}(\hatalphaSn \hatbetaMn - \alphaS \betaM )$. 
By \eqref{eq:simplecltsep}--\eqref{eq:simplecltsepboot},
we have
$n(\hatalphaSn^*\hatbetaMn^* - \hatalphaSn \hatbetaMn) =  n\{(\hatalphaSn^* - \hatalphaSn)\hatbetaMn + (\hatbetaMn^* - \hatbetaMn)\hatalphaSn +(\hatalphaSn^* - \hatalphaS )(\hatbetaMn^* -  \hatbetaMn)\} \overset{d^*}{\leadsto}  {Z_{\S,0}'Z_{\M,0}} + {Z_{\S,0}Z_{\M,0}'} +{Z_{\S,0}'Z_{\M,0}'}.$    
We can see that the limit of $n(\hatalphaSn^*\hatbetaMn^* - \hatalphaSn \hatbetaMn)$ is different from that of 
$n( \hatalphaSn \hatbetaMn - \alphaS \betaM )$,
implying inconsistency of the classical nonparametric  bootstrap.

\bigskip
\noindent \textit{Adaptive Bootstrap of the PoC Test.}\quad  
To address the challenge of correctly evaluating the distribution of the  PoC statistic, 
we utilize the local asymptotic analysis framework. 
{Intuitively, the goal is to evaluate if a small change in the target parameters leads to little change on the limit of the statistics.} 
{To this end, 
given targeted parameters $(\alphaS,\betaM)$,
we define their locally perturbed counterparts as $\alphaSn = \alphaS + n^{-1/2}\balpha$, and $\betaMn = \betaM + n^{-1/2}\bbeta$, respectively, where 
  $(\balpha, \bbeta)$ denote the local parameters of perturbations from  our targeted coefficients $(\alphaS, \betaM). $  } 
% In particular, we frame the problem under the local linear SEM as follows:  
We then frame the problem under the local linear SEM as follows:  
\begin{eqnarray}
	\M &=& \alphaSn \S +  \X^\mytrans \alphaX +\eM, \hspace{1.5em}  \Y=\betaMn \M + \X^\mytrans  \betaX + \DE \S + \eY,\label{eq:fullmodlocal2}  
\end{eqnarray}
 where $\eM$ and $\eY$ are independent error terms with mean zero and finite variances. 
 % and the local regression coefficients are set as $\alphaSn = \alphaS + n^{-1/2}\balpha$, and $\betaMn = \betaM + n^{-1/2}\bbeta$, respectively, 
 % with $(\balpha, \bbeta)$ being the local parameters of deviations from  our targeted coefficients $(\alphaS, \betaM). $ 
Fixing the target parameters $(\alphaS,\betaM)$, according to \cite{van2000asymptotic}  the formulation given in \eqref{eq:fullmodlocal2} may also be viewed as a local statistical experiment with  local parameters $(\balpha, \bbeta)$ under which we are interested in examining the limit of test statistics.  
Note that with the local parameters $(\balpha,\bbeta)=(0,0)$, \eqref{eq:fullmodlocal2} reduces to the original model  \eqref{eq:fullmod1} with  the parameters $(\alphaS,\betaM)$.
Our inference goal remains the same: that is, to test the underlying true coefficients $(\alpha_{\S}, \beta_{\M})$.  Technically, we consider a $\sqrt{n}$-vicinity of local neighboring values $(\alpha_{\S,n}, \beta_{\M, n})$ only for the theoretical investigation of  local asymptotic behaviors.  \label{page:localmodel1} 
Such an idea has also been used for studying  
 other non-regularity issues  \citep[][etc.]{mckeague2015adaptive,wang2018testing}. 
% Nevertheless, we point out that this work aims to address different challenges and 
% makes novel contributions in several major aspects; 
% we defer the detailed discussion to Remark \ref{rm:compare1} at the end of this section. 
To examine  the limit  of $\hatalphaSn \hatbetaMn - \alphaSn \betaMn$ under \eqref{eq:fullmodlocal2}, we assume the following general regularity  Condition \ref{cond:inversemoment}.

\begin{condition}\label{cond:inversemoment} 
(C1.1) $\mathrm{E}(\eM| \X, \S )=0$ and $\mathrm{E}(\eY| \X, \S, \M )=0$. 
(C1.2) $\mathrm{E}(\boldsymbol{D}\boldsymbol{D}^{\mytrans})$ 
is a positive definite matrix with bounded eigenvalues, 
where $\boldsymbol{D} = (\X^{\mytrans}, \M, \S)^{\mytrans}$.
(C1.3) The second moments of $(\eM, \eY, \Sperp, \Mperpp, \eM\Sperp, \eY\Mperpp)$ are finite,
where  $\Sperp = \S - \X^{\mytrans}Q_{1,\S}$  with  $Q_{1,\S}= \{\mathrm{E}(\X \X^{\mytrans})\}^{-1}\times \mathrm{E}(\X \S)$, and $\Mperpp =\M - \tilde{\X}^{\mytrans} Q_{2,\M}$   with  $\tilde{\X}=(\X^{\mytrans},\S)^{\mytrans}$ and $Q_{2,\M} = \{\mathrm{E}(\tilde{\X}^{\mytrans}\tilde{\X})\}^{-1}\times \mathrm{E}(\tilde{\X}\M)$.   
\end{condition}

Similarly to our above discussions under the simplified model,  Theorem \ref{thm:prodlimit2} establishes the limits of 
 $\sqrt{n}\times (\hatalphaSn \hatbetaMn - \alphaS \betaM)$  and $n\times (\hatalphaSn \hatbetaMn - \alphaS \betaM)$ 
 when $(\alphaS, \betaM) \neq (0,0)$ and $(\alphaS, \betaM) = (0,0)$, respectively.
\begin{theorem}[Asymptotic Property] \label{thm:prodlimit2}
Assume Condition \ref{cond:inversemoment}. 
Under the local model \eqref{eq:fullmodlocal2}, \vspace{-0.5em}
\begin{enumerate}
\item[(i)] when $(\alphaS, \betaM) \neq (0,0)$, 
$
\sqrt{n}\times (\hatalphaSn \hatbetaMn - \alphaSn \betaMn) \xrightarrow{d} \alphaS {Z_{\M}} + \betaM {Z_{\S}};	
$ 
\item[(ii)] when $(\alphaS, \betaM) = (0,0)$, 
$
n\times (\hatalphaSn \hatbetaMn - \alphaSn \betaMn) \xrightarrow{d} \balpha {Z_{\M}} + \bbeta {Z_{\S}} + {Z_{\M}Z_{\S}},
$
\end{enumerate}
where  $(Z_S, Z_M)^{\mytrans}$ is a mean-zero normal random vector with a covariance matrix given by that of the random vector $(\eM \Sperp/V_S, \eY \Mperpp/V_M)^{\mytrans}$ with $V_S=\mathrm{E}(\Sperp^2)$, and $V_M=\mathrm{E}(\Mperpp^2)$.  
\end{theorem}

Theorem \ref{thm:prodlimit2} suggests the 
limit 
of $\sqrt{n}(\hatalphaSn\hatbetaMn - \alphaSn \betaMn )$  is not uniform with respect to $(\alphaS , \betaM)$, and the non-uniformity occurs around $(\alphaS , \betaM)=(0,0)$. 
On the other hand, 
in the neighborhood of  $(\alphaS , \betaM)=(0,0)$, the limit of $n(\hatalphaSn\hatbetaMn - \alphaSn \betaMn )$ is continuous as a function of $(\balpha, \bbeta)^{\mytrans} \in \realnumber^2$ into the space of distribution functions. Therefore, using this local limit in the bootstrap, we expect better finite-sample accuracy, compared to the classical nonparametric bootstrap that does not take into account the local asymptotic behavior.

Moreover, to discern the null cases, 
we will consider a decomposition of the statistic.  
The idea is to isolate the possibility that $(\alphaS, \betaM ) \neq (0,0)$ by comparing the absolute value of the standardized statistics  $T_{\alpha,n}=\sqrt{n}\hatalphaSn/\hatsigmaalphan$ and $T_{\beta, n}=\sqrt{n}\hatbetaMn /\hatsigmabetan$ to some thresholds,
where $\hatsigmaalphan$  and $\hatsigmabetan$ denote the 
sample  standard deviations  of $\hatalphaSn$ and $\hatbetaMn$, respectively. 
Specifically, we decompose 
\begin{align}
\hatalphaSn \hatbetaMn - \alphaSn \betaMn=&~ (\hatalphaSn \hatbetaMn - \alphaSn \betaMn )\times  (1-\myindicator_{\alphaS, \lambda_n}\myindicator_{\betaM, \lambda_n} )\label{eq:proddecomp1}  \\
& +(\hatalphaSn \hatbetaMn - \alphaSn \betaMn)\times  \myindicator_{\alphaS, \lambda_n}\myindicator_{\betaM, \lambda_n}\notag 
\end{align} 
with the indicators 
$\myindicator_{\alphaS, \lambda_n}=\myindicator \{ |T_{\alpha,n}| \leq \lambda_{n}, \alphaS = 0\}$ and $\myindicator_{\betaM, \lambda_n}=\myindicator\{|T_{\beta,n}| \leq \lambda_{n}, \betaM = 0\}$,  
where $\myindicator \{E\}$ represents the indicator function of an event $E$, and $\lambda_n$ is a certain threshold to be specified. 
When $(\alphaS, \betaM) \neq (0,0)$, the classical  bootstrap is consistent for the first term in \eqref{eq:proddecomp1}.
For the second term in \eqref{eq:proddecomp1}, 
we next develop a bootstrap statistic motivated by  Theorem \ref{thm:prodlimit2} (ii).

To construct the bootstrap statistic, 
we introduce some notation following the convention in the empirical process literature \citep{van2000asymptotic}. 
Particularly, throughout the paper, 
$P_n$ denotes the population probability measure of $(\S, \X, \M, \Y)$, 
$\mathbb{P}_n$ denotes the empirical measure with respect to the i.i.d. observations $\{(\S_i, \X_i, \M_i, \Y_i): i=1,\ldots, n\}$, and $\mathbb{P}_n^*$ denotes the nonparametric bootstrap version of $\mathbb{P}_n$.  
For any measurable function $f(\S, \X, \M, \Y)$, we define the empirical process $\mathbb{G}_n f = \sqrt{n}(\mathbb{P}_n -P_n)f=\sqrt{n}[n^{-1}\sum_{i=1}^n f(\S_i, \X_i, \M_i, \Y_i) -\mathrm{E}\{f(\S, \X, \M, \Y)\} ]$, and its nonparametric bootstrap version is $\mathbb{G}_n^*=\sqrt{n}(\mathbb{P}_n^*- \mathbb{P}_n)$. 
With the above notation, we define the sample versions of 
$Q_{1,\S}, Q_{2,\M}, \Sperp,$ and $\Mperpp$ in Condition \ref{cond:inversemoment} as
$\hat{Q}_{1,\S}= \{\mathbb{P}_n(\X \X^{\mytrans})\}^{-1} \mathbb{P}_n(\X \S)$, 
$\hat{Q}_{2,\M} = \{\mathbb{P}_n( \tilde{\X}  \tilde{\X}^{\mytrans})\}^{-1} \mathbb{P}_n( \tilde{\X} \M)$,
$\hatSperp = \S - \X^{\mytrans}\hat{Q}_{1,\S}$, and $\hatMperpp = \M - \tilde{\X}^{\mytrans}\hat{Q}_{2,\M},$
respectively, 
where we use ``$\hat{\quad}$'' to denote the sample counterparts in this paper. 
Similarly, we define their nonparametric bootstrap counterparts $(Q_{1,\S}^*, Q_{2,\M}^*, \Sperpboot, \Mperppboot )$
by replacing $\mathbb{P}_n$ with $\mathbb{P}_n^*$ in the above definitions.

When $(\alphaS, \betaM)=(0,0)$, motivated by Theorem \ref{thm:prodlimit2} (ii), 
we construct a bootstrap statistic $\mathbb{R}_{n}^*(\balpha, \bbeta)$ as a bootstrap counterpart of of $\balpha Z_M + \bbeta Z_S+ Z_MZ_S$. 
In particular, $\mathbb{R}_{n}^*(\balpha, \bbeta) = {\balpha \mathbb{Z}_{\M,n}^*}  +{\bbeta\mathbb{Z}_{\S,n}^*} + {\mathbb{Z}_{\S,n}^*\mathbb{Z}_{\M,n}^* },$ 
where 
 $\mathbb{Z}_{\S,n}^*= \mathbb{G}_n^*(\hatenM \Sperp^*)/\mathbb{V}_{\S, n}^*$, $\mathbb{Z}_{\M,n}^*=\mathbb{G}_n^*(\hatenY \Mperpp^*)/\mathbb{V}_{\M, n}^*,$ $\mathbb{V}_{\S, n}^* = \mathbb{P}_n^*\{(\Sperp^*)^2\},$ $\mathbb{V}_{\M, n}^* = \mathbb{P}_n^*\{(\Mperpp^*)^2\}$, 
 and 
$\hatenM$ and $\hatenY$ denote the sample residuals obtained from the ordinary least squares regressions of the two models in \eqref{eq:fullmodlocal2}.   
When $(\alphaS, \betaM)\neq (0,0)$, we still consider 
the classical nonparametric bootstrap estimator $\hatalphaSn^* \hatbetaMn^*$. 
 To develop an adaptive bootstrap test, we utilize the decomposition \eqref{eq:proddecomp1} and 
 propose to replace the indicators $\myindicator_{\alphaS, \lambda_n}$ and $\myindicator_{\betaM, \lambda_n}$ in \eqref{eq:proddecomp1} by
\begin{align}
	\myindicator_{\alphaS, \lambda_{n}}^*  = \myindicator \{ |T_{\alpha,n}^*|\leq \lambda_{n}, ~|T_{\alpha,n}|\leq \lambda_{n} \} \quad \text{and} \quad  	\myindicator_{\betaM, \lambda_{n}}^* = \myindicator \{ |T_{\beta,n}^*|\leq \lambda_{n},~ |T_{\beta,n}|\leq \lambda_{n} \}, \label{eq:indicatoralphabetasep}
\end{align} 
 where $T_{\alpha,n}^*=\sqrt{n}\hatalphaSn^*/\hatsigmaalphan^*$ and $T_{\beta, n}^*=\sqrt{n}\hatbetaMn^* /\hatsigmabetan^*$ 
denote the classical nonparametric bootstrap versions of  $T_{\alpha, n}$ and $T_{\beta, n}$, respectively. 
Following the decomposition in \eqref{eq:proddecomp1}, we define a statistic 
\begin{align*}
U^*=(\hatalphaSn^* \hatbetaMn^*  - \hatalphaSn \hatbetaMn ) \times (1- \myindicator_{\alphaS, \lambda_n}^* \myindicator_{\betaM, \lambda_n}^*  ) + n^{-1} \mathbb{R}_{n}^*( \balpha, \bbeta ) \times \myindicator_{\alphaS, \lambda_n}^*  \myindicator_{\betaM, \lambda_n}^*,
\end{align*}
termed as Adaptive Bootstrap (AB) test statistic in this paper.  
  Theorem \ref{thm:bootstrapprodcomb} below establishes the bootstrap consistency of $U^*$.

\begin{theorem}[Adaptive Bootstrap Consistency] \label{thm:bootstrapprodcomb}
Assume the conditions of Theorem \ref{thm:prodlimit2} are satisfied. When  $\lambda_n=o(\sqrt{n})$ and $\lambda_n\to \infty$ as $n\to \infty$, 
$$c_nU^* \overset{d^*}{\leadsto} c_n(\hatalphaSn \hatbetaMn - \alphaS \betaM),$$
where $c_n$ is a non-random scaling factor satisfying
\begin{align}\label{eq:cndef} 
	c_n=	\begin{cases}
	\sqrt{n}, & \text{ when }(\alphaS, \betaM) \neq (0,0) \\
	n, &  \text{ when } (\alphaS, \betaM) = (0,0)
	\end{cases}. 
\end{align}
\end{theorem}

\smallskip
Theorem \ref{thm:bootstrapprodcomb} suggests that
under the original model \eqref{eq:fullmod1}, i.e., $(\balpha, \bbeta)=(0,0)$,
the AB statistic $U^*$ is a consistent bootstrap estimator for $\hatalphaSn \hatbetaMn - \alphaS \betaM$ with a proper scaling. 
Moreover, 
for any fixed targeted parameters $(\alphaS, \betaM)$, in their local neighborhoods, 
i.e.,  $(\balpha, \bbeta)\neq (0,0)$, the bootstrap consistency still holds as a smooth function of  $(\balpha, \bbeta)$.  
Intuitively, this suggests that  a small change in the target parameters does not affect the consistency property, and $U^*$ is ``regular'' under the local model.   
% This view of ``regularity'' or ``stability'' under local models shares a similar spirit of arguments on local experiments in   \cite{van2000asymptotic}. 
 \label{page:localmodel2} 
In practice,  without knowing which case is the true null we rely on $U^*$ as the bootstrap statistic for $\hatalphaSn \hatbetaMn - \alphaSn \betaMn$ generally.   
This strategy is viable because with a given finite  sample size $n$, 
using $\sqrt{n}U^*$ for bootstrapping  $\sqrt{n}(\hatalphaSn \hatbetaMn - \alphaSn \betaMn)$ is  equivalent to using $nU^*$ for bootstrapping  $n(\hatalphaSn \hatbetaMn - \alphaSn \betaMn)$. 
Therefore, as desired, $U^*$ will approximate well the distribution of $\hatalphaSn \hatbetaMn - \alphaSn \betaMn$ regardless of the underlying null case.  
\label{page:localmodel3}

\begin{remark}
As a comparison, we also discuss the naive non-parametric bootstrap when $(\alphaS, \betaM) = (0,0)$. Specifically, we obtain the following expression (in Remark \ref{rm:classboot00} of the Supplementary Material),
\begin{align} \label{eq:classboot00expand}
n(\hatalphaSn^* \hatbetaMn^* - \hatalphaSn \hatbetaMn)%\notag\\
=&~\mathbb{R}_{n}^*(\balpha, \bbeta)+  \mathbb{Z}_{\S,n}\mathbb{Z}_{\M,n}^* + \mathbb{Z}_{\M,n}\mathbb{Z}_{\S,n}^*, 
\end{align}
where 
  $\mathbb{Z}_{\S,n}= \mathbb{G}_n(\eM \hatSperp)/\mathbb{V}_{\S, n}$, $\mathbb{Z}_{\M,n}=\mathbb{G}_n(\eY \hatMperpp)/\mathbb{V}_{\M, n},$ $\mathbb{V}_{\S, n} = \mathbb{P}_n(\hatSperp^2),$ and $\mathbb{V}_{\M, n} = \mathbb{P}_n(\hatMperpp^2)$. 
In addition to the term $\mathbb{R}_{n}^*(\balpha, \bbeta)$, \eqref{eq:classboot00expand} has two extra random terms $\mathbb{Z}_{\S,n}\mathbb{Z}_{\M,n}^* + \mathbb{Z}_{\M,n}\mathbb{Z}_{\S,n}^*$, which suggests that using \eqref{eq:classboot00expand} in the bootstrap would not be consistent. 
% ruin the consistency property when \eqref{eq:classboot00expand} is used in bootstrap. 
The issue of the classical bootstrap being inconsistent at  $(\alphaS, \betaM) = (0,0)$ is circumvented by the proposed local bootstrap statistic $\mathbb{R}_n^*(\balpha, \bbeta )$.  
\end{remark} 

\bigskip 

\noindent \textit{Adaptive Bootstrap Test Procedure.} 
We introduce a consistent bootstrap test procedure for $\hatalphaSn \hatbetaMn$ 
based on 
Theorem \ref{thm:bootstrapprodcomb}. 
Given a nominal level $\omega$, let $q_{\omega/2}$ and $q_{1-\omega/2}$ denote the lower and upper $\omega/2$ quantiles, respectively, of the bootstrap estimates $U^*$. If $\hatalphaSn \hatbetaMn$ falls outside  the interval $(q_{\omega/2}, q_{1-\omega/2} )$, we reject the composite null \eqref{eq:compositenull}, and conclude that the mediation effect is statistically significant at the level $\omega$.  
We clarify that the goal is to test the underlying true coefficients $(\alpha_{\S}, \beta_{\M})$. The reason to consider their $\sqrt{n}$-local coefficients $(\alpha_{\S,n}, \beta_{\M, n})$ is merely for theoretical investigation of local asymptotic behaviors. Therefore, to test \eqref{eq:compositenull} under the original model \eqref{eq:fullmod1}, it suffices to calculate $U^*$ with $\balpha=\bbeta=0$.
We point out that the rejection region in the adaptive procedure may also be constructed  through the asymptotic distribution as an alternative to the bootstrap; nevertheless,
 the proposed bootstrap procedure is more flexible and does not rely on a particular form of the limiting distributions,  and therefore, it can be easily extended under various mediation models; see more examples in Section \ref{sec:extendmodel}.  
\label{page:asymvsboot}

\bigskip

% Theorem \ref{thm:bootstrapprodcomb} requires that $\lambda_n=o(\sqrt{n})$ and $\lambda_n\to \infty$ as $n\to \infty.$ 
 % When $\lambda_n$ satisfies the rate and , 

\noindent \textit{Choice of the Tuning Parameters.} 
 Under the conditions of Theorem \ref{thm:bootstrapprodcomb}, which specify   $\lambda_n=o(\sqrt{n})$ and $\lambda_n\to \infty$ as $n\to \infty,$ we have   
 $\lim_{n\to \infty} \Pr(|T_{\alpha, n}|>\lambda_n, |T_{\beta, n}|>\lambda_n \mid  \alphaS = \betaM = 0 ) = 0$, 
 suggesting that $\myindicator_{\alphaS, \lambda_n}\myindicator_{\betaM, \lambda_n}$ can provide a consistent test for $\alphaS=\betaM=0$. 
If $\lambda_n$ remains bounded as $n\to \infty$, 
$U^*$ asymptotically reduces to $\hatalphaSn^*\hatbetaMn^*-\hatalphaSn\hatbetaMn$, i.e., the classical nonparametric bootstrap procedure, which  may be conservative.  
In the simulation experiments, we set $\lambda_n=\lambda\sqrt{n}/\log n$ and find that a fixed constant $\lambda$, e.g.,  $\lambda=2$ can give a good performance.
In general settings, 
we can choose the tuning parameter through the double bootstrap \citep{chen2016peter}; see Section \ref{sec:tuningpar} of the Supplementary Material for more implementation details. 
\label{page:doublebootdesc}

\begin{remark}\label{rm:compare1} 
Our proposed adaptive procedures examine the non-regular asymptotic behaviors of test statistics 
through local models. 
In effect, the idea of local models may be traced back to  econometrics \citep{andrews2001testing} and was utilized in other statistical problems, such as classification and post-selection inference 
\citep{laber2011adaptive,mckeague2015adaptive,mckeague2019marginal,wang2018testing}. 
Nevertheless, 
we emphasize that there are unique statistical challenges of      
testing mediation effects. 
First, in terms of the parameter space, 
the null hypothesis of no mediation effect is essentially a \textit{union} of individual hypotheses. 
This results in a  non-standard shape of the null parameter space, on which both regular and non-regular asymptotic behaviors can occur, as illustrated in  Figure  \ref{fig:visualofcompositenull}. 
Second, in terms of the behavior of the estimator,
we unveil a fundamental zero-gradient phenomenon.   
This is caused by the special form of the product statistic and cannot be directly addressed by  the existing adaptive procedures mentioned above.  
Third, in terms of the models, 
the mediation analysis involves a system of  structural equations. 
Ignoring the model structure in the implementation could lead to slow computation;  see Section \ref{sec:suppcomputation} of the Supplementary Material for more  details on computation.  
Due to these unique challenges, new developments  in methodology, theory, and computation are necessary.      
% This work addresses those  statistical  challenges and complements the current literature on both mediation analysis and  testing non-regular hypotheses.
\end{remark}

\noindent \textsf{Adaptive Bootstrap for the Joint Significance Test.} \label{page:adjst}
In addition to the Wald-type PoC test, 
we also address the non-regularity issue of the non-Wald joint significance/maxP test through our proposed adaptive bootstrap. 
It is noteworthy that non-regular behaviors of the JS and PoC tests under the singleton  $H_{0,3}$ are distinct, as the two statistics take different forms. 
Particularly, PoC statistic has the zero-gradient issue discussed above, whereas JS statistic has a certain inconsistent convergence issue. 
Despite that difference, we can similarly develop an adaptive bootstrap for  the JS test and obtain \textit{uniformly distributed} $p$-values.   
Refer to the detail in Section  \ref{sec:jst} of  the Supplementary Material.  
This suggests that our proposed adaptive bootstrap is not restricted to the Wald-type test and may be further generalized to other tests with similar circumstances.   

\bigskip 
\noindent \textsf{On Multivariate Mediators.} \label{page:discussionmultimed} 
% So far, we have developed the AB under the model \eqref{eq:fullmod1}  with a single mediator, illustrated in Figure \ref{fig:model}.  
It is worth noting that the proposed strategy can be generalized to deal with multiple   mediators under suitable identifiability conditions.
 In the following, 
we delve into three scenarios of practical importance.   

\medskip
\textit{(i)} We consider the group-level joint Mediation Effect (ME) via a set of mediators $\boldsymbol{M}=(M_1,\ldots, M_J)$ shown by the red path in Figure \ref{fig:gourpeffect} below. 
This type of joint ME has been considered in the literature by \cite{huang2016hypothesis} and \cite{hao2022simultaneous}, among others. 
We  generalize the AB method to test the joint mediation  effect in Section \ref{sec:jointestmulti}.   

\medskip
\textit{(ii)} We consider multiple mediators that are causally uncorrelated \citep{jerolon2020causal} or governed by  the parallel path model \citep{hayes2017introduction}.  
In this case, 
the indirect effect of one single mediator 
can be identified under the known identifiability assumptions outlined in \cite{imai2013identification}. In particular, under  
the multivariate  linear SEM \eqref{eq:fullmodmult1} with no
causal interplay between mediators,  
the null hypothesis of  
no individual indirect effect via one mediator, \textit{say}, $M_1$, could be formed as 
$H_0:\alpha_{S,1}\beta_{M,1}=0$, illustrated in Figure \ref{fig:singleeffect} below. 
To apply the AB test to $\alpha_{S,1}\beta_{M,1}$,
we note that \eqref{eq:fullmodmult1} can be equivalently rewritten as 
$M_1=\alpha_{S,1}S+{\boldsymbol{X}}^{\top}\boldsymbol{\alpha}_{{\boldsymbol{X}},1}+\epsilon_{M,1},
 $ and 
$
Y=\beta_{M,1}M_1+ \boldsymbol{M}_{(-1)}\boldsymbol{\beta}_{(-1)} + \X \boldsymbol{\beta}_{\X} +\tau_S S+\epsilon_Y,$ 
where  $\boldsymbol{\beta}_{(-1)}=(\beta_{M,2},\ldots, \beta_{M,J})^{\top}$ and $ \boldsymbol{M}_{(-1)}=(M_2,\ldots, M_J)^{\top}$.
This form resembles \eqref{eq:fullmod1}, and  
 the AB method in Section \ref{sec:abt} can be employed to test $\alpha_{S,1}\beta_{M,1}=0$ by adjusting $(\boldsymbol{M}_{(-1)},\X)$ in the outcome model.
 We provide details including the identification assumptions in Section \ref{sec:parallelmediator} of the Supplementary Material.  

\medskip
\textit{(iii)} When the mediators are causally correlated, evaluating individual indirect effects along different posited paths requires correct specification of the mediators' causal structure  \citep{vanderweele2014effect}.  
 To relax such stringent assumptions, 
 researchers have proposed alternative methods, one of which is to examine the interventional indirect effects specific to each distinct mediator \citep{loh2021disentangling}. 
Intuitively, 
the interventional indirect effect via a target mediator $M_1$ is supposed to capture all of the exposure-outcome effects that are mediated by $M_1$ as well as any other mediators causally preceding $M_1$;  
see a diagram in Figure \ref{fig:singleeffectinter} below. 
Under a typical class of linear and additive mean models, estimators of interventional indirect effects take the same product form   of coefficients as that in the above Case (ii).   
Thus the proposed AB method in Section \ref{sec:abt} can be similarly applied with little effort. 
We provide relevant details including the definition and identification assumptions of the interventional indirect effects in Section \ref{sec:intervenindirect} of the Supplementary Material.

\begin{figure}[h]  
\centering 
\begin{subfigure}[b]{0.26\linewidth}
  \begin{tikzpicture}[scale=0.9]  
\node[shape=circle, thick] (S) at (-2.2,0) {$S$};   

\node (M_all) at (0,0) {$ 
     \boldsymbol{M}=\left\lbrack\begin{aligned}
        M_1\\
        \vdots\ \, \\
        M_J\\
      \end{aligned}
    \right\rbrack$};
\node[shape=circle, thick] (Y) at (2.2,0) {$Y$};   
   
\node[shape=circle, thick] (X) at (0,-2.1) {$\boldsymbol{X}$};   
	
\draw[-stealth, thick, red] (S.east) -- node[midway,above] {\footnotesize $\boldsymbol{\alpha}_S$} (M_all.west);
\draw[-stealth, thick, red] (M_all.east) -- node[midway,above] {\footnotesize $\boldsymbol{\beta}_M$} (Y.west);
\draw[-stealth, thick,black!20] (X.west) .. controls + (left:8mm) and + (down:20mm) ..  node[midway,above] {} (S);
\draw[-stealth, thick,black!20] (X) -- (M_all);
\draw[-stealth, thick,black!20] (X.east) .. controls + (right:8mm) and + (down:20mm) .. node[midway,above] {} (Y); 

\draw[-stealth, thick] (S) .. controls + (up:20mm) and + (up:20mm)  ..  (Y);
\end{tikzpicture}
  \caption{\centering  Joint ME via a group of mediators $(M_1,\ldots, M_J)$}  \label{fig:gourpeffect}  
\end{subfigure}
\hspace{1.9em} %
\begin{subfigure}[b]{0.26\linewidth}
    \begin{tikzpicture}[scale=0.9]   
\node[shape=circle, thick] (S) at (-2.2,0) {$S$};   

\node (M_all) at (0,0) {$M_1$};

\node[shape=circle, thick] (Y) at (2.2,0) {$Y$};   

\node[thick] (X) at (0,-2.1) {$\boldsymbol{X}$};   

\node[thick] (Mother) at (-1.2,-1.2) {$\boldsymbol{M}_{(-1)}$};

\draw[-stealth, thick, red] (S.east) -- node[midway,above] {\footnotesize $\alpha_{S,1}$} (M_all.west);

\draw[-stealth, thick, red] (M_all.east) -- node[midway,above] {\footnotesize $\beta_{M,1}$} (Y.west);
\draw[-stealth, black!20, thick] (X.west) .. controls + (left:8mm) and + (down:20mm) .. node[midway,above] {} (S);
\draw[-stealth, thick, black!20] (X.north) -- (M_all.south); 

\draw[dashed, thick, black!40] (Mother.north) -- (M_all.215); 

\draw[-stealth, thick, black!20] (X.175) -- (Mother.south); 

\draw[-stealth, thick, black!20] (X.east) .. controls + (right:8mm) and + (down:20mm)  ..  node[midway,above] {} (Y); 

\draw[-stealth, thick] (S) .. controls + (up:20mm) and + (up:20mm)  ..  (Y);
\end{tikzpicture}%
    \caption{\centering  Individual ME via one mediator $M_1$} \label{fig:singleeffect}  
  \end{subfigure}
  \hspace{1.9em} %
\begin{subfigure}[b]{0.26\linewidth}
    \begin{tikzpicture}[scale=0.9]   
\node[shape=circle, thick] (S) at (-2.2,0) {$S$};   

\node (M_all) at (0,0) {$M_1$};

\node[shape=circle, thick] (Y) at (2.2,0) {$Y$};   

\node[thick] (X) at (0,-2.1) {$\boldsymbol{X}$};   

\node[thick] (Mother) at (-1.2,-1.4) {$\boldsymbol{M}_{1,\text{pre}}$};

\draw[-stealth, thick, red] (S.east) -- node[midway,above] {} (M_all.west);

\draw[-stealth, thick, red] (M_all.east) -- node[midway,above] {} (Y.west);

\draw[-stealth, black!20, thick] (X.west) .. controls + (left:8mm) and + (down:20mm) .. node[midway,above] {} (S);
\draw[-stealth, thick, black!20] (X.north) -- (M_all.south); 

\draw[-stealth, thick, red] (S.east) -- node[midway,above] {} (Mother.102);

\draw[-stealth, thick, red] (Mother.78) -- (M_all.215); 

\draw[-stealth, thick, black!20] (X.175) -- (Mother.south); 

\draw[-stealth, thick, black!20] (X.east) .. controls + (right:8mm) and + (down:20mm)  ..  node[midway,above] {} (Y); 

\draw[-stealth, thick] (S) .. controls + (up:20mm) and + (up:20mm)  ..  (Y);
\end{tikzpicture}% 
    \caption{\centering  Interventional ME via one mediator $M_1$} \label{fig:singleeffectinter}  
  \end{subfigure}
\caption{Path diagram of the mediation model with multiple mediators: Dashed lines represent possible non-causal correlations or independence, and solid arrowed lines represent possible causal relationships.
In Panel (c), $\boldsymbol{M}_{1,\operatorname{pre}}$ represents mediators that are causally preceding $M_1$. }\label{fig:multiplemed}
\end{figure}
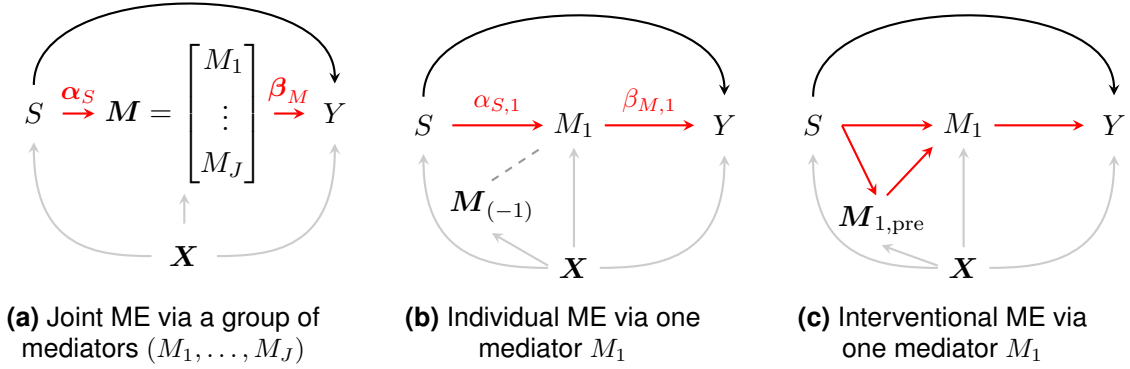

\section{Numerical Experiments} \label{sec:sim}
In this section, we conduct simulation experiments to evaluate the finite-sample performance of the proposed adaptive bootstrap PoC and JS tests. 
Particularly, we generate data through the following model:
\begin{eqnarray}
	\M &=&  \alphaS \S + \alpha_I+ \alpha_{X,1} X_1 + \alpha_{X,2} X_2  +\eM, \label{eq:simmodel}  \\ 
		\Y & = & \betaM \M +  \beta_I + \beta_{X,1} X_1 + \beta_{X,2} X_2  + \DE \S + \eY.  \notag
\end{eqnarray}
In the model \eqref{eq:simmodel}, the exposure variable $\S$ is simulated from the Bernoulli distribution with the success probability 0.5; the covariate $X_1$ is continuous and simulated from a standard normal distribution $\mathcal{N}(0,1)$; the covariate $X_2$ is discrete and  simulated from  the Bernoulli distribution with the success probability 0.5; two error terms $\eM$ and $\eY$ are simulated independently from $\mathcal{N}(0,\sigma_{\eM}^2)$ and $ \mathcal{N}(0,\sigma_{\eY}^2)$, respectively.  
We set the  parameters  $(\alpha_I, \alpha_{X,1}, \alpha_{X,2}) = (1, 1, 1)$, $(\beta_I, \beta_{X,1}, \beta_{X,2}) = (1, 1, 1)$, $\DE=1$, and $\sigma_{\eY}=\sigma_{\eM}=0.5$. 
Moreover,  we consider sample sizes $n\in \{200, 500\}$, and set the bootstrap sample size at 500.

In simulation studies, we compare eight testing methods: 
 the adaptive bootstrap for the PoC test (PoC-AB),
 the classical nonparametric bootstrap for the PoC test (PoC-B),
 Sobel's test (PoC-Sobel),
 the  adaptive bootstrap  for the JS test (JS-AB),
 the classical nonparametric bootstrap for the JS test (JS-B),
 the MaxP test (JS-MaxP), 
 the nonparametric bootstrap method in the causal mediation analysis R package  \cite{tingley2014mediation} (CMA),
 and the method in 
\cite{huang2019genome} (MT-Comp). 
It is noteworthy that \cite{huang2019genome}'s MT-Comp made specific model assumptions, 
which are not fully compatible with our simulation  settings, and we include this method just for the purpose of comparison.
Some other methods \citep[e.g.,][]{liu2020large,dai2020multiple} relied on estimating the relative proportions of the three cases, which is  not directly  applicable here  and thus not included.

\subsection{Null Hypotheses: Type \RNum{1} Error Rates} \label{sec:nulltype1error} 

\bigskip
\noindent \textit{Setting 1: Under a fixed type of null.} \quad
In the first setting, we 
simulate data under a fixed null hypothesis over 2000 Monte Carlo replications  to  estimate the distribution of $p$-values. 
Particularly, we consider 
three types of null hypotheses below:
\begin{align}
 H_{0,1}: (\alphaS, \betaM)=(0,0.5),\quad \ H_{0,2}: (\alphaS, \betaM)=(0.5,0), \quad \ H_{0,3}: (\alphaS, \betaM)=(0,0). \label{eq:simcompnull}
\end{align} 
We draw the Q-Q plots with $n=200$ in  Figure \ref{fig:fixnulln200}. 
QQ-plots under $n=500$ are similar and presented in Figure  \ref{fig:fixnulln500} of the Supplementary Material. 
% We draw the Q-Q plots with $n=200$ and $500$   in  Figures \ref{fig:fixnulln200} and  \ref{fig:fixnulln500}, respectively.
In Figure \ref{fig:fixnulln200}, three subfigures in the first row present the results of the PoC tests under three fixed  nulls $H_{0,1}, H_{0,2}$, and $H_{0,3}$, respectively, and  three subfigures in  the second row present the corresponding results of the JS tests, respectively.

 \begin{figure}[!htbp]
 \captionsetup[subfigure]{labelformat=empty} 
     \centering
       \caption{Q-Q plots of $p$-values under the fixed null with $n = 200$.} \label{fig:fixnulln200}
  \begin{subfigure}[b]{0.29\textwidth}
      \caption{\footnotesize{\quad \  $H_{0,1}: (\alphaS, \betaM)=(0, 0.5)$}}
  \includegraphics[width=\textwidth]{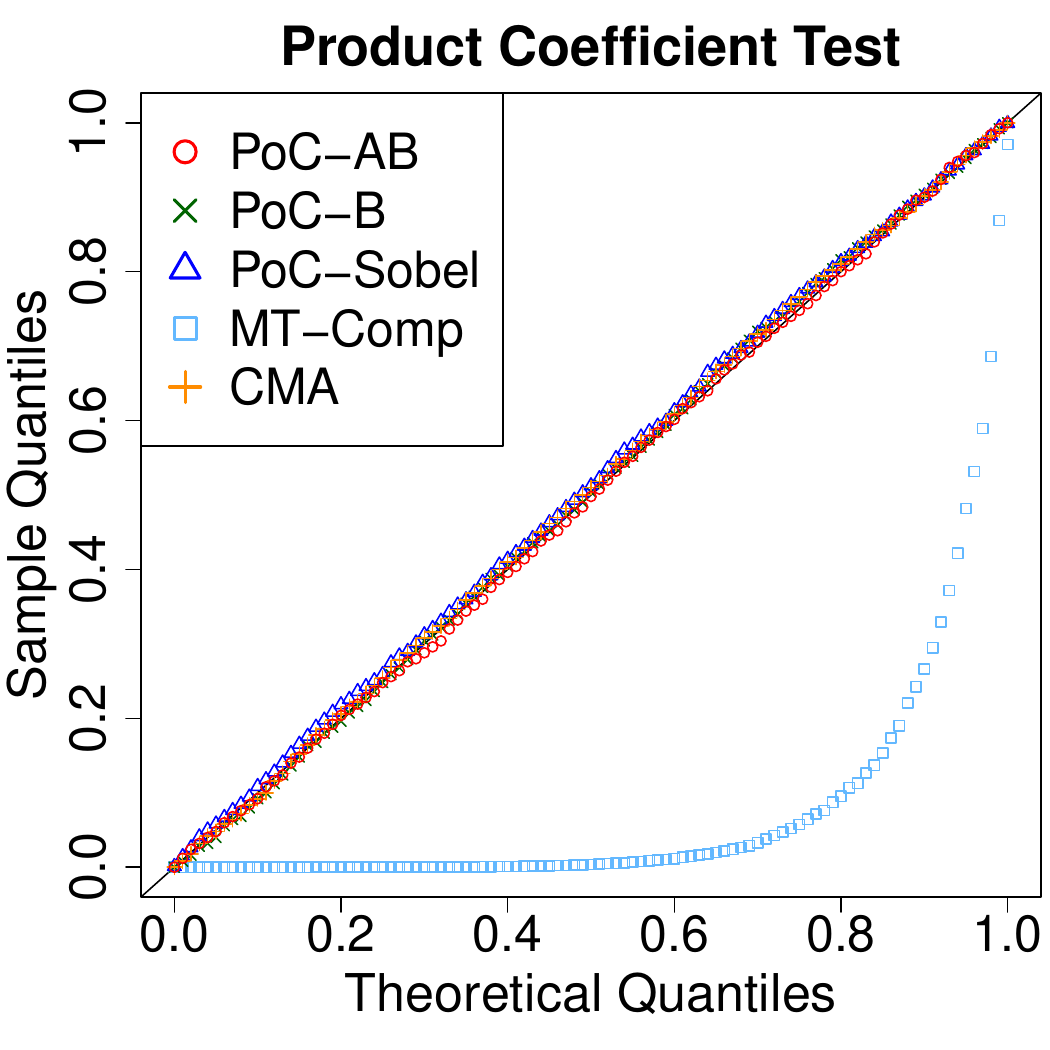}
     \end{subfigure} \  
       \begin{subfigure}[b]{0.29\textwidth}
   \caption{\footnotesize{\quad \  $H_{0,2}: (\alphaS, \betaM)=(0.5, 0)$}}
     \includegraphics[width=\textwidth]{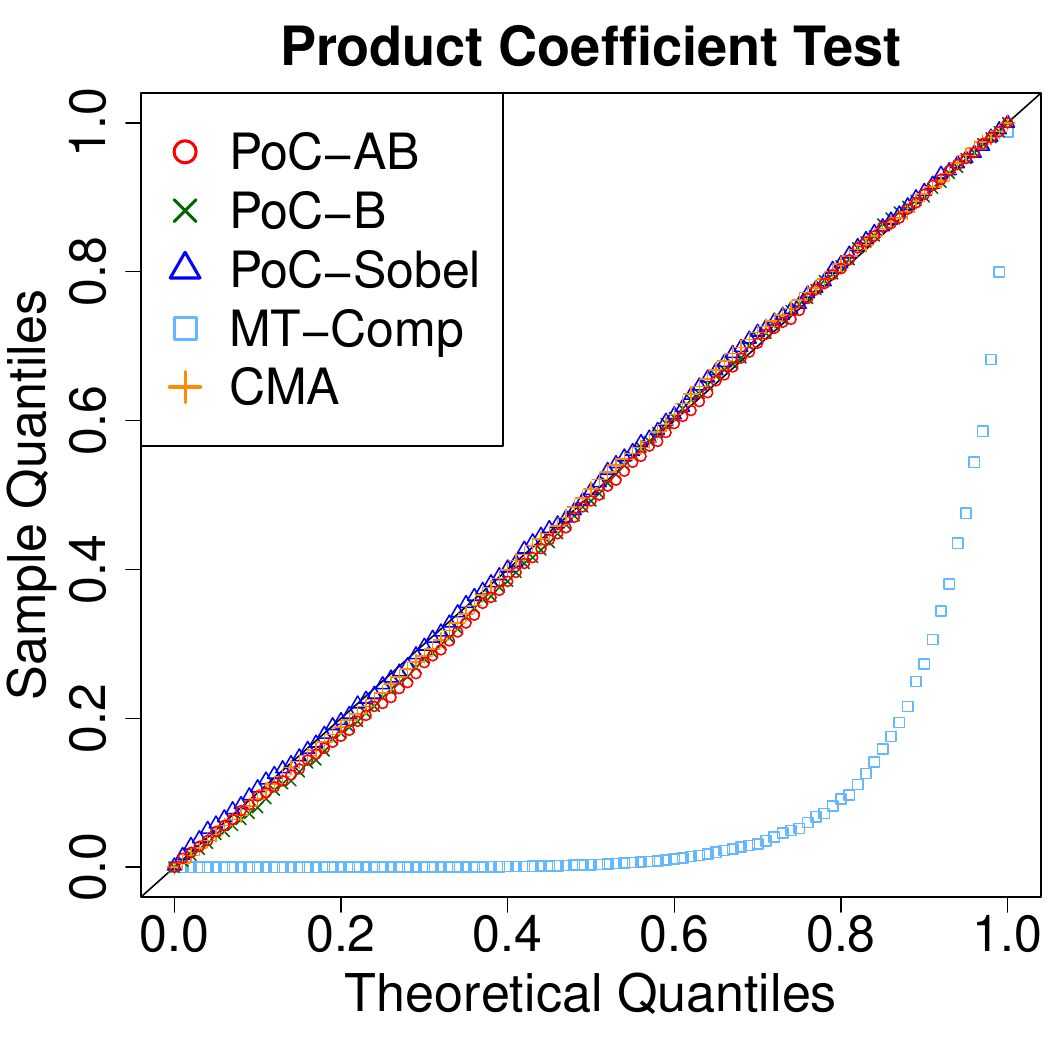}
     \end{subfigure} \
       \begin{subfigure}[b]{0.29\textwidth}
   \caption{\footnotesize{\quad \  $H_{0,3}: (\alphaS, \betaM)=(0,0)$}}
     \includegraphics[width=\textwidth]{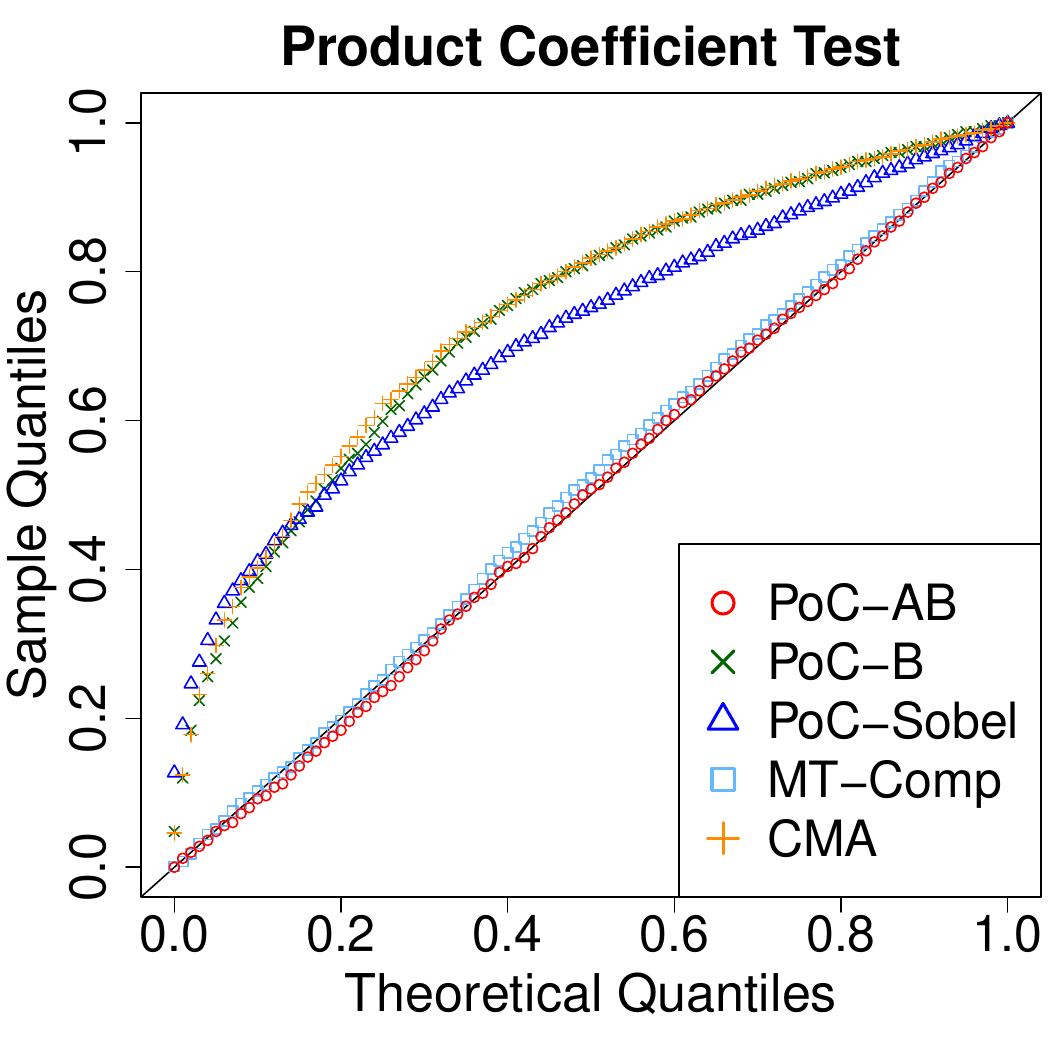}
   \end{subfigure} \\ 
       \begin{subfigure}[b]{0.29\textwidth}
             \includegraphics[width=\textwidth]{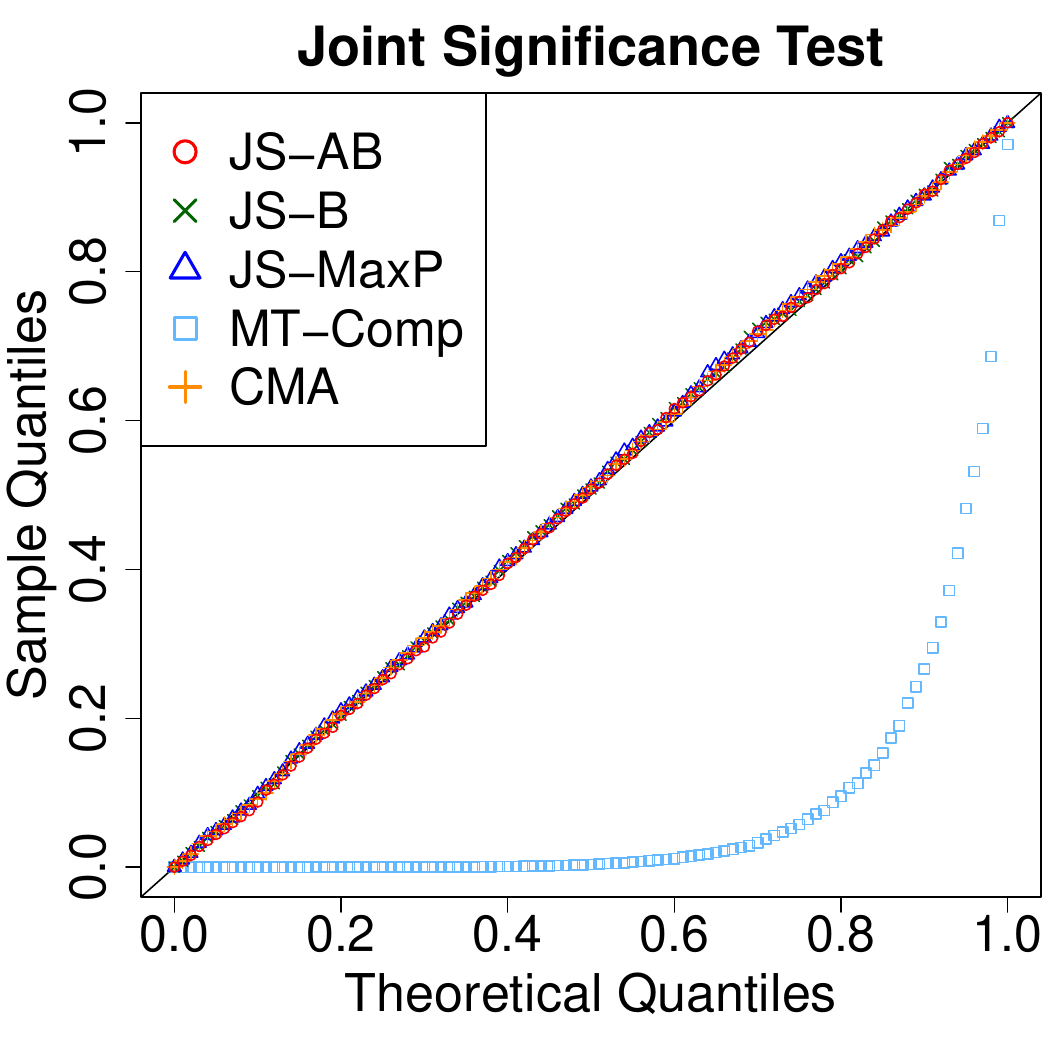}
       \end{subfigure}\
       \begin{subfigure}[b]{0.29\textwidth}
             \includegraphics[width=\textwidth]{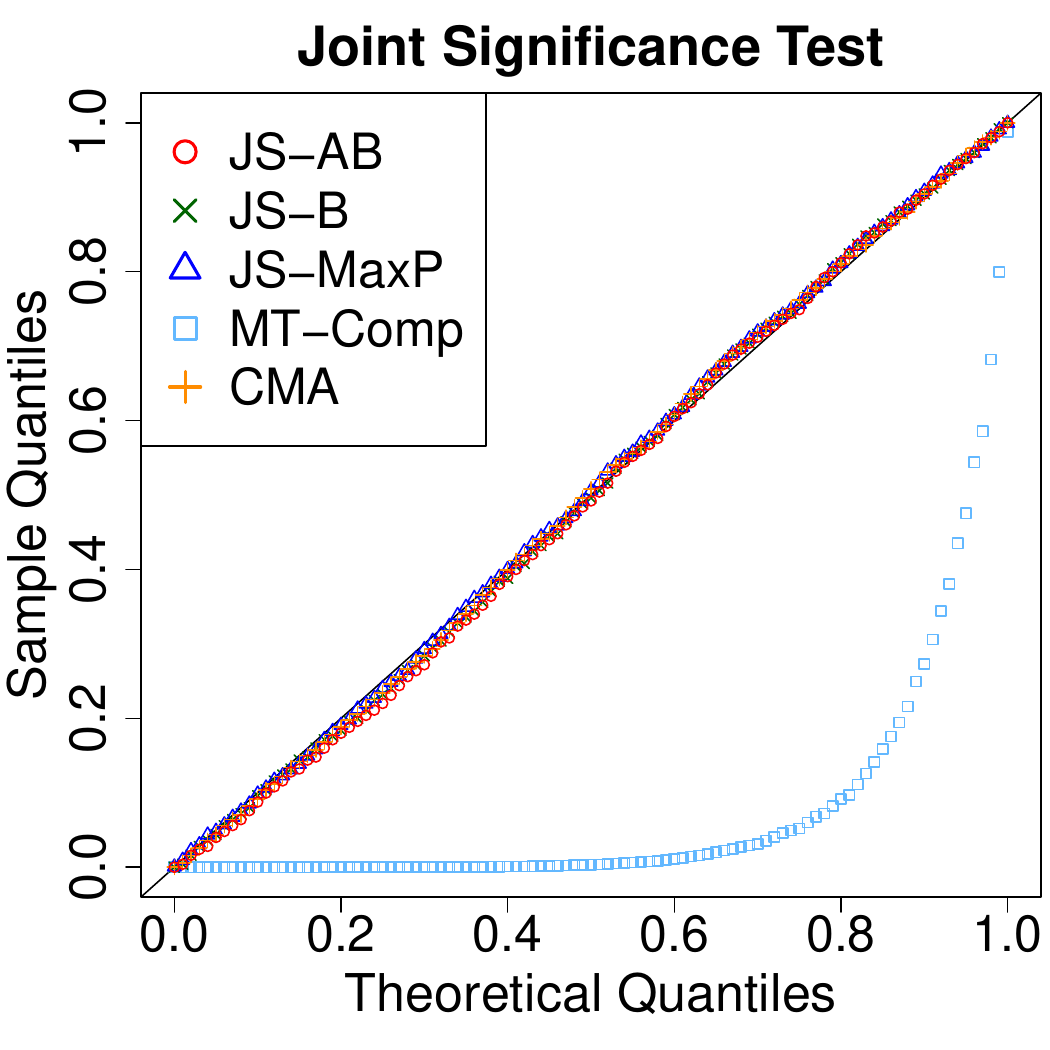}
       \end{subfigure} \ 
   \begin{subfigure}[b]{0.29\textwidth}
             \includegraphics[width=\textwidth]{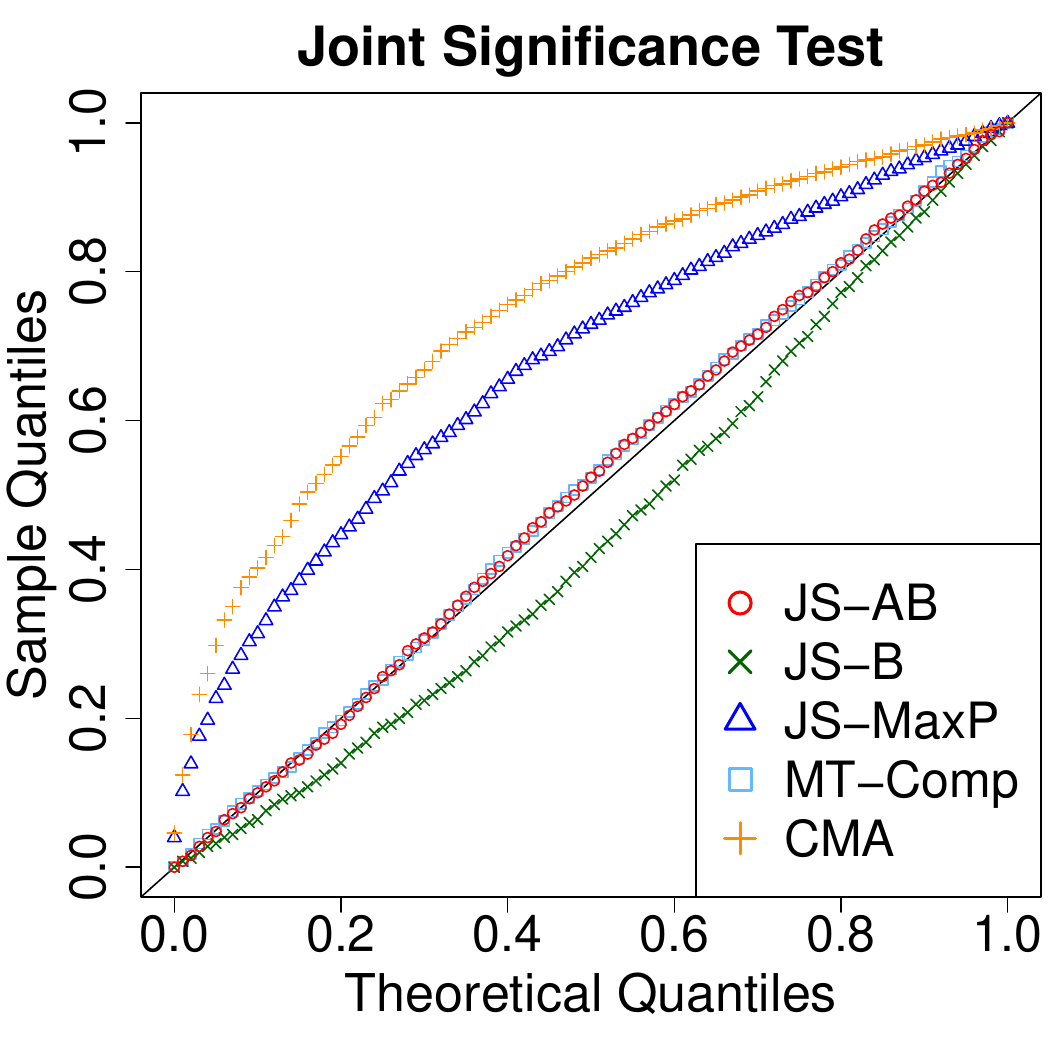}
       \end{subfigure}

 \end{figure}

Figure \ref{fig:fixnulln200}  shows that for the PoC type of tests, 
under $H_{0,1}$ or $H_{0,2}$,
the PoC-AB, the  PoC-B, and the PoC-Sobel can correctly approximate the distribution of the PoC test statistic.  
However, under $H_{0,3}$,  
the PoC-B and the PoC-Sobel  become conservative, while 
the proposed PoC-AB still approximates the distribution of the PoC statistic well.   
Similarly, for the JS type of tests,  
under $H_{0,1}$ or $H_{0,2}$,
the JS-AB, the JS-B, and the JS-MaxP all work well. 
In contrast, under $H_{0,3}$, 
the JS-B inflates,
and the JS-MaxP becomes conservative,  
while the JS-AB still exhibits a good performance. 
In addition, Figures \ref{fig:fixnulln200} and \ref{fig:fixnulln500} also display the results of both \cite{huang2019genome}'s MT-Comp and the causal mediation analysis R package CMA \citep{tingley2014mediation} for comparison.
We observe that the MT-Comp  properly controls the type \RNum{1} error under $H_{0,3}$, but fails to do so   under $H_{0,2}$ and $H_{0,3}$ with inflated type I errors.  This may be because the models considered in \cite{huang2019genome} are not compatible with our simulation settings.  
On the other hand, the causal mediation R package \citep{tingley2014mediation} produces uniformly distributed $p$-values under $H_{0,1}$ and $H_{0,2}$,
but is conservative under $H_{0,3}$.  This means that the R package CMA test is underpowered.

\medskip
\noindent \textit{Setting 2: Under a random type of null.} \quad
In the second setting, we simulate data 
over 2000 Monte Carlo replications, where  in each replication, the null hypothesis is not fixed but   randomly  selected from   $H_{0,1}$--$H_{0,3}$ in  \eqref{eq:simcompnull}. 
Specifically, for $(H_{0,1},H_{0,2}, H_{0,3})$, 
we consider 
three selection probabilities (I) $(1/3, 1/3, 1/3)$, (II) $(0.2, 0.2, 0.6)$, and (III) $(0.05, 0.05, 0.9)$,  respectively. 
We provide QQ-plots of $p$-values with $n=200$ in Figure \ref{fig:mixnulln200}, and QQ-plots under  $n=500$ are similar and  provided in Figure \ref{fig:mixnulln500} of the Supplementary Material. 
In Figure \ref{fig:mixnulln200}, three subfigures in the first row present the results of the PoC tests with three null selection  probabilities (I)--(III), respectively, 
and three subfigures in the second row present the corresponding results of the JS test, respectively.

Figures \ref{fig:mixnulln200} shows that  
the adaptive bootstrap procedures for the PoC and JS tests perform well under different settings. 
The  PoC-B test, PoC-Sobel's test, the JS-MaxP test, and the R package CMA \citep{tingley2014mediation} are conservative, and they become more conservative as the probability of choosing $H_{0,3}$ increases.  
We mention that in many biological studies such as genomics, $H_{0,3}$  predominates the null cases,
hence these  tests that are conservative under $H_{0,3}$ may not be preferred.  
Moreover, the JS-B test and the MT-Comp method can have inflated type \RNum{1} errors.  The performance of JS-B becomes worse as the proportion of $H_{0,3}$ rises, while  the MT-Comp method deteriorates as the proportions of $H_{0,1}$ and $H_{0,2}$ increase. 

 \begin{figure}[!htbp]
 \captionsetup[subfigure]{labelformat=empty}
     \centering
       \caption{Q-Q plots of $p$-values under the mixture of  nulls: $n = 200$.} \label{fig:mixnulln200}
   \begin{subfigure}[b]{0.29\textwidth}
   \caption{\footnotesize{\quad \  (I) $(1/3, 1/3, 1/3)$}}
 %    \caption{\small{\ \ $(\alphaS, \betaM)=(0,0)$}}
     \includegraphics[width=\textwidth]{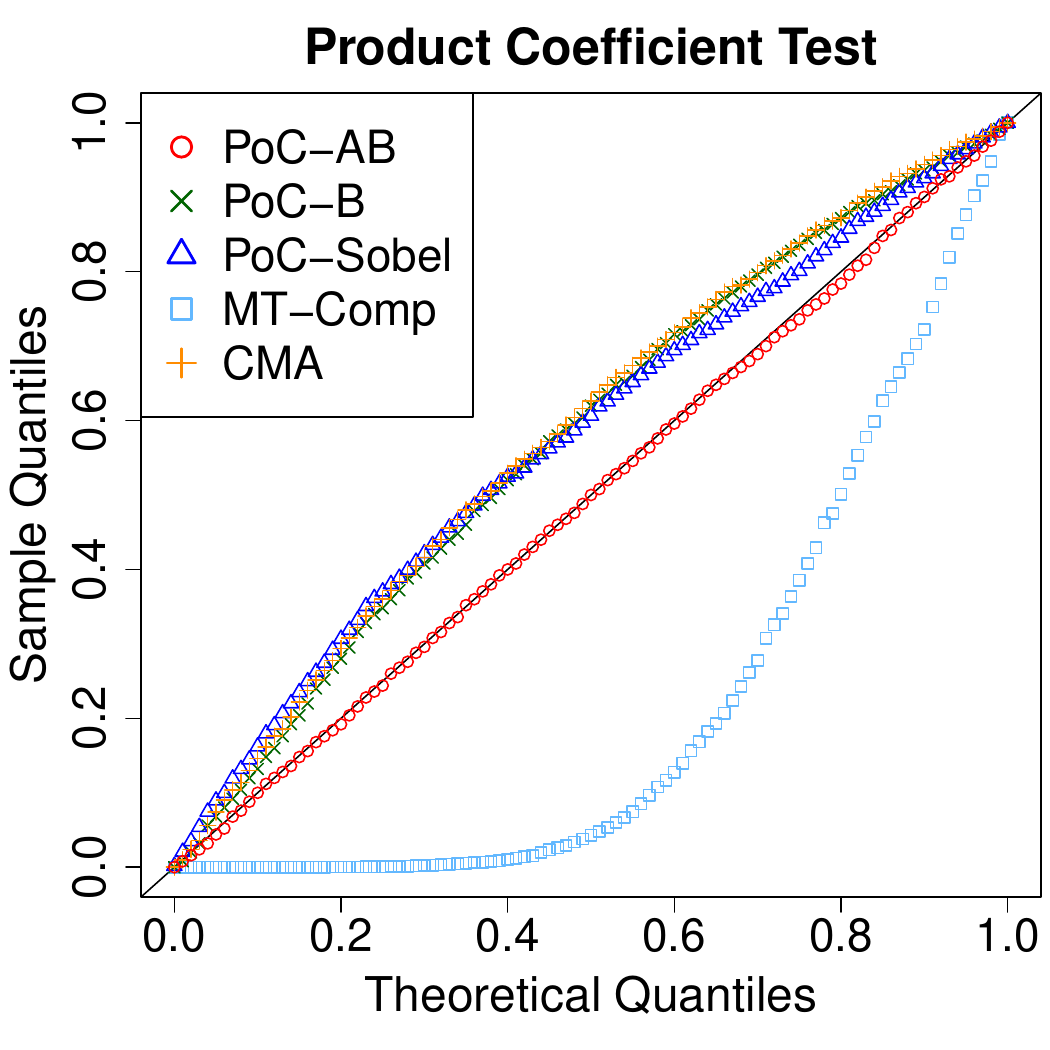}
   \end{subfigure} \    
  \begin{subfigure}[b]{0.29\textwidth}
      \caption{\footnotesize{\quad \  (II) $(0.2, 0.2, 0.6)$}}
 %   \caption{\small{\ \  $(\alphaS, \betaM)=(0,0.5)$}}
  \includegraphics[width=\textwidth]{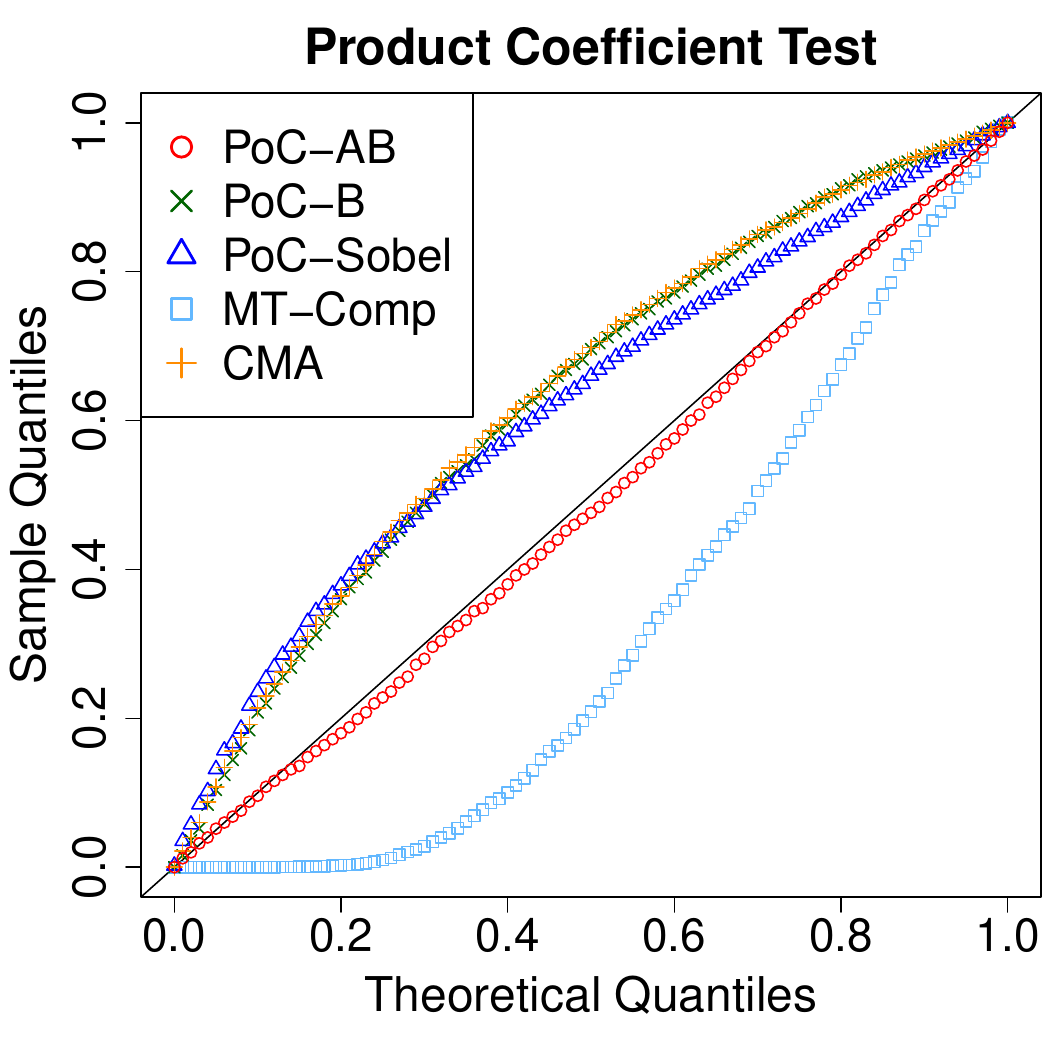}
     \end{subfigure} \  
       \begin{subfigure}[b]{0.29\textwidth}
   \caption{\footnotesize{\quad \  (III) $(0.05, 0.05, 0.9)$}}
 %      \caption{\small{\ \ $(\alphaS, \betaM)=(0.5,0)$}}
     \includegraphics[width=\textwidth]{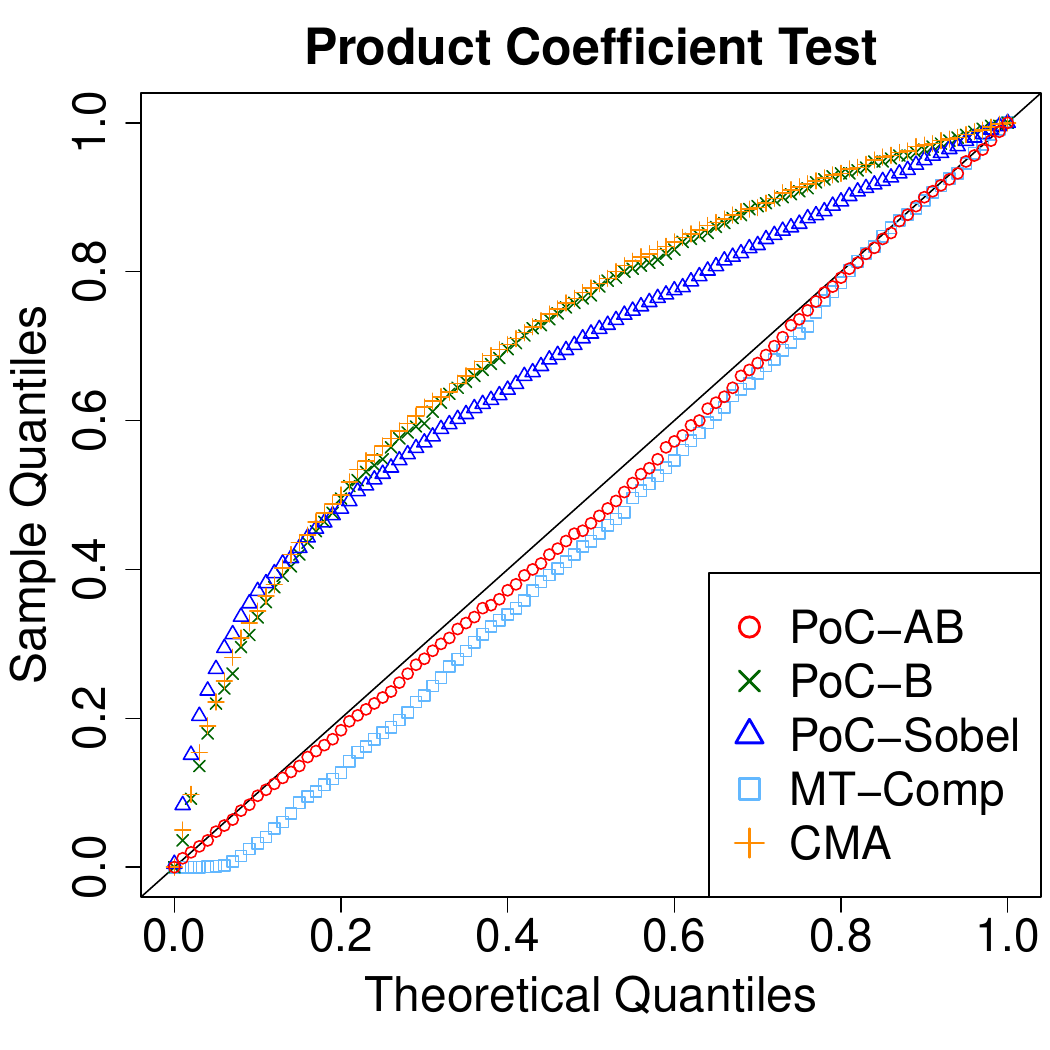}
     \end{subfigure} \\ 
       \begin{subfigure}[b]{0.29\textwidth}
             \includegraphics[width=\textwidth]{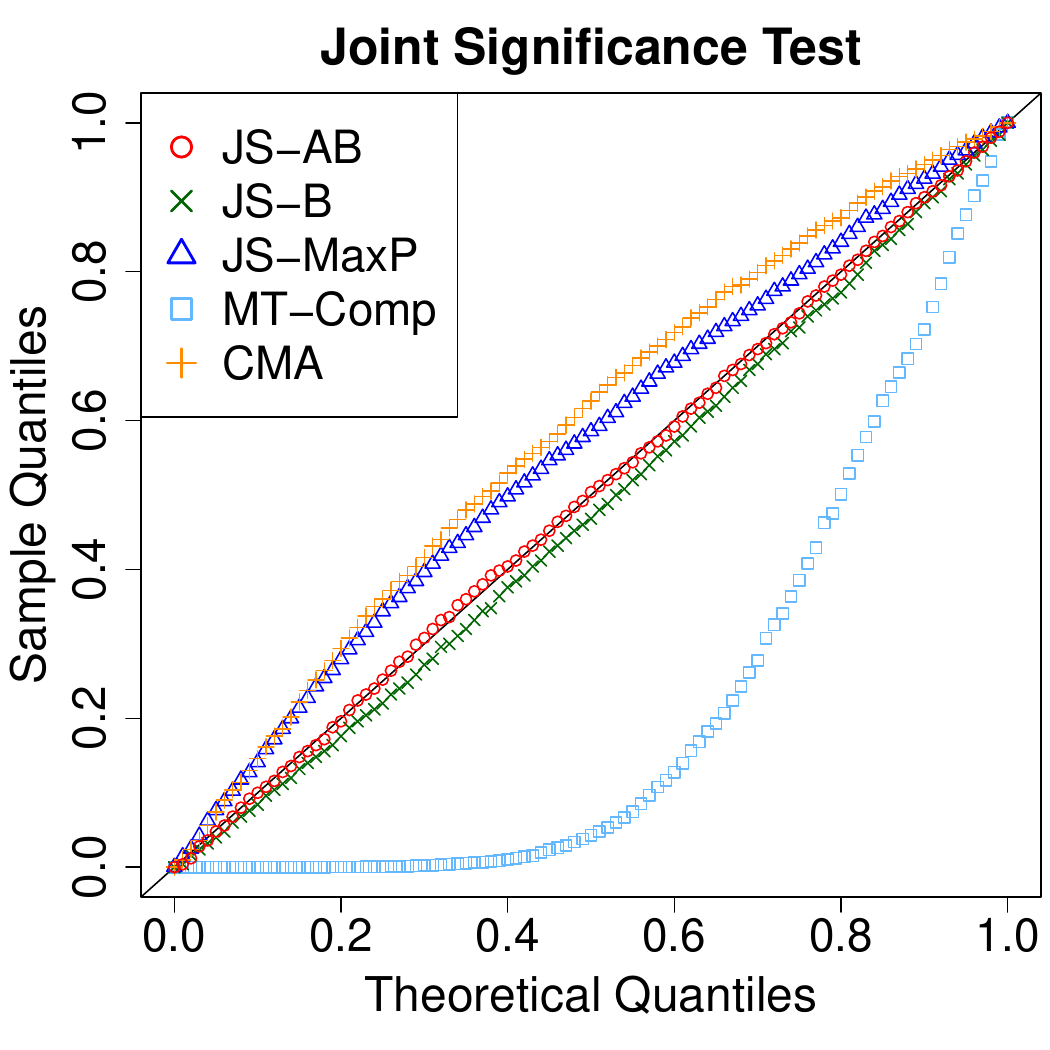}
       \end{subfigure} \  
       \begin{subfigure}[b]{0.29\textwidth}
             \includegraphics[width=\textwidth]{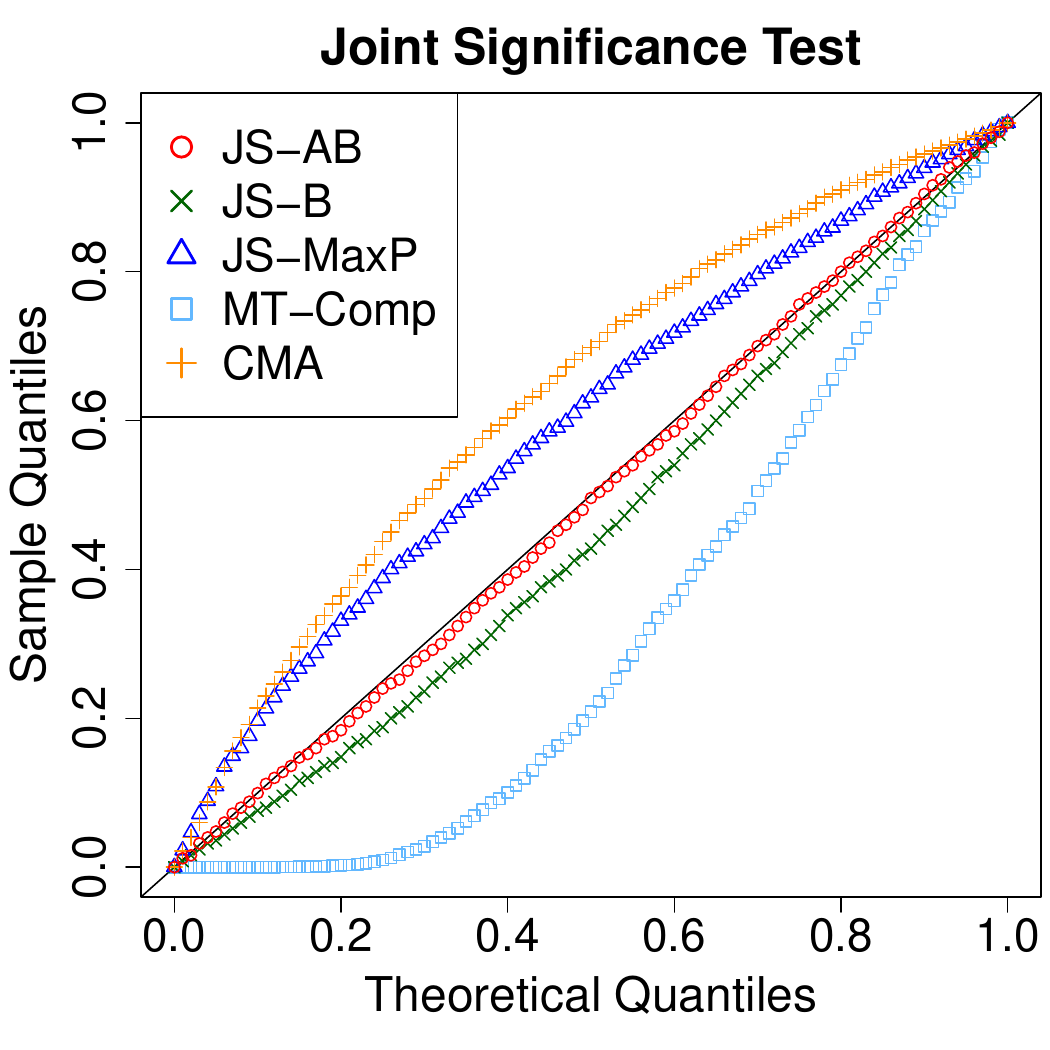}
       \end{subfigure}\
       \begin{subfigure}[b]{0.29\textwidth}
             \includegraphics[width=\textwidth]{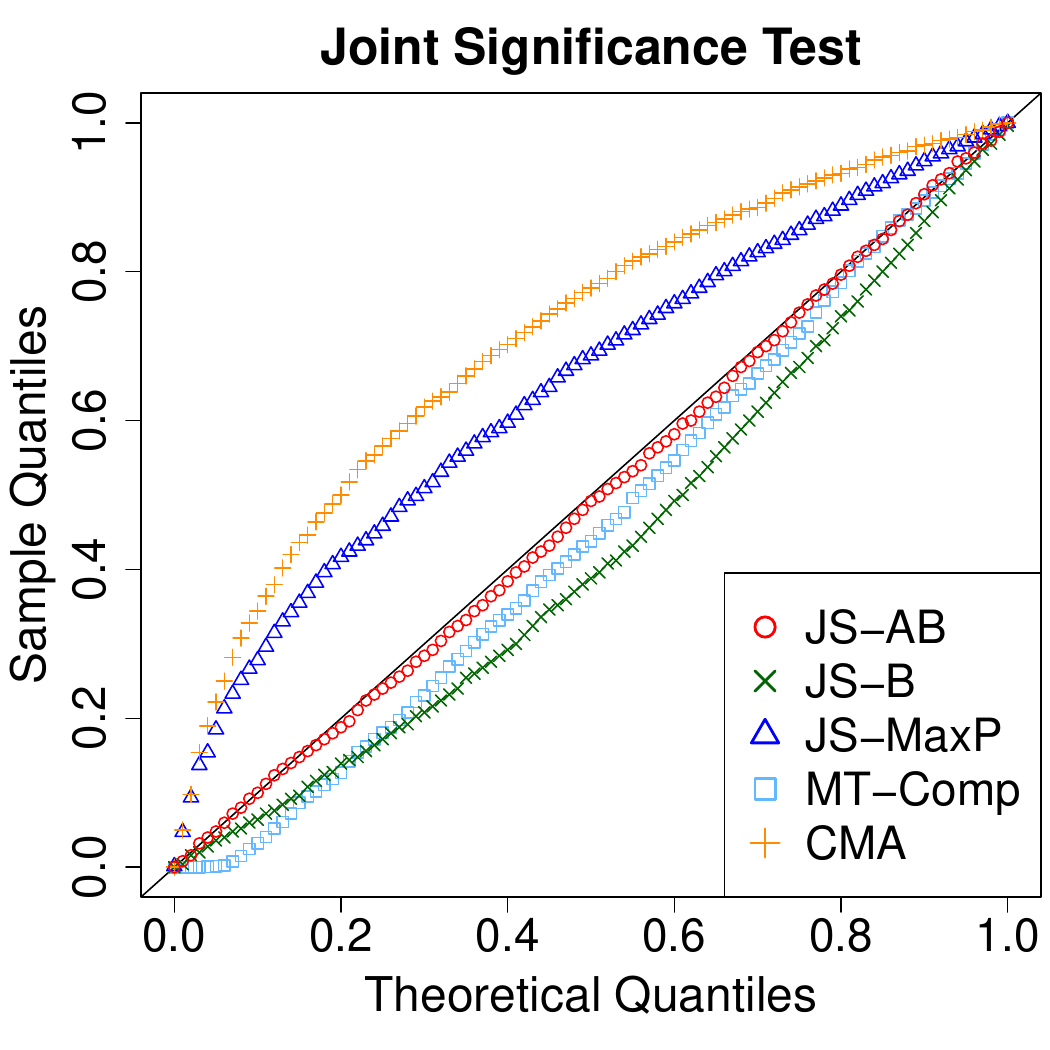}
       \end{subfigure} 
 \end{figure}

\subsection{Alternative Hypotheses: Statistical  Power}\label{sec:alterpower}

In this subsection, we evaluate the statistical  power  of the proposed AB tests under alternative hypotheses. 
Particularly, we simulate data under two settings: 
(I) fix $\alphaS=\betaM$ for the convenience of pictorial presentation, which   takes various values  beginning from zero; 
(II) fix the size of the mediation effect $\alphaS \betaM$,  and vary the ratio $\alphaS/ \betaM$. 
In the setting (I), we consider $n \in \{200, 500\}$, and then plot the empirical  rejection rates, based on 500 Monte Carlo replications, versus the signal size of $\alphaS$, which is equal to $ \betaM$ in the setting (I).
In the setting (II), we fix  $\alphaS \betaM=0.04$ when $n=200$, and $\alphaS \betaM=0.015$ when $n=500$.
Then we plot the empirical  rejection rates versus the ratio $\alphaS/\betaM$. 
The results in the two settings (I) and (II) are shown in Figures  \ref{fig:powerequal} and  \ref{fig:powerratio}, respectively.

Figures \ref{fig:powerequal} and  \ref{fig:powerratio} show that for the three PoC tests,  the PoC-AB has higher power than that of the classical nonparametric bootstrap, and both are more powerful than the Sobel's test. 
Similarly, for the JS tests, 
the JS-AB has higher power than that of the  classical bootstrap, and both have higher power than the MaxP test. In addition, the JS-B test   has slightly inflated type \RNum{1}  errors when $(\alphaS,  \betaM)=(0,0)$, which is consistent with the results in Figure \ref{fig:fixnulln200}.
Among the three classical methods 
(Sobel's test, the MaxP test, and the PoC-B), the MaxP test seems to achieve the best balance between the type \RNum{1} error and the statistical power, while Sobel's test has the lowest power.   
 These findings are consistent with those reported in the current literature \citep{mackinnon2002comparison,barfield2017testing}. 
\cite{huang2019genome}'s  MT-Comp test 
has shown seriously inflated type \RNum{1} errors 
in Figure \ref{fig:fixnulln200}, and therefore is not a fair competitor in our considered settings  despite its high power. Overall, it is clear that the proposed PoC-AB and JS-AB tests are superior over these existing methods, with the most robust control of type I error and highest power.

\begin{figure}[!htbp]
\centering
 \includegraphics[width=0.9\textwidth]{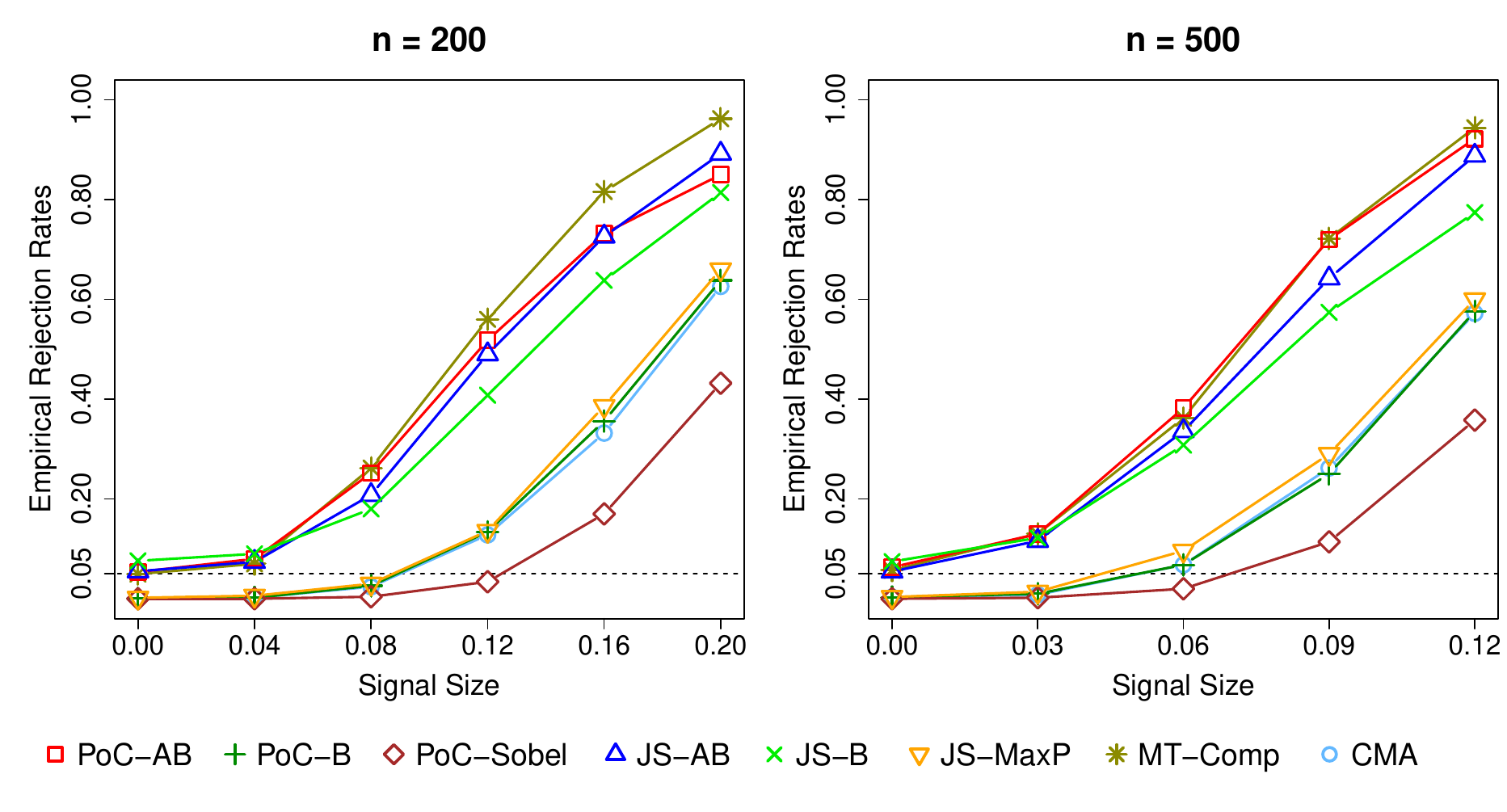}
 \caption{Empirical rejection rates (power) versus the signal strength of $\alphaS = \betaM$.} \label{fig:powerequal}
\end{figure}
\begin{figure}[!htbp] 
\centering
 \includegraphics[width=0.9\textwidth]{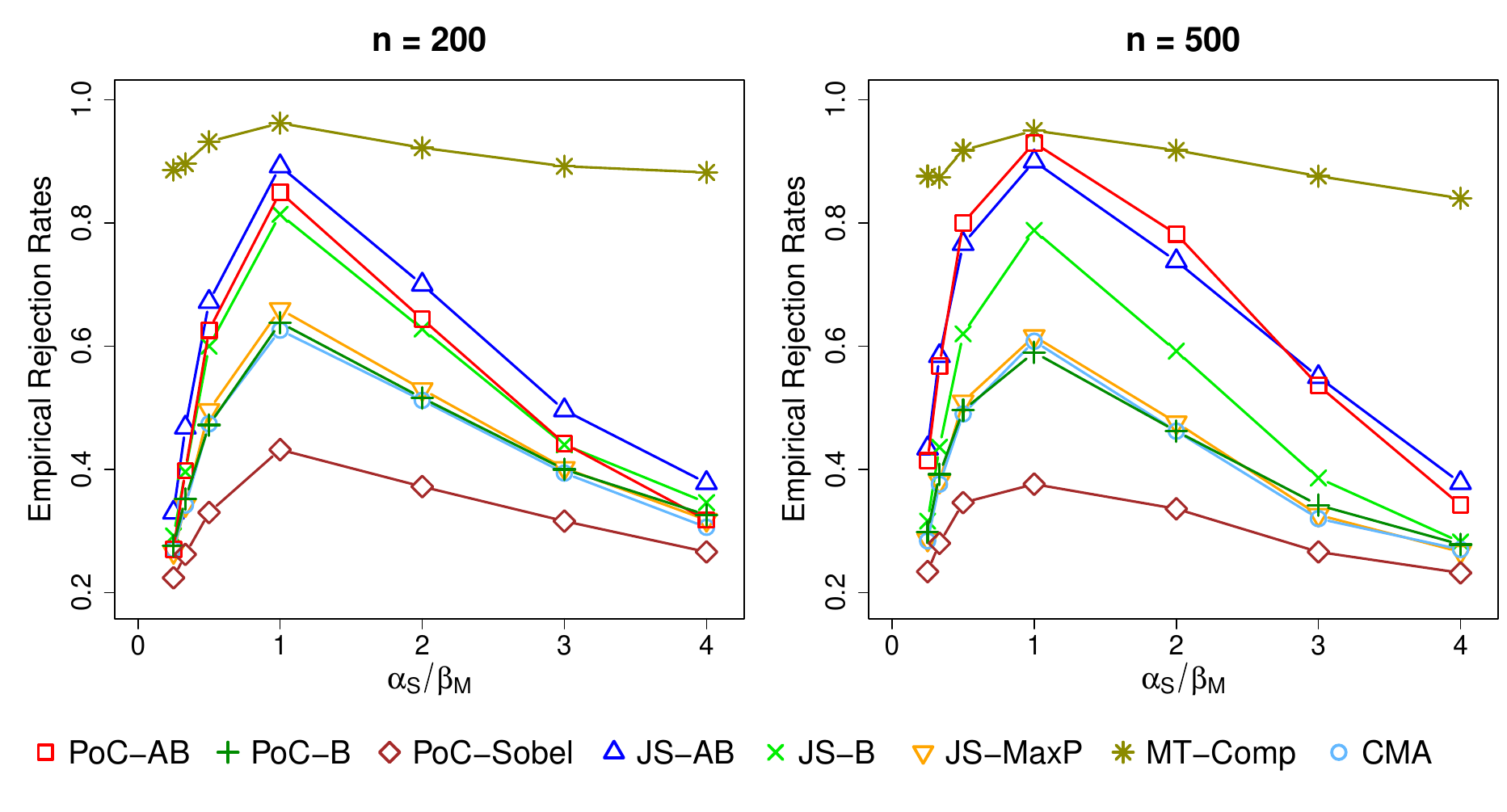}
 \caption{Empirical rejection rates (power) versus the ratio $\alphaS/\betaM$. } \label{fig:powerratio}
\end{figure}

\section{Extensions}\label{sec:extendmodel}

The adaptive bootstrap in Section \ref{sec:abt} offers a general strategy that can be extended in a wide range of scenarios beyond the model \eqref{eq:fullmod1}.   
We next examine three examples, including  testing the joint mediation effect of multivariate mediators in Section \ref{sec:jointestmulti},
testing the mediation effect in terms of odds ratio for a binary outcome in Section \ref{sec:asymbinarymedout}, 
and  testing the mediation effect in terms of risk difference when the outcome is continuous, and the mediator follows a generalized linear model in Section \ref{sec:modelbinmed}. 
In each scenario, we present details in the order of \textit{(1)} Model, \textit{(2)} Non-regularity issue, \textit{(3)} Asymptotic theory and adaptive bootstrap, and \textit{(4)} Numerical results. 

\vspace{-1.2em}

\subsection{Testing Joint Mediation Effect of  Multivariate Mediators}\label{sec:jointestmulti}

When the number of mediators is large,   
it can also be of interest to conduct  group-based mediation analyses for a set of mediators  
\citep{vanderweele2014mediation,daniel2015causal,huang2016hypothesis,sohn2019compositional,hao2022simultaneous}; 
also see a review  in \cite{blum2020challenges}. 
In this section, we show that the proposed AB method can be generalized to  test joint mediation effects. 
\medskip

\noindent 
\textit{(1) Model.} 
As an extension of \eqref{eq:fullmod1}, 
we consider the multivariate linear SEM \citep{vanderweele2014mediation,huang2016hypothesis,hao2022simultaneous},   
\begin{align}
\M_j =\alpha_{\S,j} \S +  \X^\mytrans \alpha_{\X,j} +\error_{\M,j},	\hspace{2em}\Y  = \sum_{j=1}^{\nM} \beta_{\M,j} \M_j + \X^\mytrans  \betaX + \DE \S + \eY, \label{eq:fullmodmult1} 
\end{align}
where $\X$ denotes a vector of confounders with the first element being 1 for the intercept, 
$\eY$ and $\boldsymbol{\epsilon}_M:=(\epsilon_{M,1},\ldots, \epsilon_{M,J})^{\top}$ are independent error terms with mean zero, $\mathrm{var}(\eY)=\sigmaeY^2$, and $\mathrm{cov}(\boldsymbol{\epsilon}_M)=\boldsymbol{\Sigma}_M$. 
Assume identification conditions similar to those in Section \ref{sec:review} (see Condition \ref{cond:multiidentifiability} in the Supplementary Material). The joint mediation effect through the group of mediators $\boldsymbol{M}$ is $\mathrm{E}\{\Y(s, \boldsymbol{\M}(s)) -\Y(s, \boldsymbol{\M}(s^*))\}=(s-s^*)\boldsymbol{\alpha}_{\S}^{\mytrans}\boldsymbol{\beta}_{\M}$  \citep{huang2016hypothesis}, where $\boldsymbol{\alpha}_{S}=(\alpha_{S,1},\ldots, \alpha_{S,J})^{\top}$ and  $\boldsymbol{\beta}_{M}=(\beta_{M,1},\ldots, \beta_{M,J})^{\top}$.

\medskip
\noindent 
\textit{(2) Non-Regularity Issue.}  
We are interested in $H_0:$  joint mediation effect  $=0$, which is equivalent to $H_0: \boldsymbol{\alpha}_{\S}^{\mytrans}\boldsymbol{\beta}_{\M}=0$. 
Similarly to Section \ref{sec:abt}, 
when $(\boldsymbol{\alpha}_{\S}, \boldsymbol{\beta}_{\M}) \neq \boldsymbol{0}$, i.e., there exists at least one coefficients $\alpha_{S,j}\neq 0$ or $\beta_{M,j}\neq 0$, 
we have 
${\partial  (\boldsymbol{\alpha}_{\S}^{\mytrans}\boldsymbol{\beta}_{\M})}/{\partial \alpha_{S,j}}=\beta_{M,j}\neq 0$ or ${\partial  (\boldsymbol{\alpha}_{\S}^{\mytrans}\boldsymbol{\beta}_{\M})}/{\partial \beta_{M,j}} =\alpha_{S,j}\neq 0$. 
However, when $(\boldsymbol{\alpha}_{\S}, \boldsymbol{\beta}_{\M}) = \boldsymbol{0}$, i.e.,
 $\alpha_{S,j}=\beta_{M,j}=0$ for all $j\in \{1,\ldots, J\}$, 
${\partial (\boldsymbol{\alpha}_{\S}^{\mytrans}\boldsymbol{\beta}_{\M})}/{\partial \alpha_{S,j}} = {\partial  (\boldsymbol{\alpha}_{\S}^{\mytrans}\boldsymbol{\beta}_{\M})}/{\partial \beta_{M,j}}=0$ for all $j\in\{1,\ldots, J\}$.   
We expect that  a non-regularity issue similar to that in Section \ref{sec:abt} would occur when $(\boldsymbol{\alpha}_{\S}, \boldsymbol{\beta}_{\M}) = \boldsymbol{0}$. 
This issue is also illustrated by numerical experiments in Section \ref{sec:numericresultsmulti} of the Supplementary Material.

\medskip
\noindent 
\textit{(3) Asymptotic Theory and Adaptive Bootstrap.}   
To better understand the non-regularity issue, we similarly consider a local linear SEM $\M_j =\alpha_{\S,j,n} \S +  \X^\mytrans \alpha_{\X,j} +\error_{\M,j},$ and $\Y  = \sum_{j=1}^{\nM} \beta_{\M,j,n} \M_j + \X^\mytrans  \betaX + \DE \S + \eY$, where  ${\alpha}_{\S, j, n}={\alpha}_{\S,j} +n^{-1/2}{b}_{\alpha,j}$ and ${\beta}_{\M, j, n}={\beta}_{\M,j} + n^{-1/2}{b}_{\beta,j}$.

\begin{theorem}[Asymptotic Property]\label{thm:limitmulti}
 Under Conditions \ref{cond:multiidentifiability} and  \ref{cond:multinversemoment} (the latter is a regularity condition on the design matrix similar to Condition \ref{cond:inversemoment}), and the local model, \vspace{-0.4em}
\begin{itemize}
	\item[(i)] 
 when $(\boldsymbol{\alpha}_{\S}, \boldsymbol{\beta}_{\M}) \neq \boldsymbol{0}$, 
$
\sqrt{n}\times (\hat{\boldsymbol{\alpha}}_{\S,n}^{\mytrans}\hat{\boldsymbol{\beta}}_{\M,n}-{\boldsymbol{\alpha}}_{\S,n}^{\mytrans}{\boldsymbol{\beta}}_{\M,n})\xrightarrow{d} {\boldsymbol{\alpha}}_{\S}^{\mytrans} \vec{Z}_{\M}+\boldsymbol{\beta}_{\M}^{\mytrans}\vec{Z}_{\S};$ \vspace{-0.4em}
\item[(ii)] when $(\boldsymbol{\alpha}_{\S}, \boldsymbol{\beta}_{\M}) = \boldsymbol{0}$, 
$
n\times (\hat{\boldsymbol{\alpha}}_{\S,n}^{\mytrans}\hat{\boldsymbol{\beta}}_{\M,n}-{\boldsymbol{\alpha}}_{\S,n}^{\mytrans}{\boldsymbol{\beta}}_{\M,n})\xrightarrow{d}  \boldsymbol{b}_{0,\alpha}^{\mytrans}\vec{Z}_{\M}+\boldsymbol{b}_{0,\beta}^{\mytrans}\vec{Z}_{\S} + \vec{Z}_{\S}^{\mytrans}\vec{Z}_{\M} ,$
\end{itemize}   
\noindent where $(\vec{Z}_{\S}, \vec{Z}_{\M})$ are defined to be multivariate counterparts of $({Z}_{\S},{Z}_{\M})$ in Theorem \ref{thm:prodlimit2}, and the detailed definitions are given  in Section \ref{sec:pflimitmulti} of the Supplementary Material.
\end{theorem}

To present the theory of bootstrap consistency, we define the  multivariate counterparts of 
${\mathbb{R}}_{n}^*(\balpha, \bbeta)$ in Section \ref{sec:abt} as 
$\vec{\mathbb{R}}_{n}^*(\boldsymbol{b}_{\alpha}, \boldsymbol{b}_{\beta})$.  
 The detailed forms are given in Section \ref{sec:pfmultbootjoint} of the Supplementary Material. 
Similarly to $U^*$ in Section \ref{sec:abt}, 
we define the AB statistic under the multivariate setting as
$$ 
	\vec{U}^*= (\hat{\boldsymbol{\alpha}}_{\S,n}^{*\mytrans} \hat{\boldsymbol{\beta}}_{\M, n}^*  - \hat{\boldsymbol{\alpha}}_{\S,n}^{\mytrans} \hat{\boldsymbol{\beta}}_{\M, n}) \times (1- 	\vec{\myindicator}_{\lambda_n}^* ) + n^{-1} \vec{\mathbb{R}}_{n}^*( \boldsymbol{b}_{\alpha}, \boldsymbol{b}_{\beta} ) \times 	\vec{\myindicator}_{\lambda_n}^*,  
$$
where  
$
	\vec{\myindicator}_{\lambda_n}^*=\myindicator\{\, \max\{|T_{\alpha,j,n}|,\, |T_{\alpha,j,n}^*|,\, |T_{\beta,j,n}|,\, |T_{\beta,j,n}^*|:\, 1\leqslant j\leqslant \nM \}\leqslant \lambda_n \, \},  
$
where $T_{\alpha,j,n}=\sqrt{n}\hat{\alpha}_{\S, j}/\hat{\sigma}_{\alphaS,j,n}$ and $T_{\beta, j, n}=\sqrt{n}\hat{\beta}_{\M,j,n} /\hat{\sigma}_{\betaM,j,n}$ denote the sample T-statistics of the two coefficients $\alpha_{S,j}$ and $\beta_{M,j}$, respectively, and 
$T_{\alpha,j, n}^*=\sqrt{n}\hat{\alpha}_{\S, j}^*/\hat{\sigma}_{\alphaS,j,n}^*$ and $T_{\beta, j, n}^*=\sqrt{n}\hat{\beta}_{\M,j,n}^* /\hat{\sigma}_{\betaM,j,n}^*$
denote the bootstrap counterparts of the two sample  T-statistics. 
We establish bootstrap consistency for the joint AB statistic $\vec{U}^*$ below.

\begin{theorem}[Adaptive Bootstrap Consistency] \label{thm:multbootjoint}
Under the conditions of 	Theorem \ref{thm:limitmulti}, when the tuning parameter $\lambda_n$ satisfies $\lambda_n=o(\sqrt{n})$ and $\lambda_n\to \infty$ as $n\to \infty,$  
$$c_n\vec{U}^* \overset{d^*}{\leadsto} c_n(\hat{\boldsymbol{\alpha}}_{\S,n}^{\mytrans} \hat{\boldsymbol{\beta}}_{\M, n} - {\boldsymbol{\alpha}}_{\S,n}^{\mytrans} {\boldsymbol{\beta}}_{\M, n}),$$
where $c_n$ is specified as in \eqref{eq:cndef}. 
\end{theorem} 

Based on Theorem \ref{thm:multbootjoint}, we can develop an AB test  similar to that in Section \ref{sec:abt}. 

\medskip
\noindent \textit{(4) Numerical Results.} 
To evaluate the performance of the joint AB test, 
we conduct numerical experiments, detailed in Section \ref{sec:numericresultsmulti} of the Supplementary Material.
 We compare the AB test with the classical bootstrap and two tests in \cite{huang2016hypothesis}: the Product Test based on Normal Product distribution (PT-NP) and  the  
 Product Test based on Normality (PT-N).  
We observe results similar to those in Section \ref{sec:sim}.  
Specifically, under $H_0: \boldsymbol{\alpha}_S^{\top}\boldsymbol{\beta}_M=0$, when $(\boldsymbol{\alpha}_S, \boldsymbol{\beta}_M)\neq \mathbf{0}$, both the proposed AB test and the compared methods yield uniformly distributed $p$-values.
However, when  $(\boldsymbol{\alpha}_S, \boldsymbol{\beta}_M)= \mathbf{0}$, the compared methods become overly conservative, whereas the AB test still produces uniformly distributed $p$-values. 
Under $H_A$, the AB test can achieve higher empirical power than the compared methods. 
Besides simulations, we also provide an exemplary data analysis in Section  \ref{sec:datagroup} of the Supplementary Material.

%Moreover, for the  ELEMENT  dataset in Section \ref{sec:datanalysis}, we also conduct tests for the joint mediation effect of the selected mediators after screening. 
%The details are given in Section  \ref{sec:datagroup} of the Supplementary Material. 
%We observe that AB test, PT-N, and PT-NP all reject $H_0$, suggesting that there exists significant joint mediation effect for the set of selected mediators. 
%This supports the detection of significant mediators by AB test in Section \ref{sec:datanalysis}.  
%Notably, among three joint tests, AB returns the most significant $p$-value. 
%This is consistent with our observation that  the AB  can achieve  high power in numerical experiments.

\subsection{Non-Linear Scenario I: Binary  Outcome and General Mediator} \label{sec:asymbinarymedout} 

\noindent \textit{(1) Model.} 
Suppose the outcome is binary, and consider the model
\begin{align} \label{eq:extendmodel1}
P(\Y=1 \mid \S, \M, \X ) =  &~\mathrm{logit}^{-1}( \betaM \M + \X^\mytrans  \betaX + \DE \S), \\
\mathrm{E}(\M \mid \S, \X )=&~ h^{-1}(\alphaS \S +  \X^\mytrans \alphaX ), \notag
\end{align}
where $h^{-1}(\cdot)$ is the inverse of a canonical link function in generalized linear models. 
Under Model \eqref{eq:extendmodel1}, since the outcome is binary, it is conventional to define the mediation effect as the odds ratio  \citep{vanderweele2010odds}.
Specifically, under the identification assumption given in Section \ref{sec:review},  the conditional natural indirect effect (mediation effect) can be identified as 
\begin{align*} 
\mathrm{OR}^{\mathrm{NIE}}_{s\mid s^*}(s,\boldsymbol{x})=\frac{P\left\{Y(s, M(s))=1 \mid  \X=\boldsymbol{x} \right\} /\left\{1-P\left(Y(s, M(s))=1 \mid  \X=\boldsymbol{x}\right)\right\}}{P\left\{Y(s, M(s^*))=1 \mid  \X=\boldsymbol{x}\right\} /\left\{1-P\left\{Y(s, M(s^*))=1 \mid \X =\boldsymbol{x} \right\}\right\}}, 
\end{align*}
where $M(s)$ denotes the potential value of the mediator under the exposure $\S=s$,  and $Y(s,m)$ denotes the potential outcome that would have been observed if $\S$ and $\M$ had been set to $s$ and $m$, respectively.
Under $H_0$ of no mediation effect,
\begin{align} \label{eq:ornull}
H_0: \mathrm{OR}^{\mathrm{NIE}}_{s\mid s^*}(s,\boldsymbol{x})=1 \ \Leftrightarrow \ &~\log \mathrm{OR}^{\mathrm{NIE}}_{s\mid s^*}(s,\boldsymbol{x}) =0\\
	 \ \Leftrightarrow \ &~P\left\{Y(s, M(s))=1 \mid  \X=\boldsymbol{x} \right\}-P\left\{Y(s, M(s^*))=1 \mid  \X =\boldsymbol{x}\right\}=0, \notag
\end{align}
where the second equivalence follows from the strict increasing monotonicity of the function $x/(1-x)$ when $0<x<1$.

\begin{remark} 
We consider natural indirect/mediation effects conditioning on covariates $\X=\boldsymbol{x}$ following \cite{vanderweele2010odds}.
Alternatively, \cite{imai2010general} proposed to examine the average NIE that marginalizes the distribution of $\X$. 
Examining the conditional NIE is mainly for technical convenience. 
The conditional NIE $=0$ for all $\boldsymbol{x}$ can give a sufficient condition for the average NIE $=0$. 
Conclusions of conditional NIE may be obtained for average NIE similarly. 
Please see Remark \ref{rm:extintegrationdef} in the Supplementary Material for more details. 
\end{remark}

\medskip
\noindent  \textit{(2) Non-Regularity Issue.} 
The null hypothesis of no mediation effect \eqref{eq:ornull} looks different from $H_0:\alphaS\betaM=0$ under the linear SEMs in Section \ref{sec:review}. Nevertheless, we can show that the non-regularity issue similar to that in Section \ref{sec:review} would still arise.  This is formally stated as Proposition \ref{prop:singularitylogiresponse} below.

\begin{proposition} \label{prop:singularitylogiresponse} 
Under the model \eqref{eq:extendmodel1},    Condition \ref{cond:phiintegration} (a general regularity condition on the link function $h^{-1}(\cdot)$ and the distribution of $M$), and identification conditions in Section \ref{sec:review}, \vspace{-0.7em}
\begin{enumerate}\setlength{\itemsep}{2pt}
	\item $H_0$ \eqref{eq:ornull} holds for $s\neq s^*$ if and only if $\alphaS=0$ or $\betaM=0$. 
\item For simplicity of notation, let $\mathrm{NIE}$ be a shorthand for $\log \mathrm{OR}^{\mathrm{NIE}}_{s\mid s^*}(s,\boldsymbol{x})$. We have\\[2pt]  
 \quad (i)  $\frac{\partial   \mathrm{NIE}}{\partial \alphaS} \Big|_{\betaM = 0}= \frac{\partial  \mathrm{NIE}}{\partial \betaM} \Big|_{\alphaS = 0}=0$, \quad (ii) $\frac{\partial  \mathrm{NIE}}{\partial \alphaS} \Big|_{\alphaS=0, \betaM\neq 0}\neq 0$, \quad (iii) $\frac{\partial  \mathrm{NIE}}{\partial \betaM} \Big|_{\alphaS\neq 0, \betaM = 0}\neq 0$.
\end{enumerate}
\end{proposition}
It is interesting to see that even though the conditional  mediation effect $\mathrm{OR}^{\mathrm{NIE}}_{s\mid s^*}(s,\boldsymbol{x})$ does not take a product form, 
a non-regularity issue caused by zero gradient can still arise under $H_0$ in \eqref{eq:ornull}, which is similar to the PoC statistic in Section \ref{sec:abt}. 
Specifically, 
Proposition \ref{prop:singularitylogiresponse} implies that 
 when $\alphaS=\betaM=0$,  the first-order Delta method cannot be directly applied to the inference of $\mathrm{NIE}$, which is different from the scenarios when $\alphaS\neq 0$ or $\betaM\neq 0$. 
Therefore, we expect that the ordinary estimator of NIE can behave differently under different types of null hypotheses, and a non-regularity issue can occur. 
This phenomenon is indeed demonstrated by numerical experiments in Section \ref{sec:nonlineartwo} of the Supplementary Material.
 
 \medskip
\noindent \textit{(3) Asymptotic Theory and Adaptive Bootstrap.}   
For ease of presentation, we next derive  asymptotic theory under a special case of \eqref{eq:extendmodel1}, where the mediator is binary and follows a logistic regression model. 
We point out that the analysis in this section can  be readily extended to cases where the mediator $M$ follows a linear model or other canonical generalized linear models.  
Specifically, 
 let $\M$ and $\Y$ be Bernoulli random variables with mean values in \eqref{eq:extendmodel1}, and $h^{-1}(x)=\mathrm{logit}^{-1}(x)$. 
In this case, $\log  
\mathrm{OR}^{\text{NIE}}_{s\mid s^*}(s,\boldsymbol{x})=l(P_s)-l(P_{s^*}), 
$ 
where  $P_s:=P\left\{Y(s, M(s))=1 \mid  \X=\boldsymbol{x} \right\} $,  $P_{s^*}:=P\left\{Y(s, M(s^*))=1 \mid  \X=\boldsymbol{x} \right\},
$ 
and $l(x)=\log\frac{x}{1-x}$. 
Similarly to Section \ref{sec:abt}, 
 we are interested in understanding how the local limiting behaviors of $\alphaS$ and $\betaM$ coefficients change. To this end, we consider a general  local logistic model: 
\begin{align} \label{eq:extendmodel1local}
\mathrm{E}(\M \mid \S, \X )=g(\alphaSn \S +  \X^\mytrans \alphaX),\hspace{1.5em} 
\mathrm{E}(\Y \mid \S, \M, \X )  = g( \betaMn \M + \X^\mytrans  \betaX + \DE \S), %\notag
\end{align}
where $\alphaSn=\alphaS+\balpha/\sqrt{n}$, $\betaMn=\betaM+\bbeta/\sqrt{n}$, and $g(x)=\mathrm{logit}^{-1}(x)=e^x/(1+e^x)$. 
Under the local model  \eqref{eq:extendmodel1local}, we have for $\iota \in \{s,s^*\}$, 
\begin{align}
	P_{\iota}= &~ g\big( \iota \times \alphaSn   +\boldsymbol{x}^{\mytrans}\alphaX\big)\times d_{\beta,n} + P_*, \label{eq:pspsstardef}
\end{align}
where $ d_{\beta,n} =g(\betaMn  + \boldsymbol{x}^{\mytrans} \boldsymbol{\beta}_{\X} + \tau_S s ) -g( \boldsymbol{x}^{\mytrans} \boldsymbol{\beta}_{\X} + \tau_S s )$ and $P_*=g(\boldsymbol{x}^{\mytrans} \boldsymbol{\beta}_{\X} + \tau_S s )$. (Please see   the proof of  Theorem  \ref{prop:asymbinaryoutcome} for the derivations.) 
For simplicity of notation, let NIE be a shorthand of $\log \mathrm{OR}^{\text{NIE}}_{s\mid s^*}(s, \boldsymbol{x})$, and  by \eqref{eq:ornull}, $H_0\Leftrightarrow$ NIE $=0$. 
Let $\widehat{\text{NIE}}=l(\hat{P}_s)-l(\hat{P}_{s^*})$ denote an estimator of NIE, 
where $\hat{P}_s$ and $\hat{P}_{s^*}$ are defined similarly to \eqref{eq:pspsstardef} with $({\alpha}_{S,n}, {\boldsymbol{\alpha}}_{\boldsymbol{X}}, {\beta}_{M,n}, {\boldsymbol{\beta}}_{\boldsymbol{X}}, {\tau}_S)$ replaced by their corresponding sample regression coefficient estimators $(\hat{\alpha}_{S}, \hat{\boldsymbol{\alpha}}_{\boldsymbol{X}}, \hat{\beta}_{M}, \hat{\boldsymbol{\beta}}_{\boldsymbol{X}}, \hat{\tau}_S)$.

\begin{theorem}[Asymptotic Property]  \label{prop:asymbinaryoutcome}
Assume $P_s$ and $ P_{s^*}\in (0,1)$ and Condition \ref{cond:designmatrixlogistic} in the Supplementary Material (a regularity condition on the design matrix similar to Condition \ref{cond:inversemoment}).
  Under the local model \eqref{eq:extendmodel1local} and $H_0:$ $\alphaS\betaM=0$, \vspace{-0.2em}
\begin{itemize}
	\item[(i)] when  $(\alphaS,\betaM)\neq \mathbf{0}$, $
	\sqrt{n}(\widehat{\mathrm{NIE}}-\mathrm{NIE})\xrightarrow{d}	 ( d_{\alpha} Z_{\beta} +   d_{\beta}  Z_{\alpha})\gamma_{0}; 
$ \vspace{-0.2em}
	\item[(ii)] when  $(\alphaS,\betaM)=\mathbf{0}$, 
$
	n(\widehat{\mathrm{NIE}}-\mathrm{NIE})\xrightarrow{d}	( d_{\balpha} Z_{\beta} +   d_{\bbeta}Z_{\alpha}+ Z_{\alpha}Z_{\beta})\gamma_{0},
$
\end{itemize} 
where $d_{\alpha}=g(\alphaS s + \boldsymbol{x}^{\mytrans}\alphaX)-g(\alphaSn s^* + \boldsymbol{x}^{\mytrans}\alphaX)$,  $ d_{\beta} =g(\betaM + \boldsymbol{x}^{\mytrans} \boldsymbol{\beta}_{\X} + \tau_S s ) -g(\boldsymbol{x}^{\mytrans} \boldsymbol{\beta}_{\X} + \tau_S s )$, 
$d_{\balpha}=g'(\boldsymbol{x}^{\mytrans}\alphaX)(s-s^*) \balpha$, $
d_{\bbeta}=g'(\boldsymbol{x}^{\mytrans}\betaX + \tau_S s)(s-s^*) \bbeta$, 
$(Z_{\alpha},Z_{\beta})$ represent bivariate mean-zero  normal distributions specified in Lemma \ref{lm:speratedalphabetalimit}, and $\gamma_0 = \{P_*(1-P_*)\}^{-1}$   
is a non-zero constant with $P_*$ given in \eqref{eq:pspsstardef}.  
\end{theorem}
We next study consistency of bootstrap estimators. 
Let $\widehat{\mathrm{NIE}}^*$ denote the classical nonparametric bootstrap estimator  of $\mathrm{NIE}$. 
Specifically, $\widehat{\mathrm{NIE}}^*=l( \hat{P}_s^*)-l(\hat{P}_{s^*}^*)$, 
where $ \hat{P}_s^*$, and $\hat{P}_{s^*}^*$ are defined similarly to \eqref{eq:pspsstardef} with $({\alpha}_{S,n}, {\boldsymbol{\alpha}}_{\boldsymbol{X}}, {\beta}_{M,n}, {\boldsymbol{\beta}}_{\boldsymbol{X}}, {\tau}_S)$ replaced by  their classical nonparametric bootstrap estimators 
$
	(\hat{\alpha}_S^*, \hat{\boldsymbol{\alpha}}_{\boldsymbol{X}}^*, \hat{\beta}_M^*, \hat{\boldsymbol{\beta}}_{\boldsymbol{X}}^*, \hat{\tau}_S^*).
$ 
Motivated by Theorem \ref{prop:asymbinaryoutcome}, we define the AB statistic 
$$U_{e,1}^*= (\widehat{\mathrm{NIE}}^* -\widehat{\mathrm{NIE}} ) \times (1- \myindicator_{\alphaS, \lambda_n}^* \myindicator_{\betaM, \lambda_n}^*  )+n^{-1} \big( d_{\balpha}  \mathbb{Z}_{\beta}^* +   d_{\bbeta} \mathbb{Z}_{\alpha}^*+ \mathbb{Z}_{\alpha}\mathbb{Z}_{\beta}^*\big)\hat{\gamma}_{0}^* \times \myindicator_{\alphaS, \lambda_n}^*  \myindicator_{\betaM, \lambda_n}^*,$$ 
where $\myindicator_{\alphaS, \lambda_{n}}^*$ and $\myindicator_{\betaM, \lambda_{n}}^*$ are defined similarly to \eqref{eq:indicatoralphabetasep}. 
The following theorem proves  consistency of the AB statistic $U_{e,1}^*$, based on which we can develop an AB test  similar to that in Section \ref{sec:abt}. 
\begin{theorem}[Adaptive Bootstrap Consistency] \label{thm:combinebootbinaryout}
Under the conditions of Theorem \ref{prop:asymbinaryoutcome}, when the tuning parameter $\lambda_n$ satisfies $\lambda_n=o(\sqrt{n})$ and $\lambda_n\to \infty$ as $n\to \infty,$  
$c_n U_{e,1}^* \overset{d^*}{\leadsto}c_n( \widehat{\mathrm{NIE}}- \mathrm{NIE}),$ 
 where $c_n$ is specified as in \eqref{eq:cndef}. 
\end{theorem}

\medskip

\noindent \textit{(4) Numerical Results.} 
We conduct simulation studies to compare the AB and the classical non-parametric bootstrap under the model \eqref{eq:extendmodel1}.
The detailed results are provided in Section \ref{sec:nonlineartwo} of the Supplementary Material. 
% We obtain conclusions similar to that in 
Our findings align closely with those presented 
in Section \ref{sec:sim}. 
Specifically, under $H_0:{\alpha}_S{\beta}_M=0$, when $({\alpha}_S, {\beta}_M)\neq \mathbf{0}$, 
both the proposed AB test and the classical non-parametric bootstrap  yield uniformly distributed $p$-values.
However, when  $({\alpha}_S,{\beta}_M)= \mathbf{0}$, the classical bootstrap becomes overly conservative, whereas the AB test still yields uniformly distributed $p$-values. 
Under $H_A$, the AB test can achieve higher empirical power than the classical bootstrap.

\subsection{Non-Linear Scenario II: Linear Outcome and General Mediator} \label{sec:modelbinmed} 

\noindent \textit{(1) Model.} 
Suppose the outcome follows a  linear model, and  consider 
\begin{align} \label{eq:extendmodel2}
\mathrm{E}(\M \mid \S, \X )= h^{-1}(\alphaS \S +  \X^\mytrans \alphaX),  \hspace{1.5em} \mathrm{E}(\Y \mid \S, \M, \X )  =  \betaM \M + \X^\mytrans  \betaX + \DE \S,%\notag
\end{align}
where $h^{-1}(\cdot)$ can be the inverse of a canonical link function. 
Similarly to the non-linear Scenario I, we examine the conditional natural indirect effect/mediation effect defined as the risk difference: 
\begin{align}
\mathrm{NIE}_{s\mid s^*}(s,\boldsymbol{x}) := &~\mathrm{E}\big\{\Y(s, \M(s))- \Y(s, \M(s^*)) \mid \X =\boldsymbol{x} \big\} \notag\\
=&~\betaM\left\{ h^{-1}\big(\alphaS s+ \boldsymbol{x}^{\mytrans}\alphaX  \big) -h^{-1}\big(\alphaS s^*+ \boldsymbol{x}^{\mytrans}\alphaX \big) \right\}. \label{eq:h0binarymediator} 
\end{align}

%\medskip
\noindent \textit{(2) Non-Regularity Issue.}   
We are interested in testing  $H_0: \mathrm{NIE}_{s\mid s^*}(s,\boldsymbol{x})=0$, which  
looks different from $H_0:\alphaS\betaM=0$ in Section \ref{sec:review}. Nevertheless, we can show that  the non-regularity issue similar to that in Section \ref{sec:review} would arise.  This is formally stated as Proposition \ref{prop:singularitylogitm} below. 
\begin{proposition} \label{prop:singularitylogitm}
Under the model \eqref{eq:extendmodel2}, assume $h^{-1}(\cdot)$ is strictly monotone, and the identification conditions in Section \ref{sec:review} hold. 
Let $\mathrm{NIE}$ be a shorthand for $\mathrm{NIE}_{s\mid s^*}(s,\boldsymbol{x})$ in \eqref{eq:h0binarymediator}. 
Then \vspace{-1.2em}
\begin{enumerate}
	\item $H_0:$ $\mathrm{NIE}=\eqref{eq:h0binarymediator}=0$ holds if and only if  $\alphaS=0$ or $\betaM=0$. 
	\item % We have \\[3pt]
 \quad (i)  $\frac{\partial   \mathrm{NIE}}{\partial \alphaS} \Big|_{\betaM = 0}= \frac{\partial  \mathrm{NIE}}{\partial \betaM} \Big|_{\alphaS = 0}=0$. \ (ii) $\frac{\partial  \mathrm{NIE}}{\partial \alphaS} \Big|_{\alphaS=0, \betaM\neq 0}\neq 0$. \ (iii) $\frac{\partial  \mathrm{NIE}}{\partial \betaM} \Big|_{\alphaS\neq 0, \betaM = 0}\neq 0$. 
\end{enumerate}	
\end{proposition}
Similarly to Proposition \ref{prop:singularitylogiresponse}, Proposition \ref{prop:singularitylogitm} implies that
a non-regularity issue caused by  zero gradient would arise  under $H_0$.   
Specifically, the ordinary estimator of NIE can behave differently  when $\alphaS=\betaM=0$, and when one of $\alphaS$ and $\betaM \neq 0$. 
This is similar to the PoC  statistic in Section \ref{sec:abt} and the odds ratio in Section \ref{sec:asymbinarymedout}.

\medskip
\noindent \textit{(3) Asymptotic Theory and Adaptive Bootstrap.}   
For ease of presentation, we next derive  asymptotic theory under a specific instance of \eqref{eq:extendmodel2}. 
Specifically,  the mediator $M$ is a Bernoulli random variable with its conditional mean given in \eqref{eq:extendmodel1} and $h^{-1}(x)=\mathrm{logit}^{-1}(x)=e^x/(1+e^x)$, 
and the outcome $Y$ follows the linear model in \eqref{eq:fullmod1}.  
The analysis in this section can be readily extended when the mediator $M$ follows other canonical generalized linear models.  
As we are interested in how the local limiting behavior of $\alphaS$ and $\betaM$ coefficients change, 
we consider the following general local  model 
\begin{align} \label{eq:extendmodel2local}
\mathrm{E}(\M \mid \S, \X )=\mathrm{logit}^{-1}(\alphaSn \S +  \X^\mytrans \alphaX),\hspace{2em} Y  = 	\betaMn \M + \X^\mytrans  \betaX + \DE \S +\eY. 
\end{align}
where $\alphaSn=\alphaS+\balpha/\sqrt{n}$, and $\betaMn=\betaM+\bbeta/\sqrt{n}$.

\begin{theorem}[Asymptotic Property]    \label{prop:asymbinarymediator}
Assume Condition \ref{cond:designmatrixbinarymediator} in the Supplementary Material (a regularity condition on the design matrix similar to Condition \ref{cond:inversemoment}).  
Under model \eqref{eq:extendmodel2local}, \vspace{-1.0em}
\begin{itemize}
\item[(i)] when   $(\alphaS, \betaM)\neq \mathbf{0}$, $\sqrt{n}(\widehat{\mathrm{NIE}}-\mathrm{NIE}) \xrightarrow{d} d_{\alpha} Z_{\beta}+ \betaM Z_{\alpha} $; \vspace{-0.3em} 
\item[(ii)] when $(\alphaS, \betaM)=\mathbf{0}$, $n(\widehat{\mathrm{NIE}}-\mathrm{NIE})\xrightarrow{d}  d_{\balpha}  Z_{\beta} + \bbeta Z_{\alpha} +Z_{\alpha}Z_{\beta}$, 
\end{itemize}
where $d_{\alphaS}=g(\alphaS s + \boldsymbol{x}^{\mytrans}\alphaX)-g(\alphaSn s^* + \boldsymbol{x}^{\mytrans}\alphaX)$, 
$d_{\balpha}=g'(\boldsymbol{x}^{\mytrans}\alphaX)(s-s^*) \balpha$, 
$Z_{\alpha}$ represents a normal distribution specified in Lemma \ref{lm:speratedalphabetalimit}, 
and $Z_{\beta}$ is redefined to be a mean-zero normal distribution with a covariance same as  the random vector $V_M^{-1}\epsilon_YM_{\perp'}$, where $V_M$ and $M_{\perp'}$ are defined in Theorem \ref{thm:prodlimit2}.
\end{theorem} 

We next establish bootstrap consistency theory. 
Similarly to Section \ref{sec:asymbinarymedout}, 
let $\widehat{\mathrm{NIE}}^*$ denote the nonparametric bootstrap estimator  of $\mathrm{NIE}$. 
In particular, we redefine $\widehat{\mathrm{NIE}}^*=\hat{\beta}_M^*\{g(\hat{\alpha}_S^*s + \boldsymbol{x}^{\top}\hat{\boldsymbol{\alpha}}_{\boldsymbol{X}}^* ) -g(\hat{\alpha}_S^*s^* + \boldsymbol{x}^{\top}\hat{\boldsymbol{\alpha}}_{\boldsymbol{X}}^* )\}$, 
where 
 $(\hat{\alpha}_S^*, \hat{\boldsymbol{\alpha}}_{\boldsymbol{X}}^*, \hat{\beta}_M^*)$ denotes the classical nonparametric bootstrap estimators of   $({\alpha}_S, {\boldsymbol{\alpha}}_{\boldsymbol{X}}, {\beta}_M)$. 
Motivated by Theorem \ref{prop:asymbinarymediator},  we define the AB statistic   
$$
U_{e,2}^*=(\widehat{\mathrm{NIE}}^* -\widehat{\mathrm{NIE}} )  (1- \myindicator_{\alphaS, \lambda_n}^* \myindicator_{\betaM, \lambda_n}^*  ) +n^{-1}(d_{\balpha}  \mathbb{Z}_{\beta}^* + \bbeta  \mathbb{Z}_{\alpha}^* +\mathbb{Z}_{\alpha}^*\mathbb{Z}_{\beta}^*)  \myindicator_{\alphaS, \lambda_n}^*  \myindicator_{\betaM, \lambda_n}^*, 
$$
where $\myindicator_{\alphaS, \lambda_{n}}^*$ and $\myindicator_{\betaM, \lambda_{n}}^*$ are defined similarly to \eqref{eq:indicatoralphabetasep}. 
The following theorem establishes consistency of the AB statistic $U_{e,2}^*$.
\begin{theorem}[Adaptive Bootstrap Consistency] \label{thm:combinebootbinarymediator}
Under conditions of Theorem \ref{prop:asymbinarymediator}, when the tuning parameter $\lambda_n$ satisfies $\lambda_n=o(\sqrt{n})$ and $\lambda_n\to \infty$ as $n\to \infty,$  
$c_nU_{e,2}^*\overset{d^*}{\leadsto} c_{n}( \widehat{\mathrm{NIE}}- \mathrm{NIE}),$   
where $c_n$ is specified as in \eqref{eq:cndef}. 
\end{theorem}
 
\noindent \textit{(4) Numerical Results.} 
We conduct simulation studies to compare the AB and the classical non-parametric bootstrap under the model \eqref{eq:extendmodel2}.
The detailed results are provided in Section \ref{sec:nonlineartwo} of the Supplementary Material. 
The obtained results are very similar to those in Section \ref{sec:sim} and  Section \ref{sec:asymbinarymedout} Part 4),
and therefore, we refrain from repeating the details here.

\section{Data Analysis}\label{sec:datanalysis}

We illustrate an application of our proposed method to the analysis of data from a cohort study ``Early Life Exposures in Mexico to ENvironmental Toxicants" (ELEMENT) \citep{perng2019}. 
One of the central interests in this scientific study concerns the mediation effects of metabolites, in particular, the family of lipids,
on the association between environmental exposure and children growth and development. In the literature of environmental health sciences, exposure to endocrine-disrupting chemicals (EDCs) such as phthalates have been found to be detrimental to children's health outcomes. Such findings of  direct associations need to be further assessed for possible mediation effects through metabolites, because environmental toxicants such as phthalates can alter metabolic profiles at the molecular level.

Our illustration focuses on the outcome of body mass index (BMI) and exposure to one phthalate, MEOHP (a chemical in food production and storage).  BMI is a widely used biomarker in pediatric research to measure childhood obesity. 
The dataset contains 382   adolescents aged 10-18 years old living in Mexico   City.  
Our mediation analysis involves a set of 149 lipids that are hypothesized to have potential mediation effects on children's growth and development.  Our goal is to identify the mediation pathways of exposure to MEOHP $\rightarrow$ lipids $\rightarrow$ BMI. 
Two key potential confounders, gender and age,  are included throughout the analyses. 
{
%We note that the analyses here are exemplary of applying and comparing the tests. Adjusting for gender and age may be insufficient for confounder adjustments. To obtain plausible causal interpretation, we anticipate that further scrutiny is needed to assess the causal assumptions.
It is worth noting that adjusting for gender and age may not be sufficient for proper confounding adjustments. To conduct a more plausible causal analysis and interpretation, a further investigation is deemed necessary to rigorously assess the underlying causal assumptions such as a sensitivity analysis for the sequential ignorability assumption.} \label{pageref:exemplary}
In our analyses, 
we compare the results of six tests:    
JS-AB, JS-MaxP, PoC-AB, PoC-B, PoC-Sobel, and CMA, which have been compared in our simulation studies  in Section \ref{sec:sim}. 
In particular, 
all the bootstrap methods (including JS-AB, PoC-AB, PoC-B, CMA)   
 are conducted based on $10^4$ bootstrap resamples. 
Here we no longer include the JS-B test 
and the MT-Comp method, as they are known to have  inflated type \RNum{1} errors according to our simulation studies in Section \ref{sec:sim}.

As the sample size is limited compared to the large number of mediators, we first apply
a screening analysis to identify a subset of lipids as potential candidates. 
We then jointly model the chosen lipids in the second step of our analysis.  
{To mitigate the potential issues arising from double dipping the data, we adopt a random data splitting approach by dividing the dataset into two distinct parts, each dedicated to one of the two respective analytic tasks.}
%we randomly split the data into two parts, which are used in the two steps, respectively. 
In the first screening step, we examine the effect along the path MEOHP $\to$ lipid $\to$ BMI for one lipid at a time, and the corresponding  $p$-values are obtained with the six tests, respectively. 
% For each test, we select 15 lipids with the 15 smallest $p$-values. 
For each test, we select 
a proportion ($q\%$) of lipids with the smallest $p$-values. 
The second step 
examines the path MEOHP $\to$ selected lipids $\to$ BMI, with the selected lipids being modeled jointly. 
%Given a target lipid $M$ in the selected set, 
To test the mediation effect through a target lipid $M$ within the selected set,
we adjust for non-target mediators within the outcome model, following the discussions on Page  \pageref{page:discussionmultimed}; please see
more details in  Section \ref{sec:datastep2int} of the Supplementary Material. 
\label{page:outcomedata}
%Each test 
Subsequently, we select lipids based on their $p$-values obtained in the second step, after adjusting for multiple comparisons with controlled false discovery rate (FDR)  \citep{benjamini1995controlling}.  
In our analysis, we explore a range of $q$ values $\{5, 10, 15, 20, 25\}$ and observe very similar results, indicating the robustness of our approach to the choice of the screening threshold in the first step. 
We next present the results obtained with $q = 10$ (i.e., 15 selected lipids based on their $p$-values), while results for other $q$ values are detailed in Section \ref{sec:changescreen}  of the Supplementary Material.

As an illustrative example, 
we first present the results from a single random split in Table \ref{tb:jointdetectres}. 
Table \ref{tb:jointdetectres} provides the corresponding $p$-values for the lipids selected by at least one test in the second step of the analysis.
%the $p$-values corresponding to  the lipids  selected by at least one test in the second step of the analysis. 
In this instance, the non-AB tests fail to detect any lipids. In contrast, the PoC-AB test identifies lauric acid (L.A) and FA.7.0-OH\_1 (FA.7) while controlling the FDR at 0.10, and the JS-AB test selects both L.A and FA.7 when the FDR is controlled at 0.05 and 0.10, respectively. 
% under different FDR levels. 
% . Meanwhile, the PoC-AB test identifies lauric acid (L.A) and FA.7.0-OH 1 (FA.7) with FDR controlled at 0.10. The JS-AB test, on the other hand, selects L.A and FA.7 when FDR is controlled at 0.05 and 0.01, respectively.
{To gauge the variability of results across random splits, we repeat the data-splitting analysis 400 times.}
%For each lipid, we count the instances in which the lipid is selected in the second step.  
As shown in Figure \ref{fig:countvalmain}, L.A and  FA.7 are the two most frequently selected mediators in our analysis.
%Figure \ref{fig:countvalmain} presents the top three mediators most frequently selected in our analysis.
%We then pinpointed the top three mediators most frequently chosen, and their selection frequencies are depicted in  Figure \ref{fig:countvalmain}.
% we further repeat the two-step analysis by  randomly  splitting the data 400 times and count the numbers of times that the lipids are selected in the second step. 
% We choose the top three mediators that are  most frequently selected and present their numbers of times being  selected in 
% Figure \ref{fig:countvalmain}.
% presents the number of times of the mediators that are most frequently selected. 
% that L.A and FA.7.0-OH\_1 are selected in the second step. 
% We find 
%The results imply that the mediation effects via L.A and  FA.7 are likely to be significant. 
% the two
% most significant mediators. The results imply that L.A and  FA.7 are the two most significant mediators. 
Furthermore, the  AB tests exhibit a notably higher chance of selecting L.A compared to the non-AB tests. 
This aligns with our observations from simulations in Section \ref{sec:sim},  suggesting that the AB tests can attain higher power than their non-AB counterparts.  
Lauric acid is a saturated fatty acid and is found in many vegetable fats and in coconut and palm kernel oils \citep{dayrit2015properties}. 
The results suggest  
that the exposure to MEOHP may influence the  process of breaking down fat tissue in the human body, leading to obesity and other adverse health outcomes.

% \begin{table} 
% \centering
% \caption{ Lipid mediators detected in the second step.}\label{tb:jointdetectres}
% \begin{threeparttable}
% \begin{tabular}{c|cccccc}\hline
%  FDR level  & JS-AB &  JS-MaxP &  PoC-AB &   PoC-B & PoC-Sobel &  CMA  \\ \hline
%  0.05 & L.A \& FA.7  & $\times$ & $\times$ & $\times$ & $\times$ & $\times$ \\ %\hline
%  0.10 & L.A \& FA.7 & $\times$ & L.A \& FA.7 & $\times$ & $\times$ & $\times$ \\  \hline
% % % & FA.18.0-DiC & NA & NA &  * & * & *\\ \hline 
% %  0.20 & L.A \& F.18 & $\times$ & L.A \& F.18 & $\times$  & $\times$ &  $\times$ \\\hline 
% % & FA.18.0-DiC &  &  FA.18.0-DiC &   &  & \\%\hline
% \end{tabular}
%  \begin{tablenotes} 
%         \item[] \hspace{-8pt} \small{(Abbreviations: LAURIC.ACID (L.A); FA.7.0-OH\_1 (FA.7); no mediators detected ($\times$).)}\\
%         \end{tablenotes}
%     \end{threeparttable}
% \end{table}

\begin{table}[!htbp]
\centering
\caption{ Lipids selected in the second step.}\label{tb:jointdetectres}
\begin{threeparttable}
\setlength{\tabcolsep}{10pt}
\begin{tabular}{c|cccccc}\hline
Lipids & JS-AB &  JS-MaxP &  PoC-AB &   PoC-B & PoC-Sobel &  CMA  \\ \hline
 L.A & 0.0017 ($\times *$) & 0.0399 & 0.0043 ($*$) & 0.0406 & 0.1254 & 0.0426 \\ %\hline
FA.7 & 0.0008 ($\times *$)  & 0.0146 & 0.0090 $(*)$ & 0.0236 & 0.0937 & 0.0208 \\ \hline 
\end{tabular}
 \begin{tablenotes} 
        \item[] \hspace{-8pt} \small{(Abbreviations: LAURIC.ACID (L.A); FA.7.0-OH\_1 (FA.7). $p$-values with $(\times)$ and $(*)$ indicate that the lipid specified by the row is selected by the method specified by the column under 0.05 and 0.10 FDR levels, respectively.)}\\
        \end{tablenotes}
    \end{threeparttable}
\end{table}
% The idea is that when only considering one mediator in the analysis,  the other ignored mediators can be viewed as unmeasured confounders.  
% To evaluate the effects of these unmeasured confounders/mediators, we apply the sensitivity analysis method proposed in \cite{imai2010identification}. 

Since the first screening step considers one mediator at a time,  we also conduct sensitivity analyses to evaluate the effects of the unadjusted mediators  similarly to \cite{liu2020large}. 
{We use the procedure  proposed by \cite{imai2010identification}, which  utilized the idea that  the error term in the M-S model and that in the Y-M model are likely to be correlated if the sequential ignorability assumption is violated and vice versa.} \label{page:adjust} 
The detailed results are provided in Section \ref{sec:sensanalysis} in the Supplementary Material. 
As a brief summary, the sensitivity analysis suggests that our  first  screening analysis could be robust to unadjusted  mediators.

\begin{figure}[!htbp]
 \captionsetup[subfigure]{labelformat=empty}
    \centering
      \begin{subfigure}[b]{0.32\textwidth}
      \caption{JS-AB}
      \includegraphics[width=0.95\textwidth]{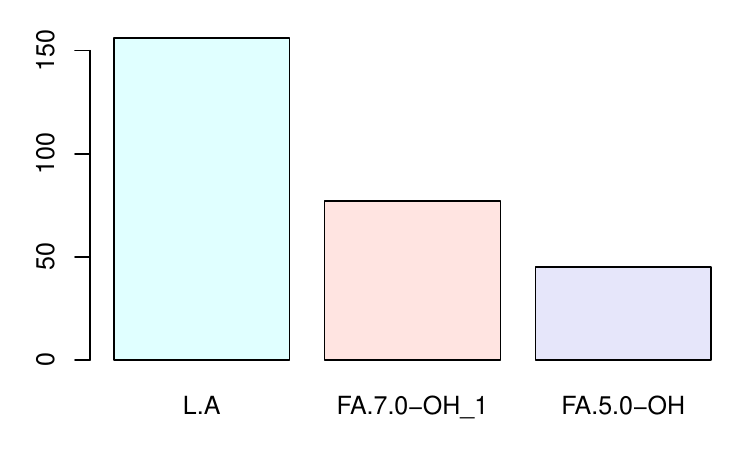}
   \end{subfigure}    
      \begin{subfigure}[b]{0.32\textwidth}
      \caption{JS-MaxP}
      \includegraphics[width=0.95\textwidth]{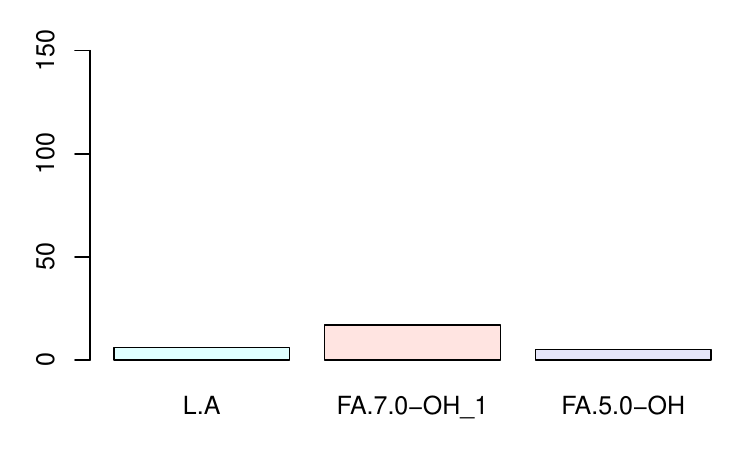}
   \end{subfigure}    
       \begin{subfigure}[b]{0.32\textwidth}
        \caption{CMA}
      \includegraphics[width=0.95\textwidth]{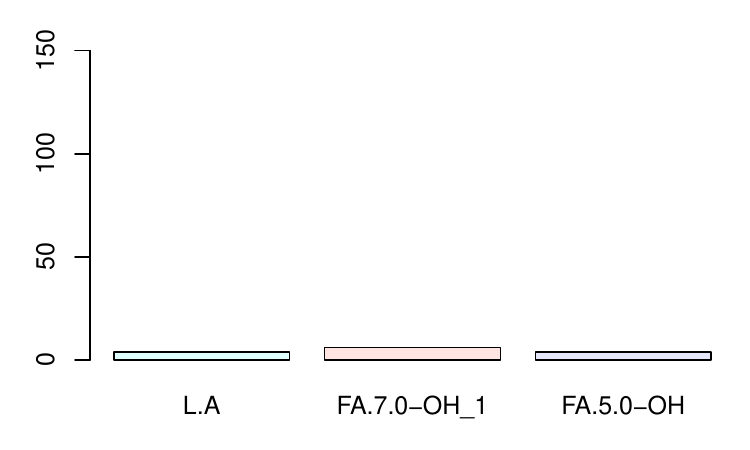}
   \end{subfigure}   
   \begin{subfigure}[b]{0.32\textwidth}
   \caption{PoC-AB}
      \includegraphics[width=0.95\textwidth]{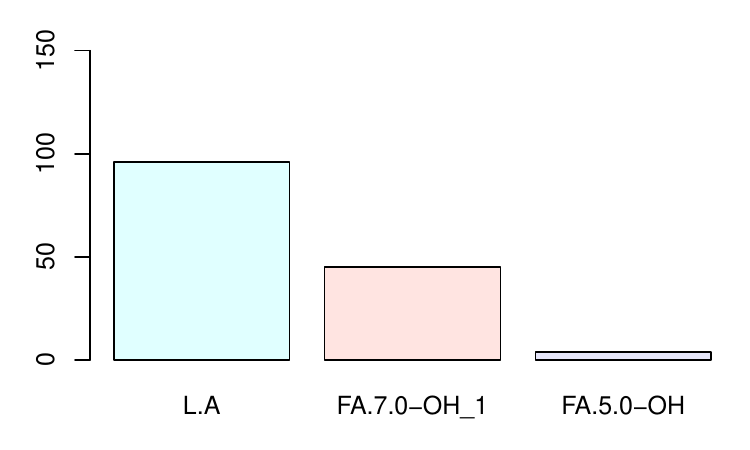}
   \end{subfigure}    
      \begin{subfigure}[b]{0.32\textwidth}
      \caption{PoC-B}
      \includegraphics[width=0.95\textwidth]{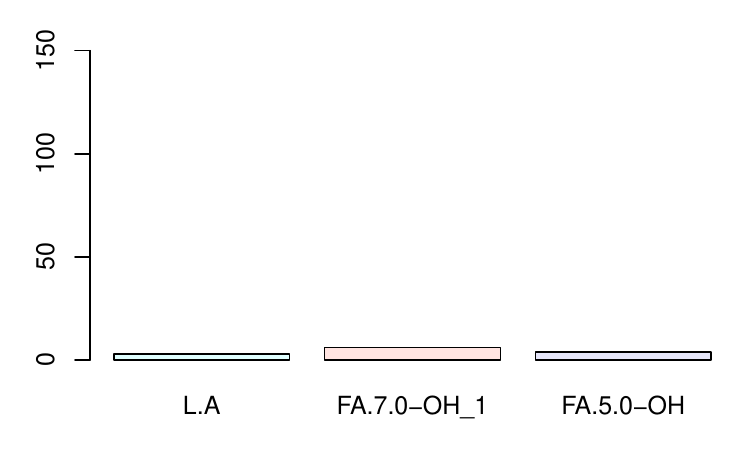}
   \end{subfigure}    
      \begin{subfigure}[b]{0.32\textwidth}
      \caption{PoC-Sobel}
      \includegraphics[width=0.95\textwidth]{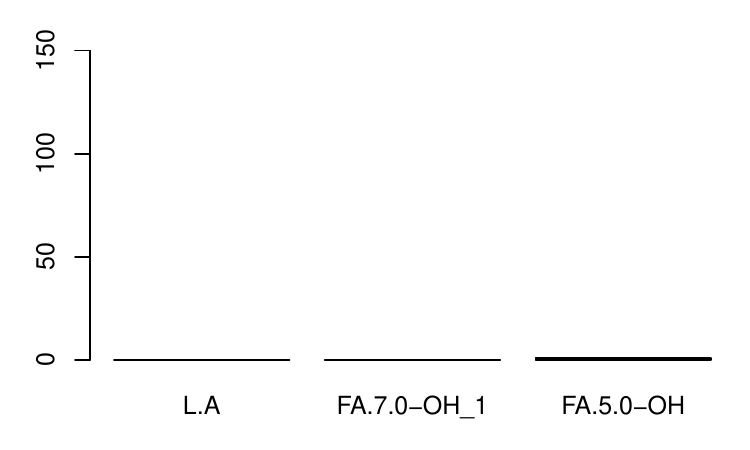}
   \end{subfigure}    
 \caption{Times of mediators being selected in Step 2 by the six tests with  FDR$=0.10$ over 400 random splits of the data. } \label{fig:countvalmain}
\end{figure}
% \begin{table} 
% \centering
% \caption{Detected  important lipid mediators based on $p$-values in Table \ref{tb:pvaljointmodel}.}\label{tb:jointdetectres}
% \begin{threeparttable}
% \begin{tabular}{c|cccccc}\hline
%  FDR level  & JS-AB &  JS-MaxP &  PoC-AB &   PoC-B & PoC-Sobel &  CMA  \\ \hline
%  0.05 & $\times$ & $\times$ & L.A & $\times$ & $\times$ & $\times$ \\ %\hline
%  0.10 & L.A \& F.18 & $\times$ & L.A & $\times$ & $\times$ & $\times$ \\  \hline
% % % & FA.18.0-DiC & NA & NA &  * & * & *\\ \hline 
% %  0.20 & L.A \& F.18 & $\times$ & L.A \& F.18 & $\times$  & $\times$ &  $\times$ \\\hline 
% % & FA.18.0-DiC &  &  FA.18.0-DiC &   &  & \\%\hline
% \end{tabular}
%  \begin{tablenotes} 
%         \item[] \hspace{-8pt} \small{(Abbreviations: LAURIC.ACID (L.A); FA.18.0-DiC (F.18); no mediators detected ($\times$).)}\\
%         \end{tablenotes}
%     \end{threeparttable}
% \end{table}

\vspace{-1em}
\section{Discussion}\label{sec:discussion}
This paper proposes a new adaptive framework  for testing composite null hypotheses in mediation pathway analysis. 
The method incorporates a consistent pre-test threshold into the bootstrap procedure, which helps  circumvent the non-regularity issue arising from the composite null hypotheses.
If at least one of the two coefficients is significant, the procedure would reduce to the classical nonparametric bootstrap; otherwise, it approximates the local asymptotic behavior of the statistics.   
 Our proposed strategy accommodates different types of null hypotheses under various models.     
Particularly, we have established similar results for both the individual and joint mediation effects under classical linear structural equation models,
and we have generalized the conclusions under generalized linear models. 
Through comprehensive simulation studies, 
we have demonstrated that the adaptive tests can properly and robustly control the type \RNum{1} error under different types of null hypotheses and improve the statistical power.

The proposed methodology offers an exemplary analytic toolbox that can be broadly extended to handle other problems of similar types involving composite null hypotheses. 
There are several interesting future research directions that are worth exploration.  
First, the non-regularity issue can similarly arise in other scenarios, such as  survival analysis  \citep{vanderweele2011causal,huang2016hypothesis},  different data types \citep{sohn2019compositional}, partially linear models  \citep{hines2021robust}\label{page:hinesref}, and models with exposure-mediator interactions; see more discussions in Section \ref{sec:partiallinear} of the Supplementary Material. 
% In our data analysis, it can be of interest to study  time-varying confounders 
% compositional data, and  
% different data types \citep{sohn2019compositional}.
% When there are time-varying covariates in the data, 
% \blue{longitudinal mediation analysis \citep{bind2016causal}}. 
These complicated models 
require special care in 
the causal interpretation of mediation effects and in the implementation of the bootstrap procedure, warranting further investigation. 
% which deserve further exploration.   
Second, 
when the dimension of mediators and covariates becomes high, 
it is of interest to extend the adaptive bootstrap under high-dimensional mediation models for both individual and joint mediation effects   
\citep{zhou2020estimation}. 
Similarly to our discussions on adjusting multivariate mediators at the end of Section \ref{sec:abt},
we might apply the adaptive bootstrap after properly adjusting high-dimensional covariates.  
In the data analysis, we have applied the marginal screening to reduce the dimension of mediators, which might potentially overlook the complicated causal dependence among mediators.    
When mediators have potential  causal dependence,     
\cite{shi2022testing} proposed to 
first estimate a directed acyclic graph of mediators and develop a testing procedure that can control the type \RNum{1} error to be less than or equal to the nominal level. 
It would be of interest to extend our proposed AB under such settings to mitigate potential conservatism. \label{pageref:shili} 
Third, the proposed AB strategy can also be utilized to examine the replicability across independent studies \citep{bogomolov2018assessing}, which is fundamental to scientific discovery.
Specifically, let $\beta_i, i=1,\ldots, K$, denote the true signals from $K$ independent studies, respectively.  
Testing whether the signals in these $K$ studies are all significant  corresponds to   $H_0: \prod_{i=1}^K \beta_i=0$ versus $H_A: \prod_{i=1}^K \beta_i\neq 0$. 
Moreover, for two studies with true signals $\beta_1$ and $\beta_2$, 
to investigate whether the effects of both studies are  significant in the same direction,
one can formulate
 the hypothesis testing problem as $H_0: \beta_1\beta_2 \leq 0$ versus $H_A: \beta_1\beta_2>0$. 
For these testing problems,  the null hypotheses are composite.
% of different scenarios.
To  properly control the type \RNum{1} error,
the adaptive strategy proposed in this paper may serve as a valuable building block, while additional effort is needed to analyze those different cases carefully. 
Last, in our data analysis, all measurements are obtained cross-sectionally at one given clinical visit within a time window of approximately three months. 
To further study potential long-term influences of toxicant exposures, it may be of interest to 
investigate how the mediation effects might vary over time. 
Such time-varying mediation effects may be naturally analyzed in the scenario of longitudinal studies that collect time-varying measurements. 
This is a very challenging research field with only minimal investigation in the current literature \citep{bind2016causal}. 
Extending the proposed AB method to analyze time-varying mediation effects would be a compelling future direction. 

\vspace{-0.2em}
\section*{Acknowledgments}
We are grateful to the joint editors, Dr. Daniela Witten and Dr. Aurore Delaigle, an associate editor, and three anonymous referees for their helpful comments and suggestions. 
This work is partially supported by  NSF DMS-1811734, DMS-2113564,  SES-1846747, SES-2150601,  NIH R01ES024732, 
R01ES033656, and Wisconsin Alumni Research Foundation. 

\textit{Conflict of interest:} None declared.

\section*{Supplementary Material}
We provide more results in the online Supplementary Material, including all the proofs, hyperparameter tuning, additional results on simulations and data analysis, and details of extensions (Joint Significance test, multiple-mediator settings, and non-linear models). 
For reproducibility of our results, we provide the R package and code on the GitHub repository:  \cite{abtestgit}.

\section*{Data Availability} 
 Due to privacy restrictions, we are unable to directly share the raw data publicly but they may be obtained offline according to a formal data request procedure outlined in the University of Michigan Data Use Agreement protocol. To satisfy the need of reproducibility, instead, we have introduced a pseudo-dataset with  added noise  on the GitHub repository: \cite{abtestgit}.

\small   
\bibliographystyle{rss} 
\bibliography{example}

\newpage

\newcommand{\PP}{\mathbb{P}_{n}}
\newcommand{\PPboot}{\mathbb{P}_{n}^*}
\newcommand{\GG}{\mathbb{G}_{n}}
\newcommand{\GGboot}{\mathbb{G}_{n}^*}

\appendix

\ 
\vspace{3em}

\begin{center}
\Large \textbf{Supplementary Material for\\ ``Adaptive Bootstrap Tests for Composite Null Hypotheses in the Mediation Pathway Analysis''}	
\end{center}

\vspace{2em} 

In this Supplementary Material, 
Section \ref{sec:defnotationconv} provides rigorous definitions of notation used in the main text and proofs.   
Section \ref{sec:jst} presents details of the Adaptive Bootstrap for the Joint Significance Test, mentioned on Page \pageref{page:adjst} of the main text.   
% Section \ref{sec:preliminary} provides notation and preliminary lemmas  to the following proofs. 
Section \ref{sec:allpfs} provides proofs of  Theorems \ref{thm:prodlimit2}--\ref{thm:bootstrapprodcomb} and Theorems 
\ref{thm:maxplimit1}--\ref{thm:bootstrapmaxp}. 
  Section \ref{sec:extendmultiplemediators} provides detailed theoretical and numerical results on multiple mediators  supplementary to Figure \ref{fig:multiplemed} and Section \ref{sec:extendmodel} in the main text.
  Section \ref{sec:discussiononglm}  provides detailed theoretical and numerical results on non-linear models supplementary to Sections \ref{sec:asymbinarymedout} and \ref{sec:modelbinmed}. 
Section \ref{sec:implement} discuss implementation details of a data-driven procedure for choosing the tuning parameter.   
Section \ref{sec:addnumres} provides additional numerical experiments and data analysis results. 
Section \ref{sec:partiallinear} provides a clarification of the partially linear model mentioned in Section \ref{sec:discussion} of the main text.

\section{{Definitions and Notation}}\label{sec:defnotationconv}

\paragraph{Notation on Convergence}
For two sequences of real numbers $\{a_n\}$ and $\{b_n\}$,
we write $a_n=o(b_n)$ if $\lim_{n\to \infty} a_n/b_n=0$. 
We say $a_n\to \infty$ as $n\to \infty$ if the value of $a_n$ can be arbitrarily large by taking $n$ to be sufficiently large. 
Given a sequence of random variables $\{X_n\}$ and a random variable $X$,
we let $X_n \xrightarrow{\mathrm{P}} X$ represent  the convergence in probability, i.e., 
 for any $\epsilon>0$, $\lim_{n\to \infty} \mathrm{P}(|X_n-X|>\epsilon)=0$.  
We let $X_n \xrightarrow{a.s.} X$ 
 represent the almost sure convergence, i.e., $ \mathrm{P}(\, \lim_{n\to \infty} X_n = X )=1$. 
We let $X_n \xrightarrow{d} X$ 
 represent the convergence in distribution, i.e., 
 $\lim_{n\to \infty} F_n(x) = F(x)$ for every $x$ at which $F(x)$ is continuous,
 where $F_n(x)$ and $F(x)$ represent the cumulative distribution functions of $X_n$ and $X$, respectively.

\paragraph{Bootstrap Consistency}
We next introduce 
the definition of bootstrap consistency relative to the Kolmogorov-Smirnov distance; also see,  Section 23.2 of \cite{van2000asymptotic}. 
Let $\hat{\theta}_n$ be an estimator of some parameter $\theta$ attached to the distribution $P_n$ of the observations.
Let $\mathbb{P}_n$ be an estimate of the underlying distribution $P_n$ of the observations. 
Let $\hat{\theta}_n^*$  denote the bootstrap estimator for  $\hat{\theta}_n$ and are obtained according to $\mathbb{P}_n$ in the same way $\hat{\theta}_n$ computed from the true observations with distribution $P_n$. 
We write 
$\sqrt{n}(\hat{\theta}_n^*-\hat{\theta}_n) \overset{d^*}{\leadsto} \sqrt{n}(\hat{\theta}_n - \theta)$ if 
\begin{align}\label{eq:bootconsistdef}
    \sup_{x\in \mathbb{R}} \left| P\left( \sqrt{n}(\hat{\theta}_n - \theta) \leqslant x \mid P_n \right) - P\left( \sqrt{n}(\hat{\theta}_n^*-\hat{\theta}_n)\leqslant x  \mid \mathbb{P}_n\right)  \right| \xrightarrow{\mathrm{P}} 0. 
\end{align}
In the following proofs, 
when the target $\sqrt{n}(\hat{\theta}_n - \theta)$ converges in distribution to a continuous distribution function $F$,
\eqref{eq:bootconsistdef} is also equivalent to that for every $x\in \mathbb{R}$, 
\begin{align*}
\text{if}\ \    P\left( \sqrt{n}(\hat{\theta}_n - \theta) \leqslant x \mid P_n \right) \to F(x),\quad\  \text{then}\quad P\left( \sqrt{n}(\hat{\theta}_n^*-\hat{\theta}_n)\leqslant x \mid \mathbb{P}_n \right) \xrightarrow{\mathrm{P}} F(x). 
\end{align*}
This type of consistency implies the asymptotic consistency of confidence intervals. 
Moreover, we let $\hat{\theta}_n^* \overset{\mathrm{P}^*}{\leadsto} \theta$ denote the convergence in conditional probability, i.e., for any $\epsilon>0$,
% conditioning on all paths of the observed observations with empirical measure  $\mathbb{P}_n$, 
\begin{align*}
  P(|\hat{\theta}_n^* - {\theta}|>\epsilon \mid \mathbb{P}_n ) \xrightarrow{a.s.} 0. 
\end{align*}

\newpage
\section{Adaptive Bootstrap  for the Joint Significance Test} \label{sec:jst} 

%\noindent \textit{(II) Joint Significance (JS) Test.} \ 
In addition to the PoC test, another popular class of methods is the joint significance test \citep{mackinnon2002comparison},
which is useful for combining a series of tests for a set of links in a causal chain \citep{judd1981estimating,baron1986moderator} such as a directed acyclic graph. 
Specifically, the JS test, also known as the MaxP test, rejects $H_0: \alphaS \betaM = 0$ if   $p_{\text{MaxP}}\equiv \max\{p_{\alphaS}, p_{\betaM}\} < \omega$,
where $\omega$ is a prespecified significance level,
and $p_{\alphaS}$  and $p_{\betaM}$ denote the $p$-values for $H_0: \alphaS=0$ (the link $\S \rightarrow \M$)  and $H_0: \betaM=0$ (the link $\M \rightarrow \Y$), respectively. 
Despite the ease of operation,  the MaxP test has also been found to be overly conservative under  $H_{0,3}$    \citep{barfield2017testing}. 
To see this analytically, 
note that when $p_{\alphaS}$ and $p_{\betaM}$ are asymptotically independent,
 under $H_{0,3}$, $\Pr(p_{\text{MaxP}}<\omega) \to  \Pr(p_{\alphaS}<\omega) \Pr(p_{\betaM}<\omega) \to \omega^2 <\omega$,
which suggests that  the MaxP test is always conservative under $H_{0,3}$ even if the sample size goes to infinity.

In this subsection, we focus on the adaptive bootstrap for the JS test.
% and develop the corresponding adaptive bootstrap test procedure. 
As  discussed in Section \ref{sec:review}, the popular MaxP test that rejects $H_0: \alphaS \betaM =0$ using the rule $p_{\mathrm{MaxP}}<\omega$  
can be conservative. 
%since  $p_{\mathrm{MaxP}}$ is not uniformly distributed under $H_{0,3}$ even if the sample size goes to infinity. 
To correctly evaluate the distribution of  $p_{\mathrm{MaxP}}$ under the composite null, 
we next develop the corresponding adaptive bootstrap test procedure for the JS test.

For ease of deriving bootstrap consistency, 
instead of directly examining  $p_{\mathrm{MaxP}}$,
%the maximum of two $p$-values  $p_{\mathrm{MaxP}}$, 
we investigate the corresponding statistic based on the t-statistics, which usually have asymptotic normality. 
%standardized statistic, which usually has asymptotic normality. 
%is usually asymptotically normal.  
%In particular, define 
%\begin{align*}
%	h(t_1, t_2)=\biggr( \myindicator\biggr\{\underset{k=1,2}{\operatorname{argmin}}\, t_k^2=1\biggr\}, ~\myindicator\biggr\{\underset{k=1,2}{\operatorname{argmin}}\, t_k^2=2 \biggr\} \biggr)^{\mytrans}\ \text{ and }\ ~ \tilde{h}(t_1,t_2)=1-h(t_1, t_2).	
%\end{align*}
%Therefore, when , $ h(|T_{\alpha,n}|, |T_{\beta,n} |  ) = \tilde{h}(p_{\alpha}, p_{\beta})$
Specifically, we introduce the statistic $\sqrt{n}\hatmaxpn = H(T_{\alpha, n}, T_{\beta, n})$, 
where $T_{\alpha,n}=\sqrt{n}\hatalphaSn/\hatsigmaalphan $ and $ T_{\beta, n}=\sqrt{n}\hatbetaMn /\hatsigmabetan$ are the standardized statistics of the two coefficients, respectively, and 
\begin{align*}
H(t_1,t_2)=(t_1, t_2) h(t_1, t_2)\   \text{ with }\ 	h(t_1, t_2)=\biggr( \myindicator\biggr\{\underset{k=1,2}{\operatorname{min}}\,|t_k|=|t_1|\biggr\}, ~\myindicator\biggr\{\underset{k=1,2}{\operatorname{min}}\, |t_k|=|t_2| \biggr\} \biggr)^{\mytrans}. 
\end{align*} 
%\begin{align*}
%H(t_1,t_2)=(t_1, t_2) h(t_1, t_2)\   \text{ with }\ 	h(t_1, t_2)=\biggr( \myindicator\biggr\{\underset{k=1,2}{\operatorname{argmin}}\, t_k^2=1\biggr\}, ~\myindicator\biggr\{\underset{k=1,2}{\operatorname{argmin}}\, t_k^2=2 \biggr\} \biggr)^{\mytrans}. 
%\end{align*} 
%\begin{align*}
%H(t_1,t_2)=\sum_{k=1}^2 t_k\times  \myindicator\Big\{\underset{j=1,2}{\operatorname{argmin}}\, t_j^2 = k \Big\}.	
%\end{align*}
%Specifically, we introduce the statistic 
%\begin{align*}
%\sqrt{n}\hatmaxpn = H(T_{\alpha, n}, T_{\beta, n}), \quad \text{ where }\quad  	H(t_1,t_2)=\sum_{k=1}^2 t_k\times  \myindicator\Big\{\underset{j=1,2}{\operatorname{argmin}}\, t_j^2 = k \Big\},
%\end{align*}
%and $T_{\alpha,n}=\sqrt{n}\hatalphaSn/\hatsigmaalphan $ and $ T_{\beta, n}=\sqrt{n}\hatbetaMn /\hatsigmabetan$ are the standardized statistics of two coefficients, respectively. 
%which  is a two-dimensional vector of  indicator functions. 
% $H(t_1,t_2)=(t_1, t_2) h(t_1, t_2)$, and $h(t_1, t_2)=( \myindicator\{\arg\min t_k^2=1\}, ~\myindicator\{\arg \min t_k^2=2 \} )^{\mytrans}$ is a two-dimensional vector of  indicator functions. 
With this construction, the value of $\sqrt{n}\hatmaxpn$ equals the one in $\{T_{\alpha, n}, T_{\beta, n}\}$ that has a smaller absolute value, and $|\sqrt{n}\hatmaxpn| = \min\{|T_{\alpha, n}|, |T_{\beta, n}| \}$.
When  $T_{\alpha, n}$ and $T_{\beta, n}$ are asymptotically normal, 
%\red{ $h( |T_{\alpha, n}|, |T_{\beta, n}|)=(\myindicator\{p_{\mathrm{MaxP}}=p_{\alpha}\},  \myindicator\{p_{\mathrm{MaxP}}=p_{\beta}\})^{\mytrans}$, thus 
 $\sqrt{n}\hatmaxpn$ equals the t-statistic 
% of the coefficient 
 whose asymptotic $p$-value equals $p_{\mathrm{MaxP}}$.  
%When  $T_{\alpha, n}$ and $T_{\beta, n}$ are asymptotically normal, 
%%the limiting distributions of $T_{\alpha, n}$ and $T_{\beta, n}$ are symmetric, 
% $\sqrt{n}\hatmaxpn$ is the t-statistic of the coefficient whose asymptotic $p$-value equals $p_{\mathrm{MaxP}}$; 
% this follows from $h( |T_{\alpha, n}|, |T_{\beta, n}|)=(\myindicator\{p_{\mathrm{MaxP}}=p_{\alpha}\},  \myindicator\{p_{\mathrm{MaxP}}=p_{\beta}\})$. 
% to see this, $\Phi(\sqrt{n}\hatmaxpn)=\max\{\Phi(|T_{\alpha, n}|), \Phi(|T_{\beta, n}|) \}= \max\{p_{\alphaS}, p_{\betaM}\} $, where $\Phi(\cdot)$ denotes the cumulative distribution of the standard normal distribution. 
% \begin{align*}
%\Phi(\sqrt{n}\hatmaxpn)=\max\{\Phi(|T_{\alpha, n}|), \Phi(|T_{\beta, n}|) \}= \max\{p_{\alphaS}, p_{\betaM}\} 	
% \end{align*} 
%that corresponds to the larger one in $\Phi(\sqrt{n}\hatmaxpn)\{p_{\alphaS}, p_{\betaM}\}$.  
% that has larger asymptotic $p$-value. 
%Therefore, 
This equivalence motivates us to further derive 
a valid and non-conservative $p$-value  
for $\sqrt{n}\hatmaxpn$ in the JS test. 
We correspondingly define the centering parameter for $\hatmaxpn$ as $\maxpnull=H(\alphaS/\sigmaalpha, \betaM/ \sigmabeta )$. 
Particularly, $\maxpnull$ follows the same form as $\hatmaxpn$ but replacing $(\hatalphaSn, \hatbetaMn, \hatsigmaalphan, \hatsigmabetan )$ in $\hatmaxpn$ with their population versions $(\alphaS, \betaM, \sigmaalpha, \sigmabeta )$, 
% true values $(\alphaS, \betaM, \sigmaalpha, \sigmabeta )$, 
and $\maxpnull=0$ under the composite null  \eqref{eq:compositenull}.  
% We will show in this paper 
%\blue
{Simulation studies in Section \ref{sec:sim}  show that directly applying the classical nonparametric bootstrap to $\sqrt{n}(\hatmaxpn-\theta_0)$  fails to provide proper type \RNum{1} error control.
We next analytically unveil the non-standard limiting behavior of $\sqrt{n}(\hatmaxpn-\theta_0)$  before introducing our adaptive bootstrap test.}

\bigskip
\noindent {\textit{Non-Regularity of the JS Test.}}\quad   
The non-regular limiting behavior of $\hatmaxpn$ is caused by 
the nonuniform convergence of the indicator vector $h(T_{\alpha, n}, T_{\beta, n})$ under different types of nulls.
% Particularly, u
Under $H_{0,1}$ or $H_{0,2}$, 
$h(T_{\alpha, n}, T_{\beta, n})$ converges to $h(\alphaS/\sigmaalpha, \betaM/\sigmabeta )$ consistently, and $\hatmaxpn$ is asymptotically normal. 
However, under $H_{0,3}$, $h(T_{\alpha, n}, T_{\beta, n})$ does not converge to $h(\alphaS/\sigmaalpha, \betaM/\sigmabeta )$,
and  $\hatmaxpn$ does not have a normal limit. 
More specifically, we characterize the limiting distribution of $\sqrt{n}(\hatmaxpn-\maxpnull)$ considering a  special case of \eqref{eq:fullmod1}: 
$\M = \alphaS \S +\eM,$ and $\Y = \betaM \M +  \eY$,  and assuming $\sigmaalpha=\sigmabeta=1$.
Under mild conditions, 
by the strong law of large numbers, 
$(\hatalphaSn, \hatbetaMn, \hatsigmaalpha,  \hatsigmabeta) \xrightarrow{a.s.} (\alphaS, \betaM, 1, 1)$.  
% where $ \xrightarrow{a.s.}$ represents the almost sure convergence. 
Then  we can write
\begin{align}
\sqrt{n}(\hatmaxpn-\maxpnull)=&~\sqrt{n}(\alphaS, \betaM)\times \big\{ h(\hatalphaSn , \hatbetaMn ) -h(\alphaS, \betaM ) \big\}	 \label{eq:jstsimple} \\
&~+\sqrt{n}\big( \hatalphaSn - \alphaS, \hatbetaMn - \betaM \big)\times h(\hatalphaSn , \hatbetaMn) + o_{P_n}(1),\notag
\end{align}
where $o_{P_n}(1)$ represents a small order term converging to 0 in probability. 
% We point out that $h(t_1,t_2)$ is continuous at $(t_1,t_2)$ if $\arg \min t_k^2$ is unique. 
% Therefore,
{Under $H_{0,1}$ or $H_{0,2}$, (or more generally, when $|\alphaS|\neq |\betaM|$), 
% we note that
$h(t_1,t_2)$ is continuous at $(t_1,t_2)=(\alphaS, \betaM)$ by the construction of $h(t_1,t_2)$; 
% ,
% by the fact that $h(t_1,t_2)$ is continuous at $(t_1,t_2)$ if 
% $|t_1|\neq |t_2|$; 
% $\arg \min_{k=1,2} t_k^2$ is unique;
the continuous mapping theorem then implies that $h(\hatalphaSn , \hatbetaMn) \xrightarrow{a.s.} h(\alphaS, \betaM ).$ 
Under $H_{0,3}$, we have  $h(\hatalphaSn , \hatbetaMn)=h(\sqrt{n}\hatalphaSn, \sqrt{n}\hatbetaMn )$,
with $\sqrt{n}(\hatalphaSn, \hatbetaMn ) \xrightarrow{d} (Z_{\S,0}, \, Z_{\M,0})$ by \eqref{eq:simplecltsep}. 
%with $\sqrt{n}(\hatalphaSn, \hatbetaMn ) \xrightarrow{d} (Z_{\S,0}/V_{\S,0}, \, Z_{\M,0}/V_{\M,0})$ by \eqref{eq:simplecltsep}. 
% In summary, by \eqref{eq:simplecltsep} and 
With these results, by \eqref{eq:jstsimple} and Slutsky's lemma, 
we have $\sqrt{n}(\hatmaxpn - \maxpnull ) \xrightarrow{d} U_{2,0}$, where}
\begin{align}\label{eq:u20def}
U_{2,0}= \begin{cases}
\displaystyle (Z_{\S,0},Z_{\M,0}) \times h(\alphaS, \betaM), &  \ \text{if }\, |\alphaS|\neq |\betaM|; \\[10pt]
\displaystyle (Z_{\S,0},Z_{\M,0}) \times h (Z_{\S,0},Z_{\M,0}), &  \ \text{if }\, (\alphaS, \betaM) = (0,0).
\end{cases}
\end{align} 
We can see that  the distribution of $\sqrt{n}(\hatmaxpn - \maxpnull)$ does not converge uniformly with respect to $(\alphaS, \betaM),$ and the nonuniformity occurs at the neighborhood of $(\alphaS, \betaM)=(0, 0).$   
This discontinuity phenomenon leads to the failure of classical nonparametric bootstrap, which was also  demonstrated by the simulation studies  in Section \ref{sec:sim}.

\bigskip
\noindent \textit{Adaptive Bootstrap of the JS Test.}\quad 
Similar to Section \ref{sec:abt}, 
to develop a consistent bootstrap procedure, we need to understand the limiting behavior of $\sqrt{n}(\hatmaxpn - \maxpnull )$ %that is uniform
in a local neighborhood of $(\alphaS, \betaM)=(0,0)$. 
To achieve this, 
%similarly to Section \ref{sec:prodcoef}, 
again we consider  the local linear SEM  \eqref{eq:fullmodlocal2}, 
where we recall that  $\alphaSn = \alphaS + n^{-1/2}\balpha$ and $\betaMn = \betaM + n^{-1/2}\bbeta.$
Then we correspondingly define the centering parameter under the local linear SEM as
$\maxpnulln = H( \alphaSn/\sigmaalpha, \betaMn/\sigmabeta )$. 
Note that  
$\maxpnulln$ takes the same form as $\maxpnull$, except that $\alphaS$ and $\betaM$ are replaced by $\alphaSn$ and $\betaMn$, respectively.
We present the limiting distribution of $\sqrt{n}(\hat{\theta}_n -\maxpnulln)$ under Model \eqref{eq:fullmodlocal2} in  the following theorem.

\begin{theorem} \label{thm:maxplimit1}
Assume Condition \ref{cond:inversemoment} holds  and   $|\alphaS/\sigmaalpha| \neq |\betaM/ \sigmabeta|$ when $(\alphaS, \betaM) \neq (0,0)$. 
Then, under the local linear SEM \eqref{eq:fullmodlocal2}, $\sqrt{n}(\hat{\theta}_n -\maxpnulln)	\xrightarrow{d} U_{2}$, as $n \rightarrow \infty$ with 
\begin{eqnarray*}
U_{2} =   \begin{cases}
\displaystyle \biggr(\frac{Z_{\S}}{\sigma_{\alphaS}}, \frac{Z_{\M}}{ \sigma_{\betaM}} \biggr)\times h\biggr(\frac{\alphaS}{\sigmaalpha }, \frac{\betaM}{\sigmabeta} \biggr), &\ \text{if }\, (\alphaS, \betaM) \neq (0,0); \\[12pt]
\displaystyle H (K_{b,\S},K_{b,\M}) -H\biggr(\frac{\balpha}{ \sigmaalpha }, \frac{\bbeta}{\sigmabeta} \biggr),
 & \ \text{if }\, (\alphaS, \betaM) = (0,0), \\[8pt]
 %\displaystyle  (K_{b,\S},K_{b,\M})\times h (K_{b,\S},K_{b,\M}) -\biggr(\frac{\balpha}{ \sigmaalpha }, \frac{\bbeta}{\sigmabeta} \biggr) h \biggr(\frac{\balpha}{ \sigmaalpha }, \frac{\bbeta}{\sigmabeta} \biggr)
% & \ \text{if }\, (\alphaS, \betaM) = (0,0). \\[8pt]
%~ \displaystyle -\biggr(\frac{\balpha}{ \sigmaalpha }, \frac{\bbeta}{\sigmabeta} \biggr) h \biggr(\frac{\balpha}{ \sigmaalpha }, \frac{\bbeta}{\sigmabeta} \biggr)
\end{cases}	
\end{eqnarray*}
where $(Z_{\S},  Z_{\M})$ are defined the same as in Theorem \ref{thm:prodlimit2}, and  
%where $(Z_{\S}, V_{\S}, Z_{\M},  V_{\M})$ are defined the same as in Theorem \ref{thm:prodlimit2}, and  
%\begin{align*}
%	K_{b,\S}= \frac{\balpha +Z_{\S}/V_{\S}}{\sigma_{\alphaS}}, \hspace{3em} K_{b,\M}= \frac{\bbeta +Z_{\M}/V_{\M}}{\sigma_{\betaM}}.
%\end{align*}
\begin{align*}
	K_{b,\S}= \frac{\balpha+Z_{\S}}{\sigma_{\alphaS}}, \hspace{3em} K_{b,\M}= \frac{\bbeta+Z_{\M}}{\sigma_{\betaM}}.
\end{align*}
\end{theorem}

The assumption $|\alphaS/\sigmaalpha| \neq |\betaM/ \sigmabeta|$ when $(\alphaS, \betaM) \neq (0,0)$ is satisfied under the composite null \eqref{eq:compositenull},
and is made mainly for the technical simplicity in the proof. 
When $(\alphaS, \betaM)=(0,0)$, 
% and under the local linear model \eqref{eq:fullmodlocal2},
% $(\balpha/\sigmaalpha, \bbeta/\sigmabeta )=( \sqrt{n}\alphaS/\sigmaalpha, \sqrt{n}\betaM/\sigmabeta) $, and 
%  we equivalently write 
 $H (K_{b,\S},K_{b,\M}) - H({\balpha}/{ \sigmaalpha }, {\bbeta}/{\sigmabeta} )$ in Theorem \ref{thm:maxplimit1} can be equivalently written as
 \begin{align*}
  \biggr(\frac{Z_{\S}}{\sigma_{\alphaS}}, \frac{Z_{\M}}{\sigma_{\betaM}} \biggr) h(K_{b,\S}, K_{b,\M} )+ \biggr(\frac{\balpha}{ \sigmaalpha }, \frac{\bbeta}{\sigmabeta} \biggr) \biggr\{h(K_{b,\S}, K_{b,\M} ) -h\biggr(\frac{\balpha}{ \sigmaalpha }, \frac{\bbeta}{\sigmabeta} \biggr)\biggr\}. 
 \end{align*}
Comparing the above expression to the form of $U_2$ when $(\alphaS, \betaM )\neq (0,0)$, 
we can see $U_2$ is discontinuous with respect to the parameters $(\alphaS, \betaM)$.
On the other hand, at the $\sqrt{n}$-neighborhood of  $(\alphaS, \betaM)=(0,0)$, Theorem \ref{thm:maxplimit1} further shows that the limiting distribution of  $\sqrt{n}(\hatmaxpn - \maxpnulln)$ is  continuous as a function of $(\balpha, \bbeta )^{\mytrans}\in \realnumber^2$ into the space of distribution functions. 
Therefore, the non-regularity at $(\alphaS, \betaM)=(0,0)$ can be  explained by the dependence 
of the limiting distribution on the (nonidentifiable) local parameters $(\balpha, \bbeta)$.
%Similarly to the literature, 
Similarly to Section \ref{sec:abt},  
we expect that developing  a bootstrap test using the local asymptotic results in Theorem \ref{thm:maxplimit1} will improve the finite-sample accuracy,    
%On the other hand, n
whereas the classical bootstrap methods that do not take into account the local asymptotic behaviors will fail to provide consistent estimates of the distribution of $\sqrt{n}(\hatmaxpn - \maxpnulln)$.

Similarly to the PoC test in Section \ref{sec:abt}, 
to overcome the local discontinuity issue,  we isolate the possibility that $(\alphaS, \betaM )\neq (0,0)$ by comparing the standardized statistics $|T_{\alpha,n}|$ and $|T_{\beta,n}|$ to certain threshold $\lambda_n$. Specifically, we decompose 
%\begin{align}
%\sqrt{n}(\hatmaxpn - \maxpnulln )= \sqrt{n}(\hatmaxpn - \maxpnulln)\times(1-\myindicator_{\alphaS, \betaM, \lambda_n } )+\sqrt{n}(\hatmaxpn - \maxpnulln)\times \myindicator_{\alphaS, \betaM, \lambda_n }, \label{eq:decompform}
%\end{align}
%where $\myindicator_{\alphaS, \betaM, \lambda_n }$ is defined same as in Section \ref{sec:prodcoef}. 
\begin{align}
\sqrt{n}(\hatmaxpn - \maxpnulln )= \sqrt{n}(\hatmaxpn - \maxpnulln)\times(1-\myindicator_{\alphaS, \lambda_n}\myindicator_{\betaM, \lambda_n} )+\sqrt{n}(\hatmaxpn - \maxpnulln)\times \myindicator_{\alphaS, \lambda_n}\myindicator_{\betaM, \lambda_n}, \label{eq:decompform}
\end{align}
where $\myindicator_{\alphaS, \lambda_n}$ and $\myindicator_{\betaM, \lambda_n}$ are defined same as  those in \eqref{eq:proddecomp1}. 
%follow the same definitions as those in \eqref{eq:proddecomp1}.
% in Section \ref{sec:prodcoef}. 
% $\myindicator_{\alphaS, \betaM, \lambda_n }$ is defined same as in Section \ref{sec:prodcoef}. 
When $(\alphaS, \betaM) \neq (0,0)$ 
and $|\alphaS/\sigmaalpha| \neq |\betaM/ \sigmabeta|$, 
the classical nonparametric bootstrap estimator $\sqrt{n}(\hatmaxpn^*- \hatmaxpn)$, where  $\sqrt{n}\hatmaxpn^*=H(T_{\alpha, n}^*, T_{\beta, n}^*)$, is consistent for the first term $\sqrt{n}(\hatmaxpn-\maxpnulln)$ in \eqref{eq:decompform}. 
For the second term in \eqref{eq:decompform},  
we have $(\alphaS, \betaM) = (0,0)$  and can write 
 $\sqrt{n}(\hatmaxpn - \maxpnulln)=\mathbb{R}_{2,n}(\balpha, \bbeta )$,   
where $ \mathbb{R}_{2,n}(\balpha, \bbeta )= H(\mathbb{K}_{b,\S}, \mathbb{K}_{b,\M})-H({\balpha}/{ \sigmaalpha }, {\bbeta}/{\sigmabeta}),$  
\begin{align*}
\mathbb{K}_{b,\S}=\frac{ \balpha+\mathbb{Z}_{\S,n}}{\hatsigmaalphan }, \hspace{2em} \mathbb{K}_{b,\M}=\frac{\bbeta+\mathbb{Z}_{\M,n}}{\hatsigmabetan },	
\end{align*}
and $(\mathbb{Z}_{\S,n}, \mathbb{Z}_{\M,n})$ are defined same as those in Section \ref{sec:abt}. 
It follows that all parts of $\mathbb{R}_{2,n}(\balpha, \bbeta )$ can be viewed as smooth functions of $\mathbb{P}_n$. 
Similarly to Section \ref{sec:abt}, it is reasonable to expect that a consistent bootstrap can
be constructed using  the nonparametric bootstrap version of $\mathbb{R}_{2,n}(\balpha, \bbeta )$. 
Specifically, we define
$\mathbb{R}_{2,n}^*(\balpha, \bbeta)=H(\mathbb{K}_{b,\S}^*, \mathbb{K}_{b,\M}^*)-H({\balpha}/{\hatsigmaalphan}, {\bbeta}/{\hatsigmabetan} ),$
where 
\begin{align*}
	\mathbb{K}_{b,\S}^*=\frac{ \balpha+\mathbb{Z}_{\S,n}^*}{ \hatsigmaalphan^*},  	\hspace{2em}\mathbb{K}_{b,\M}^*=\frac{ \bbeta+\mathbb{Z}_{\M,n}^*}{ \hatsigmabetan^*}. \notag 
\end{align*}
We are ready to develop our adaptive bootstrap test based on \eqref{eq:decompform}.
Similarly to Section \ref{sec:abt}, 
we replace the indicators $\myindicator_{\alphaS, \lambda_{n}}$ and $\myindicator_{\betaM, \lambda_{n}}$  in \eqref{eq:decompform} with $\myindicator_{\alphaS, \lambda_{n}}^*$ and $\myindicator_{\betaM, \lambda_{n}}^*$ in \eqref{eq:indicatoralphabetasep}, respectively.
Then we 
 define our proposed  adaptive bootstrap test statistic as 
\begin{align*}
U_2^*= \sqrt{n}(\hatmaxpn^* - \hatmaxpn) \times (1- \myindicator_{\alphaS, \lambda_n}^* \myindicator_{\betaM, \lambda_n}^* ) + \mathbb{R}_{2,n}^*(\balpha, \bbeta) \times \myindicator_{\alphaS, \lambda_n}^* \myindicator_{\betaM, \lambda_n}^*. 
\end{align*} 
Theorem \ref{thm:bootstrapmaxp} below proves that $\sqrt{n}(\hatmaxpn - \maxpnulln )$ can be consistently bootstrapped using the $U_2^*$ above.

\begin{theorem}\label{thm:bootstrapmaxp}
Assume that the conditions in Theorem \ref{thm:maxplimit1} are satisfied, 
and that the tuning parameter $\lambda_n$ satisfies $\lambda_n=o(\sqrt{n})$ and $\lambda_n\to \infty$ as $n\to \infty.$
%When the tuning parameter $\lambda_n$ satisfies $\lambda_n=o(\sqrt{n})$ and $\lambda_n\to \infty$ as $n\to \infty,$
Then under the local {linear SEM} \eqref{eq:fullmod1}, 
conditionally on the data, 
the adaptive test statistic $U_2^*\overset{d^*}{\leadsto} \sqrt{n}(\hatmaxpn - \maxpnulln )$.
% converges to the limiting distribution of $\sqrt{n}(\hatmaxpn - \maxpnulln )$  in probability. 
\end{theorem}

Theorem \ref{thm:bootstrapmaxp}  establishes the bootstrap consistency of $U_2^*$ for $\sqrt{n}(\hatmaxpn - \maxpnulln )$,
which is different from Theorem \ref{thm:bootstrapprodcomb} that derives  $nU_1^*$ and $\sqrt{n}U_1^*$ separately for $n(\hatalphaS \hatbetaM - \alphaSn \betaMn) $ and $\sqrt{n}(\hatalphaS \hatbetaM - \alphaSn \betaMn) $ when $(\alphaS,\betaM)=(0,0)$ and $(\alphaS,\betaM)\neq (0,0)$.    
This difference is attributed to the distinct   non-regular limiting behaviors in  the PoC test and the JS test. 
Despite this difference, based on Theorem \ref{thm:bootstrapmaxp},
we can develop an adaptive bootstrap test procedure 
for the JS test 
similarly to  that in Section \ref{sec:abt}. 
In particular, 
under the composite null \eqref{eq:compositenull}, 
we consider  $(\balpha, \bbeta)= (0,0)$ in $U_2^*$ to estimate  the distribution of $\sqrt{n}(\hatmaxpn^*-\hatmaxpn)$. 
For a given nominal level $\omega$, we redefine $q_{\omega/2}$ and $q_{1-\omega/2}$ as the lower and upper $\omega/2$ quantiles, respectively, of $U_{2}^*$. 
If $\sqrt{n}\hatmaxpn$ falls outside  the interval $(q_{\omega/2}, q_{1-\omega/2} )$, then we reject the composite null, and conclude that the mediation effect is statistically significant.  
%In the adaptive test for the joint significance test, 
In addition, 
we also choose the   value of the tuning parameter $\lambda_n$ such that $\lambda_n=o(\sqrt{n})$ and $\lambda_n\to \infty$ as $n\to \infty$ following  the discussions in Section \ref{sec:abt}.   
%and we allow two different values to be used in the two indicators in $U_2^*$. 
Moreover, similarly to Remark \ref{rm:compare1}, 
we  emphasize that the 
 analysis of the JS test statistic  is distinct from the  existing literature as we consider different problem settings and testing a composite null hypothesis, which necessitates new theoretical and methodological developments.

\section{Proofs of Theorems \ref{thm:prodlimit2}--\ref{thm:bootstrapprodcomb} and Theorems 
\ref{thm:maxplimit1}--\ref{thm:bootstrapmaxp}}\label{sec:allpfs}

In this section, we develop preliminary results for the follow-up proofs in Section \ref{sec:preliminary}.
We provide proofs of Theorems \ref{thm:prodlimit2}, \ref{thm:bootstrapprodcomb}, \ref{thm:maxplimit1}, and \ref{thm:bootstrapmaxp}
in Sections \ref{sec:pfthm1}--\ref{sec:pfbootstrapmaxp}, respectively. 

\subsection{Preliminary for the Proofs} \label{sec:preliminary} 
In the following proofs, we use the variables defined as: 
\begin{align} 
&\S_{\perp} = \S - \X^{\mytrans}Q_{1,\S}, \quad \quad \quad  \M_{\perp} = \M - \X^{\mytrans} Q_{1,\M},  \label{eq:newdefvariables} \\
&\M_{\perp'} = \M - \tilde{\X}^{\mytrans} Q_{2,\M}, \quad \quad \Y_{\perp'} = \Y - \tilde{\X}^{\mytrans}Q_{2,\Y}, \notag 
\end{align} 
where $\tilde{\X}=(\X^{\mytrans},\S)^{\mytrans}$, 
%\begin{align*}
%	& Q_{1,\S}= \{\mathrm{E}(\X \X^{\mytrans})\}^{-1} \mathrm{E}(\X \S), \quad \quad Q_{1,\M} = \{\mathrm{E}(\X \X^{\mytrans})\}^{-1} \mathrm{E}(\X \M), \\
%	&  Q_{2,\M} = \{\mathrm{E}( \tilde{\X}  \tilde{\X}^{\mytrans})\}^{-1} \mathrm{E}( \tilde{\X} \M), \quad \quad Q_{2,\Y}= \{\mathrm{E}( \tilde{\X}  \tilde{\X}^{\mytrans})\}^{-1} \mathrm{E}( \tilde{\X} \Y)
%\end{align*}
\begin{align}
	&Q_{1,\M} = \{P_n(\X \X^{\mytrans})\}^{-1} P_n(\X \M), \quad \quad Q_{1,\S}= \{P_n(\X \X^{\mytrans})\}^{-1} P_n(\X \S), \label{eq:popqmomentdef} \\
	&  Q_{2,\M} = \{P_n( \tilde{\X}  \tilde{\X}^{\mytrans})\}^{-1} P_n( \tilde{\X} \M), \quad \quad Q_{2,\Y}= \{P_n( \tilde{\X}  \tilde{\X}^{\mytrans})\}^{-1} P_n( \tilde{\X} \Y).  \notag
\end{align}
% with $\tilde{\X}=(\X^{\mytrans},\S)^{\mytrans}.$
We mention that $\Sperp$ and $\Mperpp$ are defined same as those in Condition \ref{cond:inversemoment},
where both $P_n(\cdot)$ used in \eqref{eq:popqmomentdef} and  $\mathrm{E}(\cdot)$ used in Condition \ref{cond:inversemoment} denote the expectation with respect to the distribution of $(\S,\M,\X,\Y)$. 
%is the same as $\mathrm{E}(\cdot)$ used in Condition \ref{cond:inversemoment}. 
% where $P_n(\cdot)$ is the same as $\mathrm{E}(\cdot)$ used in Condition \ref{cond:inversemoment}. 
%Correspondingly,
Based on \eqref{eq:newdefvariables}, for each index $i \in \{1,\ldots, n\}$, we define $\Sperpi =\S_i-\X_i^{\mytrans}Q_{1,\S}$, 
and also define $(\Mperpi, \Mperppi, \Yperppi)$ similarly.
% and $(\M_{r,i}, \M_{g,i},\Y_{g,i})$ similarly.
%The corresponding sample observations are $\S_{g,i}, \M_{r,i}, \M_{g,i}$ and $\Y_{g,i}$.
%By replacing $P_n(\cdot)$ with $\mathbb{P}_n(\cdot)$ in \eqref{eq:popqmomentdef}, 
Moreover, for the definitions in \eqref{eq:newdefvariables} and \eqref{eq:popqmomentdef},  
by replacing $P_n(\cdot)$ with $\mathbb{P}_n(\cdot)$,
we similarly define $(\hat{Q}_{1,\M}, \hat{Q}_{1,\S}, \hat{Q}_{2,\M}, \hat{Q}_{2,S})$, $(\hatSperp, \hatMperp, \hatMperpp, \hatYperpp)$, and $\{(\hatSperpi, \hatMperpi, \hatMperppi, \hatYperppi):  i =1,\ldots, n\}$. 
%%We then define 
%%the sampled version as
%\begin{align}
%& \hat{\S}_{g} = \S - \X^{\mytrans}\hat{Q}_{1,\S}, \quad \quad \quad  \hat{\M}_{r} = \M - \X^{\mytrans} \hat{Q}_{1,\M}, \label{eq:samplevardef} \\ 
%& \hat{\M}_{g} = \M - \tilde{\X}^{\mytrans} \hat{Q}_{2,\M}, \quad \quad  \hat{\Y}_{g} = \Y - \tilde{\X}^{\mytrans}\hat{Q}_{2,\Y}, \notag 
%\end{align}
%%\begin{align}
%%& \hat{\S}_{g,i} = \S_i - \X_i^{\mytrans}\hat{Q}_{1,\S}, \quad \quad \quad  \hat{\M}_{r,i} = \M_i - \X_i^{\mytrans} \hat{Q}_{1,\M}, \notag \\ %\label{eq:newdefvariables}
%%& \hat{\M}_{g,i} = \M_i - \X_i^{\mytrans} \hat{Q}_{2,\M}, \quad \quad  \hat{\Y}_{g,i} = \Y_i - \X_i^{\mytrans}\hat{Q}_{2,\Y}, \notag 
%%\end{align}
% where 
%\begin{align}
%		& \hat{Q}_{1,\M} = \{ \mathbb{P}_n(\X \X^{\mytrans})\}^{-1} \mathbb{P}_n(\X \M), \quad \quad \hat{Q}_{1,\S}= \{\mathbb{P}_n(\X \X^{\mytrans})\}^{-1} \mathbb{P}_n(\X \S), \label{eq:popqmomentdefsample} \\
%	&  \hat{Q}_{2,\M} = \{\mathbb{P}_n( \tilde{\X}  \tilde{\X}^{\mytrans})\}^{-1} \mathbb{P}_n( \tilde{\X} \M), \quad \quad \hat{Q}_{2,\Y}= \{\mathbb{P}_n( \tilde{\X}  \tilde{\X}^{\mytrans})\}^{-1} \mathbb{P}_n( \tilde{\X} \Y). \notag 
%\end{align}
In addition, 
by replacing $P_n(\cdot)$ with $\mathbb{P}_n^*(\cdot)$, we similarly define $(Q_{1,\M}^*, Q_{1,\S}^*, Q_{2,\M}^*, Q_{2,\S}^*)$, $(\Sperpboot, \Mperpboot, \Mperppboot, \Yperppboot)$, and 
$\{(\Sperpbooti, \Mperpbooti, \Mperppbooti, \Yperppbooti): i = 1,\ldots, n\}$.

\bigskip

To understand the motivation of defining the variables above, we point out two facts that under Condition \ref{cond:inversemoment}, 
%the conditions of Theorem \ref{thm:prodlimit2}, 
\begin{itemize}
	\item[(i)] Model \eqref{eq:fullmodlocal2} induces
\begin{align}
	\M_{\perp} = \alphaSn \S_{\perp} + \eM \quad \text{ and } \quad  \Y_{\perp'} = \betaMn \M_{\perp'} + \eY, \label{eq:transmodel}
\end{align}
where the error terms $\eM$ and $\eY$ are the same variables as the error terms in  \eqref{eq:fullmodlocal2}; 
	\item[(ii)] the ordinary least squares regression estimates of $\alphaSn$ and $\betaMn$ can be written as 
\begin{align}\label{eq:olsfwlform}
	\hatalphaSn = &~  \frac{ \sum_{i=1}^n \hatSperpi \hatMperpi }{ \sum_{i=1}^n \hatSperpi^2}=\frac{\mathbb{P}_n(\hatSperp \hatMperp) }{\mathbb{P}_n(\hatSperp^2)}	  \\% \quad  \text{and}\quad 
\hatbetaMn =  &~ \frac{ \sum_{i=1}^n  \hatMperppi \hatYperppi }{\sum_{i=1}^n  \hatMperppi^2 } = \frac{\mathbb{P}_n( \hatMperpp \hatYperpp)}{\mathbb{P}_n(\hatMperpp^2)}, \notag
\end{align}	
respectively.	
\end{itemize}
We mention that \eqref{eq:olsfwlform} directly follows from the Frisch–Waugh–Lovell theorem \citep{frisch1933partial}.  
But for self-consistency, we provide \eqref{eq:transmodel}--\eqref{eq:olsfwlform} and other conclusions induced by  \eqref{eq:transmodel}--\eqref{eq:olsfwlform} in the following Lemma \ref{lm:regresstrans},
which is proved  in Section \ref{sec:pfregresstrans} and used in the proofs of theorems below.

\begin{lemma}[Frisch–Waugh–Lovell theorem]\label{lm:regresstrans}
Under Condition \ref{cond:inversemoment}, 
\begin{enumerate}
\item[(1)] Model \eqref{eq:fullmodlocal2} induces Model \eqref{eq:transmodel}. 
\item[(2)] The ordinary least squares estimators of $\alphaSn$ and $\betaMn$ can be written as in \eqref{eq:olsfwlform}. 
\item[(3)]  The residuals of the ordinary least squares regressions of Model \eqref{eq:fullmodlocal2} can be obtained by $\hatenMi = \hat{\M}_{\perp,i}-\hatalphaSn \hat{\S}_{\perp,i}$ and $\hatenYi = \hat{\Y}_{\perp',i}-\hatbetaMn \hat{\M}_{\perp',i}$ for $i=1,\ldots, n$. 
\item[(4)] $\mathbb{P}_n(\hatenM \hat{\S}_{\perp} ) = n^{-1}\sum_{i=1}^n \hatenMi  \hat{\S}_{\perp,i} = 0 $ and  $\mathbb{P}_n(\hatenY \hat{\M}_{\perp'} ) = n^{-1}\sum_{i=1}^n \hatenYi \hat{\M}_{\perp',i} = 0$. 
\item[(5)]  The standard errors of 
%the ordinary least squares estimators 
$\hatalphaSn$ and $\hatbetaMn$ are  $\hatsigmaalphan/\sqrt{n}$ and  $\hatsigmabetan/\sqrt{n}$, respectively, where $\hatsigmaalphan^2=\mathbb{P}_n(\hatenM^2)/\mathbb{P}_n(\hat{S}_{\perp}^2)$ and  $\hatsigmabetan^2 = \mathbb{P}_n( \hatenY^2)/\mathbb{P}_n(\hat{\M}_{\perp'}^2)$. 
\end{enumerate}
\end{lemma}

\medskip

In addition, for the defined Q-moments, e.g., \eqref{eq:popqmomentdef}, Lemma \ref{lm:consistencyqmoments} below proves their consistency properties   and is used in the following proofs. 
% in Section \ref{sec:pfalllemma}. 
%the proofs below. 
%convergence of the sample version to the population version, which will be used in the proofs below. 
\begin{lemma}\label{lm:consistencyqmoments}
Under Condition \ref{cond:inversemoment}, 
%the conditions of Theorem \ref{thm:prodlimit2}, 
\begin{itemize}
	\item[(1)] $
(\hat{Q}_{1,\M}, \hat{Q}_{1,\S}, \hat{Q}_{2,\M}, \hat{Q}_{2,S}) \xrightarrow{a.s.} (Q_{1,\M}, Q_{1,\S}, Q_{2,\M}, Q_{2,\S}). 	
$
\item[(2)] 
$(Q_{1,\M}^*, Q_{1,\S}^*, Q_{2,\M}^*, Q_{2,\S}^*)  \overset{\mathrm{P}^*}{\leadsto} (Q_{1,\M}, Q_{1,\S}, Q_{2,\M}, Q_{2,\S})$. 
% $
% (Q_{1,\M}^*, Q_{1,\S}^*, Q_{2,\M}^*, Q_{2,\S}^*) \xrightarrow{P_C} (\hat{Q}_{1,\M},\hat{Q}_{1,\S},\hat{Q}_{2,\M}, \hat{Q}_{2,S}),  	
% $
% conditionally (on the data) in probability, where $\xrightarrow{P_C}$ denotes the convergence in probability conditioning on the data. 
\end{itemize}
\end{lemma}

\begin{proof}
(i) The conclusion follows from the strong law of large numbers, continuous mapping theorem, and the assumption that $P_n(\X \X^{\mytrans})$ and $P_n(\tilde{\X} \tilde{\X}^{\mytrans})$ are  invertible.  	
(ii) This follows from the bootstrap consistency $\mathbb{P}_n^*\{\X (\X^{\mytrans},\S, \M, \Y)\}
\overset{\mathrm{P}^*}{\leadsto} 
% \xrightarrow{P_C} 
{P}_n\{\X (\X^{\mytrans},\S, \M, \Y)\}$ 
% conditionally (on the data) in probability 
\cite[see, e.g.,][Theorem 2.2]{bickel1981some}. 
% \cite[see, e.g.,][Theorem 23.4]{van2000asymptotic}. 
\end{proof}

\subsection{Proof of Theorem \ref{thm:prodlimit2}} \label{sec:pfthm1}

To derive the limit of $\hatalphaSn \hatbetaMn - \alphaSn \betaMn$, we use the decomposition
\begin{align}
&~\hatalphaSn \hatbetaMn - \alphaSn \betaMn \label{eq:expanproduct1}\\
=&~ (\hatalphaSn - \alphaSn) (\hatbetaMn-\betaMn)+ \alphaSn (\hatbetaMn-\betaMn) + (\hatalphaSn-\alphaSn)\betaMn \notag
\end{align}
and the limits 
\begin{align}\label{eq:limitcoef}
	\sqrt{n}(\hatalphaSn - \alphaSn,~ \hatbetaMn - \betaMn )\xrightarrow{d} ( Z_{\S}, Z_{\M}),   
\end{align} 
\noindent which will be proved later. Based on \eqref{eq:expanproduct1} and \eqref{eq:limitcoef}, we next discuss two cases separately.

\medskip

\noindent \textit{Case 1:} 
When $(\alphaS, \betaM) \neq (0,0)$, as $n\to \infty$, we have $\alphaSn \to \alphaS$,  $\hatbetaMn \xrightarrow{P_n} \betaM$, and  $\sqrt{n}(\hatalphaSn - \alphaSn) (\hatbetaMn-\betaMn)=o_{P_n}(1)$.
Therefore, by \eqref{eq:expanproduct1}--\eqref{eq:limitcoef} and Slutsky's lemma, we know that $\sqrt{n}(\hatalphaSn \hatbetaMn - \alphaSn \betaMn) \xrightarrow{d} \alphaS Z_M + \betaM Z_S.$

\medskip
\noindent \textit{Case 2:} When $(\alphaS, \betaM) = (0,0)$,  we have $\alphaSn = n^{-1/2} \balpha $ and $\betaMn = n^{-1/2} \bbeta.$ Then by \eqref{eq:expanproduct1}, 
\begin{align*}
&~ n \times(\hatalphaSn \hatbetaMn - \alphaSn \betaMn)  \notag \\
=&~n(\hatalphaSn - \alphaSn) (\hatbetaMn-\betaMn) + \balpha	\sqrt{n}(\hatbetaMn-\betaMn) + \bbeta \sqrt{n}(\hatalphaSn - \alphaSn). 
\end{align*} 
By \eqref{eq:limitcoef},  $n(\hatalphaSn \hatbetaMn - \alphaSn \betaMn) \xrightarrow{d} Z_MZ_S+\balpha Z_M+ \bbeta Z_S$.

\medskip

To finish the proof of Theorem  \ref{thm:prodlimit2}, it remains to prove \eqref{eq:limitcoef}. In particular,
by \eqref{eq:transmodel} and \eqref{eq:olsfwlform}, 
we can write
\begin{align}
%&~\sqrt{n}(\hatalphaSn - \alphaSn)\label{eq:coeffalphaexpan} \\
\sqrt{n}(\hatalphaSn - \alphaSn) =&	\frac{\sqrt{n}\mathbb{P}_n(\hatSperp\Mperp)}{\mathbb{P}_n(\hatSperp^2)}-\sqrt{n} \alphaSn+\frac{\sqrt{n}\mathbb{P}_n\{(\hatSperp(\hatMperp-\Mperp)\}}{\mathbb{P}_n(\hatSperp^2)}\label{eq:coeffalphaexpan} \\
=&~\frac{\sqrt{n} \mathbb{P}_n(\hatSperp \eM )  }{\mathbb{P}_n(\hatSperp^2)} + \frac{\sqrt{n}\alphaSn \PP \{\hatSperp(\Sperp-\hatSperp)\} }{\mathbb{P}_n(\hat{\S}_g^2)} + \frac{\sqrt{n} \{\hatSperp(\hatMperp-\Mperp)\} }{\mathbb{P}_n(\hat{\S}_g^2)}\notag \\
:=&~\mathbb{B}_{\S,n}/\mathbb{V}_{\S,n}:=\mathbb{Z}_{\S,n},  \notag
%=&\frac{\sqrt{n} \mathbb{P}_n(\hatSperp \eM )  }{\mathbb{P}_n(\hatSperp^2)} 
\end{align}
where in the last equation, we use 
\begin{align*}
\mathbb{P}_n\{\hatSperp(\Sperp-\hatSperp)\}=&~ \mathbb{P}_n\{(\S-\hat{Q}_{1,\S}^{\mytrans}\X)\X^{\mytrans}\}( \hat{Q}_{1,\S} - Q_{1,\S} )\notag \\
 =&~\{\PP(\S \X )-\hat{Q}_{1,\S}^{\mytrans}\PP(\X\X^{\mytrans})\} ( \hat{Q}_{1,\S} - Q_{1,\S} ) =0,	
\end{align*}
%$\mathbb{P}_n\{\hatSperp(\Sperp-\hatSperp)\}= \mathbb{P}_n\{(\S-\hat{Q}_{1,\S}^{\mytrans}\X)\X^{\mytrans}\}( \hat{Q}_{1,\S} - Q_{1,\S} ) =\{\PP(\S \X )-\hat{Q}_{1,\S}^{\mytrans}\PP(\X\X^{\mytrans})\} ( \hat{Q}_{1,\S} - Q_{1,\S} ) =0, $
$\mathbb{P}_n\{(\hatSperp (\hatMperp -\Mperp)\}=\mathbb{P}_n\{(\S-\hat{Q}_{1,\S}^{\mytrans}\X)\X^{\mytrans}\}( Q_{1,\M}-\hat{Q}_{1,\M})=0$,
and $\mathrm{E}(\hatSperp \eM)=0$. 
Note that
$
\mathbb{B}_{\S,n} = \GG\{\eM (\hatSperp - \Sperp )\} + \GG( \eM \Sperp )=\GG(\eM \X^{\mytrans})( Q_{1,\S}-\hat{Q}_{1,\S} )+ \GG( \eM \Sperp ). 	
$
By Lemma \ref{lm:consistencyqmoments}, the central limit theorem, and Slutsky's lemma,  we know $\mathbb{B}_{\S,n} 	\xrightarrow{d} \eM \Sperp$. 
In addition, for $\VVSn$, we have
\begin{align*}
\VVSn-\PP(\Sperp^2)=&~\PP\{(\hatSperp-\Sperp)(\hatSperp+\Sperp)\} = \PP\{(\hatSperp-\Sperp)(\Sperp-\hatSperp)\}\notag \\
  =&~-({Q}_{1,S}-\hat{Q}_{1,S})^{\mytrans}\PP(\X\X^{\mytrans})({Q}_{1,S}-\hat{Q}_{1,S}),
\end{align*}
%\begin{align*}
%\VVSn-\PP(\Sperp^2)=&~\PP\{(\hatSperp-\Sperp)(\hatSperp+\Sperp)\}  =-({Q}_{1,S}-\hat{Q}_{1,S})^{\mytrans}\PP(\X\X^{\mytrans})({Q}_{1,S}-\hat{Q}_{1,S})
%\end{align*}
where we use $\PP\{(\hatSperp-\Sperp) \hatSperp \}=0$. 
%$\PP\{(\hatSperp-\Sperp) \hatSperp \}=(Q_{1,\S} - \hat{Q}_{1,\S})^{\mytrans} \PP\{\X (\S - \X^{\mytrans}\hat{Q}_{1,\S})\} = 0$ by the definition of $\hat{Q}_{1,\S}$. 
By $\PP(\Sperp^2)\xrightarrow{a.s.} \mathrm{E}(\Sperp^2)$, 
 Lemma \ref{lm:consistencyqmoments}, and Slutsky's lemma, 
we obtain $\VVSn \xrightarrow{P_n} V_{\S}$. 
%By Lemma \ref{lm:consistencyqmoments} and the strong law of large number, we obtain $\VVSn \xrightarrow{P_n} V_{\S}$
%$\VVSn=\PP(\hatSperp^2)$
By \eqref{eq:coeffalphaexpan} and Slutsky's lemma, we prove $\sqrt{n}(\hatalphaSn-\alphaSn) \xrightarrow{d} \eM \Sperp/V_{\S} = Z_{\S}$. 
Following similar analysis, we also have 
$\sqrt{n}(\hatbetaMn-\betaMn)=\mathbb{B}_{\M,n} /\VVMn \xrightarrow{d} \eM \Mperpp /V_{\M} =Z_{\M}$.

%Following the analysis in Section \ref{sec:pflimit1}, we can consider Model \eqref{eq:transmodel}. 

%\paragraph{Proof of indicators}

\medskip
%\subsection{Proof of Theorem \ref{thm:bootstrapprod}} \label{sec:pfbootstrapprod}

\subsection{Proof of Theorem \ref{thm:bootstrapprodcomb}} \label{sec:pfbootstrapprodcomb}

%Following the proof of Theorem \ref{thm:prodlimit2} in Section \ref{sec:pfthm1}, to prove Theorem \ref{thm:bootstrapprod}, it suffices to prove 

%To prove Part (ii) of Theorem \ref{thm:bootstrapprod}, by Theorem \ref{thm:prodlimit2}, it suffices to use  
%the following Lemma \ref{lm:limitzboot}, which is proved in Section \ref{sec:pflimitzboot}. 
%\begin{lemma}\label{lm:limitzboot}
%Under Condition \ref{cond:inversemoment}, $\big(\ZZSbootn/\VVSbootn, \ZZMbootn/\VVMbootn\big)\xrightarrow{d}(Z_{\S}/V_{\S}, Z_{\M}/V_{\M})^{\mytrans},$
%%the conditions of Theorem \ref{thm:bootstrapprod},
%%\begin{align*}
%%\big(\mathbb{Z}_{\S,n}^*/\mathbb{P}_n^*(\hatSperp^2), ~ \mathbb{Z}_{\M,n}^*/\mathbb{P}_n^*(\hatMperp^2)\big)^{\mytrans} \xrightarrow{d}(Z_{\S}/V_{\S}, Z_{\M}/V_{\M})^{\mytrans}, 
%%\end{align*} 
%%\begin{align*}
%%\Biggr(\, \frac{\mathbb{Z}_{\S,n}^*}{\mathbb{P}_n^*(\hatSperp^2)}, ~ \frac{\mathbb{Z}_{\M,n}^*}{\mathbb{P}_n^*(\hatMperp^2)}\, \Biggr)^{\mytrans} \xrightarrow{d}(Z_{\S}/V_{\S}, Z_{\M}/V_{\M})^{\mytrans}, 
%%\end{align*} 
%%\begin{align*}
%%\big(\ZZSbootn/\VVSbootn, \ZZMbootn/\VVMbootn\big)\xrightarrow{d}(Z_{\S}/V_{\S}, Z_{\M}/V_{\M})^{\mytrans},
%%\end{align*}
%conditionally (on the data) in probability. 	
%%\begin{align*}
%%\big( {\mathbb{Z}_{\S,n}^*}/{\mathbb{V}_{\S,n}^*}, ~ {\mathbb{Z}_{\M,n}^*}/{\mathbb{V}_{\M,n}^*}\, \big)^{\mytrans} \xrightarrow{d}(Z_{\S}/V_{\S}, Z_{\M}/V_{\M})^{\mytrans}, 
%%\end{align*} 	
%\end{lemma}

To prove Theorem \ref{thm:bootstrapprodcomb}, 
following the discussions in Section \ref{sec:pfthm1}, we first prove
\begin{align}\label{eq:zzbootlimit}
\sqrt{n}\big(\hatalphaSn^*- \hatalphaSn, \hatbetaMn^*-\hatbetaMn \big)
%=\big(\mathbb{Z}_{S,n}^*, \mathbb{Z}_{M,n}^*\big)
% \xrightarrow{d}
\overset{d^*}{\leadsto} 
(Z_{\S}, Z_{\M}). 
\end{align} 
Similarly to \eqref{eq:coeffalphaexpan}, we can write 
 \begin{align*}
\sqrt{n}(\hatalphaSn^*- \hatalphaSn) =&\frac{\sqrt{n}\PPboot(\Sperpboot \hatenM ) }{ \PPboot\{(\Sperpboot)^2)\} } + \frac{\sqrt{n}\hatalphaSn \PPboot\{ \Sperpboot( \hatSperp - \Sperpboot )\}}{ \PPboot\{(\Sperpboot)^2)\} }+\frac{\sqrt{n}\PPboot\{\Sperpboot (\Mperpboot -\hatMperp ) \}}{ \PPboot\{(\Sperpboot)^2)\} }\notag \\
 := &\ZZSbootn/\VVSbootn :=\mathbb{Z}_{\S,n}^*,
 \end{align*}
where in the last equation, we use $\PPboot\{ \Sperpboot( \hatSperp - \Sperpboot )\}=\PPboot[\{\S-(Q_{1,\S}^*)^{\mytrans} \X\}\X^{\mytrans}](Q_{1,\S}^* - \hat{Q}_{1,\S}) =0$,
 $\PPboot\{\Sperpboot (\Mperpboot -\hatMperp ) \}= \PPboot[\{\S-(Q_{1,\S}^*)^{\mytrans} \X\}\X^{\mytrans}](\hat{Q}_{1,\M}-Q_{1,\M}^*) =0$,
 and $\PP(\hatenM \Sperpboot)=\PP(\hatenM\S)-\PP(\hatenM\X^{\mytrans})Q_{1,\S}^* =0$. 
Similarly, we have $\sqrt{n}(\hatbetaMn^* - \hatbetaMn)=\ZZMbootn/\VVMbootn :=  \mathbb{Z}_{M,n}^*$. 
%\begin{align}\label{eq:zzbootlimit}
%\sqrt{n}\big(\hatalphaSn^*- \hatalphaSn, \hatbetaMn^*-\hatbetaMn \big)=\big(\ZZSbootn/\VVSbootn, \ZZMbootn/\VVMbootn\big)\xrightarrow{d}(Z_{\S}/V_{\S}, Z_{\M}/V_{\M})^{\mytrans},
%\end{align}
% conditionally (on the data) in probability. 
By Slutsky's lemma, to prove \eqref{eq:zzbootlimit},  it suffices to  prove
 $\VVSbootn 
 \overset{\mathrm{P}^*}{\leadsto} 
 % \xrightarrow{P_C} 
 V_{\S}$
 and $\mathbb{B}_{\S,n}^*  
 % \xrightarrow{d} 
 \overset{\mathrm{d}^*}{\leadsto} 
 Z_{\S}V_{\S}$ below.  
% conditionally (on the data) in probability, below. 
%We can obtain $\mathbb{B}_{\M,n}^* /\VVMbootn \xrightarrow{d} Z_{\M}$ conditionally (on the data) in probability following similar analysis.  

\medskip

 For $\mathbb{B}_{\S,n}^*$, we note that 
\begin{align}\label{eq:zboots}
%\ZZSbootn 
\mathbb{B}_{\S,n}^*= \GGboot(\hatenM \Sperpboot)= \GGboot\{\hatenM (\Sperpboot - \Sperp)\}+ \GGboot\{(\hatenM - \eM) \Sperp\} + \GGboot(\eM \Sperp). 
\end{align}
We next analyze the three summed terms in \eqref{eq:zboots} separately. 
First, since $\Sperpboot -\Sperp =(Q_{1,\S}-{Q}_{1,\S}^*)^{\mytrans}\X$, $\Mperpboot -\Mperp =(Q_{1,\M}- Q_{1,\M}^*)^{\mytrans}\X$, and $\hatenM=\hatMperp-\hatalphaSn \hatSperp$,
the first term in \eqref{eq:zboots} can be written as
\begin{align}
&~  \mathbb{G}_n^*\{  (\hatMperp -\hatalphaSn \hatSperp) \X^{\mytrans} \}(Q_{1,\S}-{Q}_{1,\S}^*) \label{eq:zsfirsterm} \\
=&~ \mathbb{G}_n^* [\{ ( \hatMperp -\Mperp) +\Mperp - \hatalphaSn(\hatSperp - \Sperp +\Sperp) \}\X^{\mytrans}] (Q_{1,\S}-{Q}_{1,\S}^*) \notag\\
=&~\Big[(Q_{1,\M}-\hat{Q}_{1,\M})^{\mytrans}\mathbb{G}_n^*(\X\X^{\mytrans})+ \mathbb{G}_n^*(\Mperp \X^{\mytrans})\notag \\
&~\ - \hatalphaSn (Q_{1,\S}-\hat{Q}_{1,\S})^{\mytrans}\mathbb{G}_n^*(\X\X^{\mytrans}) -\hatalphaSn \mathbb{G}_n^*(\Sperp\X^{\mytrans})\Big] (Q_{1,\S}-{Q}_{1,\S}^*). \notag
\end{align}
By Lemma \ref{lm:consistencyqmoments} and the bootstrap consistency, $\eqref{eq:zsfirsterm} \overset{\mathrm{P}^*}{\leadsto} 0$.  
Second, 
%by Lemma \ref{lm:regresstrans}, we have  
%and \ref{lm:translmandcoef}, we have 
as $\hatenM=\hatMperp-\hatalphaSn \hatSperp$ and $\eM=\Mperp-\alphaSn\Sperp$ by Lemma \ref{lm:regresstrans}, 
we write the second term in \eqref{eq:zboots} as
%Then we can write the second term in \eqref{eq:zboots} as
% Then, we can write the second term in \eqref{eq:zboots} as
\begin{align}
&~\mathbb{G}_n^*\{(\hatMperp-\hatalphaSn\hatSperp - \Mperp+ \alphaSn \Sperp) \Sperp\} \label{eq:zssecondterm} \\
=&~\mathbb{G}_n^*\{(\hatMperp - \Mperp) \Sperp\}-\hatalphaSn\mathbb{G}_n^*\{(\hatSperp - \Sperp)\Sperp\} -(\hatalphaSn - \alphaSn) \mathbb{G}_n^*(\Sperp^2) \notag \\
=&~(Q_{1,\M}-\hat{Q}_{1,\M})^{\mytrans}\mathbb{G}_n^*(\X\Sperp)-\hatalphaSn(Q_{1,\S}-\hat{Q}_{1,\S})^{\mytrans}\mathbb{G}_n^*(\X\Sperp) -(\hatalphaSn - \alphaSn) \mathbb{G}_n^*(\Sperp^2). \notag 
\end{align}
Similarly to \eqref{eq:zsfirsterm}, by  Lemma \ref{lm:consistencyqmoments},  $\hatalphaSn \xrightarrow{a.s.} \alphaSn$, and the bootstrap consistency, 
%$ \hat{Q}_{1,\S}\xrightarrow{a.s.} Q_{1,\S}$, $ \hat{Q}_{1,\M}\xrightarrow{a.s.} Q_{1,\M}$, $\hatalphaSn \xrightarrow{a.s.} \alphaSn$, and the bootstrap consistency, 
we know that $\eqref{eq:zssecondterm}\overset{\mathrm{P}^*}{\leadsto}  0$.
% conditionally  (on the data) in probability.  
Third,  by the bootstrap consistency, $\mathbb{G}_n^*(\eM \Sperp) \overset{d^*}{\leadsto} 
Z_{\S}V_{\S}$.  
% conditionally (on the data) in probability.  
In summary, by \eqref{eq:zboots} and Slutsky's lemma, we prove $\mathbb{B}^*_{\S,n}\overset{d^*}{\leadsto} 
Z_{\S}V_{\S}$.

\medskip

 For $\VVSbootn$, we have
\begin{align}
\VVSbootn-\PPboot(\Sperp^2)=&~	\PPboot\{ (\Sperpboot - \Sperp ) (\Sperpboot + \Sperp )\} = 	\PPboot\{ (\Sperpboot - \Sperp ) \Sperp \} \label{eq:denominatorboot1}\\
=&~ (Q_{1,\S}-Q_{1,\S}^*)^{\mytrans}\PPboot(\X\Sperp), \notag
\end{align}
where we use $\PPboot\{(\Sperpboot-\Sperp)\Sperpboot\}=(Q_{1,\S}-Q_{1,\S}^*)^{\mytrans}\PPboot\{\X(\S-\X^{\mytrans}Q_{1,\S}^* )\}=0$.
By Lemma \ref{lm:consistencyqmoments},
$\eqref{eq:denominatorboot1} \overset{\mathrm{P}^*}{\leadsto}  0$.
% conditionally (on the data) in probability. 
Therefore, by Slutsky's  lemma and the bootstrap consistency, 
we know  $\VVSbootn \overset{\mathrm{P}^*}{\leadsto} V_{\S}$.
% conditionally (on the data) in probability. 

\medskip

In summary, by the above arguments and the decomposition
\begin{align}
&~\hatalphaSn^* \hatbetaMn^* - \hatalphaSn \hatbetaMn \notag\\
=&~ (\hatalphaSn^* - \hatalphaSn) (\hatbetaMn^*-\hatbetaMn)+ \alphaSn^* (\hatbetaMn^*-\hatbetaMn) + (\hatalphaSn^*-\hatalphaSn)\hatbetaMn, \notag
\end{align}
we obtain that 
\begin{enumerate}
	\item[(i)] when $(\alphaS, \betaM) \neq (0,0)$,  $\sqrt{n}(\hatalphaSn^* \hatbetaMn^*  - \hatalphaSn \hatbetaMn ) \overset{d^*}{\leadsto}$  $\sqrt{n}( \hatalphaSn \hatbetaMn - \alphaSn \betaMn)$;
	\item[(ii)]  when $(\alphaS, \betaM) = (0,0)$, 	$\mathbb{R}_{1,n}^*( \balpha, \bbeta )\overset{d^*}{\leadsto}$   $n( \hatalphaSn \hatbetaMn - \alphaSn \betaMn)$.
\end{enumerate}
To finish the proof of Theorem \ref{thm:bootstrapprodcomb}, 
it remains to prove 
\begin{align}
&1 - \myindicator_{\alphaS, \lambda_n}^* \overset{\mathrm{P}^*}{\leadsto}  \myindicator\{\alphaS \neq 0\},& \hspace{-5em}   \myindicator_{\alphaS, \lambda_n}^* \overset{\mathrm{P}^*}{\leadsto} \myindicator\{\alphaS = 0\},& \label{eq:indicatorconvg} \\
&1 - \myindicator_{\betaM, \lambda_n}^* \overset{\mathrm{P}^*}{\leadsto} \myindicator\{\betaM \neq 0\},& \hspace{-5em}  \myindicator_{\betaM, \lambda_n}^* \overset{\mathrm{P}^*}{\leadsto}  \myindicator\{\betaM = 0\}. 	& \notag
\end{align}
% conditionally (on the data) in probability.
To prove \eqref{eq:indicatorconvg},  we use the following Lemma \ref{lm:sigmaconsis}, which is proved in Section \ref{sec:pfsigmaconsis}. 
\begin{lemma}\label{lm:sigmaconsis}
Under Condition \ref{cond:inversemoment}, 
$(\hatsigmaalphan, \hatsigmabetan) \xrightarrow{a.s.} (\sigmaalpha, \sigmabeta)$ 
and  $(\hatsigmaalphan^*, \hatsigmabetan^*) \overset{\mathrm{P}^*}{\leadsto} 
 (\sigmaalpha, \sigmabeta )$,
 % conditionally (on the data) in probability,
where $\sigmaalpha^2=\mathrm{E}(\eM^2)/\mathrm{E}(\Sperp^2)$ and $\sigmabeta^2=\mathrm{E}(\eY^2)/\mathrm{E}(\M_g^2)$. 
\end{lemma}

\smallskip

%Similarly to \eqref{eq:pctstran}, we have
Note that
\begin{align*}
	P_C(|T_{\alpha,n}^*| > \lambda_n )=&~ P_C \big\{ | ( \hatalphaSn^* - \hatalphaSn ) + (\hatalphaSn - \alphaSn ) + \alphaSn | > n^{-1/2}\lambda_n \hatsigmaalphan^*\big\} \notag \\
 	\leq &~P_C \big(\, |\alphaSn |+  |\hatalphaSn^* - \hatalphaSn| + |\hatalphaSn - \alphaSn |\,  > n^{-1/2}\lambda_n \hatsigmaalphan^* \big). 
\end{align*}
When $\alphaS = 0$, 
 by   the limits in \eqref{eq:limitcoef} and \eqref{eq:zzbootlimit},  $\alphaSn=n^{-1/2}\balpha$, Lemma \ref{lm:sigmaconsis}, and  $\lambda_n \to \infty$,
we have $T_{\alpha,n}^*/\lambda_n = \sqrt{n} \hatalphaSn^*/(\hatsigmaalphan^* \lambda_n) \overset{\mathrm{P}^*}{\leadsto}  0$.
% conditionally (on the data)  in probability. 
Similarly, we have
 \begin{align} 
 	P_C(|T_{\alpha,n}^*| \leq \lambda_n ) %=&~ P_C \big\{ | ( \hatalphaSn^* - \hatalphaSn ) + (\hatalphaSn - \alphaSn ) + \alphaSn | \leq n^{-1/2}\lambda_n \hatsigmaalphan^*\big\} \notag \\
 	\leq &~P_C \big(\, |\alphaSn | \leq n^{-1/2}\lambda_n \hatsigmaalphan^* +  |\hatalphaSn^* - \hatalphaSn| + |\hatalphaSn - \alphaSn |\, \big). \label{eq:pctstran}
 \end{align}  
When $\alphaS \neq 0$, \eqref{eq:pctstran} $\xrightarrow{P_n} 0$ 
% tends to zero in probability  
by $|\alphaS| > 0$, $n^{-1/2}\lambda_n = o(1)$, \eqref{eq:limitcoef}, and \eqref{eq:zzbootlimit}. 
%When $\alphaS = 0$, by   the limits in \eqref{eq:limitcoef} and \eqref{eq:zzbootlimit}, Lemma \ref{lm:sigmaconsis}, and  $\lambda_n \to \infty$,
%%and  Lemmas \ref{lm:limitcoef} and \ref{lm:separatecoeflimitboot}, 
%%by Lemmas \ref{lm:limitcoef} and \ref{lm:separatecoeflimitboot}, and $\lambda_n \to \infty$ as $n\to \infty$, 
%we have $T_{\alpha,n}^*/\lambda_n = \sqrt{n} \hatalphaSn^*/(\hatsigmaalphan^* \lambda_n) \xrightarrow{P_C} 0$  conditionally (on the data)  in probability. 
% When $\alphaS = 0,$ by the limit in \eqref{eq:limitalphabeta}, ... below, and $\lambda_n \to \infty$ as $n\to \infty$, we have $T_{\alpha,n}^*/\lambda_n \xrightarrow{P_C} 0$ in probability.
% When $\alphaS \neq 0$, we write   
% \begin{align}
% 	P_C(|T_{\alpha,n}^*| \leq \lambda_n ) =&~ P_C \big\{ | ( \hatalphaSn^* - \hatalphaSn ) + (\hatalphaSn - \alphaSn ) + \alphaSn | \leq n^{-1/2}\lambda_n \hatsigmaalphan^*\big\} \notag \\
% 	\leq &~P_C \big(\, |\alphaSn | \leq n^{-1/2}\lambda_n \hatsigmaalphan^* +  |\hatalphaSn^* - \hatalphaSn| + |\hatalphaSn - \alphaSn |\, \big). \label{eq:pctstran}
% \end{align}  
%% tends to zero in probability, where by
%We then know that \eqref{eq:pctstran} tends to zero in probability  by $|\alphaS| > 0$, $\lambda_n = o(n^{1/2})$, \eqref{eq:limitcoef}, and \eqref{eq:zzbootlimit}.  
%Lemmas \ref{lm:limitcoef} and \ref{lm:separatecoeflimitboot}, and $\hat{\sigma}_{\alphaS, n}\to  \sigma_{\alphaS} > 0$. 
% , where we use the limit in  % \eqref{eq:limitalphabeta}, .... to be proved below, and $\lambda_n=o(\sqrt{n})$. 
Let $\mathrm{E}_C$ denote the expectation conditioning on the data, and then 
%It follows that
\begin{eqnarray*}
&&\mathrm{E}_C|\,\myindicator\{ |T_{\alpha, n}^*| > \lambda_n \} - \myindicator\{\alphaS \neq 0\}  \, | \notag \\
&\leq & P_C\big(\, |T_{\alpha, n}^*| > \lambda_n,\alphaS = 0 \, \big) + P_C \big(\, |T_{\alpha, n}^*| \leq \lambda_n,\alphaS \neq 0\, \big) \notag \\
&=& P_C\big(\, |T_{\alpha, n}^*| > \lambda_n | \alphaS = 0  \,\big) \times \myindicator\{\alphaS =0 \}+ P_C\big( |T_{\alpha, n}^*| \leq \lambda_n | \alphaS \neq 0 \big) \times \myindicator\{\alphaS \neq 0\}
\end{eqnarray*} 
which $\xrightarrow{P_n} 0$. 
% tends to zero in probability. 
This implies $\myindicator\{ |T_{\alpha, n}^*| > \lambda_n \} \overset{\mathrm{P}^*}{\leadsto} 
 \myindicator\{ \alphaS \neq 0\}$ and $\myindicator\{ |T_{\alpha, n}^*| \leq \lambda_n \} \overset{\mathrm{P}^*}{\leadsto} \myindicator\{ \alphaS = 0\}$.
 % conditionally (on the data) in probability. 
As $\myindicator\{|T_{\alpha,n}|\leq \lambda \} \xrightarrow{\mathrm{P}} \myindicator\{\alphaS = 0 \}$,
% in probability,  
by Slutsky's lemma, 
we know that the first and the second limits in \eqref{eq:indicatorconvg} hold. 
Following similar analysis, we also obtain the third and the fourth limits in  \eqref{eq:indicatorconvg}.

\begin{remark}[Calculation of classical  non-parametric bootstrap]\label{rm:classboot00}
By calculations in Sections \ref{sec:preliminary} and \ref{sec:pfthm1}, 
%and \ref{sec:pfbootstrapprod}, 
we have when $\alphaS=\betaM=0$, 
\begin{align*}
&~ n(\hatalphaSn^* \hatbetaMn^* - \hatalphaSn \hatbetaMn)\notag\\
 =&~ \sqrt{n} \hatalphaSn (\hatbetaMn^* - \hatbetaMn ) + \sqrt{n} \hatbetaMn (\hatalphaSn^* -\hatalphaSn) +  (\hatalphaSn^* -\hatalphaSn) (\hatbetaMn^* - \hatbetaMn ) \notag\\
 =&~ \sqrt{n} \hatalphaSn \mathbb{Z}_{\M,n}^* + \sqrt{n} \hatbetaMn \mathbb{Z}_{\S,n} +  \mathbb{Z}_{\S,n}^* \mathbb{Z}_{\M,n}^* \notag\\
=&~(\balpha+ \mathbb{Z}_{\S,n})\mathbb{Z}_{\M,n}^*+ (\bbeta+\mathbb{Z}_{\M,n})\mathbb{Z}_{\S,n}^*+\mathbb{Z}_{\S,n}^* \mathbb{Z}_{\M,n}^* \notag\\
=&~\mathbb{R}_{n}^*(\balpha, \bbeta)+  \mathbb{Z}_{\S,n}\mathbb{Z}_{\M,n}^* + \mathbb{Z}_{\M,n}\mathbb{Z}_{\S,n}^*, 	
\end{align*}
and
\begin{align*}
&~	n(\hatalphaSn \hatbetaMn - \alphaSn \betaMn) \notag\\
 =&~ \sqrt{n} \alphaSn (\hatbetaMn - \betaMn ) + \sqrt{n} \betaMn (\hatalphaSn -\alphaSn) +  (\hatalphaSn -\hatalphaSn) (\hatbetaMn - \hatbetaMn ) \notag\\
=&~	\balpha \mathbb{Z}_{\M,n} + \bbeta \mathbb{Z}_{\S,n} +  \mathbb{Z}_{\S,n} \mathbb{Z}_{\M,n}. \notag
\end{align*}	
\end{remark}

\subsection{Proof of Theorem \ref{thm:maxplimit1}}\label{sec:pfmaxplimit1}

\noindent \textit{Case 1:} When $(\alphaS, \betaM) \neq (0,0)$, since $|\alphaS/\sigmaalpha| \neq |\betaM/ \sigmabeta|$ is assumed, and 
$h(t_1,t_2)$ is continuous at $(t_1,t_2)$ if $\arg \min t_k^2$ is unique,
we know that 
%$h(\alphaS/ \sigmaalpha, \betaM/\sigmabeta )$ 
the function $h(t_1, t_2)$ 
is continuous at $( \alphaS/\sigmaalpha, \betaM/\sigmabeta)$.
Then by 
$n^{-1/2}(T_{\alpha,n},  T_{\beta,n})= (\hatalphaSn/\hatsigmaalphan,   \hatbetaMn/\hatsigmabetan )\xrightarrow{a.s.} ( \alphaS/\sigmaalpha, \betaM/\sigmabeta)$  
%\begin{align} \label{eq:tstatlimits}
%	n^{-1/2}(T_{\alpha,n},  T_{\beta,n})= (\hatalphaSn/\hatsigmaalphan,   \hatbetaMn/\hatsigmabetan )\xrightarrow{a.s.} ( \alphaS/\sigmaalpha, \betaM/\sigmabeta) 
%\end{align}
and the continuous mapping theorem, 
we have 
\begin{align}\label{eq:indicalimith}
h(T_{\alpha, n}, T_{\beta, n})=h\biggr(\frac{\hatalphaSn}{\hatsigmaalphan}, \frac{\hatbetaMn}{\hatsigmabetan} \biggr) \xrightarrow{a.s.} h\biggr(\frac{\alphaSn}{\sigmaalpha}, \frac{\betaMn}{\sigmabeta} \biggr).	
\end{align}
% $h(T_{\alpha, n}, T_{\beta, n})\xrightarrow{a.s.} h(\alphaS/ \sigmaalpha, \betaM/\sigmabeta )$. 
% where we use the 
By the definitions of $\hatmaxpn$ and $\maxpnulln$, we write 
\begin{align}
%&~\sqrt{n}(\hatmaxpn - \maxpnulln)\label{eq:thetacenterexp} \\
\sqrt{n}(\hatmaxpn - \maxpnulln)=&~\sqrt{n}\biggr(\frac{\hatalphaSn}{\hatsigmaalphan } -\frac{ \alphaSn }{\sigmaalpha }, \frac{\hatbetaMn}{\hatsigmabetan }-\frac{ \betaMn}{\sigmabeta } \biggr)\times h\biggr(\frac{\hatalphaSn}{\hatsigmaalphan}, \frac{\hatbetaMn}{\hatsigmabetan} \biggr)\notag \\ %\label{eq:thetacenterexp} \\
&~ + \sqrt{n}\biggr(\frac{\alphaSn}{\sigmaalpha}, \frac{\betaMn}{\sigmabeta} \biggr)\times \biggr\{h\biggr(\frac{\hatalphaSn}{\hatsigmaalphan}, \frac{\hatbetaMn}{\hatsigmabetan} \biggr) -h\biggr(\frac{\alphaSn}{\sigmaalpha}, \frac{\betaMn}{\sigmabeta} \biggr) \biggr\}. \notag
\end{align}
%By \eqref{eq:limitcoef} and Lemma \ref{lm:sigmaconsis}, 
%\begin{align*}
%h\biggr(\frac{\hatalphaSn}{\hatsigmaalphan}, \frac{\hatbetaMn}{\hatsigmabetan} \biggr) -h\biggr(\frac{\alphaSn}{\sigmaalpha}, \frac{\betaMn}{\sigmabeta} \biggr)\xrightarrow{P_n} 0,	
%\end{align*}
%and thus 
%\begin{align*}
%\sqrt{n}(\hatmaxpn - \maxpnulln)=\sqrt{n}\biggr(\frac{\hatalphaSn-\alphaSn}{\sigmaalpha}, \frac{\hatbetaMn-\betaMn}{\sigmabeta }\biggr)	h\biggr(\frac{\alphaSn}{\sigmaalpha}, \frac{\betaMn}{\sigmabeta} \biggr)+o_{P_n}(1).
%\end{align*}
%By \eqref{eq:indicalimith} and Lemma \ref{lm:sigmaconsis}, 
%\begin{align*}
%\sqrt{n}(\hatmaxpn - \maxpnulln)=\sqrt{n}\biggr(\frac{\hatalphaSn-\alphaSn}{\sigmaalpha}, \frac{\hatbetaMn-\betaMn}{\sigmabeta }\biggr)	h\biggr(\frac{\alphaSn}{\sigmaalpha}, \frac{\betaMn}{\sigmabeta} \biggr)+o_{P_n}(1).
%\end{align*}
%Therefore, by \eqref{eq:limitcoef} and Slutsky's lemma, we obtain
Combining  \eqref{eq:limitcoef}, \eqref{eq:indicalimith}, Lemma \ref{lm:sigmaconsis}, Slutsky's lemma, and the continuous mapping theorem, we obtain 
\begin{align*}
	\sqrt{n}(\hatmaxpn - \maxpnulln ) \xrightarrow{d} \biggr(\frac{Z_{\S}}{\sigma_{\alphaS}}, \frac{Z_{\M}}{\sigma_{\betaM}} \biggr)\times h\biggr(\frac{\alphaS}{\sigmaalpha }, \frac{\betaM}{\sigmabeta} \biggr).
\end{align*}

% $\sqrt{n}(\hatmaxpn - \maxpnulln ) \xrightarrow{d} (Z_{\S}/V_{\S}, Z_{\M}/V_{\M}) \times h(\alphaS/\sigmaalpha, \betaM/\sigmabeta ).$ 

%Then by \eqref{eq:limitcoef}, \eqref{eq:thetacenterexp}, Lemma \ref{lm:sigmaconsis}, and Slutsky's lemma, 
%we obtain $\sqrt{n}(\hatmaxpn - \maxpnulln ) \xrightarrow{d} (Z_{\S}/V_{\S}, Z_{\M}/V_{\M}) \times h(\alphaS/\sigmaalpha, \betaM/\sigmabeta ).$ 
%$(\hatsigmaalphan, \hatsigmabetan )\xrightarrow{a.s.} (\sigmaalpha, \sigmabeta)$, and Slutsky's lemma,   

\medskip
 
\noindent \textit{Case 2:} When $(\alphaS, \betaM) = (0,0)$, 
we have $\sqrt{n}(\alphaSn, \betaMn)=(\balpha, \bbeta)$, and then  $\sqrt{n}\maxpnulln=H(\balpha/\sigmaalpha, \bbeta/\sigmabeta )$. 
Moreover, by  \eqref{eq:limitcoef} and Lemma \ref{lm:sigmaconsis}, we obtain
\begin{align*}
(T_{\alpha,n},\, T_{\beta,n})=&~\sqrt{n}\biggr(\frac{\hatalphaSn -\alphaSn}{\hatsigmaalphan },\, \frac{\hatbetaMn - \betaMn }{\hatsigmabetan } \biggr)+\biggr(\frac{\balpha}{\hatsigmaalphan},\, \frac{ \bbeta}{\hatsigmabetan} \biggr)\notag \\
\xrightarrow{d} &~	(K_{b,\S}, K_{b,\M}). 
\end{align*}
Since $(Z_{\S}, Z_{\M})$ is a normal random vector and 
 $|\mathrm{corr}(Z_{\S}, Z_{\M})| <1,$ we have  $|K_{b,\S}|\neq |K_{b,\M}|$ a.s..  
 As $h(t_1,t_2)$ is continuous at $(t_1,t_2)$ if $\arg \min t_k^2$ is unique,
we can apply the continuous mapping theorem, 
and obtain 
$\sqrt{n}(\hatmaxpn - \maxpnulln)	\xrightarrow{d} H(K_{b,\S}, K_{b,\M})-H(\balpha/\sigmaalpha, \bbeta/\sigmabeta )$.

\bigskip

\subsection{Proof of Theorem \ref{thm:bootstrapmaxp}} \label{sec:pfbootstrapmaxp}
%To prove Theorem \ref{thm:bootstrapmaxp}, by Lemma \ref{lm:indicatorconsistency}, it suffices to discuss  two cases $(\alphaS, \betaM) \neq (0,0)$ and  $(\alphaS, \betaM)= (0,0)$ separately.

%Similarly to the proofs of Theorems \ref{thm:bootstrapprod} and \ref{thm:bootstrapprodcomb}, by ..., 

To prove Theorem \ref{thm:bootstrapmaxp}, by Theorem \ref{thm:maxplimit1} and \eqref{eq:indicatorconvg}, it suffices to 
prove 
\begin{itemize}
	\item[(i)]  when $(\alphaS, \betaM) \neq (0,0)$ and $|\alphaS/ \sigmaalpha |\neq|\betaM/ \sigmabeta|$, 
	\begin{align*}
		\sqrt{n}(\hatmaxpn^* - \hatmaxpn)
  % \xrightarrow{d}
  \overset{d^*}{\leadsto} 
  \biggr(\frac{Z_{\S}}{\sigma_{\alphaS}}, \frac{Z_{\M}}{\sigma_{\betaM}} \biggr)\times h\biggr(\frac{\alphaS}{\sigmaalpha }, \frac{\betaM}{\sigmabeta} \biggr); 
	\end{align*}
%	\begin{align*}
%	\sqrt{n}(\hatmaxpn^* - \hatmaxpn )=&~
%		H(T_{\alpha, n}^*,T_{\beta, n}^*) -  H(T_{\alpha,n},T_{\beta,n})\notag \\
%		\xrightarrow{d}&~ H(T_{\alpha,n},T_{\beta,n})-\sqrt{n}H(\alphaS/\sigmaalpha, \betaM/\sigmabeta )
%	\end{align*}
	 % conditionally (on the data) in probability; 
	\item[(ii)] when  $(\alphaS, \betaM)=(0,0)$, $\mathbb{V}_{2,n}^*(\balpha, \bbeta) 
 % \xrightarrow{d} 
 \overset{d^*}{\leadsto} 
 H (K_{b,\S},K_{b,\M}) -H({\balpha}/{ \sigmaalpha }, {\bbeta}/{\sigmabeta})$.
 % conditionally (on the data) in probability.
\end{itemize}
%discuss two cases separately.
%We prove Theorem \ref{thm:bootstrapmaxp} by discussing two cases.
%\begin{itemize}
%	\item[(i)]  when $(\alphaS, \betaM) \neq (0,0)$ and $|\alphaS/ \sigmaalpha |\neq|\betaM/ \sigmabeta|$, $\sqrt{n}(\hatmaxpn^* - \hatmaxpn )\xrightarrow{d}\sqrt{n}(\hatmaxpn - \maxpnulln)$
%%	\begin{align*}
%%	\sqrt{n}(\hatmaxpn^* - \hatmaxpn )=&~
%%		H(T_{\alpha, n}^*,T_{\beta, n}^*) -  H(T_{\alpha,n},T_{\beta,n})\notag \\
%%		\xrightarrow{d}&~ H(T_{\alpha,n},T_{\beta,n})-\sqrt{n}H(\alphaS/\sigmaalpha, \betaM/\sigmabeta )
%%	\end{align*}
%	 conditionally (on the data) in probability; 
%	\item[(ii)] when  $(\alphaS, \betaM)=(0,0)$, $\mathbb{V}_{2,n}^*(\balpha, \bbeta) \xrightarrow{d} \mathbb{V}_{2,n}(\balpha, \bbeta)$ conditionally (on the data) in probability.
%\end{itemize}
%
%When $(\alphaS, \betaM) \neq (0,0)$, we note that
\noindent For (i), we note that 
\begin{align*}
	\sqrt{n}(\hatmaxpn^* - \hatmaxpn )=&~\sqrt{n}\biggr(\frac{\hatalphaSn^* }{\hatsigmaalphan^* } -\frac{\hatalphaSn}{\hatsigmaalphan}, \frac{\hatbetaMn^* }{\hatsigmabetan^* }-\frac{\hatbetaMn}{\hatsigmabetan}\biggr)\times h\biggr(\frac{\hatalphaSn^* }{\hatsigmaalphan^* }, \frac{\hatbetaMn^* }{\hatsigmabetan^* } \biggr)	&~ \notag \\
	&~\sqrt{n}\biggr(\frac{\hatalphaSn }{\hatsigmaalphan }, \frac{\hatbetaMn }{\hatsigmabetan } \biggr)\times \biggr\{h\biggr(\frac{\hatalphaSn^*}{\hatsigmaalphan^*}, \frac{\hatbetaMn^*}{\hatsigmabetan^*}\biggr)  -h\biggr(\frac{\hatalphaSn}{\hatsigmaalphan}, \frac{\hatbetaMn}{\hatsigmabetan} \biggr) \biggr\}.
\end{align*}
%By \eqref{eq:limitcoef}, \eqref{eq:coefbootlimitsep}, and Lemma \ref{lm:sigmaconsis}, 
%$({\hatalphaSn^*}/{\hatsigmaalphan^*}, {\hatbetaMn^*}/{\hatsigmabetan^*}) \xrightarrow{P_C} ({\alphaSn}/{\sigmaalpha}, {\betaMn}/{\sigmabeta})$ 
%conditionally (on the data) in probability.
%Then similarly to \eqref{eq:indicalimith}, by the continuous mapping theorem, we have $h({\hatalphaSn^*}/{\hatsigmaalphan^*}, {\hatbetaMn^*}/{\hatsigmabetan^*}) \xrightarrow{P_C} h({\alphaSn}/{\sigmaalpha}, \betaMn/\sigmabeta)$. 
Similarly to Section \ref{sec:pfmaxplimit1}, by \eqref{eq:limitcoef}, \eqref{eq:zzbootlimit}, Lemma \ref{lm:sigmaconsis}, Slutsky's lemma,  and  the continuous mapping theorem, we have
$h({\hatalphaSn^*}/{\hatsigmaalphan^*}, {\hatbetaMn^*}/{\hatsigmabetan^*}) 
\overset{\mathrm{P}^*}{\leadsto} 
% \xrightarrow{P_C} 
h({\alphaSn}/{\sigmaalpha}, \betaMn/\sigmabeta)$,
% conditionally (on the data) in probability,
and then 
%by  \eqref{eq:indicalimith}, 
(i) is obtained.
In addition, for (ii), 
by  \eqref{eq:zzbootlimit} and Lemma \ref{lm:sigmaconsis},
%by Lemmas \ref{lm:limitzboot}  and \ref{lm:sigmaconsis}, 
we have $(\mathbb{K}_{b,\S}^*, \mathbb{K}_{b,\M}^*) 
% \xrightarrow{d} 
\overset{d^*}{\leadsto} 
(K_{b,\S}, K_{b,\M})$, 
% conditionally (on the data) in probability,  
and then  
(ii) follows by the continuous mapping theorem similarly to  Section \ref{sec:pfmaxplimit1}.

\subsection{Proofs of Assisted Lemmas}\label{sec:pfalllemma}

\subsubsection{Proof of Lemma \ref{lm:regresstrans}}\label{sec:pfregresstrans}

%By Frisch–Waugh–Lovell theorem, 
%The conclusions follow from the Frisch–Waugh–Lovell theorem and the properties of the ordinary least squares estimator. For self-consistence, we provide a  proof below.

 \textit{Part (1).}  
%Given the model ,  
Multiplying both sides of $\M = \alphaSn \S +  \X^\mytrans \alphaX +\eM$ by $\{P_n(\X^{\mytrans}\X)\}^{-1}\X$ and taking expectation, yields
$Q_{1,\M} = Q_{1,\S} \alphaSn + \alphaX $,
 where we use $\mathrm{E}(\X \eM)=\mathbf{0}$. It follows that
 \begin{align*}
 \M-\X^{\mytrans}Q_{1,\M}= &~ \alphaSn \S +  \X^\mytrans \alphaX +\eM - \X^{\mytrans}(Q_{1,\S} \alphaSn + \alphaX )\\
 = &~ (\S-\X^{\mytrans}Q_{1,\S})\alphaSn + \eM,
 \end{align*}
% \begin{align*}
% \M-\X^{\mytrans}Q_{1,\M}=(\S-\X^{\mytrans}Q_{1,\S})\alphaSn + (\X-\X)^{\mytrans}\alphaX + \eM = (\S-\X^{\mytrans}Q_{1,\S})\alphaSn + \eM,  
% \end{align*} 
 that is, $\Mperp=\Sperp\alphaSn+\eM$. The second model in \eqref{eq:transmodel} can be obtained similarly. 
% The conclusions in \eqref{eq:transformmodel1} follow from the form of the ordinary least squares   estimator.

\medskip
\noindent \textit{Part (2).}  
For $n$ independent and identically distributed observations $\{(\S_i, \M_i, \X_i): i=1,\ldots, n\}$,  
we write $\mathcal{S}_n=(\S_1,\ldots, \S_n)^{\mytrans}$, $\mathcal{M}_n=(\M_1,\ldots, \M_n)^{\mytrans}$, and $\mathcal{X}_n=(\X_1,\ldots, \X_n)^{\mytrans}$. It follows that  the ordinary least square estimator of $(\alphaS, \alphaX^{\mytrans})^{\mytrans}$ is
\begin{align*}
\begin{bmatrix}
		\hatalphaSn \\[3pt]
		\hatalphaXn
	\end{bmatrix}=\begin{bmatrix}
		\mathcal{S}_n^{\mytrans}\mathcal{S}_n &\mathcal{S}_n^{\mytrans} \mathcal{X}_n \\[3pt]
	\mathcal{X}_n^{\mytrans}\mathcal{S}_n & \mathcal{X}_n^{\mytrans}\mathcal{X}_n
	\end{bmatrix}^{-1}\begin{bmatrix}
		\mathcal{S}_n^{\mytrans} \mathcal{M}_n\\[3pt]
		\mathcal{X}_n^{\mytrans} \mathcal{M}_n
	\end{bmatrix}.
\end{align*}
By the blockwise matrix inversion, we obtain
\begin{align}
&~\begin{bmatrix}
		\mathcal{S}_n^{\mytrans}\mathcal{S}_n &\mathcal{S}_n^{\mytrans} \mathcal{X}_n \\[3pt]
	\mathcal{X}_n^{\mytrans}\mathcal{S}_n & \mathcal{X}_n^{\mytrans}\mathcal{X}_n
	\end{bmatrix}^{-1} \notag \\[2pt]
=&~ 
	\begin{bmatrix}
		(\mathcal{S}_n^{\mytrans}\mathbb{O}_{\X}\mathcal{S}_n)^{-1} & -(\mathcal{S}_n^{\mytrans}\mathbb{O}_{\X}\mathcal{S}_n)^{-1}\mathcal{S}_n^{\mytrans}\mathcal{X}_n(\mathcal{X}_n^{\mytrans}\mathcal{X}_n)^{-1} \\[3pt]
		-(\mathcal{X}_n^{\mytrans}\mathbb{O}_{\S}\mathcal{X}_n)^{-1}\mathcal{X}_n^{\mytrans}\mathcal{S}_n(\mathcal{S}_n^{\mytrans}\mathcal{S}_n)^{-1}& (\mathcal{X}_n^{\mytrans}\mathbb{O}_{\S} \mathcal{X}_n)^{-1}
	\end{bmatrix}, \label{eq:inversematrix}
\end{align} where 
$\mathbb{O}_{\X}=\myident_n - \mathcal{X}_n(\mathcal{X}_n^{\mytrans}\mathcal{X}_n )^{-1}\mathcal{X}_n^{\mytrans}$, 
%$\mathbb{O}_{\X}=\myident_n -  (\X_1,\ldots, \X_n)^{\mytrans}(\sum_{i=1}^n \X_i \X_i^{\mytrans} )^{-1}(\X_1,\ldots, \X_n)$,  
$\mathbb{O}_{\S}=\myident_n- \mathcal{S}_n (\mathcal{S}_n^{\mytrans} \mathcal{S}_n)^{-1}\mathcal{S}_n^{\mytrans}$,
and $\myident_n$ denotes an $n\times n$ identity matrix. 
%$\mathbb{O}_{\S}=\myident_n- (\S_1,\ldots, \S_n)^{\mytrans} \times (\sum_{i=1}^n \S_i \S_i^{\mytrans} )^{-1}(\S_1,\ldots, \S_n)$. 
It follows that 
\begin{align}\label{eq:alphahatsxjoint}
\begin{bmatrix}
		\hatalphaSn \\[3pt]
		\hatalphaXn
	\end{bmatrix} = 
\begin{bmatrix}
(\mathcal{S}_n^{\mytrans}\mathbb{O}_{\X}\mathcal{S}_n)^{-1}\mathcal{S}_n^{\mytrans} \mathbb{O}_{\X} \mathcal{M}_n\\[3pt]
(\mathcal{X}_n^{\mytrans}\mathbb{O}_{\S} \mathcal{X}_n)^{-1} \mathcal{X}_n^{\mytrans}\mathbb{O}_{\S} \mathcal{M}_n
\end{bmatrix}. 
\end{align}
Since $\mathbb{O}_{\X} = \mathbb{O}_{\X}\mathbb{O}_{\X}$ and $\mathbb{O}_{\X}=\mathbb{O}_{\X}^{\mytrans}$, 
$\hatalphaSn=\{(\mathbb{O}_{\X}\mathcal{S}_n)^{\mytrans}\mathbb{O}_{\X}\mathcal{S}_n\}^{-1}(\mathbb{O}_{\X}\mathcal{S}_n)^{\mytrans} \mathbb{O}_{\X} \mathcal{M}_n$.
Then we obtain $\hatalphaSn = \mathbb{P}_n(\hatSperp\hatMperp)/\mathbb{P}_n(\hatSperp^2) $ by noting that $\mathbb{O}_{\X}\mathcal{S}_n = ( \hat{\S}_{\perp,1},\ldots,  \hat{\S}_{\perp,n})$ and $\mathbb{O}_{\X}\mathcal{M}_n = (\hat{\M}_{\perp,1},\ldots, \hat{\M}_{\perp,n})$. 
Following similar analysis, we also have $ \hatbetaMn = \mathbb{P}_n( \hatMperpp \hatYperpp)/\mathbb{P}_n(\hatMperpp^2)$. 

\medskip 
\noindent \textit{Part (3).} 
%Let $\mathcal{E}_{\M,n}=(\hat{\epsilon}_{\M,1},\ldots, \hat{\epsilon}_{\M,n})^{\mytrans}$, i.e., the vector of $n$ residuals. By the definition of residuals, $\mathcal{E}_{\M,n}=\mathcal{M}_n-\mathcal{S}_n\hatalphaSn - \mathcal{X}_n\hatalphaXn$.
Let $\mathcal{E}_{\M,n}=(\hat{\epsilon}_{\M,1},\ldots, \hat{\epsilon}_{\M,n})^{\mytrans} $ denote the vector of $n$ residuals from the ordinary least squares regression,
and by the definitions, we can write $\mathcal{E}_{\M,n}=\mathcal{M}_n-\mathcal{S}_n\hatalphaSn - \mathcal{X}_n\hatalphaXn$. 
Since  $\mathbb{O}_{\X}\mathcal{S}_n = ( \hat{\S}_{\perp,1},\ldots,  \hat{\S}_{\perp,n})$ and $\mathbb{O}_{\X}\mathcal{M}_n = (\hat{\M}_{\perp,1},\ldots, \hat{\M}_{\perp,n})$, 
proving $\hatenMi = \hatMperpi-\hatalphaSn \hatSperpi$ for $i=1,\ldots, n$ can be written as $\mathcal{E}_{\M,n}=\mathbb{O}_{\X}\mathcal{M}_n - \mathbb{O}_{\X}\mathcal{S}_n\hatalphaSn$, which  is equivalent to showing $\mathcal{M}_n-\mathcal{S}_n\hatalphaSn - \mathcal{X}_n\hatalphaXn = \mathbb{O}_{\X}\mathcal{M}_n - \mathbb{O}_{\X}\mathcal{S}_n\hatalphaSn$
By \eqref{eq:alphahatsxjoint}, it suffices to prove 
\begin{align}
\mathbb{N}_{\X}-\mathbb{N}_{\X}\mathcal{S}_n(\mathcal{S}_n^{\mytrans}\mathbb{O}_{\X}\mathcal{S}_n)^{-1}\mathcal{S}_n^{\mytrans}(\myident_n - \mathbb{N}_{\X}) = \mathcal{X}_n(\mathcal{X}_n^{\mytrans}\mathbb{O}_{\S} \mathcal{X}_n)^{-1} \mathcal{X}_n^{\mytrans}(\myident_n - \mathbb{N}_{\S}), \label{eq:pfproj0}
\end{align} where we define $\mathbb{N}_{\X}=\myident_n-\mathbb{O}_{\X}$ and $ \mathbb{N}_{\S}=\myident_n - \mathbb{O}_{\S}$. 

By the symmetricity of the matrix in \eqref{eq:inversematrix}, 
\begin{align}
(\mathcal{X}_n^{\mytrans}\mathcal{X}_n)^{-1}\mathcal{X}_n^{\mytrans}	\mathcal{S}_n(\mathcal{S}_n^{\mytrans}\mathbb{O}_{\X}\mathcal{S}_n)^{-1} = (\mathcal{X}_n^{\mytrans}\mathbb{O}_{\S}\mathcal{X}_n)^{-1}\mathcal{X}_n^{\mytrans}\mathcal{S}_n(\mathcal{S}_n^{\mytrans}\mathcal{S}_n)^{-1}. \label{eq:pfproj1}
\end{align}
Multiplying the left and right hand sides of \eqref{eq:pfproj1} by $\mathcal{X}_n$ and $\mathcal{S}_n^{\mytrans}$, respectively, yields
\begin{align}
\mathbb{N}_{\X}\mathcal{S}_n(\mathcal{S}_n^{\mytrans}\mathbb{O}_{\X}\mathcal{S}_n)^{-1}\mathcal{S}_n^{\mytrans} = \mathcal{X}_n(\mathcal{X}_n^{\mytrans}\mathbb{O}_{\S}\mathcal{X}_n)^{-1} \mathcal{X}_n^{\mytrans}\mathbb{N}_{\S}. \label{eq:pfproj2}
\end{align}
In addition, by the Woodbury identity, 
\begin{align*}
	&~ ( \mathcal{X}_n^{\mytrans}\mathbb{O}_{\S} \mathcal{X}_n)^{-1} \notag \\
=&~ ( \mathcal{X}_n^{\mytrans} \mathcal{X}_n -  \mathcal{X}_n^{\mytrans}\mathcal{S}_n(\mathcal{S}_n^{\mytrans}\mathcal{S}_n)^{-1}\mathcal{S}_n^{\mytrans} \mathcal{X}_n)^{-1} \notag \\
=&~ ( \mathcal{X}_n^{\mytrans} \mathcal{X}_n)^{-1}-( \mathcal{X}_n^{\mytrans} \mathcal{X}_n)^{-1}  \mathcal{X}_n^{\mytrans}\mathcal{S}_n \{-\mathcal{S}_n^{\mytrans}\mathcal{S}_n+ \mathcal{S}_n^{\mytrans} \mathcal{X}_n( \mathcal{X}_n^{\mytrans} \mathcal{X}_n)^{-1}  \mathcal{X}_n^{\mytrans}\mathcal{S}_n \}^{-1}\mathcal{S}_n^{\mytrans} \mathcal{X}_n( \mathcal{X}_n^{\mytrans} \mathcal{X}_n)^{-1}.\notag
\end{align*}
Therefore, 
\begin{align}
\mathcal{X}_n(\mathcal{X}_n^{\mytrans}\mathbb{O}_{\S}\mathcal{X}_n)^{-1} \mathcal{X}_n^{\mytrans} = \mathbb{N}_{\X}+\mathbb{N}_{\X}\mathcal{S}_n(\mathcal{S}_n^{\mytrans}\mathbb{O}_{\X}\mathcal{S}_n)^{-1}\mathcal{S}_n^{\mytrans}\mathbb{N}_{\X}.\label{eq:pfproj3}
\end{align}
Combining \eqref{eq:pfproj2} and \eqref{eq:pfproj3}, \eqref{eq:pfproj0} is proved. 
Therefore, we obtain  
 $\hatenMi = \hatMperpi-\hatalphaSn \hatSperpi$ for $i=1,\ldots, n$.
 Similarly, we also have $\hatenYi = \hatYperppi-\hatbetaMn \hatMperppi$ for $i=1,\ldots, n$.

\medskip
\noindent \textit{Part (4).} By the property of the ordinary least squares regression, we know $ \mathcal{E}_{\M,n}^{\mytrans}\mathcal{S}_n=0$ and $ \mathcal{E}_{\M,n}^{\mytrans}\mathcal{X}_n=\mathbf{0}$. Therefore, 
\begin{align*}
	n\mathbb{P}_n(\hatenM \hatSperp ) =  \mathcal{E}_{\M,n}^{\mytrans}\mathbb{O}_{\X} \mathcal{S}_n =\mathcal{E}_{\M,n}^{\mytrans} \mathcal{S}_n - \mathcal{E}_{\M,n}^{\mytrans} \mathcal{X}_n(\mathcal{X}_n^{\mytrans}\mathcal{X}_n)^{-1}\mathcal{X}_n^{\mytrans}\mathcal{S}_n  = 0. 
\end{align*}
Following similar analysis, we also have $\mathbb{P}_n(\hatenY \hat{\M}_g ) = 0$. 

\medskip
\noindent \textit{Part (5).} 
By the property of the ordinary least squares regressions, we know the square of the standard error of $\hatalphaSn$, that is, $\hatsigmaalphan^2/n$, is the entry in the first row and the first column of
\begin{align*}
	\mathbb{P}_n(\hatenM^2)\times \begin{bmatrix}
		\mathcal{S}_n^{\mytrans}\mathcal{S}_n &\mathcal{S}_n^{\mytrans} \mathcal{X}_n \\[3pt]
	\mathcal{X}_n^{\mytrans}\mathcal{S}_n & \mathcal{X}_n^{\mytrans}\mathcal{X}_n
	\end{bmatrix}^{-1}.
\end{align*}
By \eqref{eq:inversematrix} and $\mathbb{O}_{\X}=\mathbb{O}_{\X}^{\mytrans}\mathbb{O}_{\X}$, $\hatsigmaalphan^2 = n\mathbb{P}_n(\hatenM^2)(\mathcal{S}_n^{\mytrans}\mathbb{O}_{\X}^{\mytrans}\mathbb{O}_{\X}\mathcal{S}_n)^{-1}=\mathbb{P}_n(\hatenM^2)/\mathbb{P}_n(\hatSperp^2)$, where we use $\mathbb{O}_{\X}\mathcal{S}_n = ( \hat{\S}_{\perp,1},\ldots,  \hat{\S}_{\perp,n})$. 
Similarly, we also have $\hatsigmabetan^2 = \mathbb{P}_n( \hatenY^2)/\mathbb{P}_n(\hatMperp^2)$. 

\subsubsection{Proof of Lemma \ref{lm:sigmaconsis}} \label{sec:pfsigmaconsis}
\noindent \textit{Part (1)} ~ In the first part, we prove $(\hatsigmaalphan, \hatsigmabetan) \xrightarrow{a.s.} (\sigmaalpha, \sigmabeta)$. 
By  Lemma \ref{lm:regresstrans},  $\hatsigmaalphan^2 =\mathbb{P}_n(\hatenM^2)/\mathbb{P}_n(\hatSperp^2)$. 
To prove $\hatsigmaalphan^2\xrightarrow{a.s.}\sigmaalpha^2$, 
by Slutsky's lemma, 
it suffices to prove $\mathbb{P}_n(\hatenM^2)\xrightarrow{a.s.} \mathrm{E}(\eM^2)$ and $\mathbb{P}_n(\hatSperp^2) \xrightarrow{a.s.} \mathrm{E}(\Sperp^2)$ separately. 
In particular, 
\begin{align*}
&~\mathbb{P}_n(\hatenM^2) - \mathrm{E}(\eM^2)\notag \\
=&~	\mathbb{P}_n \{( \hatenM - \eM)(\hatenM + \eM )\} + \mathbb{P}_n(\eM^2)-\mathrm{E}(\eM^2) \notag \\
=&~ \mathbb{P}_n \{(\alphaSn\S -\hatalphaSn\S + \alphaX^{\mytrans}\X - \hatalphaXn^{\mytrans}\X)(\hatenM + \eM )\} +\mathbb{P}_n(\eM^2)-\mathrm{E}(\eM^2)   \notag \\
=&~( \alphaSn - \hatalphaSn )\mathbb{P}_n(\S\eM) + (\alphaX - \hatalphaXn )^{\mytrans}\mathbb{P}_n(\X\eM) +  \mathbb{P}_n(\eM^2)-\mathrm{E}(\eM^2), 
\end{align*}
where in the last equation, we use $\mathbb{P}_n( \S \hatenM)=0$ and $\mathbb{P}_n( \X \hatenM)=\mathbf{0}$ by the property of the ordinary least squares regressions. 
%Lemma \ref{lm:regresstrans}. 
By the strong law of large numbers, we have $\mathbb{P}_n(\S\eM)\xrightarrow{a.s.} \mathrm{E}(\S\eM) =0$,  $\mathbb{P}_n(\X\eM) \xrightarrow{a.s.} 0$,  $\mathbb{P}_n(\eM^2)-\mathrm{E}(\eM^2)\xrightarrow{a.s.} \mathbf{0}$, 
 $\hatalphaSn - \alphaSn\xrightarrow{a.s.}0$, and $\hatalphaXn - \alphaX \xrightarrow{a.s.} \mathbf{0}$.
 Therefore, $\mathbb{P}_n(\hatenM^2) - \mathrm{E}(\eM^2)\xrightarrow{a.s.} 0$ is proved.
In addition, 
\begin{align*}
%&~ \mathbb{P}_n(\hatSperp^2) - \mathrm{E}(\Sperp^2) \notag \\
 \mathbb{P}_n(\hatSperp^2) - \mathrm{E}(\Sperp^2)=&~	\mathbb{P}_n\{(\hatSperp-\Sperp)(\hatSperp+\Sperp)\}+\mathbb{P}_n(\Sperp^2)- \mathrm{E}(\Sperp^2) \notag \\
=&~	\mathbb{P}_n\{(\hatSperp-\Sperp)(\Sperp-\hatSperp)\}+\mathbb{P}_n(\Sperp^2)- \mathrm{E}(\Sperp^2) \notag \\
=&~		(Q_{1,\S}-\hat{Q}_{1,\S})^{\intercal}\mathbb{P}(\X\X^{\intercal})(\hat{Q}_{1,\S}-Q_{1,\S})+\mathbb{P}_n(\Sperp^2)- \mathrm{E}(\Sperp^2), 
\end{align*}
where in the second equation, we use $\PP\{(\hatSperp-\Sperp)\hatSperp\}=(Q_{1,\S}-\hat{Q}_{1,\S})^{\mytrans}\mathbb{P}_n\{\X(\S-\X^{\mytrans}\hat{Q}_{1,\S})\}=0$.
Then  $\mathbb{P}_n(\hatSperp^2) \xrightarrow{a.s.} \mathrm{E}(\Sperp^2)$ holds by Lemma \ref{lm:consistencyqmoments} and $\mathbb{P}_n(\Sperp^2)\xrightarrow{a.s.} \mathrm{E}(\Sperp^2)$.

% where in the last equation, we use $\mathbb{P}_n\{\X(\S-\X^{\mytrans}\hat{Q}_{1,\S})\}=\mathbf{0}$ by the definition of $\hat{Q}_{1,\S}$. Then  $\mathbb{P}_n(\hatSperp^2) \xrightarrow{a.s.} \mathrm{E}(\Sperp^2)$ holds by Lemma \ref{lm:consistencyqmoments} and $\mathbb{P}_n(\Sperp^2)\xrightarrow{a.s.} \mathrm{E}(\Sperp^2)$.
%\begin{align*}
%&~ \mathbb{P}_n(\hatSperp^2) - \mathrm{E}(\Sperp^2) \notag \\
%=&~  \mathbb{P}_n \{(\S - \X^{\mytrans}\hat{Q}_{1,\S})^2\}- P_n\{( \S - \X^{\mytrans}Q_{1,\S})^2\} \notag \\
%%=&~ \mathbb{P}_n[(Q_{1,\S}-\hat{Q}_{1,\S})^{\mytrans}\X\{ 2(\S-\X^{\mytrans}\hat{Q}_{1,\S}) + \X^{\mytrans}\hat{Q}_{1,\S} - \X^{\mytrans}Q_{1,\S}\}]  \notag \\
%=&~ 2(Q_{1,\S}-\hat{Q}_{1,\S})^{\mytrans}\mathbb{P}_n\{\X(\S-\X^{\mytrans}\hat{Q}_{1,\S})\}+(Q_{1,\S}-\hat{Q}_{1,\S})^{\mytrans}\mathbb{P}_n(\X\X^{\mytrans}) (\hat{Q}_{1,\S} - Q_{1,\S} ) \notag \\
%=&~(Q_{1,\S}-\hat{Q}_{1,\S})^{\mytrans}\mathbb{P}_n(\X\X^{\mytrans}) (\hat{Q}_{1,\S} - Q_{1,\S}),
%\end{align*} where in the last equation, we use $\mathbb{P}_n\{\X(\S-\X^{\mytrans}\hat{Q}_{1,\S})\}=\mathbf{0}$ by the definition of $\hat{Q}_{1,\S}$. Then  $\mathbb{P}_n(\hatSperp^2) \xrightarrow{a.s.} \mathrm{E}(\Sperp^2)$ holds by Lemma \ref{lm:consistencyqmoments}.

%\paragraph{Part 2}

\medskip
\noindent \textit{Part (2)}~  
In the second part, we prove  $(\hatsigmaalphan^*, \hatsigmabetan^*) 
% \xrightarrow{P_C} 
\overset{\mathrm{P}^*}{\leadsto} 
(\sigmaalpha, \sigmabeta )$.
% conditionally (on the data) in probability. 
Particularly, we focus on $\hatsigmaalphan^*$ below, and $\hatsigmabetan^*$ can be analyzed similarly. 
Let $\hatenM^*$ denote the residuals obtained form the the  nonparametric bootstrap, i.e., the paired bootstrap regression,
and then following (5) in Lemma \ref{lm:regresstrans},
we write $(\hatsigmaalphan^*)^2 = {\mathbb{P}_n^*\{ (\hatenM^*)^2\}}/{ \mathbb{P}_n^* \{ (\Sperp^*)^2\}}$,
which is obtained by replacing $\hatenM$ and $\mathbb{P}_n(\cdot)$ in the formula of $\hatsigmaalphan^2$ with their nonparametric bootstrap versions $\hatenM^*$ and $\mathbb{P}_n^*(\cdot)$, respectively.

By the analysis of \eqref{eq:denominatorboot1}, 
we know $ \mathbb{P}_n^* \{ (\Sperp^*)^2\}  \overset{\mathrm{P}^*}{\leadsto} \mathrm{E}(\Sperp^2)$.
% conditionally (on the data) in probability. 
To prove $\hatsigmaalphan^* 
\overset{\mathrm{P}^*}{\leadsto} 
% \xrightarrow{P_C} 
\sigmaalpha$,
% conditionally (on the data) in probability, 
by Slutsky's lemma, it remains to show $ \mathbb{P}_n^*\{ (\hatenM^*)^2\}\overset{\mathrm{P}^*}{\leadsto} \mathrm{E}(\eM^2)$,
% conditionally (on the data) in probability. 
Particularly, 
%\begin{align}
%\mathbb{P}_n^*\{ (\hatenM^*)^2\} - P_n(\eM^2)=  \mathbb{P}_n^*\{ (\hatenM^* - \eM ) ( \hatenM^*+ \eM)\}+  \mathbb{P}_n^*(\eM^2)-P_n(\eM^2).	 \label{eq:bootresq}
%\end{align}
\begin{align}
\mathbb{P}_n^*\{ (\hatenM^*)^2\} - \mathrm{E}(\eM^2)=  \mathbb{P}_n^*\{ (\hatenM^* - \eM ) ( \hatenM^*+ \eM)\}+  \mathbb{P}_n^*(\eM^2)-\mathrm{E}(\eM^2).	 \label{eq:bootresq}
\end{align}
%By \eqref{eq:newdefvariables} and \eqref{eq:transmodel}, 
%$\eM = \M - \X^{\mytrans}{Q}_{1,\M} - \alphaSn (\S - \X^{\mytrans}{Q}_{1, \S})$,
%and then 
By the definitions in Section \ref{sec:preliminary} and \eqref{eq:transmodel}, 
%(1) in Lemma \ref{lm:regresstrans}, 
we have $\eM = \M - \X^{\mytrans}{Q}_{1,\M} - \alphaSn (\S - \X^{\mytrans}{Q}_{1, \S})$,
and then 
 the first summed term in \eqref{eq:bootresq} satisfies
% Following Lemma \ref{lm:regresstrans} similarly, under the nonparametric bootstrap, we can write $\hatenM^* = \M - \X^{\mytrans}{Q}_{1,\M}^* - \alphaSn^* (\S - \X^{\mytrans}{Q}_{1, \S}^*)$. 
% and similarly $\hatenM^* = \M - \X^{\mytrans}{Q}_{1,\M}^* - \alphaSn^* (\S - \X^{\mytrans}{Q}_{1, \S}^*)$. 
%Then the first summed term in \eqref{eq:bootresq} satisfies
\begin{align}
%	&~ \mathbb{P}_n^*\{ (\hatenM^*)^2\} - P_n(\eM^2) \notag \\
%=&~ \mathbb{P}_n^*\{ (\hatenM^*)^2\} - \mathbb{P}_n^*(\eM^2) +  \mathbb{P}_n^*(\eM^2)-P_n(\eM^2) \notag \\
&~  \mathbb{P}_n^*\{ (\hatenM^* - \eM ) ( \hatenM^*+ \eM)\} \label{eq:bootresq1} \\
=&~ \mathbb{P}_n^*[\{(Q_{1,\M}-Q_{1,\M}^*)^{\mytrans}\X + (\alphaSn - \hatalphaSn^*)\S - (\alphaSn Q_{1,\S} - \hatalphaSn^*Q_{1,\S}^*)^{\mytrans}\X\} ( \hatenM^*+ \eM)] \notag \\
=&~(Q_{1,\M}-Q_{1,\M}^*-\alphaSn Q_{1,\S}+\hatalphaSn^*Q_{1,\S}^*)^{\mytrans}\mathbb{P}_n^*(\X\eM)+(\alphaSn - \hatalphaSn^*)\mathbb{P}_n^*(\S\eM),\notag 
\end{align}
where in the last equation, we use $\mathbb{P}_n^*(\X \hatenM^*)=\mathbf{0}$ and $\mathbb{P}_n^*(\S\hatenM^*)=0$
by the property of the ordinary least squares regressions under the nonparametric bootstrap. 
Since  $\hatalphaSn^*-\alphaSn\overset{\mathrm{P}^*}{\leadsto}  0$,
% conditionally  (on the data) in probability, 
and by Lemma \ref{lm:consistencyqmoments}, 
we know 
$\eqref{eq:bootresq1}\overset{\mathrm{P}^*}{\leadsto}  0$.
% conditionally  (on the data) in probability. 
In addition, by the bootstrap consistency, $\mathbb{P}_n^*(\eM^2)-\mathrm{E}(\eM^2)\overset{\mathrm{P}^*}{\leadsto}  0$.
% conditionally (on the data) in probability.
In summary, we obtain $\mathbb{P}_n^*\{ (\hatenM^*)^2\} - \mathrm{E}(\eM^2)\overset{\mathrm{P}^*}{\leadsto}  0$.
% conditionally (on the data) in probability. 

\subsubsection{Proof of Lemma \ref{lm:perpconsisall}}\label{sec:pfperpconsisall}

\noindent \textit{Part (1)}~
By the formulae in \eqref{eq:olsfwlform}, 
we can write the ordinary least squares estimates of $\alphaSn$ and $\betaMn$ under the nonparametric bootstrap as
\begin{align*}
\hatalphaSn^*=\frac{\mathbb{P}_n^*(\Sperpboot  \Mperpboot)}{\mathbb{P}_n^*\{(\Sperpboot )^2\}}\hspace{2em} \text{and} \hspace{2em} \hatbetaMn^* = \frac{ \mathbb{P}_n^*(\Mperppboot \Yperppboot) }{\mathbb{P}_n^*\{(\Mperppboot)^2\}}, 		
\end{align*} 
%respectively,
which are obtained by replacing 
%where we replace 
$(\hatSperp, \hatMperp, \hatMperpp, \hatYperpp)$ and  $\mathbb{P}_n$ in \eqref{eq:olsfwlform} with their  
nonparametric bootstrap versions 
 $(\Sperpboot, \Mperpboot, \Mperppboot, \Yperppboot)$ and $\mathbb{P}_n^*$, respectively.  
%\begin{remark}
%We clarify the differences between $(\hatalphaSn^*, \hatbetaMn^*)$ and $(\hatalphaSgn^*, \hatbetaMgn^*)$	 by showing the difference between  $\mathbb{P}_n(\hatSperp)$ and $\mathbb{P}^*_n(\Sperp^*)$ as an example. 
%\end{remark}
%where $\Sperp^*, \Mperp^*, \M_g^*$, and $\Y_g^*$ are defined in  \eqref{eq:bootvardef}. 
%We then have 
%\noindent Similarly to \eqref{eq:coefdiffsample}, by the formulae of $\hatalphaSn^*$ and $\hatalphaSgn^*$, we have
Then by the formulae of $\hatalphaSn^*$ and $\hatalphaSgn^*$, we have
\begin{align}
	&~ \sqrt{n}(\hatalphaSn^* -  \hatalphaSgn^* ) \times \{\mathbb{P}_n^*(\Sperp^2)\mathbb{P}_n^* (\hatSperp^2)\} \notag \\
%=&~ \sqrt{n}\Big[ \mathbb{P}_n^*(\Sperp^* \Mperp^* ) \mathbb{P}_n^* (\hat{\S}_{g}^2) - \mathbb{P}_n^*(\hatSperp \hatMperp ) \mathbb{P}_n^*\{(\S_{g}^*)^2\}  \Big]\notag \\
=&~ \sqrt{n}\big\{  \mathbb{P}_n^*(\Sperpboot \Mperpboot ) - \mathbb{P}_n^*(\hatSperp \hatMperp ) \big\}\mathbb{P}_n^* (\hatSperp^2)+  \sqrt{n} \mathbb{P}_n^*(\hatSperp \hatMperp )\big[ \mathbb{P}_n^* (\hatSperp^2) - \mathbb{P}_n^*\{(\Sperpboot)^2\}  \big].\label{eq:diffexp1}	
%=&~ \sqrt{n}\Big[ \big\{  \mathbb{P}_n^*(\Sperp^* \Mperp^* ) - \mathbb{P}_n^*(\hatSperp \hatMperp ) \big\}\mathbb{P}_n^* (\hat{\S}_{g}^2)+  \mathbb{P}_n^*(\hatSperp \hatMperp )\big[ \mathbb{P}_n^* (\hat{\S}_{g}^2) - \mathbb{P}_n^*\{(\S_{g}^*)^2\}  \big] \Big].\label{eq:diffexp1}	
\end{align}
%By the definitions of $(\hatSperp, \hatMperp, \Sperpboot, \Mperpboot)$, we have
For the first summed term in \eqref{eq:diffexp1}, we have
\begin{align}
&~\mathbb{P}_n^*(\Sperpboot \Mperpboot)-\mathbb{P}_n^*(\hatSperp \hatMperp)\label{eq:coefdiffsampleboot2} \\
=&~\mathbb{P}_n^*\{\Sperpboot(\Mperpboot - \hatMperp) \}+ \PPboot\{(\Sperpboot- \hatSperp)\Mperpboot \}+\PPboot\{(\Sperpboot- \hatSperp)(\hatMperp-\Mperpboot)\} \notag \\
= &~ -(Q_{1,\S}^*-\hat{Q}_{1,\S})^{\mytrans}\mathbb{P}_n^*(\X \mathbf{\X}^{\mytrans})(Q_{1,\M}^*-\hat{Q}_{1,\M}), \notag 
\end{align}
where in the last equation, we use 
$\mathbb{P}_n^*\{\Sperpboot(\Mperpboot - \hatMperp) \}=(\hat{Q}_{1,\M}-Q_{1,\M}^*)^{\mytrans}\PPboot\{\X(\S-\X^{\mytrans} Q_{1,\S}^*)\}=0$
 and 
$\PPboot\{(\Sperpboot- \hatSperp)\Mperpboot \}=(\hat{Q}_{1,\S}-Q_{1,\S}^*)^{\mytrans}\PPboot\{\X(\M-\X^{\mytrans} Q_{1,\M}^*)\}=0$. 
Therefore, by Lemma \ref{lm:consistencyqmoments}, the bootstrap consistency, and Slutsky's lemma, 
we obtain $\sqrt{n}\{\mathbb{P}_n^* (\Sperpboot \Mperpboot) - \mathbb{P}_n^* (\hatSperp \hatMperp)\}\overset{\mathrm{P}^*}{\leadsto}  0$.
% conditionally (on the data) in probability.   
Following similar analysis, we have
\begin{align}\label{eq:vsmomentasymconsis}
\mathbb{P}_n^*\{(\Sperpboot)^2\} - \mathbb{P}_n^*(\hatSperp^2)=-(Q_{1,\S}^* - \hat{Q}_{1,\S})^{\mytrans}\PPboot(\X\X^{\mytrans})(Q_{1,\S}^* - \hat{Q}_{1,\S}),
\end{align}
and then 
%Following similar analysis, we also obtain
$\sqrt{n}[\mathbb{P}_n^*\{(\Sperpboot)^2\} - \mathbb{P}_n^*(\hatSperp^2)]\overset{\mathrm{P}^*}{\leadsto}  0$.
% $\sqrt{n}\{ \mathbb{P}_n^*(\hatSperp^2) - \mathbb{P}_n^*(\Sperp^2)\} \xrightarrow{P_C} 0$ 
 % conditionally (on the data)  in probability. 
 Therefore, by \eqref{eq:diffexp1} and Slutsky's lemma, 
% $\eqref{eq:diffexp1} \xrightarrow{P_C} 0$, and then 
we prove $\sqrt{n}(\hatalphaSn^* -  \hatalphaSgn^* )\overset{\mathrm{P}^*}{\leadsto} 0$.
% conditionally (on the data) in probability. 
Similarly, we can also prove $\sqrt{n}(\hatbetaMn^* -  \hatbetaMgn^* )\overset{\mathrm{P}^*}{\leadsto} 0$. 
 The details are   similar and thus skipped.

\bigskip
\noindent \textit{Part (2)}~
Similarly to \eqref{eq:zsfirsterm}, 
%by  $\hatSperp-\Sperpboot=(Q_{1,\S}^* - \hat{Q}_{1,\S})^{\mytrans}\X$ and $\hatenM=\hatMperp - \hatalphaSn \hatSperp$, 
we have
%\begin{align*}
%\ZZSperpbootn-\ZZSbootn =&~ \GGboot (\hatenM \hatSperp) -  \GGboot (\hatenM \Sperpboot)\notag \\
%=&~	\GGboot (\hatenM \X^{\intercal}) (Q_{1,\S}^* - \hat{Q}_{1,\S})
%\end{align*}
\begin{align*}
\ZZSperpbootn-\ZZSbootn =&~\GGboot \{ \hatenM (\hatSperp - \Sperpboot ) \} \notag \\ %=	\GGboot (\hatenM \X^{\intercal}) (Q_{1,\S}^* - \hat{Q}_{1,\S}) \notag \\
=&~\Big[(Q_{1,\M}-\hat{Q}_{1,\M})^{\mytrans}\mathbb{G}_n^*(\X\X^{\mytrans})+ \mathbb{G}_n^*(\Mperp \X^{\mytrans})\notag \\
&~\ - \hatalphaSn (Q_{1,\S}-\hat{Q}_{1,\S})^{\mytrans}\mathbb{G}_n^*(\X\X^{\mytrans}) -\hatalphaSn \mathbb{G}_n^*(\Sperp\X^{\mytrans})\Big] ({Q}_{1,\S}^*-\hat{Q}_{1,\S}), \notag
\end{align*}
and then by Lemma \ref{lm:consistencyqmoments}, $\ZZSperpbootn-\ZZSbootn\overset{\mathrm{P}^*}{\leadsto}  0$.
% conditionally (on the data) in probability. 
Following similar analysis, we can also prove  $\ZZMperpbootn-\ZZMbootn\overset{\mathrm{P}^*}{\leadsto}  0$.
% conditionally (on the data) in probability. 

%where we use $\hatSperp-\Sperpboot=(Q_{1,\S}^* - \hat{Q}_{1,\S})^{\mytrans}\X$ and $\hatenM=\hatMperp - \hatalphaSn \hatSperp$. 
%Therefore, $\ZZSperpbootn-\ZZSbootn\xrightarrow{P_C} 0$  conditionally (on the data) in probability. 

\bigskip
\noindent \textit{Part (3)}~
%By \eqref{eq:vsmomentasymconsis} and Lemma \ref{lm:consistencyqmoments}, 
Note that $\VVSperpbootn=\PPboot(\hatSperp^2)$ and $\VVSbootn=\PPboot\{(\Sperpboot)^2\}$. 
Thus by the analysis of \eqref{eq:vsmomentasymconsis}, 
we have $\VVSperpbootn-\VVSbootn \overset{\mathrm{P}^*}{\leadsto} 0$.
% conditionally (on the data) in probability. 
Following similar analysis, we can also prove  $ \VVMperpbootn - \VVMbootn\overset{\mathrm{P}^*}{\leadsto}  0$.
% conditionally (on the data) in probability. 

%\subsection{Proof of Lemma }
%\begin{proof}

\bigskip
\noindent \textit{Part (4)}~
We focus on discussing $(\hatsigmaalphan^*, \hatsigmaalphaperpboot)$ below and $(\hatsigmabetan^*, \hatsigmabetaperpboot)$ can be analyzed similarly. 
%By the definition of $(\hatsigmaalphaperpboot)^2$, 
Note that we can write 
$(\hatsigmaalphaperpboot)^2=\mathbb{P}_n^*\{ (\hatenMperpboot)^2\}/\mathbb{P}_n^*(\hatSperp^2)$
and   $(\hatsigmaalphan^*)^2 = {\mathbb{P}_n^*\{ (\hatenM^*)^2\}}/{ \mathbb{P}_n^* \{ (\Sperp^*)^2\}   }$,
which is obtained  by replacing $\hatenM$ and $\mathbb{P}_n(\cdot)$ in the formula of $\hatsigmaalphan^2$ in Part (5) of Lemma \ref{lm:regresstrans} with their nonparametric bootstrap versions $\hatenM^*$ and $\mathbb{P}_n^*(\cdot)$, respectively, 
 where $\hatenM^*$ denotes the residuals from the ordinary least squares regressions under the nonparametric bootstrap. 
%By the definitions above, we have $(\hatsigmaalphaperpboot)^2=\mathbb{P}_n^*\{ (\hatenMperpboot)^2\}/\mathbb{P}_n^*(\hatSperp^2)$.
% and $(\hatsigmaalphan^*)^2 = {\mathbb{P}_n^*\{ (\hatenM^*)^2\}}/{ \mathbb{P}_n^* \{ (\Sperp^*)^2\}   }$. 
%Moreover, for $(\hatsigmaalphan^*)^2$ from the nonparametric bootstrap, 
% following Part (5) in Lemma \ref{lm:regresstrans}, we have
% $(\hatsigmaalphan^*)^2 = {\mathbb{P}_n^*\{ (\hatenM^*)^2\}}/{ \mathbb{P}_n^* \{ (\Sperp^*)^2\}   }$,
% which is obtained by replacing $\hatenM$ and $\mathbb{P}_n(\cdot)$ in the formula of $\hatsigmaalphan^2$ with their nonparametric bootstrap versions $\hatenM^*$ and $\mathbb{P}_n^*(\cdot)$, respectively, 
% where $\hatenM^*$ denotes the residuals from the ordinary least squares regressions 
% under the nonparametric bootstrap.  
%To prove   $ (\hatsigmaalphan^*)^2 - (\hatsigmaalphaperpboot)^2 \xrightarrow{P_C}0$ conditionally (on the data) in probability,
Since we know $\mathbb{P}_n^*(\hatSperp^2)-  \mathbb{P}_n^* \{ (\Sperp^*)^2\}\overset{\mathrm{P}^*}{\leadsto}  0$ 
% conditionally (on the data) in probability 
by \eqref{eq:vsmomentasymconsis},
it suffices to prove $\mathbb{P}_n^*\{ (\hatenMperpboot)^2\} - \mathbb{P}_n^*\{ (\hatenM^*)^2\}\overset{\mathrm{P}^*}{\leadsto}  0$. 
% conditionally (on the data) in probability. 

In particular, similarly to \eqref{eq:bootresq1}, 
we can write
\begin{align*}
&\mathbb{P}_n^*\{ (\hatenMperpboot)^2\} - \mathbb{P}_n^*\{ (\hatenM^*)^2\} \notag \\
%=&~ \mathbb{P}_n^*\{ (\hatenMperpboot - \hatenM^*)(\hatenMperpboot + \hatenM^*) \}\notag \\
=&\mathbb{P}_n^*\{  (\hatMperp - \hatalphaSgn^* \hatSperp - \Mperpboot + \hatalphaSn^* \Sperp^* )(\hatenMperpboot + \hatenM^*) \}\notag \\
=& \mathbb{P}_n^*[\{(Q_{1,\M}^*-\hat{Q}_{1,\M})^{\mytrans}\X +(\hatalphaSn^*-\hatalphaSgn^*)\S- ( \hatalphaSn^*Q_{1,\S}^*-\hatalphaSgn^* \hat{Q}_{1,\S} )^{\mytrans}\X\} (\hatenMperpboot+ \hatenM^*)] \notag \\ 
=& (Q_{1,\M}^*-\hat{Q}_{1,\M}+\hatalphaSgn^* \hat{Q}_{1,\S} - \hatalphaSn^*Q_{1,\S}^*)^{\mytrans}\mathbb{P}_n^*(\X\hatenMperpboot) +(\hatalphaSn^* - \hatalphaSgn^*)\mathbb{P}_n^*(\S\hatenMperpboot), 
\end{align*}
where in the last equation, we use $\mathbb{P}_n^*(\X \hatenM^*)=\mathbf{0}$ and $\mathbb{P}_n^*(\S\hatenM^*)=0$ by the property of the ordinary least squares regression. 
Following similar analysis, we have
\begin{align*}
\mathbb{P}_n^*(\X\hatenMperpboot) = &~\mathbb{P}_n^*(\X\X^{\mytrans}) (Q_{1,\M}^*-\hat{Q}_{1,\M}+\hatalphaSgn^* \hat{Q}_{1,\S} - \hatalphaSn^*Q_{1,\S}^*),\notag \\
\mathbb{P}_n^*(\S\hatenMperpboot) =&~ \mathbb{P}_n^*(\S^2) (Q_{1,\M}^*-\hat{Q}_{1,\M}+\hatalphaSgn^* \hat{Q}_{1,\S} - \hatalphaSn^*Q_{1,\S}^*).
\end{align*}
Then by Lemma \ref{lm:consistencyqmoments},  Part (1) of Lemma \ref{lm:perpconsisall}, and the bootstrap consistency, we have $\mathbb{P}_n^*\{ (\hatenMperpboot)^2\} - \mathbb{P}_n^*\{ (\hatenM^*)^2\}\overset{\mathrm{P}^*}{\leadsto}  0$.
% conditionally in probability. 

\newpage

\section{Extensions under  Multiple Mediators}\label{sec:extendmultiplemediators}
%\begin{itemize}
%\setlength{\itemsep}{0pt}
%    \item define counterfactual outcome notation (finished)
%    \item joint mediation effect:  
%    \item parallel multiple mediator path model 
%    \item individual mediator:
%    \begin{itemize}
%        \item   interventional indirect effect: definition, identification condition; regression model (no restriction on the relationship between $M$)
%        \item natural indirect effect: 
%    \end{itemize}
%\end{itemize}

In this section,
we introduce two types of individual indirect effects under multiple-mediator settings in Section  \ref{sec:differentypesofeffects}, 
and then we present supplementary results on testing joint mediation effect of multiple mediators in Section \ref{sec:jointmedmulti}. 

% Section \ref{sec:numericresultsmulti}

\subsection{Two Types of Indirect Effects under Multiple Mediators: Supplementary Material of Figure \ref{fig:multiplemed}}\label{sec:differentypesofeffects}
We discuss the two scenarios when multiple mediators are  causally uncorrelated or causally correlated in Sections \ref{sec:parallelmediator} and \ref{sec:intervenindirect}, respectively.

\subsubsection{Causally Uncorrelated   Mediators}\label{sec:parallelmediator}
%Consider the potential outcome framework 
%and the indirect effect mediated by $M_1$ defined in 
%\cite{imai2013identification}. 
As an example, we next focus on a target mediator, \textit{say} $M_1$, and then denote the non-target mediators by  $\boldsymbol{M}_{(-1)}$. 
Let $M_1(s)$ denote the potential value of the target mediator $M_1$ under the exposure $\S=s$, 
let  $\boldsymbol{M}_{(-1)}(s)$ denote the potential value of non-target mediators $\boldsymbol{M}_{(-1)}$ under the exposure $\S=s$, 
and let $\Y(s, m_1, \boldsymbol{m}_{(-1)})$ denote the potential outcome that would have been observed if $\S$, $\M_1$, and $\boldsymbol{M}_{(-1)}$ had been set to $s$, $m_1$, and $\boldsymbol{m}_{(-1)}$,  respectively. 
Consider the individual indirect effect mediated by $M_1$ defined in 
\cite{imai2013identification}: 
\begin{align}\label{eq:indirecteffect1}
	\mathrm{E}\big\{ \Y(s, \M_1(s), \boldsymbol{M}_{(-1)}(s)) - \Y(s, \M_1(s^*), \boldsymbol{M}_{(-1)}(s)) \big\}.
\end{align}
%The individual natural indirect effect $\mathrm{E}\big\{ \Y(s, \M_1(s)) - \Y(s, \M_1(s^*))  \big\}$ 
The effect \eqref{eq:indirecteffect1} 
can be nonparametrically identified  given the following condition on sequential ignorability with multiple causally independent  mediators; see, e.g., \cite{jerolon2020causal}. 
% and  \cite{imai2013identification}. 
%This is a direct generalization from \cite{imai2010identification} by treating $(\X,\boldsymbol{M}_{(-1)})$ as observed confounders. 
\begin{condition}\label{cond:identificatoinmultiplemed} \quad 
Let $\boldsymbol{X}$ denote all the observed pretreatment covariates (variables unaffected by the treatment). 
For $s, s^*, s'$, and $\boldsymbol{x}$ in the support set, 
\begin{enumerate}
	\item[(i)] $\{\Y(s, m_1, \boldsymbol{m}_{(-1)}),\  M_1(s^*),\ \boldsymbol{M}_{(-1)}(s')\}\perp S\mid \{\boldsymbol{X}=\boldsymbol{x}\}$, 
	\item[(ii)] $\Y(s^*, m_1, \boldsymbol{m}_{(-1)})\, \perp \,  \{ M_1(s),\ \boldsymbol{M}_{(-1)}(s)\}\mid  \{ S=s,\boldsymbol{X}=\boldsymbol{x} \}$,
	\item[(iii)] $\Y(s, m_1, \boldsymbol{m}_{(-1)})\, \perp \,   \{ M_1(s^*),\ \boldsymbol{M}_{(-1)}(s)\}\mid  \{S=s,\boldsymbol{X}=\boldsymbol{x}\}$, 
%	\item $\Y\big(s, m_1, \boldsymbol{M}_{(-1)}(s)\big) \perp M_1 \mid \{S=s^*, \boldsymbol{X}=\boldsymbol{x}\}$,
%	\item $Y\big(s, M_1(s), \boldsymbol{m}_{(-1)}\big)\perp \boldsymbol{M}_{(-1)} \mid \{S=s^*, \boldsymbol{X}=\boldsymbol{x}\}$,
\end{enumerate}
where $P(S=s \mid \boldsymbol{X}=\boldsymbol{x}) >0$, 
and 
the conditional density (mass)  function 
of $\boldsymbol{M}=\boldsymbol{m}$ (when $\boldsymbol{M}$ is discrete) 
$f(\boldsymbol{M}=\boldsymbol{m} \mid S=s,\boldsymbol{X}=\boldsymbol{x}) >0$. 
%\begin{enumerate} 
%\item[(i)] $\{\Y(s^*,m_1),\ M_1(s) \} \perp \S \mid \{\X = \boldsymbol{x},\, \boldsymbol{M}_{(-1)}=\boldsymbol{m}_{(-1)}\}$. 
%\item[(ii)]\, $Y(s^*, m_1) \perp M_1(s) \mid \{\S=s,\,  \X = \boldsymbol{x},\, \boldsymbol{M}_{(-1)}=\boldsymbol{m}_{(-1)}\}$.
%\item[(iii)] $ P(S=s\mid  \X=\boldsymbol{x},\, \boldsymbol{M}_{(-1)}=\boldsymbol{m}_{(-1)}) > 0,$ and   \\[4pt] 
%\ $P\big( M_1(s) = m_1 \mid S=s, \, \X=\boldsymbol{x},\,  \boldsymbol{M}_{(-1)}=\boldsymbol{m}_{(-1)}\big) > 0$. 
%%	\item[(i)] $\Y(s,m_1) \perp \S \mid \{\X,\boldsymbol{M}_{(-1)}\}$, no unmeasured confounder for the relation of $\Y$ and $\S$.
%%%	\item[(ii)]$\Y(s,m_1) \perp \M_1 \mid \{\S, \X, \boldsymbol{M}_{(-1)}\}$, no unmeasured confounder for the relation of $\Y$ and $\M$ conditioning on $\S$.
%%	\item[(iii)]$\M_1(s)\perp \S \mid \{\X,\boldsymbol{M}_{(-1)}\}$, no unmeasured confounder for the relation of $\M$ and $\S$.
%%	\item[(iv)]$\Y(s,m_1) \perp \M_1(s^*) \mid \{\X,\boldsymbol{M}_{(-1)}\}$,  
%% no unmeasured confounder for the $\M$-$\Y$ relation that is affected by $\S$.
%\end{enumerate}
\end{condition}
%Under Condition \ref{cond:identificatoinmultiplemed}, 
%the individual mediation effect $\mathrm{E}\big\{ \Y(s, \M_1(s)) - \Y(s, \M_1(s^*))  \big\}$ is identifiable. 

% The first hypothesis
Condition \ref{cond:identificatoinmultiplemed}-(i) 
suggests that there are no unobserved  pretreatment confounders between the treatment and the outcome and between the treatment and the individual mediators, after conditioning on all observed covariates. 
Condition \ref{cond:identificatoinmultiplemed}-(ii) and (iii)  imply that:   
% two scenarios:
(a) %the mediators and the outcome are confounded by an
there are no 
unobserved pretreatment variables between the mediators taken jointly and the outcome, 
and (b) the mediators and the outcome are confounded by an observed or unobserved posttreatment variable. 
% the existence of two distinct types of confounders between the mediators taken jointly and the outcome: the confounding by an unobserved pretreatment variable and the confounding by an observed or unobserved posttreatment variable. 
We point out that 
Condition \ref{cond:identificatoinmultiplemed} allow that the mediators are \textit{uncausally correlated}, e.g.,
there exist unobserved pretreatment confounder $\boldsymbol{U}$ affecting the mediators jointly. 
We give specific examples below. 
% under the following scenarios.
\begin{example}\label{eg:1}
Assume the multivariate linear model where for $j=1,\ldots, J$, 
%\eqref{eq:fullmodmultsupp}, 
\begin{align}
\M_j =\alpha_{\S,j} \S +  \X^\mytrans \alpha_{\X,j} +\error_{\M,j},	\hspace{2em}\Y  = \sum_{j=1}^{\nM} \beta_{\M,j} \M_j + \X^\mytrans  \betaX + \DE \S + \eY, \label{eq:fullmodmultsupp}
\end{align}
Assume (i) $\epsilon_{M,1}\ldots, \epsilon_{M,J},  \epsilon_Y,$ and $ S$ are mutually independent conditioning on $ \X$; (ii) $\mathrm{E}(\boldsymbol{\error}_{\M}| \X, \S )=\boldsymbol{0}$  and $\mathrm{E}(\eY| \X, \S, \boldsymbol{\M} )=0$. 
% , and they are independent with $\epsilon_Y$.
% conditioning on $\{S, \X \}$. 
\end{example} 

\begin{example}\label{eg:2}
Under the multivariate linear model \eqref{eq:fullmodmultsupp}, 
there exists unobserved confounders $\boldsymbol{U}$ such that: 
(i) $\epsilon_{M,1}\ldots, $ and $\epsilon_{M,J}$ are mutually  independent conditioning on 	$\{\X, \boldsymbol{U}\}$;  
%and  $\boldsymbol{U}\perp \S$; 
(ii) 
$(\epsilon_{M,1}\ldots, \epsilon_{M,J}), $ $\epsilon_Y$, and $S$ are independent conditioning on $ \X$; 
% $ \perp \boldsymbol{U} \mid \X$;   
(iii) $\mathrm{E}(\boldsymbol{\error}_{\M}| \X, \S )=\boldsymbol{0}$  and $\mathrm{E}(\eY| \X, \S, \boldsymbol{\M} )=0$.  
%, and (iii) $\epsilon_{M,1},\ldots, \epsilon_{M,J}$ are independent with $\epsilon_Y$.
% conditioning on $\{S, \X, \boldsymbol{U}\}$.
\end{example}

\begin{lemma}
Under Examples \ref{eg:1} or  \ref{eg:2}, 
%Assume marginal 
$\mathrm{E}\big\{ \Y(s, \M_1(s)) - \Y(s, \M_1(s^*))  \big\}=\alpha_{S,1}\beta_{M,1}(s-s^*)$. 
\end{lemma}

\begin{proof}
Let $\mathcal{E}_{\boldsymbol{x},\boldsymbol{u},\boldsymbol{m_{(-1)}}}=\{\X=\boldsymbol{x}, \boldsymbol{U}=\boldsymbol{u},\boldsymbol{M}_{(-1)}=\boldsymbol{m}_{(-1)}\}$
and define the other events similarly. We have
%$\mathcal{E}_{\boldsymbol{x},\boldsymbol{u},\boldsymbol{m},s}=\{\X=\boldsymbol{x}, \boldsymbol{U}=\boldsymbol{u},\boldsymbol{M}=\boldsymbol{m}, S=s\}$.   
%\small
\begin{align*}
&~\mathrm{E}\big\{ \Y(s, \M_1(s)) - \Y(s, \M_1(s^*))  \big\}\notag\\ 
=&~\int \mathrm{E}\big\{ \Y(s, \M_1(s)) - \Y(s, \M_1(s^*)) \mid \mathcal{E}_{\boldsymbol{x},\boldsymbol{u},\boldsymbol{m_{(-1)}}}\big\}\mathrm{d}F(\mathcal{E}_{\boldsymbol{x},\boldsymbol{u},\boldsymbol{m_{(-1)}}})\notag\\ 
%=&\int \mathrm{E}(\Y\mid \S=s, M_1=m_1) \mathrm{d}F(M_1=m_1\mid \S=s)- \int \mathrm{E}(\Y\mid \S=s, M_1=m_1) \mathrm{d}F(M_1=m_1\mid \S=s^*) \notag\\
=&~\int\int \mathrm{E}\big\{ \Y\mid \S=s, M_1=m_1, \mathcal{E}_{\boldsymbol{x},\boldsymbol{u},\boldsymbol{m_{(-1)}}}\big\} \  \notag\\
&~\times\Big\{\mathrm{d}F(\M_1=m_1\mid \mathcal{E}_{\boldsymbol{x},\boldsymbol{u},\boldsymbol{m_{(-1)}},s})-\mathrm{d}F(M_1=m_1\mid \mathcal{E}_{\boldsymbol{x},\boldsymbol{u},\boldsymbol{m_{(-1)}},s^*}) \Big\}\mathrm{d}F(\mathcal{E}_{\boldsymbol{x},\boldsymbol{u},\boldsymbol{m_{(-1)}}})\notag\\
%&~\hspace{1.5em}\times \mathrm{d}F(\mathcal{E}_{\boldsymbol{x},\boldsymbol{u},\boldsymbol{m_{(-1)}}})\notag\\
=&~\int \int\big\{ \beta_S s + \beta_{M_1}m_1+\boldsymbol{\beta}_X^{\top}\boldsymbol{x}+\boldsymbol{\beta}_{M_{(-1)}}^{\top}\boldsymbol{m}_{(-1)}+\mathrm{E}(\epsilon_Y\mid \mathcal{E}_{\boldsymbol{x},\boldsymbol{u},\boldsymbol{m},s})\big\}\  \notag\\
&~ \times\Big\{\mathrm{d}F(M_1=m_1\mid \mathcal{E}_{\boldsymbol{x},\boldsymbol{u},\boldsymbol{m_{(-1)}},s})-\mathrm{d}F(M_1=m_1\mid \mathcal{E}_{\boldsymbol{x},\boldsymbol{u},\boldsymbol{m_{(-1)}},s^*}) \Big\}\mathrm{d}F(\mathcal{E}_{\boldsymbol{x},\boldsymbol{u},\boldsymbol{m_{(-1)}}})\notag\\
=&~\beta_{M_1}\int \big\{ \mathrm{E}(M_1\mid \mathcal{E}_{\boldsymbol{x},\boldsymbol{u},\boldsymbol{m}_{(-1)},s})-\mathrm{E}(M_1\mid \mathcal{E}_{\boldsymbol{x},\boldsymbol{u},\boldsymbol{m}_{(-1)},s^*}) \big\}\mathrm{d}F(\mathcal{E}_{\boldsymbol{x},\boldsymbol{u},\boldsymbol{m_{(-1)}}})\notag\\
=&~\alpha_{S,1}\beta_{M_1}(s-s^*).
\end{align*}
\end{proof}

\normalsize

\subsubsection{Causally Correlated Mediators: Interventional Indirect Effect}\label{sec:intervenindirect}
We briefly introduce the definition of interventional indirect effect through one mediator $M_1$ when there are $J$ mediators;  more details can be found in \cite{vanderweele2014effect}, 
\cite{vansteelandt2017interventional}, and 
\cite{loh2021disentangling}. 
% In particular, we took the definition under multiple mediators in the supplementary material of \cite{loh2021disentangling}. 
% More discussions on the interventional  effect can be found in \cite{vanderweele2014effect} and \cite{vansteelandt2017interventional}. 
% The definition  was introduced in  \cite{vanderweele2014effect}; see also  
% \cite{vansteelandt2017interventional}, and 
% \cite{loh2021disentangling}. 
% \paragraph{Notation}

Let $\boldsymbol{M}=(M_1,M_2,\ldots, M_J)$ and 
$\boldsymbol{M}_{(-1)}=(M_2,\ldots, M_J)$.
Similarly, we let $\boldsymbol{m}=(m_1,\ldots, m_J) $, and $\boldsymbol{m}_{(-1)}=(m_2,\ldots, m_J)$.
% $(m_1,\boldsymbol{m}_{(-1)})=\boldsymbol{m}$, where $\boldsymbol{m}_{(-1)}=(m_2,\ldots, m_J)$
% conditioning on covariates $\boldsymbol{X}=\boldsymbol{x}$
Let $M_j(s)$ denote the potential value of the mediator $M_j$ if $S$ had been set to $s$. 
Let $Y(s,\boldsymbol{m})$ denote the potential value of the outcome $Y$ if $S$ and $\boldsymbol{M}$ had been assigned to $s$ and $\boldsymbol{m}$, respectively. 
Equivalently, we also write  $Y(s,m_1,\boldsymbol{m}_{(-1)})$ to separate $m_1$ from the other mediators.

In the following, we took the definition under multiple mediators in the supplementary material of \cite{loh2021disentangling}. 
Assume the identification assumptions in Condition \ref{cond:multiidentifiability}.
The interventional indirect effect of treatment on the outcome via a given mediator $M_1$ is defined as
\begin{align*}
\mathrm{IE}_j=\mathrm{E} \Biggr[ \sum_{\boldsymbol{m}} \mathrm{E}\big(Y(s,\boldsymbol{m})\mid \boldsymbol{x}\big)\big\{P\big(M_1(s)=m_1\mid \boldsymbol{x}\big)-P\big(M_1(s^*)=m_1 \mid \boldsymbol{x}\big)\big\}\\
\prod_{k=1}^{j-1} P\big({M}_{k}(s)={m}_{k}\mid \boldsymbol{x}\big) \times \prod_{l=j+1}^{J} P\big({M}_{l}(s^*)={m}_{l}\mid \boldsymbol{x}\big)\Biggr].
\end{align*} 
The
interventional indirect effect of treatment on outcome via the mediator $M_j$ is
interpreted as the combined effect along all (underlying) causal pathways from $S$ to $M_j$ (possibly intersecting any other mediators that are causes of $M_j$), then lead directly
from $M_j$ to Y .
% This expresses the effect of shiting the distribution of mediator $M_1$ from the counterfactual distribution (given covariates) at exposure level $s^*$ to that at level $s$, while fixing the exposure at $s$ and the other mediators $\boldsymbol{M}_{(-1)}$ to a random subject-specific draw from the counterfactual distribution (given covariates) at level $s^*$ for all subjects \citep{vansteelandt2017interventional}. 

% \paragraph{Identification Assumptions}
% We list the unconfounded assumptions given in \cite{loh2021disentangling},
To identify the individual interventional  indirect effect,
\cite{loh2021disentangling} list the following unconfounded assumptions,  
including:  
(i) the effect of exposure $\S$ on outcome $\Y$ is unconfounded conditional on $\X$, 
(ii) the effect of mediators $\boldsymbol{M}$ on outcome $\Y$ is unconfounded conditional on $\{\S, \X\}$,
and 
(iii) the  effect of treatment $\S$ on both mediators is unconfounded conditional on $\X$. 

% \begin{condition}\label{cond:intervenassump}
% \quad 
% \begin{enumerate}
%     \item The effect of exposure $\S$ on outcome $\Y$ is unconfounded conditional on $\X$.
%     \item The effect of mediators $\boldsymbol{M}$ on outcome $\Y$ is unconfounded conditional on $\{\S, \X\}$. 
%     \item The effect of treatment $\S$ on both mediators is unconfounded conditional on $\X$. 
% \end{enumerate}
% \end{condition}

% \paragraph{Estimation under the linear and additive mean model} 
When we further assume the linear and additive mean model below:
\begin{align}\label{eq:linearaddmeaniie}
    \mathrm{E}(Y \mid S, \boldsymbol{M}, \boldsymbol{X})= &~\sum_{j=1}^J \beta_{M, j} M_j+\boldsymbol{X}^{\top} \boldsymbol{\beta}_{\boldsymbol{X}}+\tau_S S; \quad \\  
   \mathrm{E}(M_j\mid S, \boldsymbol{X}) = &~\alpha_{S,j}S+ \boldsymbol{X}^{\top}\boldsymbol{\alpha}_{\boldsymbol{X},j}, \notag
\end{align}
for $j=1,\ldots, J$. 
% When $S$ takes binary value  0 or 1, 
Then the interventional indirect effect for the $j$-th mediator has been obtained as 
$\mathrm{IE}_j=\beta_{M,j} \alpha_{S,j}(s-s^*)$. 
Therefore, the interventional indirect effects can be estimated by 
%(partial) 
regression coefficients; see \cite{loh2021disentangling}. 
% \begin{align*}
%     \beta_{M,j}\{\mathrm{E}(M_j\mid S=1)- \mathrm{E}(M_j\mid S=0)\}=\beta_{M,j} \alpha_{S,j}. 
% \end{align*}

% \newpage
\subsection{Joint Mediation Effect: Supplementary Material of Section \ref{sec:jointestmulti}}\label{sec:jointmedmulti}
% Extensions for Multiple Mediators} 
\label{sec:simulationmultiple}
% This section presents theoretical and numerical results supplementary to Section \ref{sec:jointestmulti} on testing joint mediation effect in the main text. 
In the following, 
we present regularity conditions in Section \ref{sec:jointmedcond},
we prove   Theorems \ref{thm:limitmulti}--\ref{thm:multbootjoint} in Sections \ref{sec:pflimitmulti}--\ref{sec:pfmultbootjoint}, respectively,
and then we provide numerical results of testing joint mediation effect in Section \ref{sec:numericresultsmulti}.

% additional theoretical results in Section 

\subsubsection{Conditions}\label{sec:jointmedcond}

Consider the potential outcome framework.
Let $\boldsymbol{M}(s)$ denote the potential value of all the target mediators  $\boldsymbol{M}$ under the exposure $\S=s$, and let $\Y(s, \boldsymbol{m})$ denote the potential outcome that would have been observed if $\S$ and $\boldsymbol{M}$ had been set to $s$ and $\boldsymbol{m}$, respectively.
\begin{condition}[Identification]\label{cond:multiidentifiability}\quad 
\begin{enumerate}
%\item[(i)] $\{\Y(s^*, \boldsymbol{m}),\, \boldsymbol{M}(s) \}\perp \S \mid \{\boldsymbol{X}=\boldsymbol{x}\}$.
%\item[(ii)]  $Y(s^{*}, \boldsymbol{m}) \perp \boldsymbol{\M}\left(s\right) \mid \{\S=s, \boldsymbol{X}=\boldsymbol{x}\}$.
%\item[(iii)] $P(S=s \mid \boldsymbol{X}=\boldsymbol{x})>0 $ and 
%$P(\boldsymbol{M}(s) = \boldsymbol{m} \mid S=s, \, \X=\boldsymbol{x} ) >0$.
	\item $\Y(s, \boldsymbol{m}) \perp \S \mid \boldsymbol{X}$, i.e., no unmeasured  confounding for the relationship of the exposure $S$ and the outcome $Y$.  
	\item $\Y(s, \boldsymbol{m}) \perp \boldsymbol{M} \mid \{\boldsymbol{X}, S\}$, i.e., no unmeasured confounding for the relationship of the mediators $\boldsymbol{M}$ and the outcome $Y$, conditional on the exposure $\S$.
	\item $ \boldsymbol{M}(s) \perp \S \mid \boldsymbol{X}$, i.e., no unmeasured confounding for the relationship of the exposure $\S$ and mediators $\boldsymbol{\M}$. 
	\item $Y(s, \boldsymbol{m}) \perp \boldsymbol{\M}\left(s^{*}\right) \mid \boldsymbol{X}$.
i.e., no unmeasured confounder for the mediators-outcome $\boldsymbol{\M}$-$\Y$ relationship that is affected by the exposure $\S$, 
\end{enumerate}
where $P( S=s \mid \boldsymbol{X}=\boldsymbol{x}) >0$,  and 
the conditional density (mass)  function 
of $\boldsymbol{M}=\boldsymbol{m}$ (when $\boldsymbol{M}$ is discrete)  
% $f(\boldsymbol{m} \mid S=s,\boldsymbol{X}=\boldsymbol{x})$ are strictly bounded between constants. 
$f(\boldsymbol{M}=\boldsymbol{m} \mid S=s,\boldsymbol{X}=\boldsymbol{x}) >0$. 
\end{condition}

\begin{lemma}\label{lm:jointident}
Under Condition 	\ref{cond:multiidentifiability}, 
the joint mediation effect $\mathrm{E}\{\Y(s, \boldsymbol{\M}(s)) -\Y(s, \boldsymbol{\M}(s^*))\}$  is identifiable. 
If we further assume the multivariate linear structural equation model \eqref{eq:fullmodmult1}, the joint mediation effect equals $(s-s^*)\boldsymbol{\alpha}_{\S}^{\mytrans}\boldsymbol{\beta}_{\M}$.  
\end{lemma}

Lemma \ref{lm:jointident} is straightforward, and thus the proof is omitted. More details can be found in \cite{huang2016hypothesis}. 

\begin{condition}\label{cond:multinversemoment} 
(C2.1) $\mathrm{E}(\boldsymbol{\error}_{\M}| \X, \S )=\boldsymbol{0}$  and $\mathrm{E}(\eY| \X, \S, \boldsymbol{\M} )=0$. 
(C2.2) $\mathrm{E}(\boldsymbol{D}\boldsymbol{D}^{\mytrans})$ 
is a positive definite matrix with bounded eigenvalues, 
where $\boldsymbol{D} = (\X^{\mytrans}, \boldsymbol{\M}^{\mytrans}, \S)^{\mytrans}$.
(C2.3) The second moments of $(\boldsymbol{\error}_{\M},\eY, \Sperp, \boldsymbol{\M}_{\perp'}, \eM\Sperp, \eY\boldsymbol{\M}_{\perp'})$ are finite,
%Expectations of $(\eM^2, \eY^2, \S_{g}^2, \M_g^2, \eM^2\S_g^2, \eY^2\M_g^2)$ are finite,
where $\Sperp = \S - Q_{1,\S}^{\mytrans}\X$  with  $Q_{1,\S}= \{\mathrm{E}(\X \X^{\mytrans})\}^{-1}\times \mathrm{E}(\X \S)$, and $\boldsymbol{\M}_{\perp'} =\boldsymbol{\M} - Q_{2,\boldsymbol{\M}}^{\mytrans}\tilde{\X}$   with  $\tilde{\X}=(\X^{\mytrans},\S)^{\mytrans}$ and $Q_{2,\boldsymbol{\M}} = \{\mathrm{E}(\tilde{\X}^{\mytrans}\tilde{\X})\}^{-1}\times \mathrm{E}(\tilde{\X}\boldsymbol{\M}^{\mytrans})$.   
%where  $\S_{g} = \S - \X^{\mytrans}Q_{1,\S}$, $Q_{1,\S}= \{\mathrm{E}(\X \X^{\mytrans})\}^{-1} \mathrm{E}(\X \S)$,  $\M_{g}=\M - \tilde{\X}^{\mytrans} Q_{2,\M}$, and 
\end{condition}

%\subsection{Proofs}

\subsubsection{Proof of Theorem \ref{thm:limitmulti}}\label{sec:pflimitmulti} 
By the property of OLS of the linear SEM and following  the proof in Section \ref{sec:pfthm1}, we have
$
	\sqrt{n}(\hat{\alpha}_{S,j}-\alpha_{S,j})=\{{\mathbb{P}_n(\hat{S}_{\perp}^2)}\}^{-1}{\sqrt{n}\mathbb{P}_n(\hat{S}_{\perp}\epsilon_{M,j} )} \xrightarrow{d} \vec{Z}_{S,j}, 
$
where $\vec{Z}_{S,j}$ is a mean-zero normal variable with covariance same as $\epsilon_{M,j}S_{\perp}/\vec{V}_{S}$ and $\vec{V}_{S}=\mathrm{E}(S_{\perp}^2)$.
It follows that $\sqrt{n}(\hat{\boldsymbol{\alpha}}_S-\boldsymbol{\alpha}_{S,n}) \xrightarrow{d} \vec{Z}_{S}$, where $\vec{Z}_S=(\vec{Z}_{S,1},\ldots, \vec{V}_{Z,J})^{\top}$.  
Moreover, 
$
	\sqrt{n}(\hat{\boldsymbol{\beta}}_{M}-\boldsymbol{\beta}_{M,n})= \{\mathbb{P}_n(\hat{\boldsymbol{M}}_{\perp}\hat{\boldsymbol{M}}_{\perp}^{\top})\}^{-1}\sqrt{n}\mathbb{P}_n(\hat{\boldsymbol{M}}_{\perp} \epsilon_{Y}) \xrightarrow{d} \vec{Z}_{M},  
$ 
where $\vec{Z}_{M}$ is a normal vector with mean-zero and covariance same as $\vec{V}_{M}^{-1}\boldsymbol{M}_{\perp'}\eY$, 
$\vec{V}_{M} = \mathrm{E} (\boldsymbol{\M}_{\perp'}\boldsymbol{\M}_{\perp'}^{\mytrans})$, $\boldsymbol{\M}_{\perp'}$ is defined in Condition \ref{cond:multinversemoment}, $\hat{\boldsymbol{\M}}_{\perp'}$ represents sample version of $\boldsymbol{\M}_{\perp'}$ similarly to that in Section \ref{sec:preliminary}.   
Then Part (i) of Theorem  \ref{thm:limitmulti} follows by 
$
	\hat{\boldsymbol{\alpha}}_{\S,n}^{\mytrans}\hat{\boldsymbol{\beta}}_{\M,n}-{\boldsymbol{\alpha}}_{\S,n}^{\mytrans}{\boldsymbol{\beta}}_{\M,n}={\boldsymbol{\alpha}}_{\S,n}^{\mytrans}(\hat{\boldsymbol{\beta}}_{\M,n}-{\boldsymbol{\beta}}_{\M,n})+{\boldsymbol{\beta}}_{\M,n}^{\mytrans}(\hat{\boldsymbol{\alpha}}_{\S,n}-{\boldsymbol{\alpha}}_{\S,n})+(\hat{\boldsymbol{\alpha}}_{\S,n}-{\boldsymbol{\alpha}}_{\S,n})^{\top}(\hat{\boldsymbol{\beta}}_{\M,n}-{\boldsymbol{\beta}}_{\M,n}).
$ 
  Part (ii) of Theorem  \ref{thm:limitmulti} follows by
$n(\hat{\boldsymbol{\alpha}}_{\S,n}^{\mytrans}\hat{\boldsymbol{\beta}}_{\M,n}-{\boldsymbol{\alpha}}_{\S,n}^{\mytrans}{\boldsymbol{\beta}}_{\M,n})={\boldsymbol{b}}_{\alpha,n}^{\mytrans}(\hat{\boldsymbol{\beta}}_{\M,n}-{\boldsymbol{\beta}}_{M,n})+{\boldsymbol{b}}_{\beta,n}^{\mytrans}(\hat{\boldsymbol{\alpha}}_{\S,n}-{\boldsymbol{\alpha}}_{\S,n})+(\hat{\boldsymbol{\alpha}}_{\S,n}-{\boldsymbol{\alpha}}_{\S,n})^{\top}(\hat{\boldsymbol{\beta}}_{\M,n}-{\boldsymbol{\beta}}_{\M,n}).
$

%\subsubsection{Proof of Theorem \ref{thm:theorembootsep}}\label{sec:pftheorembootsep}

\subsubsection{Proof of Theorem \ref{thm:multbootjoint}}\label{sec:pfmultbootjoint}

\textit{Notation}. We first define some notation, similarly to those in Theorem \ref{thm:bootstrapprodcomb}.   
In particular, we let 
\begin{align*}
\vec{\mathbb{R}}_{1,n}(\boldsymbol{b}_{\alpha}, \boldsymbol{b}_{\beta}) = &~\vec{\mathbb{Z}}_{\S,n}^{\mytrans}\vec{\mathbb{Z}}_{{\M},n} +	\boldsymbol{b}_{\alpha}^{\mytrans}\vec{\mathbb{Z}}_{{\M},n} + \boldsymbol{b}_{\beta}^{\mytrans}\vec{\mathbb{Z}}_{\S,n},\notag\\	
\vec{\mathbb{R}}_{1,n}^*(\boldsymbol{b}_{\alpha}, \boldsymbol{b}_{\beta}) = &~\vec{\mathbb{Z}}_{\S,n}^{*\mytrans}\vec{\mathbb{Z}}_{\M,n}^* +\boldsymbol{b}_{\alpha}^{\mytrans}\vec{\mathbb{Z}}_{\M,n}^* + \boldsymbol{b}_{\beta}^{\mytrans}\vec{\mathbb{Z}}_{\S,n}^*,
\end{align*}
 where 
$(\vec{\mathbb{Z}}_{\S,n}, \vec{\mathbb{Z}}_{\M,n})$ and $(\vec{\mathbb{Z}}_{\S,n}^*, \vec{\mathbb{Z}}_{\M,n}^*)$ 
are multivariate counterparts of 
$({\mathbb{Z}}_{\S,n}, {\mathbb{Z}}_{\M,n})$ and $({\mathbb{Z}}_{\S,n}^*, {\mathbb{Z}}_{\M,n}^*)$ in Section \ref{sec:abt}, respectively.
Specifically, we define  
$\vec{\mathbb{Z}}_{\S,n}= (\vec{\mathbb{V}}_{{\S}, n})^{-1}\mathbb{G}_n(\boldsymbol{\error}_{\M} \hatSperp)$, and $\vec{\mathbb{Z}}_{{\M},n}=(\vec{\mathbb{V}}_{{\M}, n})^{-1}\mathbb{G}_n(\eY \hat{\boldsymbol{\M}}_{\perp'}),$ where  $\vec{\mathbb{V}}_{\S, n} = \mathbb{P}_n(\hatSperp^2),$ $\vec{\mathbb{V}}_{{\M}, n} = \mathbb{P}_n(\hat{\boldsymbol{\M}}_{\perp'}\hat{\boldsymbol{\M}}_{\perp'}^{\mytrans})$,  
$\hatSperp = \S-\hat{Q}_{1,\S}^{\mytrans}\X$,  $\hat{\boldsymbol{\M}}_{\perp'} = \boldsymbol{\M} - \hat{Q}_{2,\boldsymbol{\M}}^{\mytrans}\tilde{\X},$
$\hat{Q}_{1,\S}= \{\mathbb{P}_n(\X \X^{\mytrans})\}^{-1}\times \mathbb{P}_n(\X \S)$, 
and 
$\hat{Q}_{2,\boldsymbol{\M}} = \{\mathbb{P}_n( \tilde{\X}  \tilde{\X}^{\mytrans})\}^{-1} \times\mathbb{P}_n( \tilde{\X} \boldsymbol{\M}^{\mytrans})$.
Moreover, we can similarly define the bootstrap counterparts $\vec{\mathbb{Z}}_{\S,n}^*= (\vec{\mathbb{V}}_{{\S}, n}^*)^{-1}\mathbb{G}_n^*(\hat{\boldsymbol{\error}}_{\M} \Sperp^*)$, $\vec{\mathbb{Z}}_{{\M},n}^*=(\vec{\mathbb{V}}_{{\M}, n}^*)^{-1}\mathbb{G}_n^*(\hatenY \boldsymbol{\M}_{\perp'}^*),$ $\vec{\mathbb{V}}_{\S, n}^* = \mathbb{P}_n^*\{(\Sperp^*)^2\},$ and $\vec{\mathbb{V}}_{{\M}, n}^* = \mathbb{P}_n^*(\boldsymbol{\M}_{\perp'}^*\boldsymbol{\M}_{\perp'}^{*\mytrans})$.

\bigskip

\noindent \textit{Proof.} By the property of OLS estimator of multivariate linear model, and following the proof in Section \ref  {sec:pfbootstrapprodcomb}, 
we have 
$
	\sqrt{n}(\hat{\boldsymbol{\alpha}}_{S,n}^*-\hat{\boldsymbol{\alpha}}_{S,n})= \vec{\mathbb{Z}}_{S,n}^* \xrightarrow{d}\vec{Z}_{S}, 
$
and $	\sqrt{n}(\hat{\boldsymbol{\beta}}_{M,n}^*-\hat{\boldsymbol{\beta}}_{M,n})= \vec{\mathbb{Z}}_{M,n}^*  \xrightarrow{d} \vec{Z}_{M}$.  
Then we obtain
\begin{itemize}
    \item[(i)] when $(\boldsymbol{\alpha}_S, \boldsymbol{\beta}_M) \neq \mathbf{0}$,  
    % conditionally on the data, 
    $\sqrt{n}(\hat{\boldsymbol{\alpha}}_{\S,n}^{*\mytrans} \hat{\boldsymbol{\beta}}_{\M, n}^*  - \hat{\boldsymbol{\alpha}}_{\S,n}^{\mytrans} \hat{\boldsymbol{\beta}}_{\M, n})\overset{d^*}{\leadsto} 
$ 
% converges to the limiting distribution of 
$\sqrt{n}(\hat{\boldsymbol{\alpha}}_{\S,n}^{\mytrans} \hat{\boldsymbol{\beta}}_{\M, n} - {\boldsymbol{\alpha}}_{\S,n}^{\mytrans} {\boldsymbol{\beta}}_{\M, n})$ by 
$	\hat{\boldsymbol{\alpha}}_{\S,n}^{*\mytrans}\hat{\boldsymbol{\beta}}_{\M,n}^*-\hat{\boldsymbol{\alpha}}_{\S,n}^{\mytrans}\hat{\boldsymbol{\beta}}_{\M,n}=\hat{\boldsymbol{\alpha}}_{\S,n}^{\mytrans}(\hat{\boldsymbol{\beta}}_{\M,n}^*-\hat{\boldsymbol{\beta}}_{\M,n})+{\boldsymbol{\beta}}_{\M,n}^{\mytrans}(\hat{\boldsymbol{\alpha}}_{\S,n}^*-\hat{\boldsymbol{\alpha}}_{\S,n})+(\hat{\boldsymbol{\alpha}}_{\S,n}^*-\hat{\boldsymbol{\alpha}}_{\S,n})^{\top}(\hat{\boldsymbol{\beta}}_{\M,n}^*-\hat{\boldsymbol{\beta}}_{\M,n}) 
$; % in probability;
 \item[(ii)]  when $(\boldsymbol{\alpha}_S^{\top}, \boldsymbol{\beta}_M^{\top}) = \mathbf{0}$,
 % conditionally on the data, 	
 $\vec{\mathbb{R}}_{n}^*( \boldsymbol{b}_{\alpha}, \boldsymbol{b}_{\beta} )\overset{d^*}{\leadsto} 
n(\hat{\boldsymbol{\alpha}}_{\S,n}^{\mytrans} \hat{\boldsymbol{\beta}}_{\M, n} - {\boldsymbol{\alpha}}_{\S,n}^{\mytrans} {\boldsymbol{\beta}}_{\M, n})$ as $n(\hat{\boldsymbol{\alpha}}_{\S,n}^{\mytrans}\hat{\boldsymbol{\beta}}_{\M,n}-{\boldsymbol{\alpha}}_{\S,n}^{\mytrans}{\boldsymbol{\beta}}_{\M,n})=\vec{\mathbb{R}}_{1,n}(\boldsymbol{b}_{\alpha}, \boldsymbol{b}_{\beta})$. % in probability. 
\end{itemize}
Similarly to the proof in Section \ref{sec:pfbootstrapprodcomb}, 
by the property of OLS estimator of coefficients in the linear model, we can obtain 
\begin{align*}
	\vec{\myindicator}_{\lambda_n}^* \overset{\mathrm{P}^*}{\leadsto}  \myindicator\{\boldsymbol{\alpha}_S =\mathbf{0}  \text{ and } \boldsymbol{\beta}_M =\mathbf{0} \}, \hspace{1em} \text{ and } \hspace{1em}  1-\vec{\myindicator}_{\lambda_n}^* \overset{\mathrm{P}^*}{\leadsto}  \myindicator\{\boldsymbol{\alpha}_S\neq \mathbf{0} \text{ or } \boldsymbol{\beta}_M \neq \mathbf{0} \}.
\end{align*}
Theorem \ref{thm:multbootjoint} follows similarly to the arguments in Section \ref{sec:pfbootstrapprodcomb}. 

%Results similar to \eqref{eq:indicatorconvg} can be established as under the logistic models by the asymptotic normality in  \eqref{eq:normallimitbinary2} and \eqref{eq:normallimitbinary2beta}. Then the proof follows by Theorem  \ref{prop:bootconsisbinaryoutcome} and the arguments in Section \ref{sec:pfbootstrapprodcomb}. 

\clearpage
\subsubsection{Numerical Results of the Multivariate  Joint Test} \label{sec:numericresultsmulti}
We extend the simulation model in Section  \ref{sec:sim} to settings with multiple mediators. 
Specifically, we consider the following linear structural equation model, 
\begin{eqnarray}
	\M_j &=&  \alpha_{\S,j} \S + \alpha_{I,j}+ \alpha_{X,1,j} X_1 + \alpha_{X,2,j} X_2  +\epsilon_{M,j}, \quad \text{ for } j=1,\ldots, J,  \label{eq:simmodelmult}  \\ 
		\Y & = & \sum_{j=1}^{\nM}\beta_{\M,j} \M_j +  \beta_I + \beta_{X,1} X_1 + \beta_{X,2} X_2  + \DE \S + \eY.  \notag
\end{eqnarray}
In the model \eqref{eq:simmodelmult}, the exposure variable $\S$ is simulated from a Bernoulli distribution with the success probability equal to 0.5; the covariate $X_1$ is continuous and simulated from a standard normal distribution $\mathcal{N}(0,0.5^2)$; the covariate $X_2$ is discrete and  simulated from  a Bernoulli distribution with the success probability equal to 0.5; the error terms $\epsilon_{M,j}$ and $\eY$ are simulated independently from $\mathcal{N}(0,\sigma_{\eM}^2)$ and $ \mathcal{N}(0,\sigma_{\eY}^2)$, respectively.
We set the  parameters  $(\alpha_{I,j}, \alpha_{X,1,j}, \alpha_{X,2,j}) = (1, 1, 1)$ for $j=1,\ldots,\nM$, $(\beta_I, \beta_{X,1}, \beta_{X,2}) = (1, 1, 1)$, $\DE=1$, and $\sigma_{\eY}=\sigma_{\eM}=0.5$. 
We present the simulation results when $n=200,$ and $\nM=20$, and we use a fixed tuning parameter $\lambda=2$ across all scenarios. 
For each simulated data,  the adaptive bootstrap is conducted 500 times. Under each fixed null hypothesis, we simulate data over 1000 Monte Carlo replications to estimate the distribution of $p$-values. 
Two existing approaches to testing this group-level mediation effect include: 
Product Test based on Normal Product distribution (PT-NP) \citep{huang2016hypothesis}
and Product Test based on Normality (PT-N) \citep{huang2016hypothesis}. 

\paragraph{Part 1: Type \RNum{1} error under $H_0$.}\quad 
We consider different types of null hypotheses given in Table \ref{tb:nulltypes}.  
\begin{table}[!htbp]
\caption{Different types of null hypotheses for multivariate mediators}\label{tb:nulltypes}
\centering
%\begin{threeparttable}
\begin{tabular}{c|c|c}
Case & $\boldsymbol{\alpha}_{\S}$ & $\boldsymbol{\beta}_{\M}$  \\ \hline
1 & $\boldsymbol{0}_{\nM}$ & $\boldsymbol{0}_{\nM}$ \\[5pt]
2 & $\boldsymbol{1}_{\nM}$ & $\boldsymbol{0}_{\nM}$ \\[5pt]
3 & $\boldsymbol{0}_{\nM}$ & $\boldsymbol{1}_{\nM}$ \\[5pt]
4 & $(\boldsymbol{1}_{\nM/2},\, \boldsymbol{0}_{\nM/2})$ & $(\boldsymbol{0}_{\nM/2},\, \boldsymbol{1}_{\nM/2})$ \\[5pt]
5 & $(\boldsymbol{0}_{\nM/2},\, \boldsymbol{1}_{\nM/2})$ &  $(\boldsymbol{1}_{\nM/2},\, \boldsymbol{0}_{\nM/2})$ \\[5pt]
6 & $\boldsymbol{1}_{\nM}$ & $(\boldsymbol{1}_{\nM/2},\, \boldsymbol{-1}_{\nM/2})$ \\[5pt]
7 & $(\boldsymbol{1}_{\nM/2},\, \boldsymbol{-1}_{\nM/2})$ & $\boldsymbol{1}_{\nM}$ \\[5pt]
\end{tabular}
\end{table}

 \begin{figure}[!htbp]
 \captionsetup[subfigure]{labelformat=empty} 
     \centering
       \caption{Q-Q plots of $p$-values under different types of null hypotheses with $n = 200$ and $\nM=20$.} \label{fig:multifixnulln200}
\hspace{10em}
   \begin{subfigure}[b]{0.29\textwidth}
  \caption{\footnotesize{Case 1}}
     \includegraphics[width=\textwidth]{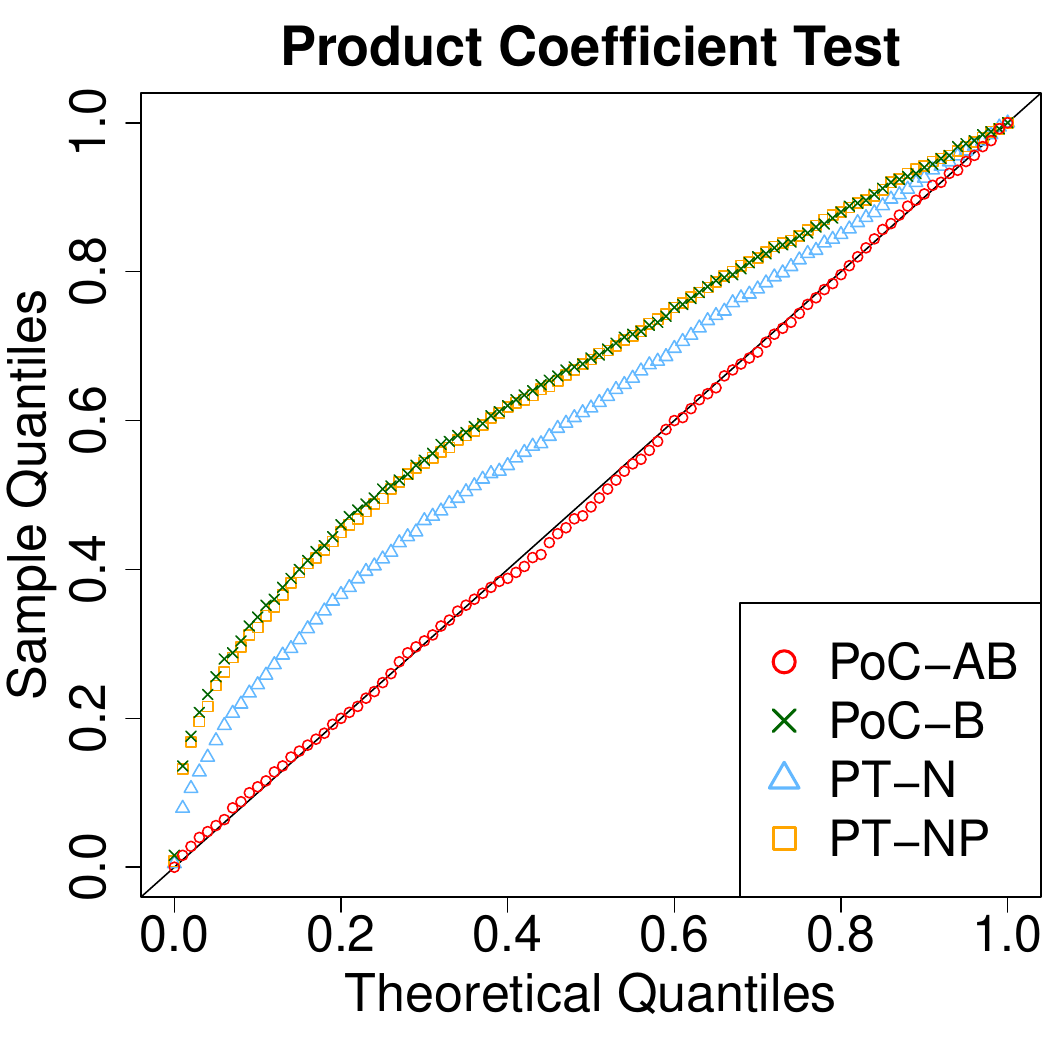}
   \end{subfigure} \ 
\hspace{10em}

  \begin{subfigure}[b]{0.29\textwidth}
      \caption{\footnotesize{Case 2}}
  \includegraphics[width=\textwidth]{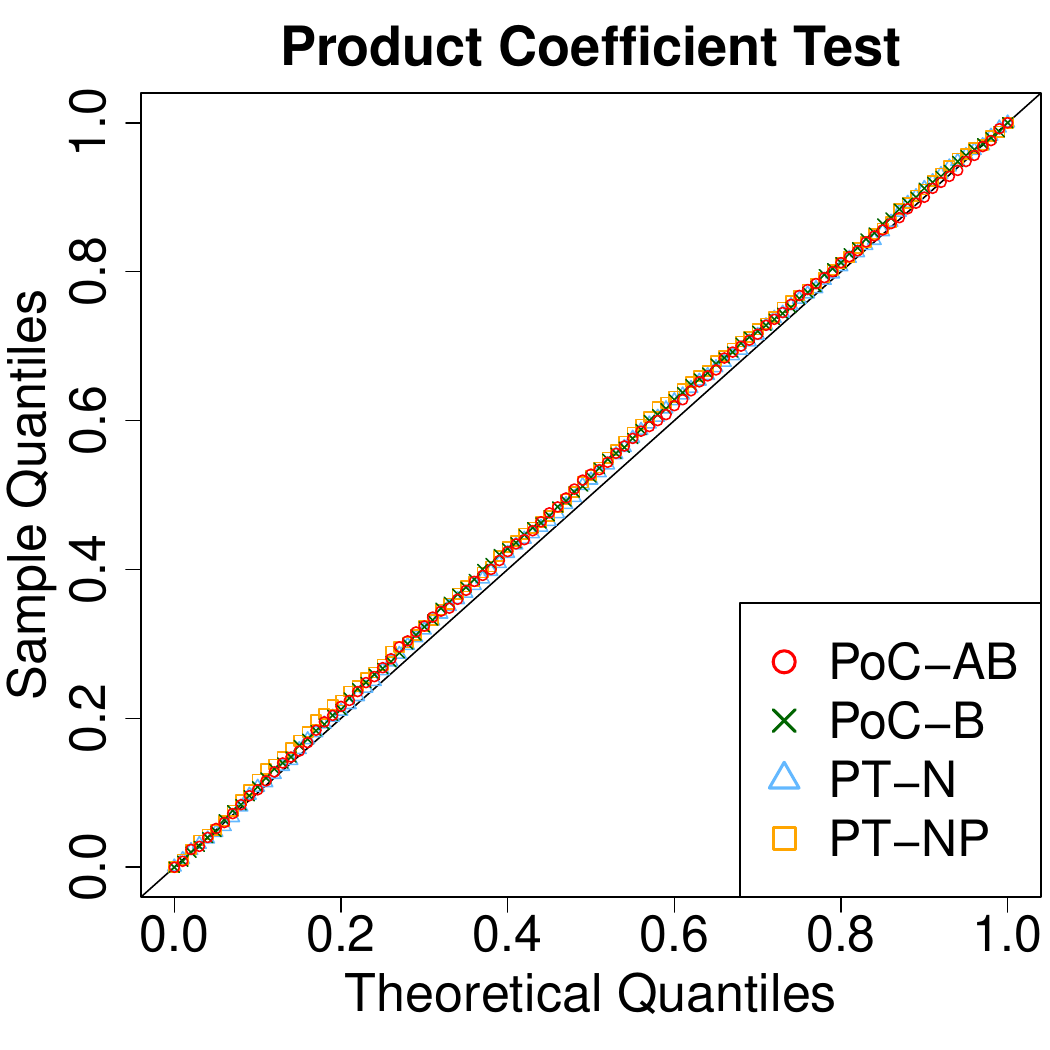}
     \end{subfigure} \  
       \begin{subfigure}[b]{0.29\textwidth}
   \caption{\footnotesize{Case 3}}
     \includegraphics[width=\textwidth]{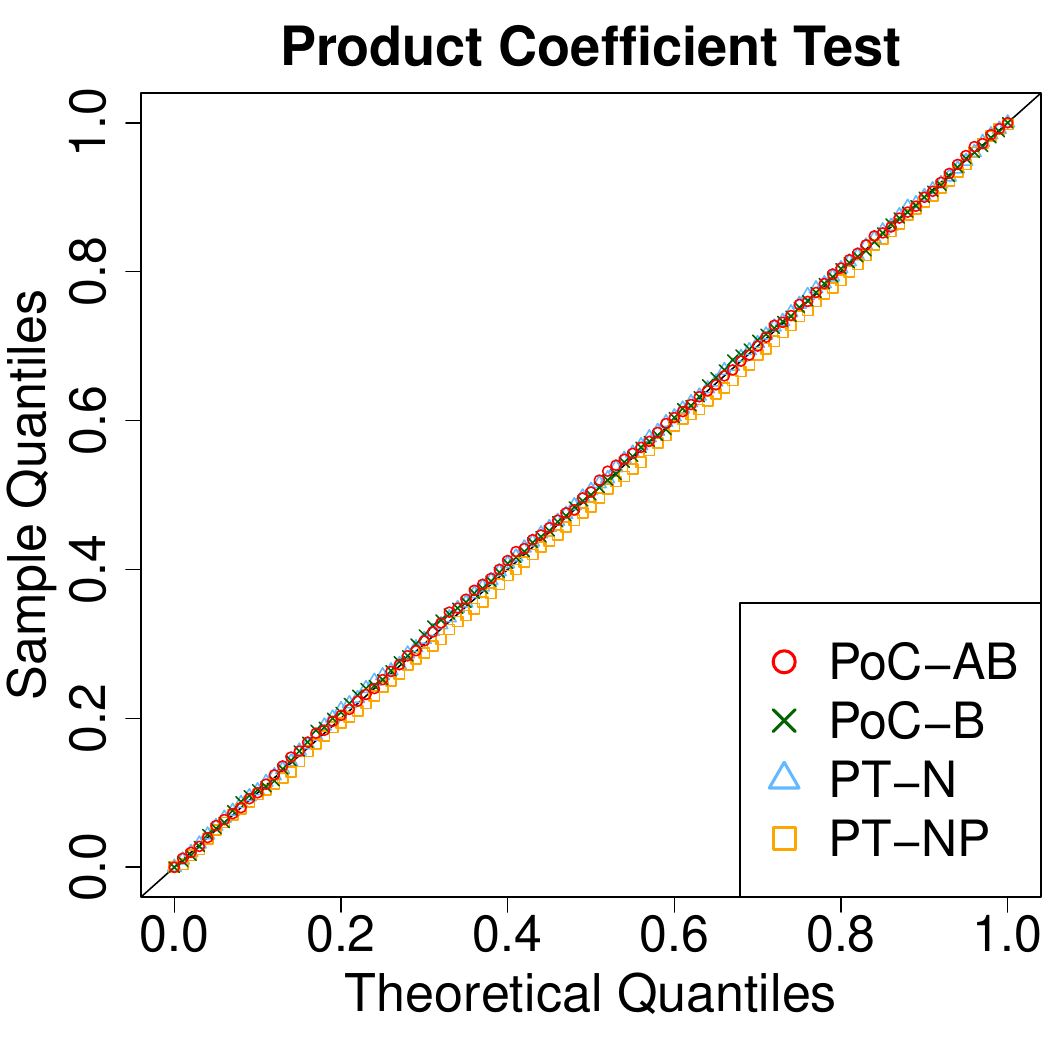}     \end{subfigure} \ 
  \begin{subfigure}[b]{0.29\textwidth}
  \caption{\footnotesize{Case 4}}
     \includegraphics[width=\textwidth]{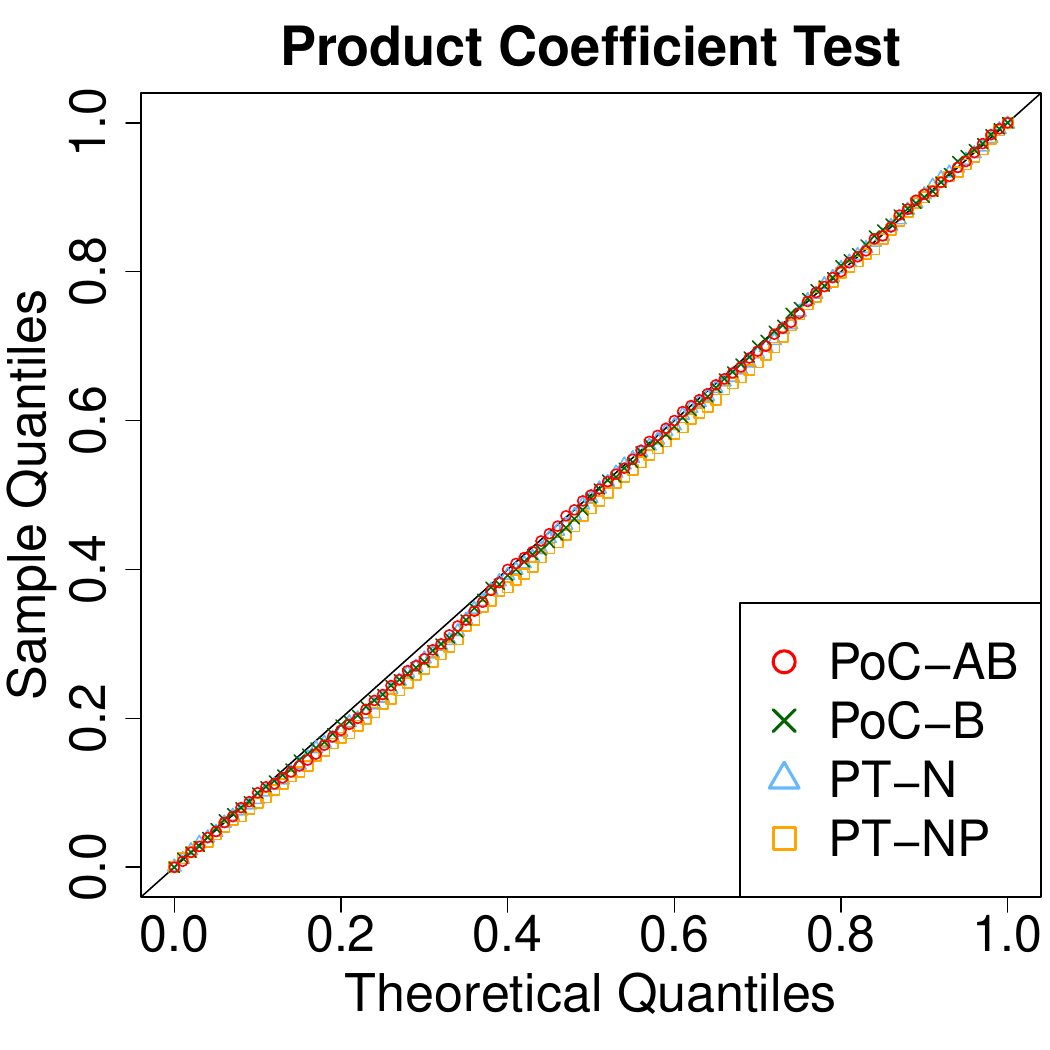}
   \end{subfigure} 
 \vspace{0.7em}

         \begin{subfigure}[b]{0.29\textwidth}
  \caption{\footnotesize{Case 5}}
     \includegraphics[width=\textwidth]{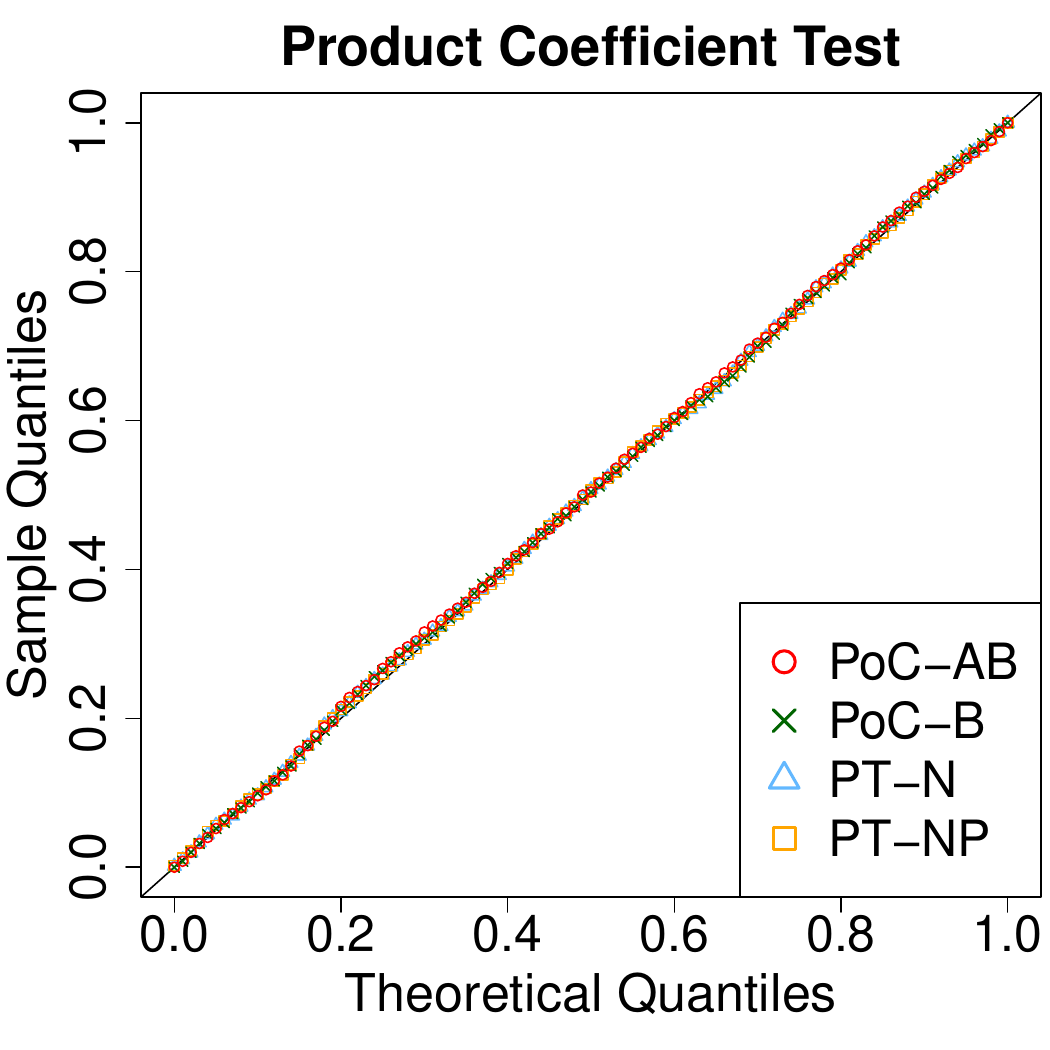}
   \end{subfigure} 
         \begin{subfigure}[b]{0.29\textwidth}
  \caption{\footnotesize{Case 6}}
     \includegraphics[width=\textwidth]{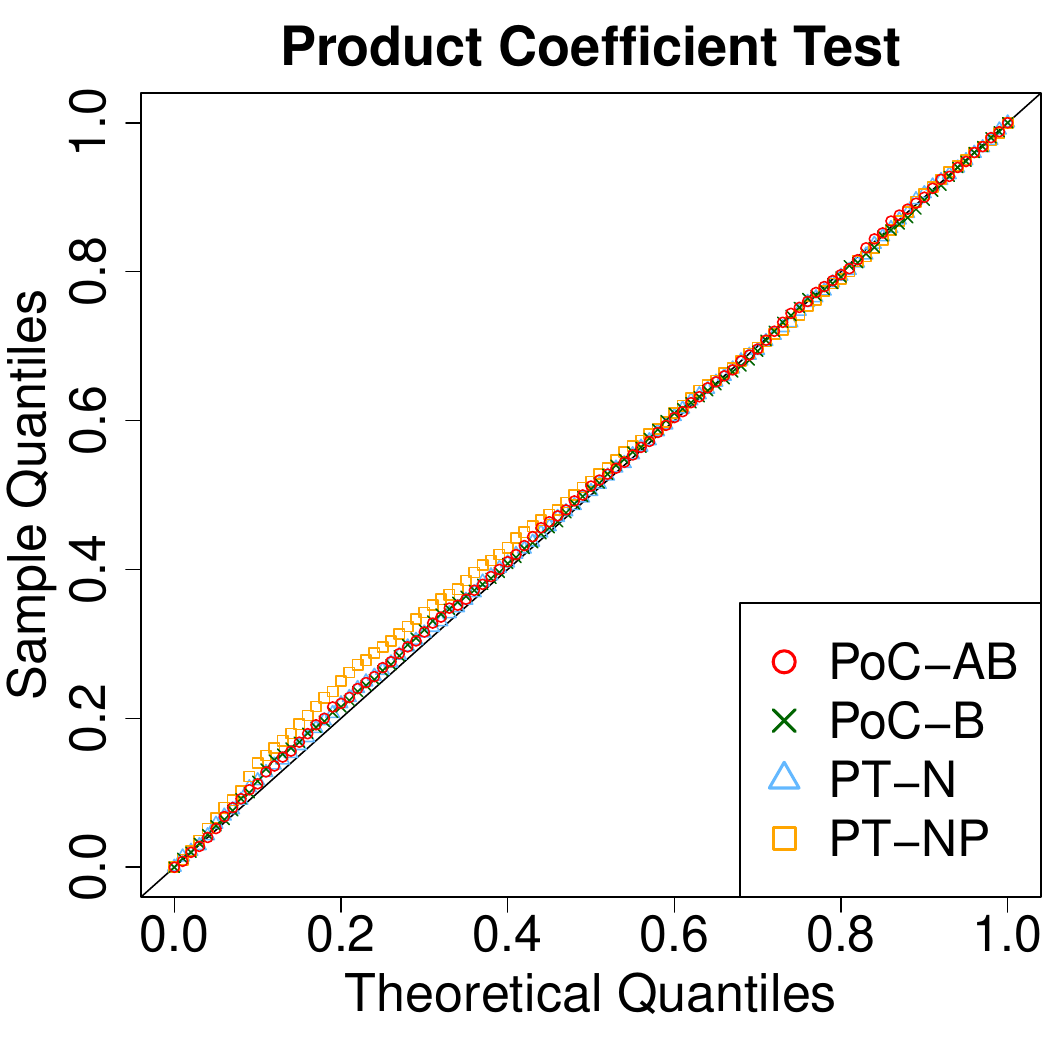}
   \end{subfigure} 
        \begin{subfigure}[b]{0.29\textwidth}
  \caption{\footnotesize{Case 7}}
     \includegraphics[width=\textwidth]{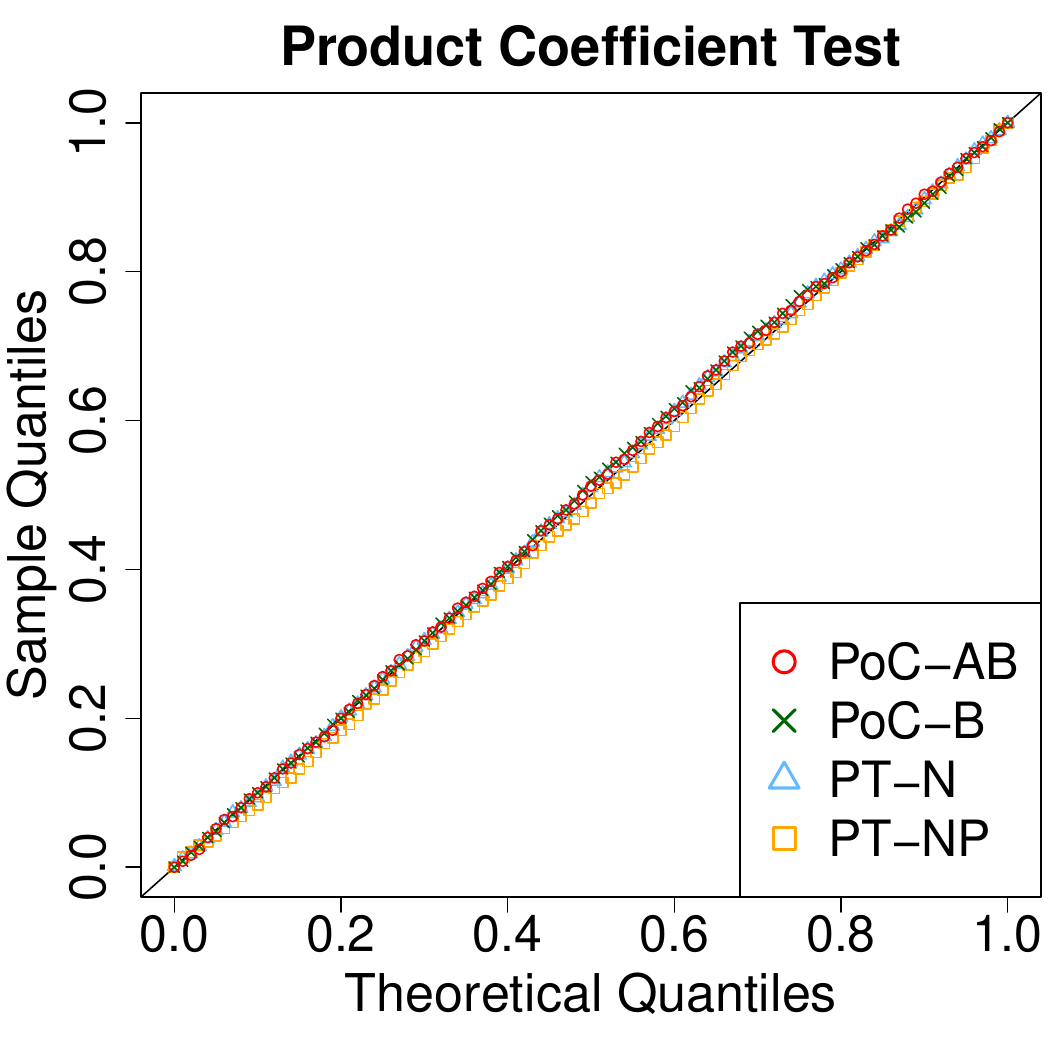}
   \end{subfigure}    
\end{figure}

\newpage
\paragraph{Part 2: Statistical power under $H_A$.}\quad 
Under the alternative hypotheses, we consider $\boldsymbol{\alpha}_{\S}=a\times \mathbf{1}_{\nM}$ and $\boldsymbol{\beta}_{\M}=b\times \mathbf{1}_{\nM}$.
 We fix the size of the mediation effect $\boldsymbol{\alpha}_{\S}^{\mytrans}\boldsymbol{\beta}_{\M}$ and vary the ratio $a/b$.
Figure \ref{fig:powermulti}  presents the empirical rejection rates (power) versus the ratio $a/b$ for $n\in \{200, 500\}$, respectively. 
When $n=200$, we fix $\boldsymbol{\alpha}_{\S}^{\mytrans}\boldsymbol{\beta}_{\M}=0.1$; when $n=500$, we fix $\boldsymbol{\alpha}_{\S}^{\mytrans}\boldsymbol{\beta}_{\M}=0.04$. 
\begin{figure}[!htbp] 
\centering
\caption{\large Empirical rejection rates (power) versus $a/b$} \label{fig:powermulti}
\includegraphics[width=1\textwidth]{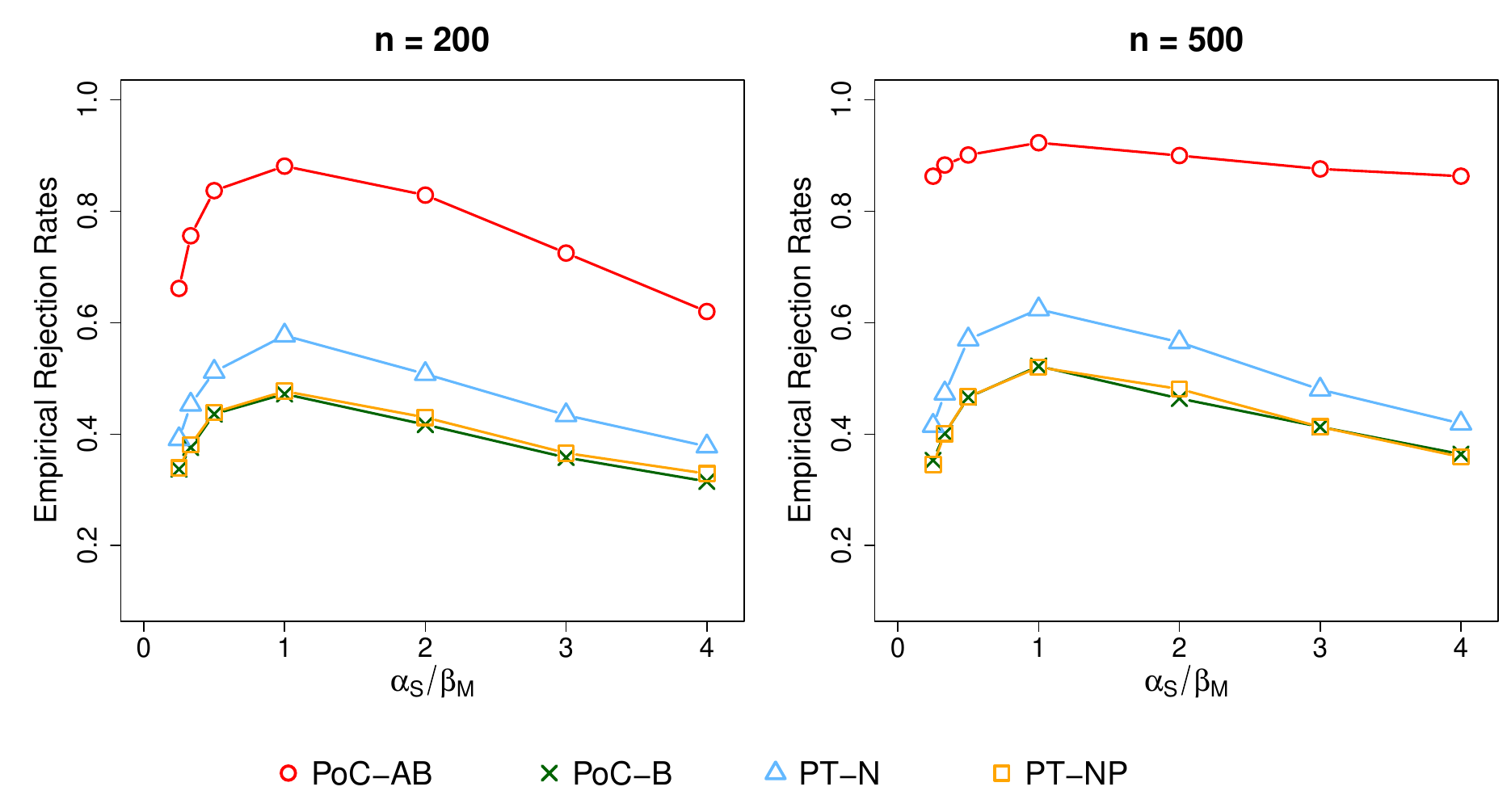} 
\end{figure}

\clearpage 
\section{Extensions Beyond Linear Models} \label{sec:discussiononglm}
This section provides supplementary details to Sections \ref{sec:asymbinarymedout} and \ref{sec:modelbinmed}  of the main text.

\subsection{Supplementary Theoretical Details}  \label{sec:singissue}

\begin{condition} \label{cond:phiintegration}
%Let $\phi(m; \nu)$ denote the conditional density of $M\mid (\tau_S S+ \X^{\mytrans}\alphaX = \nu)$,
The link function $h^{-1}(\cdot)$ in \eqref{eq:extendmodel1} is strictly monotone. 
Moreover, let  $P_{\nu}( M \leqslant m)$  denote the cumulative distribution function of $M\mid \tau_S S+ \X^{\mytrans}\alphaX = \nu$. 
Assume that given any $m$ in the support of distribution,  $P_{\nu}( M \leqslant m)$ is continuously differentiable with respect to $\nu$, and $\frac{\partial P_{\nu}( M \leqslant m)}{\partial \nu}$ is always positive or always negative when  $P_{\nu}( M \leqslant m)$ is not a constant with respect to $\nu$. %$\phi(m;\nu)$ is strictly monotone in $\nu$, or 
%the quantile of the distribution is strictly increasing in $\nu$. 
%$\nu$ is a location shift parameter of the density $\phi(m;\nu)$. 
\end{condition}
%\textit{Examples of the conditional distribution:} 
Condition \ref{cond:phiintegration} can be satisfied under various distributions. 
(i) Bernoulli distribution (logistic regression): 
% $\phi(m; \nu)=\{g(\nu)\}^m\{1-g(\nu)\}^{1-m}$, where $g(\nu)=\frac{e^{\nu}}{1+e^{\nu}}$; 
 $h^{-1}(\nu)=g(\nu)$. 
(ii) Normal distribution (linear regression) with fixed variance: $h^{-1}(\nu)=\nu$.
(iii) Poisson distribution: $h^{-1}(\nu)=\exp(\nu)$. 
In the following, Proposition \ref{prop:singularitylogiresponse} Part 1 shows that the null hypothesis \eqref{eq:ornull} is composite similarly to that under the linear SEMs, and the Part 2  further  specifies the singularity issue under the composite null hypothesis.

\begin{condition}\label{cond:designmatrixlogistic}
Let $D_{\alpha}=(S, \X^{\mytrans})^{\mytrans}$,   
 $g_{\alpha}=g(\S \alphaS + \X^{\mytrans}\alphaX)$,
 $D_{\beta}=(M, S, \X^{\mytrans})^{\mytrans}$, 
 and 
$g_{\beta}=g(M\betaM +\S \tau_S + \X^{\mytrans}\betaX)$.  Assume $ \mathrm{E}\{g_{\alpha}(1-g_{\alpha})D_{\alpha}D_{\alpha}^{\mytrans} \}  
$ and $\mathrm{E}\{g_{\beta}(1-g_{\beta})D_{\beta}D_{\beta}^{\mytrans}\}$ are positive definite matrices with bounded eigenvalues. 
\end{condition}
Condition \ref{cond:designmatrixlogistic} is a general regularity condition on the design matrix, which is similar to Condition \ref{cond:inversemoment} under linear SEM in the main text.

\begin{remark}\label{rm:extintegrationdef}
Sections \ref{sec:asymbinarymedout} and \ref{sec:modelbinmed}  consider natural indirect effects/mediation effects conditioning on covariates $\X$ following \cite{vanderweele2010odds}.
On the other hand, \cite{imai2010general} proposed to examine the average NIE that marginalizes the distribution of the covariates $\X$. 
Examining the conditional NIE is mainly for technical convenience. 
%can be technically more convenient than examining the average NIE. 
The conditional NIE $=0$ can give a sufficient condition for the average NIE $=0$. 
Some results of conditional NIE could be  established for average NIE similarly.    
For instance, under Scenario II, if $h^{-1}(\cdot)$ is strictly monotone, conclusions in Proposition \ref{prop:singularitylogitm} also hold for the average $\mathrm{NIE}_{s\mid s^*}(s):=\int\mathrm{NIE}_{s\mid s^*}(s,\boldsymbol{x})\mathrm{d}P_{\X}(\boldsymbol{x})$, as the sign of $\mathrm{NIE}_{s\mid s^*}(s,\boldsymbol{x})$ for all $\boldsymbol{x}$ are the same.    
Similarly, under Scenario I, results similar to Proposition \ref{prop:singularitylogiresponse} can also be established for the average $\mathrm{OR}_{s\mid s^*}(s):=\int\mathrm{OR}_{s\mid s^*}(s,\boldsymbol{x})\mathrm{d}P_{\X}(\boldsymbol{x})$ if $\mathrm{OR}_{s\mid s^*}(s,\boldsymbol{x})-1$ is non-negative/non-positive for all $\boldsymbol{x}$.  
As an example, this could hold when $\M$ in  \eqref{eq:extendmodel1} is binary and follows a logistic regression model, which is a case studied in Section \ref{sec:asymbinarymedout} in detail.  
\end{remark}

\begin{condition}\label{cond:designmatrixbinarymediator}
Let $D_{\alpha}=(S, \X^{\mytrans})^{\mytrans}$ and   
 $g_{\alpha}=g(\S \alphaS + \X^{\mytrans}\alphaX)$. Assume $ \mathrm{E}\{g_{\alpha}(1-g_{\alpha})D_{\alpha}D_{\alpha}^{\mytrans} \}  
$ is a positive definite matrix with bounded eigenvalues.
Assume conditions on the model of the outcome $Y$ in  Condition \ref{cond:inversemoment} in the main text . 
\end{condition}
Condition \ref{cond:designmatrixbinarymediator} is similar to Condition  \ref{cond:designmatrixlogistic} and  Condition  \ref{cond:inversemoment}.

\subsection{Simulations under Non-Linear Models}\label{sec:nonlineartwo}

%\paragraph{Part 1: Type-I error under $H_0$.} %\quad 
%Under the alternative hypotheses, we consider
\paragraph{Non-linear Scenario I.}
\noindent 
For $i=1,\ldots, n$, we generate binary mediators $M_i$ and  outcomes $\Y_i$ follow Bernoulli distributions with conditional means  
$\mathrm{E}(\M_i \mid \S_i, \X_i )=g(\alphaS \S_i + \alpha_I + \alpha_X X_i )$, and $\mathrm{E}(\Y_i \mid \S_i, \M_i, \X_i )  = g(\betaM \M_i+ \beta_I + \beta_X X_i +\tau_S S_i)$,  respectively, where $g(x)=\mathrm{logit}^{-1}(x)=e^x/(1+e^x)$.  
We take $S_i$ and $X_i \sim \mathrm{Bernoulli}(0.5)$, independently,    
$\alpha_I=\beta_I=-1$, and $\alpha_X=\beta_X=\tau_S=1$. 
Following the definition in Section \ref{sec:asymbinarymedout}, 
we examine the NIE when $s=0$, $s^*=1$, $X=0$, that is, $
\log  \mathrm{OR}^{\mathrm{NIE}}_{s\mid s^*}(s,\boldsymbol{x})=\log  \mathrm{OR}^{\mathrm{NIE}}_{0\mid 1}(0,0) = 	l(P_0) -l(P_1),
$ 
where 
$
P_0 = 	 d_{\beta,n} \times g\big( \alpha_I \big) + g(\beta_I )$, and $P_1= d_{\beta} \times g(\alphaS + \alpha_I) + g(\beta_I)$. We set $\lambda_{\alpha}=1.9\sqrt{n}/\log n$ and $\lambda_{\beta}=1.9\sqrt{n}/\log n$.   

\paragraph{Non-linear Scenario II}
\noindent For $i=1,\ldots, n$, we generate $\M_i$ as i.i.d. Bernoulli random variables with the conditional mean $\mathrm{E}(\M_i \mid \S_i, X_i )=g(\alphaS \S_i + \alpha_I + \alpha_X X_i )$,  
and $\Y_i =\betaM \M_i  + \beta_I + \beta_X X_i + \tau_S S_i + \epsilon_i$. 
Similarly to the  scenario I above,
we take  $\S_i$ and $X_i \sim$   Bernoulli(0.5) and $\epsilon_i\sim N(0,0.5^2)$   i.i.d.,
%$\alpha_I=\beta_I=-1$,
%and $\alpha_X=\beta_X=\tau_S=1$. 
%%where $\S_i$ are i.i.d. Bernoulli(0.5), $X_i $ are i.i.d. $\text{Bernoulli}$ with the success probability 0.5, 
%$\alpha_I=\alpha_X=1$. 
%Moreover, we generate 
%where  $\epsilon_i$ are i.i.d. $N(0,0.5^2)$, $\beta_I=\beta_X=\tau_S=1$. 
In this case, we test the conditional natural indirect effect  in \eqref{eq:h0binarymediator} when $x=0$, $s=1$, and $s^*=0$, i.e., 
$
\text{NIE}_{1|0}(0) = \betaM\{g(\alphaS + \alpha_I) - g( \alpha_I) \}. 	
$ 
 $\lambda_{\alpha}=1.9\sqrt{n}/\log n$ and $\lambda_{\beta}=3.3\sqrt{n}/\log n$. 

\paragraph{Results.}
\noindent \textit{(1) Under $H_0:$ Type-I error control.} We estimate $p$-values under three different types of  null hypothesis over 1000 repetitions.
 We present QQ plots of  $p$-values obtained under Scenarios I and II in Figures \ref{fig:logisticoutcomen500} and \ref{fig:logisticn500}, respectively.  
Similarly to the linear cases, 
we observe that both the classical bootstrap (B) and the adaptive bootstrap (AB) give uniformly distributed $p$-values under $H_{0,1}$ and $H_{0,2}$, 
whereas under $H_{0,3}$, 
the classical bootstrap would become overly conservative,
and the adaptive bootstrap can still give uniformly distributed $p$-values. 
Specifically, we fix $\alphaS\betaM=0.5^2$ in Scenario I and $\alphaS\betaM=0.25^2$ in Scenario II. 
We observe that the adaptive bootstrap can improve the power of the classical bootstrap. 

\noindent \textit{(2) Under $H_A:$ Power.}
Under the alternative hypotheses,
we fix the product $\alphaS\betaM$,
and vary the ratio $\alphaS/\betaM$.
%In Scenario I, we fix $\alphaS\betaM = 0.5^2$,  and vary the ratio $\alphaS/\betaM$.
We present the empirical power versus the ratio $\alphaS/\betaM$ in Figure    \ref{fig:logisticpowern500}. 
 
%Note that $g(\alphaS + \alpha_I)$ and $g(\betaM+\beta_I)$ are model fitted value when $X=1$. 
%$\lambda_{\alpha}=1.3\sqrt{n}/\log n$ and $\lambda_{\beta}=1.3 \sqrt{n}/\log n$. 

 \begin{figure}[!htbp]
 \captionsetup[subfigure]{labelformat=empty} 
     \centering
       \caption{Scenario 1: Q-Q plots of $p$-values with $n = 500$.} \label{fig:logisticoutcomen500}
   \begin{subfigure}[b]{0.3\textwidth}
  \caption{{$H_{0,1}: (\alphaS,\betaM)=(0,0)$}}
     \includegraphics[width=\textwidth]{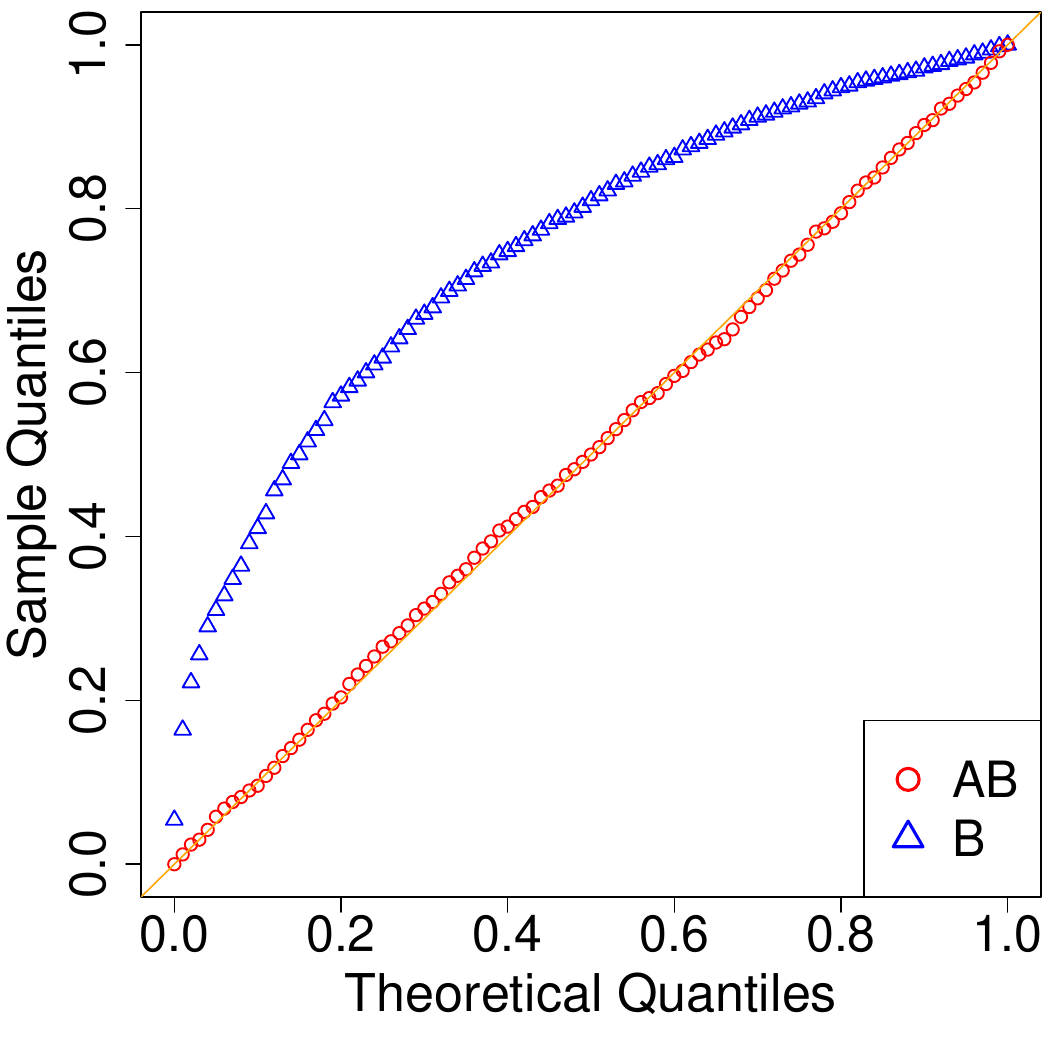}
   \end{subfigure} \ \ \ 
  \begin{subfigure}[b]{0.3\textwidth}
  \caption{$H_{0,2}: (\alphaS,\betaM)=(2,0)$}
     \includegraphics[width=\textwidth]{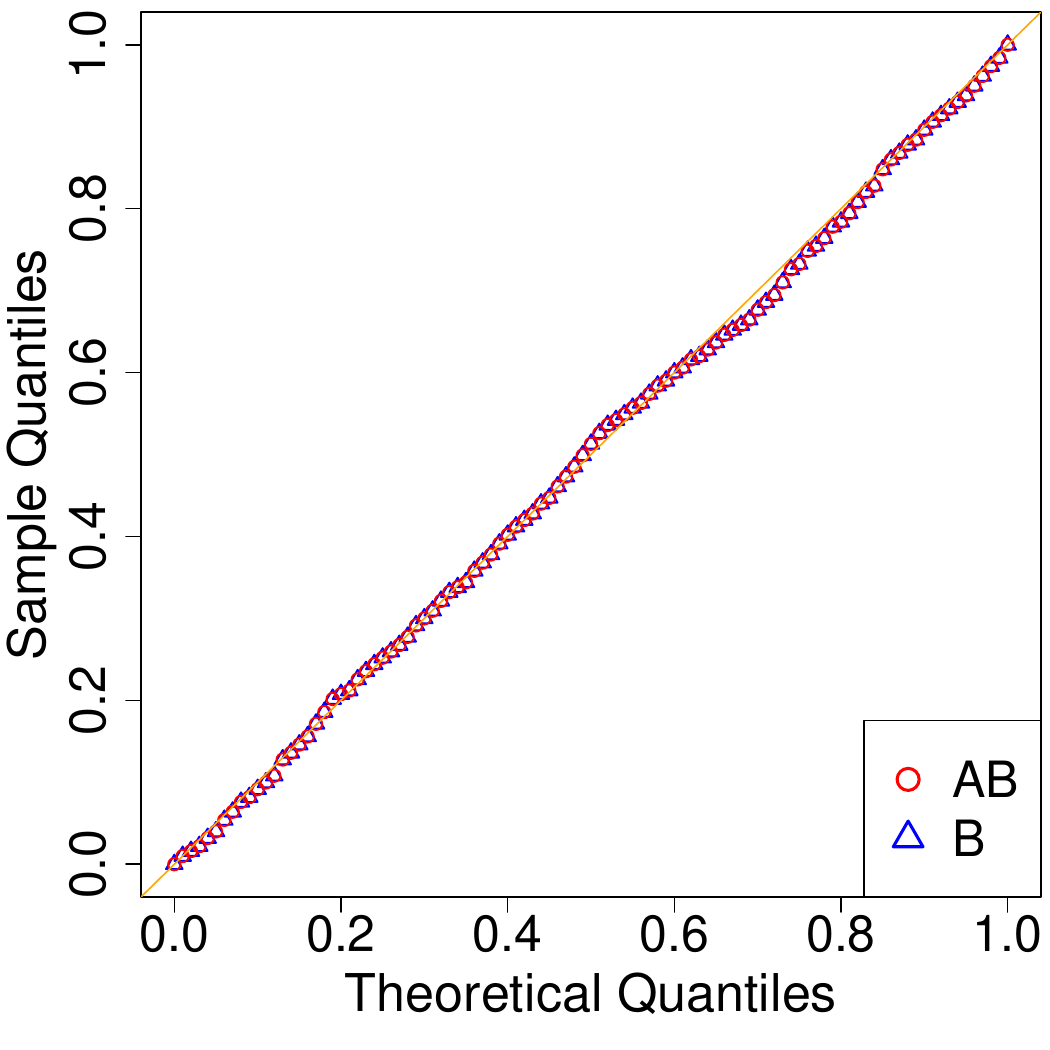}
   \end{subfigure}\ \ \ 
  \begin{subfigure}[b]{0.3\textwidth}
  \caption{$H_{0,3}: (\alphaS,\betaM)=(0,2)$}
     \includegraphics[width=\textwidth]{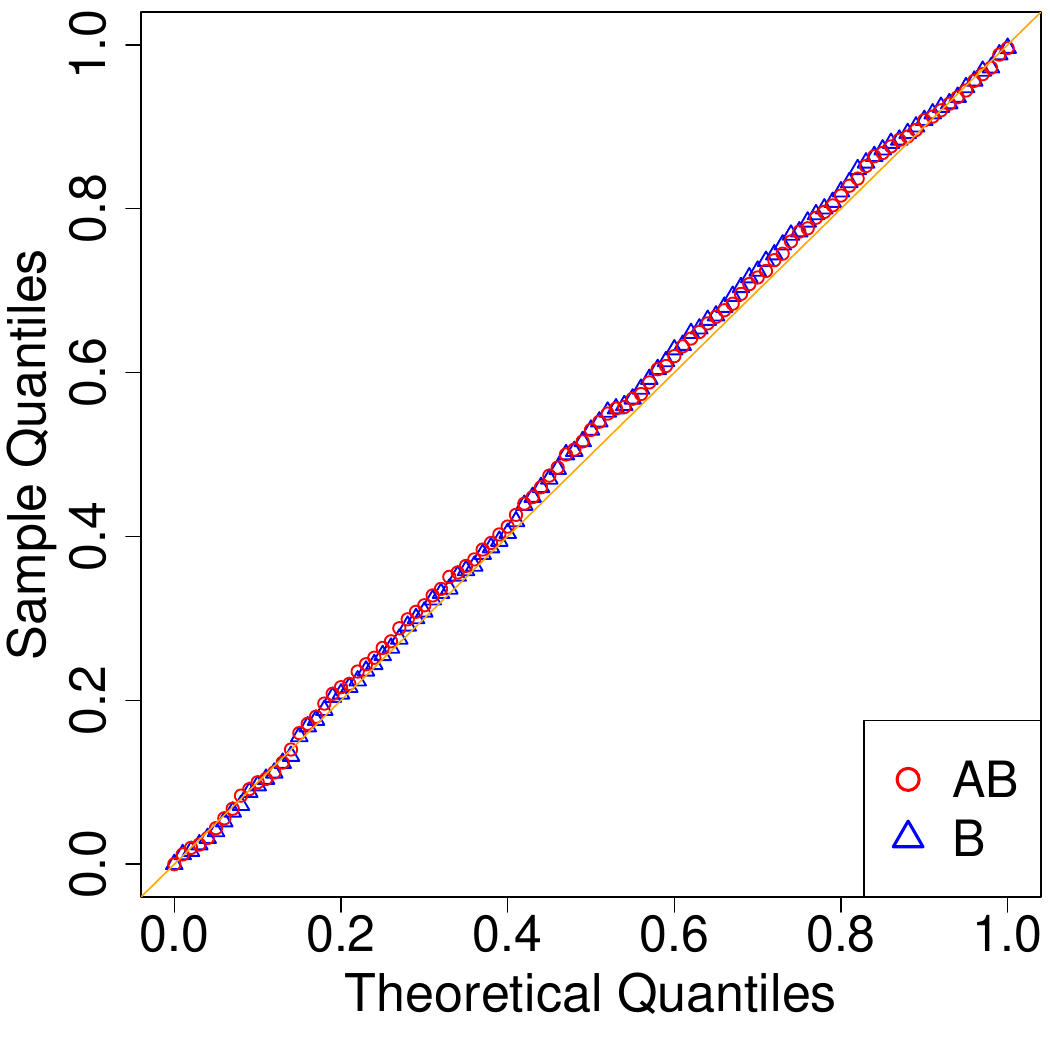}
   \end{subfigure}
\end{figure}

 \begin{figure}[!htbp]
 \captionsetup[subfigure]{labelformat=empty} 
     \centering
       \caption{Scenario 2: Q-Q plots of $p$-values with $n = 500$.} \label{fig:logisticn500}
   \begin{subfigure}[b]{0.3\textwidth}
  \caption{{$H_{0,1}: (\alphaS,\betaM)=(0,0)$}}
     \includegraphics[width=\textwidth]{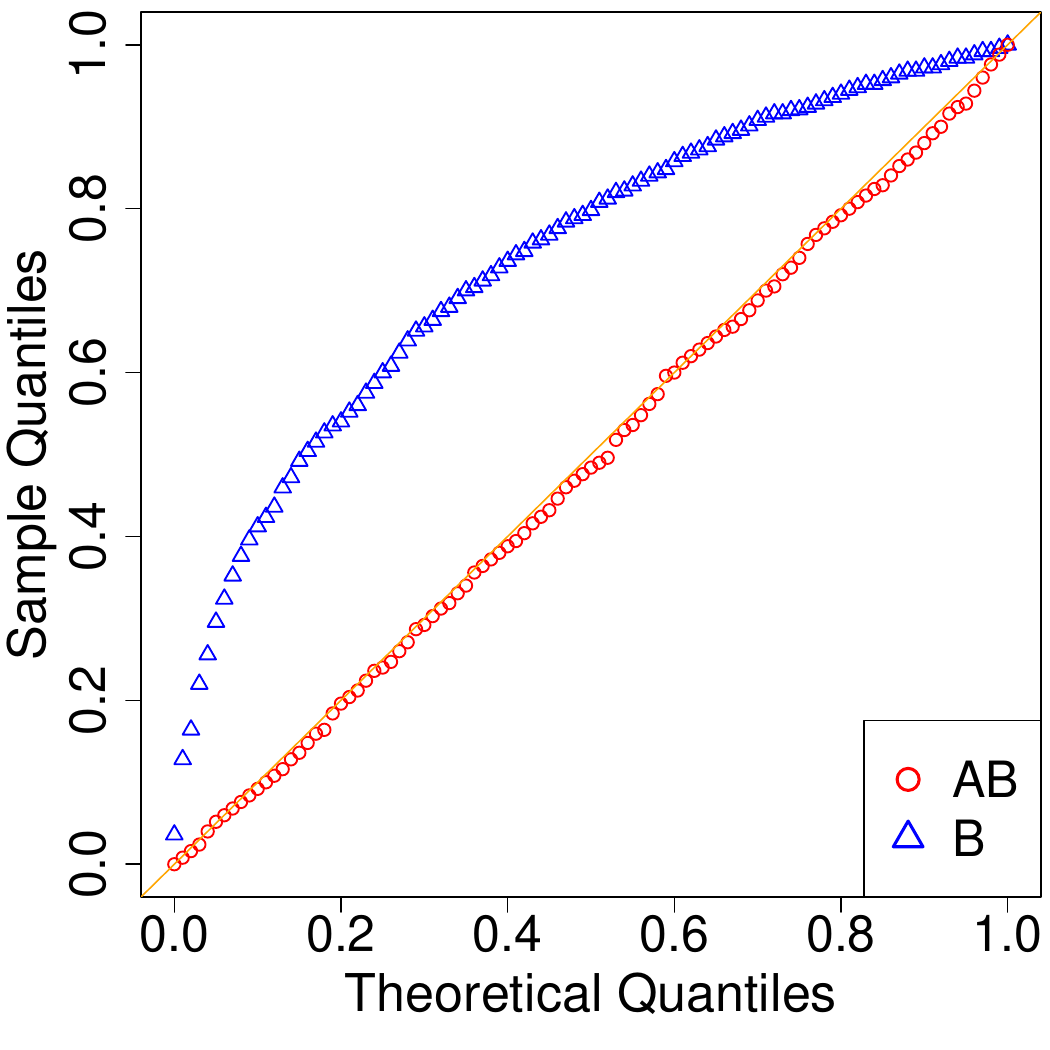}
   \end{subfigure} \ \ \ 
  \begin{subfigure}[b]{0.3\textwidth}
  \caption{$H_{0,2}: (\alphaS,\betaM)=(2,0)$}
     \includegraphics[width=\textwidth]{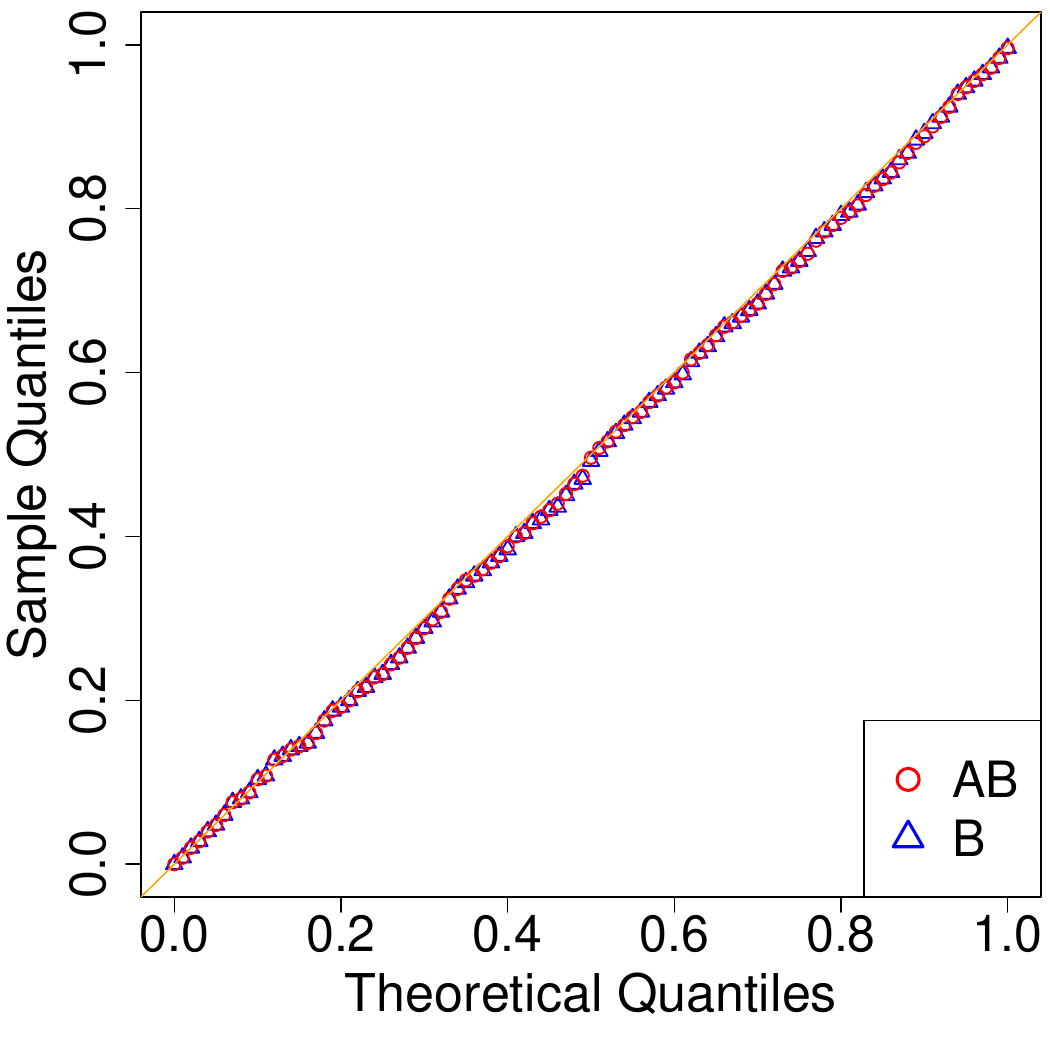}
   \end{subfigure}\ \ \ 
  \begin{subfigure}[b]{0.3\textwidth}
  \caption{$H_{0,3}: (\alphaS,\betaM)=(0,2)$}
     \includegraphics[width=\textwidth]{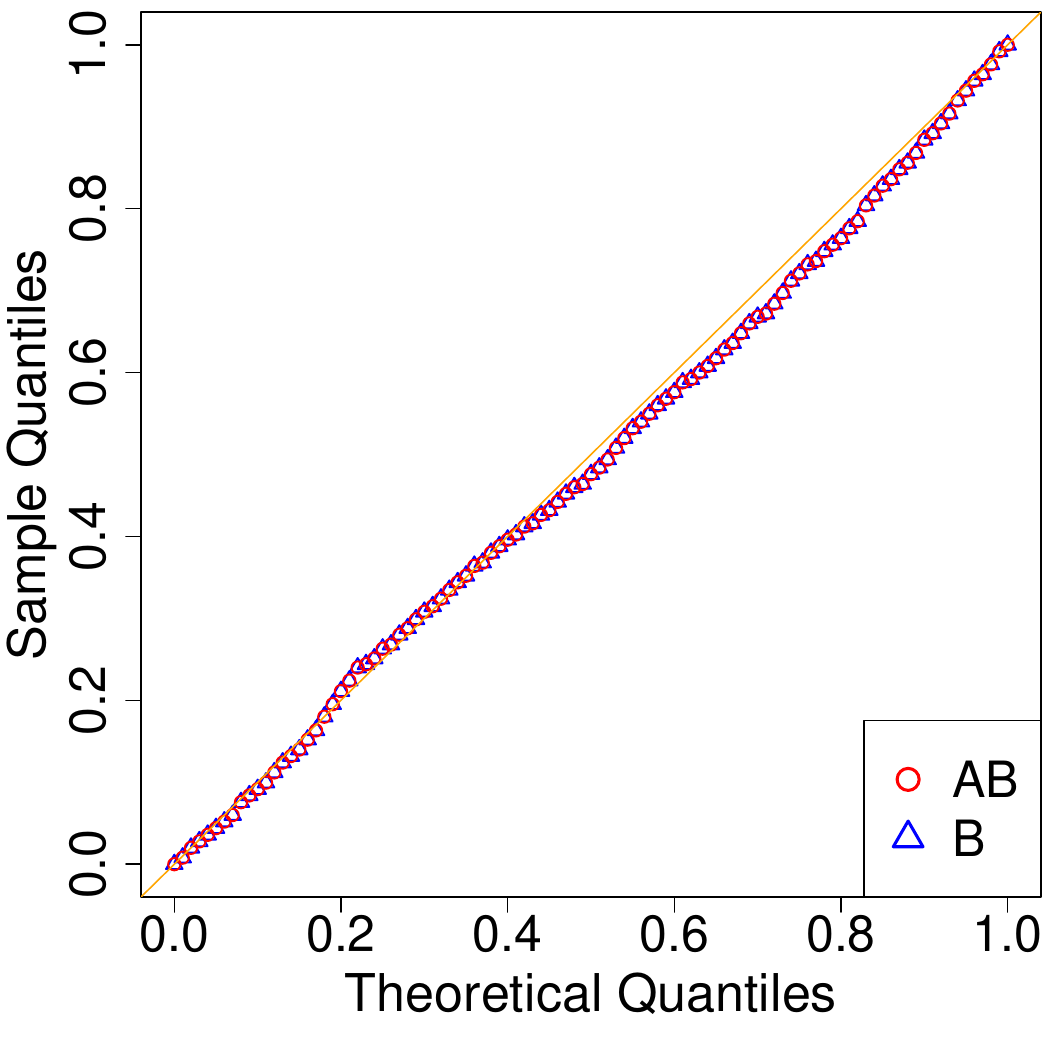}
   \end{subfigure}
\end{figure}

\begin{figure}[!htbp]
 \captionsetup[subfigure]{labelformat=empty} 
     \centering
       \caption{Empirical power versus the ratio $\alphaS/\betaM$.} \label{fig:logisticpowern500}
   \begin{subfigure}[b]{0.48\textwidth} 
%  \caption{Scenario 2}
     \includegraphics[width=\textwidth]{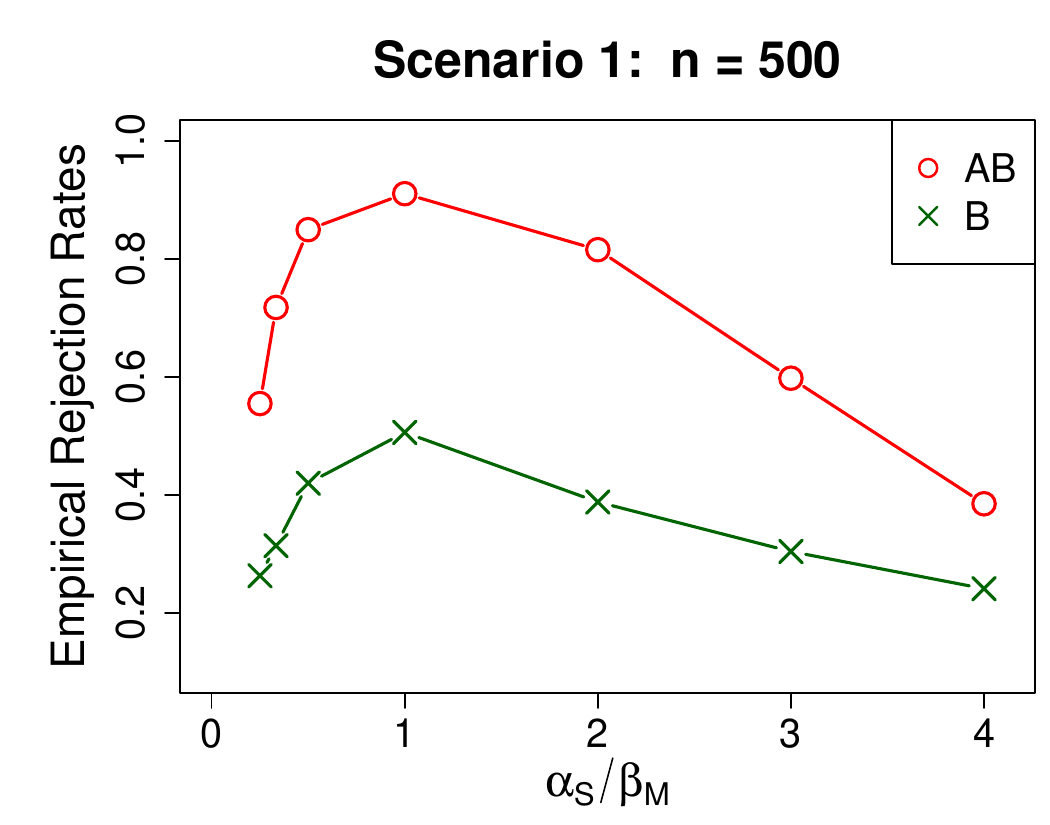}
   \end{subfigure}\quad  
  \begin{subfigure}[b]{0.48\textwidth}
%  \caption{Scenario 2}
     \includegraphics[width=\textwidth]{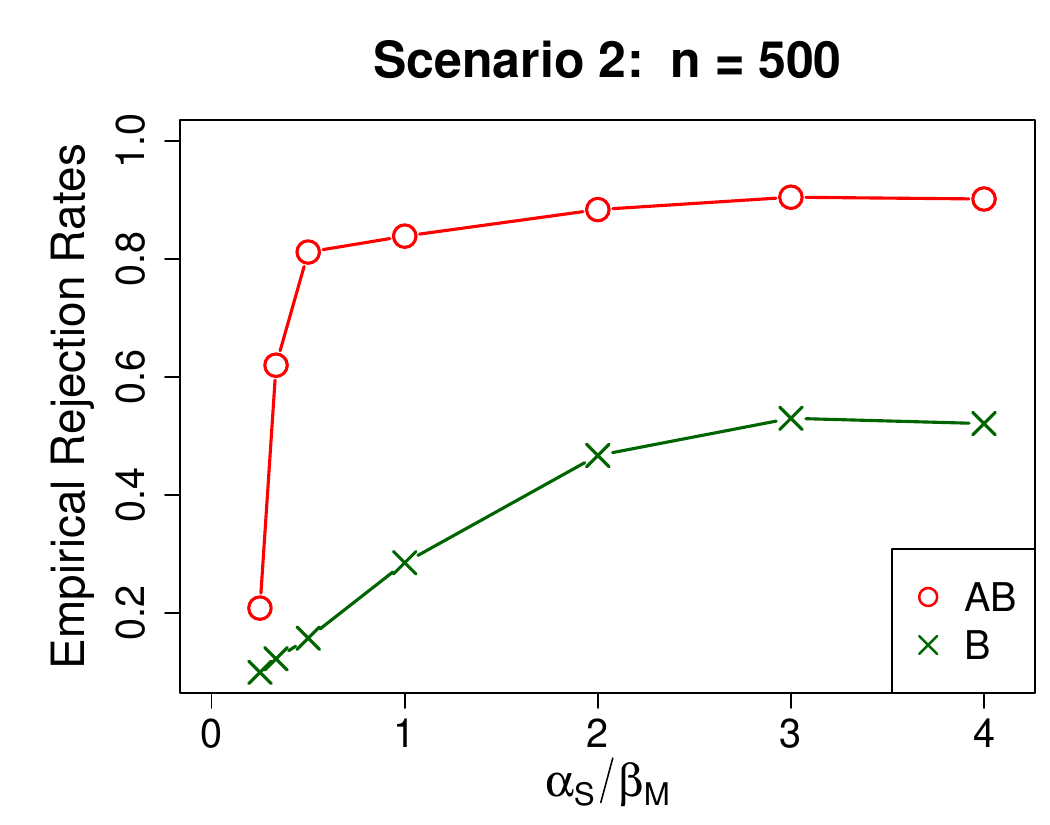}
   \end{subfigure}
\end{figure}

\subsection{Proof of Proposition \ref{prop:singularitylogiresponse}}
\paragraph{Proof of Part 1.}
Let $\phi(m; \nu)$ denote the conditional density of $M\mid (\tau_S s+ \X^{\mytrans}\alphaX = \nu)$. We have 
\begin{align*}
P_s:=&~P\left\{Y(s, M(s))=1 \mid  \X = \boldsymbol{x} \right\}=\int g( \betaM m + \mu_{1,s} ) \phi( m ; \alphaS s+\mu_{2} )	 \mathrm{d} m, \notag\\
P_{s^*}:=&~P\left\{Y(s, M(s^*))=1 \mid  \X= \boldsymbol{x}  \right\}=\int g( \betaM m + \mu_{1,s} ) \phi( m ; \alphaS s^*+\mu_{2} )	 \mathrm{d} m, 
\end{align*}
where we define  $\mu_{1,s}=\tau_S s + \boldsymbol{x}^{\mytrans}\betaX$, $\mu_{2}=\boldsymbol{x}^{\mytrans}\alphaX$, and  $g(x)=\mathrm{logit}^{-1}(x)=e^x/(1+e^x)$. 
%Then
%\begin{align}
%&~P\left\{Y(s, M(s))=1 \mid  \X \right\}-P\left\{Y(s, M(s^*))=1 \mid  \X \right\}\notag\\
%=&~\int g( \betaM m + \mu_{1,s} ) \times \big\{\phi( m ; \alphaS s+\mu_{2} )-\phi( m ; \alphaS s^*+\mu_{2} )\big\} \mathrm{d} m. \label{eq:diffprobform}
%\end{align} 
$H_0$ in \eqref{eq:ornull} is equivalent to $P_{s}-P_{s^*}=0$. 
First, if $\betaM=0$, 
\begin{align*}
P_{s}-P_{s^*}	=g(\mu_{1,s})\int  \big\{\phi( m ; \alphaS s+\mu_{2} )-\phi( m ; \alphaS s^*+\mu_{2} )\big\} \mathrm{d} m =0. 
\end{align*}
%the above equation holds.
Second, if $\betaM\neq 0$, $g( \betaM m + \mu_{1,s} ) >0$.
% is strictly monotone in $m$. 
%(i) If $\phi(m;\nu)$ is strictly monotone in $\nu$,
%assume without loss of generality that $\alphaS s > \alphaS s^* $. Then $\phi( m ; \alphaS s+\mu_{2} )>\phi( m ; \alphaS s^*+\mu_{2} )$ for all possible $m$. By $g(x) > 0$, 
% $\eqref{eq:diffprobform}\neq 0$. 
%(ii) Otherwise, 
By $g(x) > 0$ and the integrated tail probability expectation formula, we have $\mathrm{E}_X\{g(X)\}=\int g(X)\mathrm{d}F(X)=\int_{0}^{\infty} P\{g(X) > t\} \mathrm{d}t$ for any integrable random variable $X$. It follows that
\begin{align*}
P_{s}-P_{s^*}=\int_{0}^{\infty} \Big[ P_{\alphaS s +\mu_2}\big\{g( \betaM M + \mu_{1,s} ) > t \big\} - P_{\alphaS s^* +\mu_2}\big\{g( \betaM M + \mu_{1,s} ) > t \big\}\Big]  \mathrm{d} t  	
\end{align*}
where for $\nu=\alphaS s +\mu_2$ or $\nu=\alphaS s^* +\mu_2$, we define
\begin{align*}
P_{\nu}\Big\{g( \betaM M + \mu_{1,s}) > t \Big\} =&~	 \int \mathrm{I}\big\{ g( \betaM M + \mu_{1,s}) > t \big\}  \phi(m;\nu)\mathrm{d} m\notag\\
=&~	P_{\nu}\Big\{ \betaM M > g^{-1}(t) -  \mu_{1,s}\Big\},
\end{align*}
where the second equation holds as $g(x)$ is strictly increasing.
%, and we assume without loss of generality that $\betaM>0$. 
By Condition \ref{cond:phiintegration}, given any $m$, $P_{\nu}(\betaM M>m)$ is strictly monotone in $\nu$. Therefore, when $\betaM\neq 0$,  $P_{s}-P_{s^*} = 0$ if and only if $\alphaS s +\mu_{1,s} = \alphaS s^* + \mu_{1,s} \Leftrightarrow \alphaS=0$. 
In summary, $H_0$ \eqref{eq:ornull} holds for $s\neq s^*$ if and only if $\alphaS=0$ or $\betaM=0$. 	

\paragraph{Proof of Part 2.}
For the simplicity of notation, we let 
$P_s = P\left\{Y(s, M(s))=1 \mid  \X \right\}$
and $P_{s^*} = P\left\{Y(s, M(s))=1 \mid  \X \right\}$.
For the parameter $\theta \in \{\alphaS, \betaM\}$, 
\begin{align}\label{eq:logniepartial}
\frac{\partial \log \mathrm{NIE}}{\partial \theta} =\frac{1}{P_s(1-P_s)}\frac{\partial P_s}{\partial \theta}-\frac{1}{P_{s^*}(1-P_{s^*})}\frac{\partial P_{s^*}}{\partial \theta}.
\end{align}	
\noindent (i.1) When $\alphaS=0$, $P_s= P_{s^*}$ and $\phi(m;\alphaS s + \mu_{2}) = \phi(m; \alphaS s^* + \mu_{2})$, where $\mu_2 = \X^{\mytrans}\alphaX$. By $\alphaS=0$, 
\begin{align*}
\frac{\partial P_s}{\partial \betaM} = &\int \frac{g(\betaM m +\mu_{1,s})}{\partial \beta_M} \phi(m; \alphaS s + \mu_2)  \mathrm{d}m = \int g'(\betaM m +\mu_{1,s}) m \phi(m;  \mu_2)  \mathrm{d}m, \\
\frac{\partial P_{s^*}}{\partial \betaM}=&\int \frac{g(\betaM m +\mu_{1,s})}{\partial \beta_M} \phi(m; \alphaS s^* + \mu_{2})  \mathrm{d}m =  \int g'(\betaM m +\mu_{1,s}) m \phi(m;  \mu_2)  \mathrm{d}m, 
\end{align*}
where
% $g(x)=\mathrm{logit}^{-1}(x)=e^x/(1+e^x)$  and 
 $g'(x)=e^x/(1+e^x)^2$. 
It follows that $\frac{\partial P_s}{\partial \betaM}=\frac{\partial P_{s^*}}{\partial \betaM}\mid_{\alphaS=0}$.  
By \eqref{eq:logniepartial}, 
$\frac{\partial \log \mathrm{NIE}}{\partial \betaM} \mid_{\alphaS=0} = 0$.

\noindent (i.2) When $\betaM=0$,  we have 
$
	P_s=\int g(\mu_{1,s}) \phi(m; \alphaS s + \mu_2) \mathrm{d} m 
	=g(\mu_{1,s}),
%P_{s^*} = \int \mathrm{logit}^{-1}(\mu_{1,s}) \phi(m; \mu_{s^*}) \mathrm{d} m. 
$
where we use $\int \phi(m; \alphaS s + \mu_2) \mathrm{d} m =1$ by the property of  density. 
Similarly, we have $P_{s^*} = P_s=g(\mu_{1,s})$. 
Moreover, when $\betaM=0$, 
\begin{align*}
\frac{\partial P_s}{\partial \alphaS}\biggr|_{\betaM=0}=	&~\int g(\betaM m + \mu_{1,s}) \frac{\partial \phi(m; \alphaS s + \mu_{2})}{\partial \alphaS} \mathrm{d} m \biggr|_{\betaM=0}= g(\mu_{1,s})\int \frac{\partial \phi(m; \alphaS s+\mu_{2})}{\partial \alphaS} \mathrm{d} m,  
\end{align*} 
so 
$
	\frac{\partial P_s}{\partial \alphaS} = \frac{\partial P_{s^*}}{\partial \alphaS} \mid_{\betaM=0}.  
$ 
By \eqref{eq:logniepartial}, 
$\frac{\partial \log \mathrm{NIE}}{\partial \alphaS} \mid_{\betaM=0} = 0$. 

\begin{itemize}
%	\item[(i)] When $\alphaS=\betaM=0$,  the first-order Delta method cannot be applied to $\log \mathrm{NIE}$.
%Therefore, asymptotic normality or bootstrap cannot be applied for the inference of $\log \mathrm{NIE}$ following the discussion in Section \ref{sec:prodcoef}. 
\item[(ii)] When $\alphaS=0$ and $\betaM\neq 0$, we have $P_s=P_{s^*}$, and by \eqref{eq:logniepartial}, 
\begin{align*}
\frac{\partial \log \mathrm{NIE}}{\partial \alphaS} =\frac{1}{P_s(1-P_s)}\left( \frac{\partial P_s}{\partial \alphaS} -   \frac{\partial P_{s^*}}{\partial \alphaS} \right).
\end{align*} 
When $\alphaS=0$, we have 
\begin{align*}
\left(\frac{\partial P_s}{\partial \alphaS}-   \frac{\partial P_{s^*}}{\partial \alphaS} \right) \biggr|_{\alphaS=0}=(s-s^*)\int_{0}^{\infty} \frac{\partial P_{\nu}\big\{\betaM M > g^{-1}(t)-  \mu_{1,s}\big\}}{\partial \nu}\biggr|_{\nu=\mu_2} \mathrm{d}t \neq 0
\end{align*}
which follows by Condition \ref{cond:phiintegration}. 
\item[(iii)] When $\betaM=0$ and $\alphaS\neq 0$, we have $P_s=P_{s^*}$, 
\begin{align} 
\left(\frac{\partial P_s}{\partial \betaM}-   \frac{\partial P_{s^*}}{\partial \betaM} \right) \biggr|_{\betaM=0}=&g'(\mu_{1,s})\int  m\big\{ \phi(m; \alphaS s + \mu_2) - \phi(m; \alphaS s^* + \mu_2) \big\} \mathrm{d}m, \notag\\
=&g'(\mu_{1,s}) \big\{ h^{-1}(\alphaS s + \mu_2)-h^{-1}(\alphaS s^* + \mu_2) \big\} \label{eq:partialbetadiff}
\end{align}
which is obtained by the definition of calculating population mean and the model  \eqref{eq:extendmodel1}. 
When $h^{-1}(\cdot)$ is strictly monotone, by $g'(\mu_{1,s})>0$, $\eqref{eq:partialbetadiff}\neq 0$. 
\end{itemize}

\subsection{Proof of  Proposition \ref{prop:singularitylogitm}}
\paragraph{Proof of Part 1.} The conclusion follows by the form of $\mathrm{NIE}$ in \eqref{eq:h0binarymediator}. 
\paragraph{Proof of Part 2.}  Note that
\begin{align*}
\frac{\partial \mathrm{NIE}}{\partial \alphaS} = &~\betaM\Big\{g'\big(\alphaS s+ \boldsymbol{x}^{\mytrans}\alphaX  \big)\times s - g'\big(\alphaS s^*+ \boldsymbol{x}^{\mytrans}\alphaX \big)\times s^*\Big\}\\
\frac{\partial \mathrm{NIE}}{\partial \betaM} =&~ g\big(\alphaS s+ \boldsymbol{x}^{\mytrans}\alphaX  \big) -g\big(\alphaS s^*+ \boldsymbol{x}^{\mytrans}\alphaX \big) 
\end{align*}
where $g'(x)=e^x(1+e^x)^2$. \\
(i) $\frac{\partial \mathrm{NIE}}{\partial \alphaS}\mid_{\betaM=0} =0$ and $\frac{\partial \mathrm{NIE}}{\partial \alphaS}\mid_{\alphaS=0}=g(\boldsymbol{x}^{\mytrans}\alphaX )-g(\boldsymbol{x}^{\mytrans}\alphaX )=0$. \\
(ii) $\frac{\partial \mathrm{NIE}}{\partial \alphaS}\mid_{\alphaS=0, \betaM\neq 0} = \betaM g'( \boldsymbol{x}^{\mytrans}\alphaX)(s -s^*)\neq 0$ when $s\neq s^*$.\\
(iii) $\frac{\partial \mathrm{NIE}}{\partial \betaM}\mid_{\alphaS\neq 0, \betaM=0} = g(\alphaS s+ \boldsymbol{x}^{\mytrans}\alphaX  ) -g(\alphaS s^*+ \boldsymbol{x}^{\mytrans}\alphaX)  \neq 0 $, which follows by the strict monotonicity of the function  $g(x)$.

\subsection{Proof of Theorem \ref{prop:asymbinaryoutcome}} \label{sec:proofasymbinaryoutcome}
By the definition $P_s=P\left\{Y(s, M(s))=1 \mid  \X=\boldsymbol{x} \right\}$ and the model  \eqref{eq:extendmodel1local},  
\begin{align*}
P_s = 	&~\sum_{m\in \{0,1\}} P(Y(s,m)=1 \mid M(s)=m, \X=\boldsymbol{x})P(M(s)=m\mid \X=\boldsymbol{x})\notag\\
=&~\sum_{m\in \{0,1\}} g(\betaMn m + \boldsymbol{x}^{\mytrans} \boldsymbol{\beta}_{\X} + \tau_S s ) \big\{g\big(\alphaSn s+\boldsymbol{x}^{\mytrans}\alphaX\big)\big\}^m  \big\{1-g\big( \alphaSn s+\boldsymbol{x}^{\mytrans}\alphaX\big)\big\}^{1-m}\notag\\ 
=&~g(\betaMn + \boldsymbol{x}^{\mytrans} \boldsymbol{\beta}_{\X} + \tau_S s ) g\big(\alphaSn s+\boldsymbol{x}^{\mytrans}\alphaX\big)+g(\boldsymbol{x}^{\mytrans} \boldsymbol{\beta}_{\X} + \tau_S s ) \big\{1-g\big( \alphaSn s+\boldsymbol{x}^{\mytrans}\alphaX\big)\big\}\notag\\ 
=&~  d_{\beta,n} \times g\big( \alphaSn s+\boldsymbol{x}^{\mytrans}\alphaX\big) + g(\boldsymbol{x}^{\mytrans} \boldsymbol{\beta}_{\X} + \tau_S s ),
\end{align*}
where for the simplicity of notation, we define $ d_{\beta,n} =g(\betaMn + \boldsymbol{x}^{\mytrans} \boldsymbol{\beta}_{\X} + \tau_S s ) -g(\boldsymbol{x}^{\mytrans} \boldsymbol{\beta}_{\X} + \tau_S s )$. 
Similarly, by $P_{s^*}=P\left\{Y(s, M(s))=1 \mid  \X=\boldsymbol{x} \right\}$,
\begin{align*}
P_{s^*}= &~ \sum_{m\in \{0,1\}} P(Y(s,m)=1 \mid M(s^*)=m, \X=\boldsymbol{x})P(M(s^*)=m\mid \X=\boldsymbol{x})\notag\\
=&~\sum_{m\in \{0,1\}} g(\betaMn m + \boldsymbol{x}^{\mytrans} \boldsymbol{\beta}_{\X} + \tau_S s ) \big\{g\big(\alphaSn s^*+\boldsymbol{x}^{\mytrans}\alphaX\big)\big\}^m  \big\{1-g\big( \alphaSn s^*+\boldsymbol{x}^{\mytrans}\alphaX\big)\big\}^{1-m}\notag\\ 
%=&~g(\betaM + \boldsymbol{x}^{\mytrans} \boldsymbol{\beta}_{\X} + \tau_S s ) g\big(\alphaS s^*+\boldsymbol{x}^{\mytrans}\alphaX\big)+g(\boldsymbol{x}^{\mytrans} \boldsymbol{\beta}_{\X} + \tau_S s ) \big\{1-g\big( \alphaS s^*+\boldsymbol{x}^{\mytrans}\alphaX\big)\big\}\notag\\ 
=&~ d_{\beta,n}  \times g\big( \alphaSn s^*+\boldsymbol{x}^{\mytrans}\alphaX\big) + g(\boldsymbol{x}^{\mytrans} \boldsymbol{\beta}_{\X} + \tau_S s ). 
\end{align*}
By consistency of regression coefficients of the logistic regressions, we have $\hat{P}_r-P_r= O_p(n^{-1/2})$ for $r\in \{s,s^*\}$.

Let $l(x)=\log \frac{x}{1-x}$ and its derivative is  $l'(x)=\frac{1}{x(1-x)}$.  
By $\hat{P}_r - P_r = o_p(1)$ for $r\in \{s,s^*\}$ and Taylor's expansion,  
\begin{align}
\widehat{\mathrm{NIE}}-\mathrm{NIE}=&~l(\hat{P}_s)-l(\hat{P}_{s^*}) - l(P_{s}) + l(P_{s^*})\notag\\
=&~ \big\{l'(P_s)(\hat{P}_s - P_s) - l'(P_{s^*})(\hat{P}_{s^*} - P_{s^*})\big\}\times\{1+O_p(n^{-1/2})\}. \label{eq:loghatdiff}
\end{align}
In particular, for $r\in \{s,s^*\}$, we have
\begin{align*}
\hat{P}_r-P_r
= &~\hat{d}_{{\beta}}\times g\big( \hat{\alpha}_S r+\boldsymbol{x}^{\mytrans}\hat{\boldsymbol{\alpha}}_{\X}\big)  -d_{\beta,n}\times g\big( \alphaSn r+\boldsymbol{x}^{\mytrans}\alphaX\big) \notag\\
 &~ +g(\boldsymbol{x}^{\mytrans} \hat{\boldsymbol{\beta}}_{\X} + \hat{\tau}_S s )- g(\boldsymbol{x}^{\mytrans} \boldsymbol{\beta}_{\X} + \tau_S s ), \notag	
\end{align*}
where we define $\hat{d}_{{\beta}}=g(\hat{\beta}_M + \boldsymbol{x}^{\mytrans} \hat{\boldsymbol{\beta}}_{\X} + \hat{\tau}_S s ) -g(\boldsymbol{x}^{\mytrans}\hat{\boldsymbol{\beta}}_{\X} + \hat{\tau}_S s )$,
and $(\hat{\beta}_M, \hat{\beta}_{\X},\hat{\tau}_S)$ denote  sample estimates of $({\beta}_M, {\beta}_{\X}, {\tau}_S)$ under the logistic regression.  
As $P_s=P_{s^*}$ under $H_0$,  
\begin{align*}
\eqref{eq:loghatdiff}=&~ l'(P_s)\{(\hat{P}_s - P_s)- (\hat{P}_{s^*} - P_{s^*}) \}\times \{1+O_p(n^{-1/2})\}\notag\\
=&~l'(P_s)\big(\hat{d}_{{\alpha}}\hat{d}_{{\beta}} - d_{\alpha,n}d_{\beta,n} \big)\times \big\{1+O_p(n^{-1/2})\big\},
\end{align*}
where we define  
$d_{\alpha,n}=g(\alphaSn s + \boldsymbol{x}^{\mytrans}\alphaX)-g(\alphaSn s^* + \boldsymbol{x}^{\mytrans}\alphaX)$, 
 $\hat{d}_{\alpha}=g( \hat{\alpha}_S s+\boldsymbol{x}^{\mytrans}\hat{\boldsymbol{\alpha}}_{\X})-g( \hat{\alpha}_S s^*+\boldsymbol{x}^{\mytrans}\hat{\boldsymbol{\alpha}}_{\X})$, and $(\hat{\alpha}_S,\hat{\boldsymbol{\alpha}}_{\X})$ denote sample estimates of $({\alpha}_S,{\boldsymbol{\alpha}}_{\X})$ under the logistic regression. 
%Note that we can decompose 
%\begin{align*}
%\hat{d}_{{\alpha}_S}\hat{d}_{{\beta}_M} - d_{\alphaS,n}d_{\betaM,n} =&~(\hat{d}_{{\alpha}_S}-d_{\alphaS,n})d_{\betaM,n}+(\hat{d}_{{\beta}_M}-d_{\betaM,n}) d_{\alphaS,n}\notag\\
%&~+(\hat{d}_{{\alpha}_S}-d_{\alphaS,n})(\hat{d}_{{\beta}_M}-d_{\betaM,n}).  
%\end{align*}
%Moreover, we can decompose 
%\begin{align*}
%\eqref{eq:loghatdiff}=&~	l'(P_s)\big\{(\hat{d}_{{\alpha}_S}-d_{\alphaS,n})d_{\betaM,n}+(\hat{d}_{{\beta}_M}-d_{\betaM,n}) d_{\alphaS,n}\notag\\
%&~\quad \quad +(\hat{d}_{{\alpha}_S}-d_{\alphaS,n})(\hat{d}_{{\beta}_M}-d_{\betaM,n}) \big\}\times \big\{1+O_p(n^{-1/2})\big\}
%\end{align*}
To analyze the asymptotic property of $\eqref{eq:loghatdiff}$,  we next discuss under three cases of $H_0$, respectively. 

\medskip
\noindent (i) When $\alphaS\neq 0$ and $\betaM=0$, we have 
%no singularity issue show asymptotic normality.
$d_{\beta,n}\to d_{\beta}=0$, 
whereas $d_{\alpha,n} \to d_{\alpha}\neq 0$ by strict monotonicity of the function $g(\cdot)$.
Moreover, by  $d_{\beta}=0$,  
$l'(P_s)$ and $l'(P_{s^*}) \to \gamma_{0}=l'\{g(\boldsymbol{x}^{\mytrans} \boldsymbol{\beta}_{\X} + \tau_S s )\}\}\neq 0$ as $n\to \infty$.
%By $d_{\betaM}=0$, we also have $\gamma_{0}=l'\{g(\boldsymbol{x}^{\mytrans} \boldsymbol{\beta}_{\X} + \tau_S s )\}\neq 0$. 
Therefore, as $n\to \infty$, 
\begin{align*} 
%	\eqref{eq:loghatdiff}=&~l'(P_s)\big\{(\hat{d}_{{\beta}_M}-0)d_{\alphaS} +  (\hat{d}_{{\beta}_M}-0)(\hat{d}_{{\alpha}_S}-d_{\alphaS})\big\} \big\{1+O_p(n^{-1/2})\big\}\notag\\
\sqrt{n}\times \eqref{eq:loghatdiff}\to &~\gamma_{0} \times d_{\alpha}\times \sqrt{n}(\hat{d}_{{\beta}}-0) 
\xrightarrow{d}\gamma_{0}  \times d_{\alpha}\times Z_{\beta}.  
%\xrightarrow{d}\gamma_{0}  \times d_{\alpha}\times W_{\beta}^{\mytrans}B_{\beta}.  
\end{align*}

\noindent (ii) When $\alphaS=0$, and $\betaM\neq 0$,  we have $d_{\alpha,n}\to d_{\alpha}=0$, whereas $d_{\beta,n}\to d_{\beta}\neq 0$ by strict monotonicity of the function $g(\cdot)$. 
%no singularity issue. 
Moreover, $l'(P_s)$ and $l'(P_{s^*})\to  \gamma_{0}= l'\{d_{\betaM}g(\boldsymbol{x}^{\mytrans}\alphaX\big) + g(\boldsymbol{x}^{\mytrans} \boldsymbol{\beta}_{\X} + \tau_S s )\}\neq 0$, and 
%and $l'( \gamma_{\beta,s}) = 1/\{ \gamma_{\beta,s}(1- \gamma_{\beta,s})\}\neq 0$. 
\begin{align*}
\sqrt{n}\times\eqref{eq:loghatdiff}\to &~\gamma_{0}\times  \sqrt{n}(\hat{d}_{{\alpha}}-0)\times d_{\beta} 
\xrightarrow{d} \gamma_{0}   \times Z_{\alpha}\times d_{\beta}.
%\xrightarrow{d} \gamma_{0}   \times W_{\alpha}^{\mytrans}B_{\alpha}\times d_{\beta}.
\end{align*}

\noindent (iii) When $\alphaS=\betaM=0$, $l'(P_s)$ and $l'(P_{s^*})\to \gamma_0=l'\{g(\boldsymbol{x}^{\mytrans} \boldsymbol{\beta}_{\X} + \tau_S s )\}\neq 0$, and
\begin{align*}
\sqrt{n}d_{\alpha,n} = &\sqrt{n}\left\{ g(\alphaSn s + \boldsymbol{x}^{\mytrans}\alphaX)-g(\alphaSn s^* + \boldsymbol{x}^{\mytrans}\alphaX)\right\}\to g'(\boldsymbol{x}^{\mytrans}\alphaX)(s-s^*) \balpha=d_{\balpha}, \notag\\
\sqrt{n}d_{\beta,n} = &\sqrt{n}\left\{ g(\betaMn s + \boldsymbol{x}^{\mytrans}\betaX +\tau_S s)-g(\betaMn s^* + \boldsymbol{x}^{\mytrans}\betaX+\tau_S s)\right\} \notag\\
\to & g'(\boldsymbol{x}^{\mytrans}\betaX + \tau_S s)(s-s^*) \bbeta=d_{\bbeta}. \notag 	
\end{align*}
It follows that 
\begin{align*}
n\times \eqref{eq:loghatdiff}\to &~\gamma_{0} \Big\{d_{\balpha} \sqrt{n}(\hat{d}_{{\beta}}-0) + d_{\bbeta} \sqrt{n}(\hat{d}_{{\alpha}}-0) +  \sqrt{n}(\hat{d}_{{\alpha} }-0)\times\sqrt{n}(\hat{d}_{{\beta}}-0)\Big\}\notag\\
%\xrightarrow{d}	&~\gamma_{0} \Biggr( d_{\balpha} W_{\beta}^{\mytrans} Z_{\beta} + d_{\bbeta}  W_{\alpha}^{\mytrans}Z_{\alpha}+ \prod_{\pi\in \{\alpha, \beta\}}W_{\pi}^{\mytrans}Z_{\pi}\Biggr). \notag\\
\xrightarrow{d}	&~\gamma _{0} \left( d_{\balpha} Z_{\beta} + d_{\bbeta}  Z_{\alpha}+ Z_{\alpha}Z_{\beta}\right). 
\end{align*}

\vspace{0.5em} 

\begin{lemma}[Individual limits of $\hat{d}_{\alpha}$ and $\hat{d}_{\beta}$.]\label{lm:speratedalphabetalimit}
Under the condition of Theorem \ref{prop:asymbinaryoutcome}, 
\begin{align*}
\sqrt{n}(\hat{d}_{{\alpha}}-d_{\alpha,n})\xrightarrow{d} Z_{\alpha},\hspace{2em}\sqrt{n}(\hat{d}_{{\beta}}-d_{\beta,n})\xrightarrow{d} 	Z_{\beta}, 
\end{align*}	
where  $Z_{\alpha}=W_{\alpha}^{\mytrans}B_{\alpha}$, $	Z_{\beta}=W_{\beta}^{\mytrans}B_{\beta}$,  $B_{\alpha}$ and $B_{\beta}$ represent mean-zero multivariate normal distributions specified in \eqref{eq:normallimitbinary2} and \eqref{eq:normallimitbinary2beta}, respectively,
and  $W_{\alpha}$ and $W_{\beta}$ are vectors defined in \eqref{eq:walphadefbinary}  and \eqref{eq:wbetadefbinary}, respectively. 
%with mean zero and covariance specified after \eqref{eq:normallimitbinary2}, 
%$W_{\alpha}$ is a vector defined in \eqref{eq:walphadefbinary}, 
% $Z_{\beta}$ represents a mean-zero multivariate normal distribution specified in \eqref{eq:normallimitbinary2beta},
%% with mean zero and covariance specified in \eqref{eq:normallimitbinary2beta}, 
% and 
%$W_{\beta}$ is a vector defined in \eqref{eq:wbetadefbinary}. 
\end{lemma}

\noindent \textit{Proof of Lemma \ref{lm:speratedalphabetalimit}.} By Taylor's expansion and the property of logistic regression, we have
%under the assumption of consistency, we have
\begin{align}
\sqrt{n}\begin{pmatrix}
 \hat{\alpha}_S- \alphaSn \\ 
 \hat{\boldsymbol{\alpha}}_{\X}- \boldsymbol{\alpha}_{\X}
\end{pmatrix} 
=&\sqrt{n}\left\{\sum_{i=1}^n g_{\alpha,i} (1-g_{\alpha,i})D_{\alpha,i}D_{\alpha,i}^{\mytrans}\right\}^{-1} \left\{ \sum_{i=1}^n (M_i-g_{\alpha, i})D_{\alpha,i}\right\} \{1+o_p(1)\}\notag\\
\xrightarrow{d}& B_{\alpha}, \label{eq:normallimitbinary2}	
\end{align}
where in the first equation,  we define
%, for the simplicity of notation, 
$D_{\alpha,i}=(S_i, \X_i^{\mytrans})^{\mytrans}$ and  $g_{\alpha,i}=g(\S_i  \alphaS + \X_i^{\mytrans}\alphaX)$, and in \eqref{eq:normallimitbinary2}, 
 $B_{\alpha}$ represents a multivariate normal distribution with $\mathrm{E}(B_{\alpha})=\mathbf{0}$
 and $\mathrm{cov}(B_{\alpha})=\mathrm{cov}\{V_{\alpha}^{-1}(M-g_{\alpha})D_{\alpha}\}$, 
% mean zero and covariance same as that of $V_{\alpha}^{-1}(M-g_{\alpha})D_{\alpha}$ 
 where 
we define 
$D_{\alpha}=(S, \X^{\mytrans})^{\mytrans}$, 
 $g_{\alpha}=g(\S \alphaS + \X^{\mytrans}\alphaX)$,   and $ 
V_{\alpha}	= \mathrm{E}\big\{g_{\alpha}(1-g_{\alpha})D_{\alpha}D_{\alpha}^{\mytrans} \big\}  
$. 
By  Taylor's expansion, 
\begin{align*}
%\sqrt{n}(\hat{d}_{{\alpha}_S}-d_{\alphaS})
%=&~\sqrt{n}\big\{g( \hat{\alpha}_S s+\boldsymbol{x}^{\mytrans}\hat{\boldsymbol{\alpha}}_{\X})-g( \hat{\alpha}_S s^*+\boldsymbol{x}^{\mytrans}\hat{\boldsymbol{\alpha}}_{\X}) \notag\\
%&~\quad \quad - g(\alphaS s + \boldsymbol{x}^{\mytrans}\alphaX)+g(\alphaS s^* + \boldsymbol{x}^{\mytrans}\alphaX)	\big\} \notag \\
\sqrt{n}(\hat{d}_{{\alpha}}-d_{\alpha})=&~ {W}_{\alpha}^{\mytrans}\times \sqrt{n}\begin{pmatrix}
 \hat{\alpha}_S- \alphaSn \\ 
 \hat{\boldsymbol{\alpha}}_{\X}-  {\boldsymbol{\alpha}}_{\X,n}
\end{pmatrix}\{1+o_p(1)\}, 
\end{align*}
where we define $\mu_{s,\alpha}=\alphaS s + \boldsymbol{x}^{\mytrans}\alphaX$, $\mu_{s^*,\alpha}=\alphaS s^* + \boldsymbol{x}^{\mytrans}\alphaX$, and
\begin{align} \label{eq:walphadefbinary} 
{W}_{\alpha}^{\mytrans} = 	g'(\mu_{s,\alpha}) \times \big(
	s, \, 
	\boldsymbol{x}^{\mytrans}\big)
 - g'(\mu_{s^*,\alpha})\times \big(
	s^*, \, 
	\boldsymbol{x}^{\mytrans}\big). 
\end{align}
%\begin{align} \label{eq:walphadefbinary}
%{W}_{\alpha} = 	 g'(\mu_{s,\alpha}) \begin{pmatrix}
%	s\\
%	\boldsymbol{x}
%\end{pmatrix}
% - g'(\mu_{s^*,\alpha}) \begin{pmatrix}
%	s^*\\
%	\boldsymbol{x}
%\end{pmatrix}. 
%\end{align}
Therefore, by \eqref{eq:normallimitbinary2}, 
$\sqrt{n}(\hat{d}_{{\alpha}}-d_{\alpha})\xrightarrow{d} W_{\alpha}^{\mytrans}B_{\alpha}=Z_{\alpha}$. 

Similarly, by Taylor's expansion, 
\begin{align}
\sqrt{n}\begin{pmatrix}
 \hat{\beta}_M- \beta_{\M,n} \\ 
 \hat{\boldsymbol{\beta}}_{\X}-  {\boldsymbol{\beta}}_{\X}\\
  \hat{\tau}_S- \tau_S
\end{pmatrix} 
=&\left\{\sum_{i=1}^n g_{\beta,i} (1-g_{\beta,i})D_{\beta,i}D_{\beta,i}^{\mytrans}\right\}^{-1} \left\{ \sum_{i=1}^n (Y_i-g_{\beta, i})D_{\beta,i}\right\} \{1+o_p(1)\}\notag\\
\xrightarrow{d}& B_{\beta}, \label{eq:normallimitbinary2beta}	
\end{align} 
where in the first equation,  we define
%, for the simplicity of notation,
$D_{\beta,i}=(M_i, S_i, \X_i^{\mytrans})^{\mytrans}$, $g_{\beta,i}=g(M_i \betaM +\S_i  \tau_S + \X_i^{\mytrans}\betaX)$, and in \eqref{eq:normallimitbinary2beta},  
 $B_{\beta}$ represents a normal distribution with $\mathrm{E}(B_{\beta})=\mathbf{0}$
 and $\mathrm{cov}(B_{\beta})=\mathrm{cov}\{V_{\beta}^{-1}(Y-g_{\beta})D_{\beta}\}$, 
% mean zero and covariance same as that of $V_{\beta}^{-1}(Y-g_{\beta})D_{\beta}$ 
 where 
we define 
 $g_{\beta}=g(M\betaM +\S \tau_S + \X^{\mytrans}\betaX)$,  $D_{\beta}=(M, S, \X^{\mytrans})^{\mytrans}$, and $ 
V_{\beta}	= \mathrm{E}\big\{g_{\beta}(1-g_{\beta})D_{\beta}D_{\beta}^{\mytrans} \big\}  
$. 
Moreover, by Taylor's expansion,
\begin{align*}
	\sqrt{n}(\hat{d}_{\beta}-d_{\beta,n}) = W_{\beta}^{\mytrans}\times \sqrt{n}\begin{pmatrix}
 \hat{\beta}_M- \beta_{\M,n} \\ 
 \hat{\boldsymbol{\beta}}_{\X}- {\boldsymbol{\beta}}_{\X}\\
 \hat{\tau}_S- \tau_S
\end{pmatrix}\times \{1+o_p(1)\}, 
\end{align*}
where we define $\mu_{0,\beta}=\boldsymbol{x}^{\mytrans}\alphaX + s\tau_S$, $\mu_{1,\beta}=\betaM + \mu_{0,\beta}$, 
%$\mu_{s^*,\alpha}=\alphaS s^* + \boldsymbol{x}^{\mytrans}\alphaX$, 
and
\begin{align} \label{eq:wbetadefbinary}
{W}_{\beta}^{\mytrans} =  g'(\mu_{1,\beta})\times \big(1,\,  \boldsymbol{x}^{\mytrans},\, s\big) - g'(\mu_{0,\beta})\times \big(0, \, \boldsymbol{x}^{\mytrans}, \, s\big).
\end{align}
Therefore, by \eqref{eq:normallimitbinary2}, 
$\sqrt{n}(\hat{d}_{{\beta}}-d_{\beta,n})\xrightarrow{d} W_{\beta}^{\mytrans}B_{\beta}=Z_{\beta}$. 

\subsection{Proof of Theorem  \ref{thm:combinebootbinaryout}} \label{sec:pfbootbinaryboth}

\noindent \textit{Notation.}
%In Theorem \ref{prop:bootconsisbinaryoutcome}, 
Define $(\mathbb{Z}_{\alpha}^*, \mathbb{Z}_{\beta}^*)$  as the bootstrap counterparts of $(Z_{\alpha},Z_{\beta})$. 
In particular, 
by the definitions of $Z_{\alpha}$ and $Z_{\beta}$ in Lemma \ref{lm:speratedalphabetalimit}, 
%and following the notation in the main text, 
 we let
%Let $\mathcal{P}^*$ be a bootstrap index set of $\{1,\ldots, n\}$.  
$
	\mathbb{Z}_{\alpha}^*={W}_{\alpha}^{*\top}\left[\mathbb{P}_n^*\{ g_{\alpha}(1-g_{\alpha})D_{\alpha}D_{\alpha}^{\top}\}\right]^{-1}\times\mathbb{G}_n^*\{(M-g_{\alpha})D_{\alpha}\}
$ and 
$
	\mathbb{Z}_{\beta}^*={W}_{\beta}^{*\top}[\mathbb{P}_n^*\{ g_{\beta}(1-g_{\beta})D_{\beta}D_{\beta}^{\top}]^{-1}\times\mathbb{G}_n^*\{(M-g_{\beta})D_{\beta}\}, 
$
where $D_{\alpha}$, $D_{\beta}$, $g_{\alpha}$, and $g_{\beta}$ are defined in Condition \ref{cond:designmatrixlogistic},
and ${W}_{\alpha}^{*\mytrans}$ and ${W}_{\beta}^{*\mytrans}$ represent bootstrap estimators of 
${W}_{\alpha}^{\mytrans}$ in \eqref{eq:walphadefbinary} and ${W}_{\beta}^{\mytrans}$ in \eqref{eq:wbetadefbinary}, respectively.  
Specifically, ${W}_{\alpha}^{*\mytrans} = 	g'(\hat{\mu}_{s,\alpha}^*) \times (
	s, \, 
	\boldsymbol{x}^{\mytrans})
 - g'(\hat{\mu}_{s^*,\alpha}^*)\times (
	s^*, \, 
	\boldsymbol{x}^{\mytrans})$
and ${W}_{\beta}^{*\mytrans} =  g'(\hat{\mu}_{1,\beta}^*)\times (1,\,  \boldsymbol{x}^{\mytrans},\, s) - g'(\hat{\mu}_{0,\beta}^*)\times (0, \, \boldsymbol{x}^{\mytrans}, \, s)$,
where $(\hat{\mu}_{s^*,\alpha}^*,\hat{\mu}_{s^*,\alpha}^*,\hat{\mu}_{1,\beta}^*, \hat{\mu}_{0,\beta}^*)$ represent bootstrap estimators of  $({\mu}_{s^*,\alpha},{\mu}_{s^*,\alpha},{\mu}_{1,\beta}, {\mu}_{0,\beta})$. 
The definitions are similar to $ \mathbb{Z}_{\S,n}^*$ and $ \mathbb{Z}_{\M,n}^*$ in Section \ref{sec:abt}. 
Moreover, $\hat{\gamma}_{0}^*= \{\hat{P}_*^*(1-\hat{P}_*^*)\}^{-1}$,
where $\hat{P}_*^*=g(\boldsymbol{x}^{\mytrans} \hat{\boldsymbol{\beta}}_{\X}^* + \hat{\tau}_S^* s )$,
and $\hat{\boldsymbol{\beta}}_{\X}^*$ and $\hat{\tau}_S^*$ denote  non-parametric bootstrap estimators of $\boldsymbol{\beta}_{\X}$ and ${\tau}_S$, respectively. 

\medskip
\noindent\textit{Proof.} The proof is very similar to that of Theorem \ref{thm:bootstrapprodcomb}. 
%s   \ref{thm:bootstrapprod}  and \ref{thm:bootstrapprodcomb}.
We describe the key steps, and the details follow similarly to that in Section
%s \ref{sec:pfbootstrapprod} and 
\ref{sec:pfbootstrapprodcomb}. 
%\textit{(1) For Theorem  \ref{prop:bootconsisbinaryoutcome}.} 
When $(\alphaS,\betaM)\neq (0,0)$, 
the bootstrap estimator $\widehat{\text{NIE}}^*$ is consistent by the asymptotic expansion in Section \ref{sec:proofasymbinaryoutcome} and asymptotic normality. 
When $(\alphaS,\betaM)=(0,0)$, the bootstrap estimator $
( d_{\balpha}  \mathbb{Z}_{\beta}^* +   d_{\bbeta} \mathbb{Z}_{\alpha}^*+ \mathbb{Z}_{\alpha}^*\mathbb{Z}_{\beta}^*)\hat{\gamma}_{0}^*	
$ is consistent by the asymptotic expansion in \ref{sec:proofasymbinaryoutcome} and its limit form. 
% in Theorem \ref{prop:bootconsisbinaryoutcome}.  
%\textit{(2) For Theorem  \ref{thm:combinebootbinaryout}.} 
To prove  Theorem  \ref{thm:combinebootbinaryout}, 
results similar to \eqref{eq:indicatorconvg} can be established as under the logistic models by the asymptotic normality in  \eqref{eq:normallimitbinary2} and \eqref{eq:normallimitbinary2beta}. Then the proof follows by the arguments in Section \ref{sec:pfbootstrapprodcomb}.

%Theory of Binary Mediator, Linear Outcome
\subsection{ Proof of Theorem  \ref{prop:asymbinarymediator} }
%We next derive asymptotic expansions under Scenario 2.
Note that 
$
\widehat{\mathrm{NIE}}-\mathrm{NIE}= \hat{\beta}_M \hat{d}_{\alphaS} - \betaM d_{\alphaS,n} 
$  
where $\hat{d}_{\alphaS}$ and $d_{\alphaS,n}$ are defined in Section \ref{sec:proofasymbinaryoutcome}. 
By the proof of Theorem \ref{thm:prodlimit2} in the main text, 
we have 
$ 
	\sqrt{n}(\hat{\beta}_M - \beta_M)\xrightarrow{d} Z_{\beta} 
$ where $Z_{\beta}$ denotes a mean-zero multivariate normal distribution with a covariance same as the random vector $V_M^{-1}\epsilon_YM_{\perp'}$ defined in Theorem \ref{thm:prodlimit2}. 
Moreover, by Lemma \ref{lm:speratedalphabetalimit}, 
%the proof in Section \ref{sec:proofasymbinaryoutcome}, 
$\sqrt{n}(\hat{d}_{\alphaS} - d_{\alphaS,n}) \xrightarrow{d} Z_{\alpha}$. 

%\begin{align*}
%	Z_{\beta}= 
%\end{align*}
%by the analysis in the main text. 
%We next examine  $\hat{\alpha}_S- \alphaS$.

\medskip
\noindent (i) When $\alphaS\neq 0$ and $\betaM=0$, we have $\betaMn \to \betaM = 0$ and $d_{\alphaS,n}\to d_{\alphaS} \neq 0$. 
Therefore, as $n\to \infty$,
$
\sqrt{n} (\widehat{\mathrm{NIE}}-\mathrm{NIE})\to d_{\alphaS} \times \sqrt{n}(\hat{\beta}_{M}- \betaM)\xrightarrow{d} d_{\alphaS} \times Z_{\beta}. 
$
\medskip

\noindent (ii) When $\alphaS = 0$ and $\betaM\neq 0$, 
we have $d_{\alphaS,n} \to d_{\alphaS} =0$ and $\hat{\beta}_M\to \betaM \neq 0$. 
Therefore, as $n\to \infty$,
$
\sqrt{n} (\widehat{\mathrm{NIE}}-\mathrm{NIE})\to 	\betaM \times \sqrt{n}(\hat{d}_{\alphaS} - d_{\alphaS,n} ) \xrightarrow{d} \betaM \times Z_{\alpha}. 
$
\medskip

\noindent (iii) When $\alphaS = \betaM = 0$, 
\begin{align*}
&~\sqrt{n} (\widehat{\mathrm{NIE}}-\mathrm{NIE}) \notag\\
= &~\sqrt{n}(\hat{\beta}_M -\betaMn )\sqrt{n}d_{\alphaS,n} + (\hat{d}_{\alphaS}-d_{\alphaS,n}) \sqrt{n}\betaMn +n(\hat{d}_{\alphaS}-d_{\alphaS,n})(\hat{\beta}_M -\betaMn ) \notag\\
\xrightarrow{d}&~ Z_{\beta}d_{\balpha} + Z_{\alpha} \bbeta + Z_{\alpha}Z_{\beta}. 
%\xrightarrow{d} Z_{\beta}d_{\balpha} + W_{\alpha}^{\top}Z_{\alpha} \bbeta +   W_{\alpha}^{\top}Z_{\alpha}Z_{\beta}. 
\end{align*}

\subsection{Proofs of Theorem \ref{thm:combinebootbinarymediator}}\label{sec:pfthemscenario2twoboot}
 
%In Theorem \ref{prop:bootconsisbinarymediator}, 
We let  $\mathbb{Z}_{\alpha}^*$  and  $\mathbb{Z}_{\beta}^*$  denote bootstrap counterparts of $Z_{\alpha}$ and $Z_{\beta}$, respectively. 
Similarly to Section \ref{sec:asymbinarymedout}, by the definition of $Z_{\alpha}$ in Lemma \ref{lm:speratedalphabetalimit}, 
%and following the notation in the main text,  
we have
%In particular, 
$
	\mathbb{Z}_{\alpha}^*=W_{\alpha}^{*\top}[\mathbb{P}_n^*\{ g_{\alpha}(1-g_{\alpha})D_{\alpha}D_{\alpha}^{\top}\}]^{-1}\mathbb{G}_n^*\{(M-g_{\alpha})D_{\alpha}\}.  
$
%, which  is defined same as in Section \ref{sec:asymbinarymedout},
Moreover, we redefine $\mathbb{Z}_{\beta}^* = (\mathbb{V}_{M,n}^*)^{-1} \times\mathbb{G}^*(\hatenY M^*_{\perp'})$, which is same as $\mathbb{Z}_{M,n}^*$ 
%in Theorem \ref{thm:bootstrapprod} 
in the main text. 
Theorem  \ref{thm:combinebootbinarymediator}  can be similarly obtained following the arguments in Sections \ref{sec:pfbootstrapprodcomb} and  \ref{sec:pfbootbinaryboth}. 
We therefore skip the details here.

\newpage

% \addtocontents{toc}{\setcounter{tocdepth}{2}}

\section{Implementation Details}\label{sec:implement}

\subsection{Double Bootstrap for Choosing the Tuning Parameter}\label{sec:tuningpar}

\textit{Overview.} Double bootstrap (DB) has two layers of bootstrap. The first layer applies ordinary bootstrap to a given data $\mathcal{D}$, and the second layer applies AB to the bootstrapped data from the first layer and returns  an estimated $p$-value. 
Repeating the procedure multiple times yields a sample of estimated $p$-values,
which, intuitively, can approximate the  distribution of $p$-values given by directly applying AB to $\mathcal{D}$. 
Therefore, the $p$-values estimated by double bootstrap can guide the choice of tuning parameters. 

\bigskip
\noindent \textit{Implementation Details.} 
Our goal is to choose $\lambda$ value such that the AB test would return uniformly distributed $p$-values under $H_0: \alphaS\betaM=0$. 
Given observed data $\mathcal{D}_{obs}$, we mimic $H_0$ by processing the observed data $\mathcal{D}_{obs}$ so that the sample estimate of  mediation effect based on the processed data would be $0$.    
To achieve that, we specify two methods of  data  processing 
after which sample estimates of $\alphaS$ and $\betaM$ become zero, respectively. 
% $\hatalphaS =0$ and $\hatbetaM = 0$, respectively. 
Technically,  
% to specify the data processing, 
we define a projection mapping 
$\mathcal{P}_{S}^{\perp}(M)= \{\mathrm{I}- S(S^{\top}S)^{-1}S^{\top}\}M$, 
which denotes the projection of observations $M=(M_1,\ldots, M_n)^{\top}$ onto the space orthogonal to observations $S=(S_1,\ldots, S_n)^{\top}$. 
Two data processing methods are specified as follows. 
\begin{enumerate}\setlength{\itemsep}{-4pt}
    \item[(i)]  In the mediator-exposure model $M\sim S + \boldsymbol{X}$, replace $(M,\boldsymbol{X})$ by the projected data $(\mathcal{P}_{S}^{\perp}(M), \mathcal{P}_{S}^{\perp}(\boldsymbol{X}))$, and then the sample coefficient of $S$ is 0 by Section \ref{sec:preliminary}. %The processed data is denoted as $\mathcal{D}_{\alpha}$. 
    \item[(ii)]  In the outcome-mediator model $Y\sim M+S+\boldsymbol{X}$, replace $(Y, S, \boldsymbol{X})$ by the projected data $(\mathcal{P}_{M}^{\perp}(Y), \mathcal{P}_{M}^{\perp}(S), \mathcal{P}_{M}^{\perp}(\boldsymbol{X}))$. Then the sample coefficient of $M$ is 0 by Section \ref{sec:preliminary}. %The processed data is denoted as $\mathcal{D}_{\beta}$. 
%   \item[(c)] $\hatalphaS = \hatbetaM = 0$. \ Apply methods (a) and (b) simultaneously, i.e., in the mediator model, replace $(M,\boldsymbol{X})$ by the projected data $(\mathcal{P}_{S}^{\perp}(M), \mathcal{P}_{S}^{\perp}(\boldsymbol{X}))$, and in the outcome model, replace $(Y, S, \boldsymbol{X})$ by the projected data $(\mathcal{P}_{M}^{\perp}(Y), \mathcal{P}_{M}^{\perp}(S), \mathcal{P}_{M}^{\perp}(\boldsymbol{X}))$. 
\end{enumerate}
% Given original data $\mathcal{D}$, 
%  let $\mathcal{D}_{P_{\alpha}}$, $\mathcal{D}_{P_{\beta}}$, and $\mathcal{D}_{P_{\alpha,\beta}}$ denote the data after processing $\mathcal{D}$ by the above three methods, respectively.
% Sample estimates of the mediation effect  based on the processed data $\mathcal{D}_{P_{\alpha}}$, $\mathcal{D}_{P_{\beta}}$, and $\mathcal{D}_{P_{\alpha,\beta}}$ would be always 0. 
% Given original data $\mathcal{D}$, 
%  let $\mathcal{D}_{P_{\alpha}}$ and $\mathcal{D}_{P_{\beta}}$ denote the data after processing $\mathcal{D}$ by the above two procedures, respectively.
% Sample estimates of the mediation effect  based on processed data $\mathcal{D}_{P_{\alpha}}$ and $\mathcal{D}_{P_{\beta}}$ would be always 0. 
 % Given original data $\mathcal{D}$, 
%  let $\mathcal{D}_{P_{\alpha}}$ and $\mathcal{D}_{P_{\beta}}$ denote the data after processing $\mathcal{D}$ by the above two procedures, respectively.
% Sample estimates of the mediation effect  based on processed data $\mathcal{D}_{P_{\alpha}}$ and $\mathcal{D}_{P_{\beta}}$ would be always 0. 

%  \bigskip
% \noindent \textit{Summary of the procedure.}  \\
% \noindent \textit{Double bootstrap procedure.} 
% We first conduct double bootstrap with $\lambda=0$ as follows.\\ 

 We let $\mathcal{D}_{{\alpha}}$ and $\mathcal{D}_{{\beta}}$ denote the processed data using the aforementioned methods (i) and (ii) only, respectively. 
 Moreover, we let $\mathcal{D}_{{\alpha,\beta}}$  denote the processed data using (i) and (ii) simultaneously. 
  A detailed double bootstrap procedure is specified as follows. 

\medskip
\noindent \textit{Step 1.} 
Given original data $\mathcal{D}_{obs}$, 
apply processing methods (i) and (ii) to obtain two processed data $\mathcal{D}_{{\alpha}}$ and $\mathcal{D}_{{\beta}}$, respectively. \\[4pt]
\textit{Step 2.} Apply DB to  $\mathcal{D}_{{\alpha}}$ and $\mathcal{D}_{{\beta}}$. For $b=1,\ldots, B$, 
\begin{itemize} \setlength{\itemsep}{-4pt}
    \item[--] apply ordinary bootstrap to  $\mathcal{D}_{{\alpha}}$ and $\mathcal{D}_{{\beta}}$ and obtain bootstrapped data $\mathcal{D}_{{\alpha},b}^*$ and $\mathcal{D}_{{\beta},b}^*$, respectively; 
    \item[--] apply adaptive bootstrap with fixed $\lambda = 0$ to $\mathcal{D}_{{\alpha},b}^*$ and $\mathcal{D}_{{\beta},b}^*$ to obtain estimated $p$-values $p_{\alpha,b}^*(0)$ and  $p_{\beta,b}^*(0)$, respectively.  
\end{itemize} 
Let $\mathcal{P}^*_{\alpha}(0)=\{p_{\alpha,b}^*(0): b=1,\ldots, B \}$ and
 $\mathcal{P}^*_{\beta}(0)=\{p_{\beta,b}^*(0): b=1,\ldots, B \}$  denote two sets of estimated $p$-values. 
We would observe different patterns of  $\mathcal{P}^*_{\alpha}(0)$ and $\mathcal{P}^*_{\beta}(0)$ under different scenarios of the  true parameters.   
Specifically, if $\alphaS=\betaM=0$, both $\mathcal{P}^*_{\alpha}(0)$ and $\mathcal{P}^*_{\beta}(0)$ are conservative; if one of $\alphaS$ and $\betaM$ is non-zero, at least one of $\mathcal{P}^*_{\alpha}(0)$ and $\mathcal{P}^*_{\beta}(0)$ is non-conservative.
Step 3 take different strategies of parameter choice based on observed patterns of  $\mathcal{P}^*_{\alpha}(0)$ and $\mathcal{P}^*_{\beta}(0)$ in Step 2. 

\smallskip
\noindent \textit{Step 3.} 
\begin{itemize}
 \item If $\mathcal{P}^*_{\alpha}(0)$ and $\mathcal{P}^*_{\beta}(0)$ are conservative, both $\alphaS$ and $\betaM$ are likely to be 0. 
 To find a good tuning parameter when both parameters are $0$, we apply processing methods (i) and (ii) to $\mathcal{D}_{obs}$ simultaneously and obtain a processed data $\mathcal{D}_{{\alpha,\beta}}$, which satisfies $\hat{\alpha}_S(\mathcal{D}_{{\alpha,\beta}}) = \hat{\beta}_M(\mathcal{D}_{{\alpha,\beta}})=0$ and mimics the scenario $\alphaS=\betaM=0$.  
 Let $\mathcal{P}^*_{\alpha,\beta}(\lambda)$ denote the set of estimated $p$-values when applying the double bootstrap to  $\mathcal{D}_{{\alpha,\beta}}$ with a fixed $\lambda$. 
 We increase $\lambda$ until  $\mathcal{P}^*_{\alpha,\beta}(\lambda)$ is close to $U[0,1]$. 
 \item If at least one of $\mathcal{P}^*_{\alpha}(0)$ and $\mathcal{P}^*_{\beta}(0)$ is non-conservative, we can choose any $\lambda$ such that $\mathcal{P}^*_{\alpha}(\lambda)$ and $\mathcal{P}^*_{\beta}(\lambda)$ are similar to $\mathcal{P}^*_{\alpha}(0)$ and $\mathcal{P}^*_{\beta}(0)$, respectively. 
\end{itemize}
We point out that multiple $\lambda$ values may yield similar properties and are all acceptable.  
We next provide a simple numerical illustration on the use of the double bootstrap.

\subsubsection{Numerical Example of the Double Bootstrap}
We present a numerical illustration  
under the model in Section \ref{sec:sim} with $n=200$
and three scenarios including: 
(1) $\alphaS=\betaM=0$;
(2) $\alphaS=0.5$, $\betaM=0$;
(3) $\alphaS= 0 $, $\betaM=0.5$. 
In the double bootstrap, the number of resampling of the two layers of bootstrap are both 500.
% in the ordinary bootstrap and the 
% $B = 500$, and the number of adaptive bootstrap is $500$.   
% Moreover, we consider three scenarios including: 
% (1) $\alphaS=\betaM=0$;
% (2) $\alphaS=0.5$, $\betaM=0$;
% (3) $\alphaS= 0 $, $\betaM=0.5$. 
Under each scenario, 
we present Q-Q plots of estimated $p$-values following the procedure above. 
% , which are obtained by the double bootstrap described above. 
As a comparison, we also simulate data from the underlying true model $M=500$ times, and apply the AB with the same fixed $\lambda$ to each simulated data to obtain estimated $p$-values $\mathcal{P}^*(\lambda)=\{p^*_m(\lambda): m=1,\ldots, M\}$.   
Confidence intervals presented in the figures are obtained through the  Kolmogorov–Smirnov test at the significance level $0.01$.  

% apply the AB to the original data $\mathcal{D}$ and obtain an estimated $p$-value $p^*(\lambda)$. 
% By generating the data repeatedly, 
% $\mathcal{P}^*(\lambda)$ when applying the AB with $\lambda$ 
% $B=400$, and the number of adaptive bootstrap is 200.  
% \begin{itemize}
%     \item Scenario 1: $\alphaS=\betaM=0$. 
%     \item Scenario 2: $\alphaS=0.5$, $\betaM=0$. 
%     \item Scenario 3: $\alphaS= 0 $, $\betaM=0.5$. 
% \end{itemize} 

\paragraph{Scenario 1: $\alphaS=\betaM=0$.} In (a) and (b) of Figure \ref{fig:tuning00_0},  both sets of $p$-values $\mathcal{P}_{\alpha}^*(0)$ and $\mathcal{P}_{\beta}^*(0)$ in Step 2 are  conservative.   
As a comparison, (c) presents the estimated distribution of $p$-values obtained from AB with $\lambda=0$, which can be viewed as the ground truth.   
We can see that (a) and (b) indeed captures the over-conservativness in (c).  
% Therefore, the results suggest that . 
% Therefore,
% Results in (a) and (b) suggest that the underlying true parameters are likely to be both $0$. 
% Therefore,
To find a good tuning parameter when $\alphaS=\betaM=0$, 
in Step 3, we process the data  to get $\mathcal{D}_{{\alpha, \beta}}$ and find that increasing $\lambda$ to $4$ in the double bootstrap can return uniformly distributed $p$-values.  
As a validation, (e) presents the estimated distribution of $p$-values obtained from AB with $\lambda=4$.  
We can see that  (d) indeed captures the uniform distribution (e), suggesting $\lambda=4$ is a good tuning parameter in this case. 
 \begin{figure}[!htbp]
 \graphicspath{{figures/for_comments/R2_C3_tuning/}} 
%  \captionsetup[subfigure]{labelformat=empty} 
     \centering
       \caption{Scenario 1: $\alphaS=\betaM=0$. Q-Q plots of  sampled $p$-values.} \label{fig:tuning00_0}
  \begin{subfigure}[b]{0.19\textwidth}
  \caption{$\mathcal{P}_{\alpha}^*(0)$}
   \includegraphics[width=\textwidth]{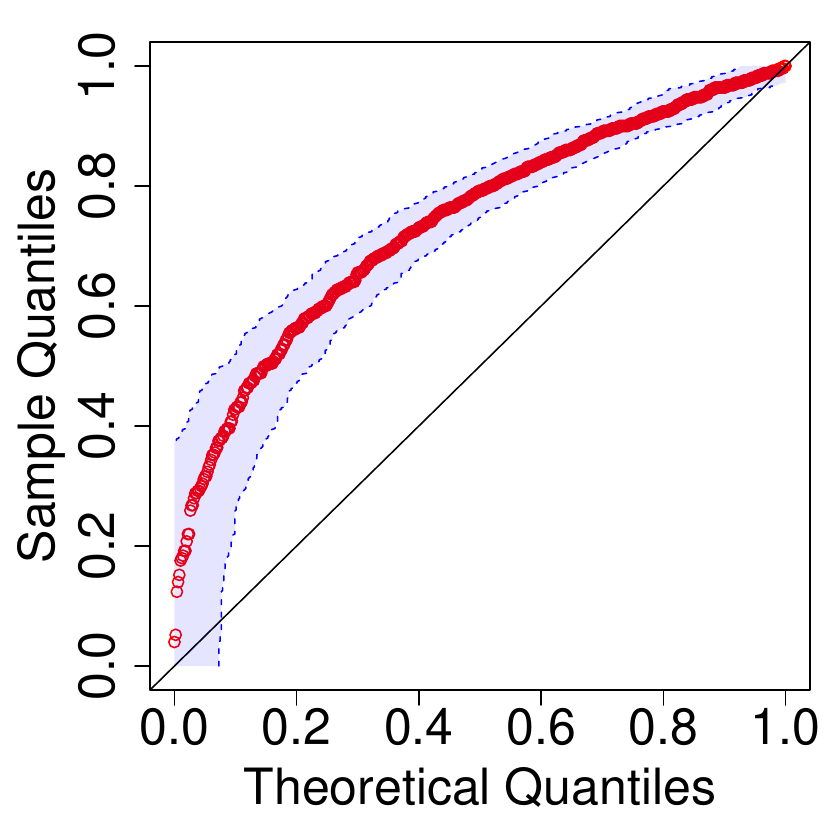}
  \end{subfigure}\ 
   \begin{subfigure}[b]{0.19\textwidth}
    \caption{$\mathcal{P}_{\beta}^*(0)$}
   \includegraphics[width=\textwidth]{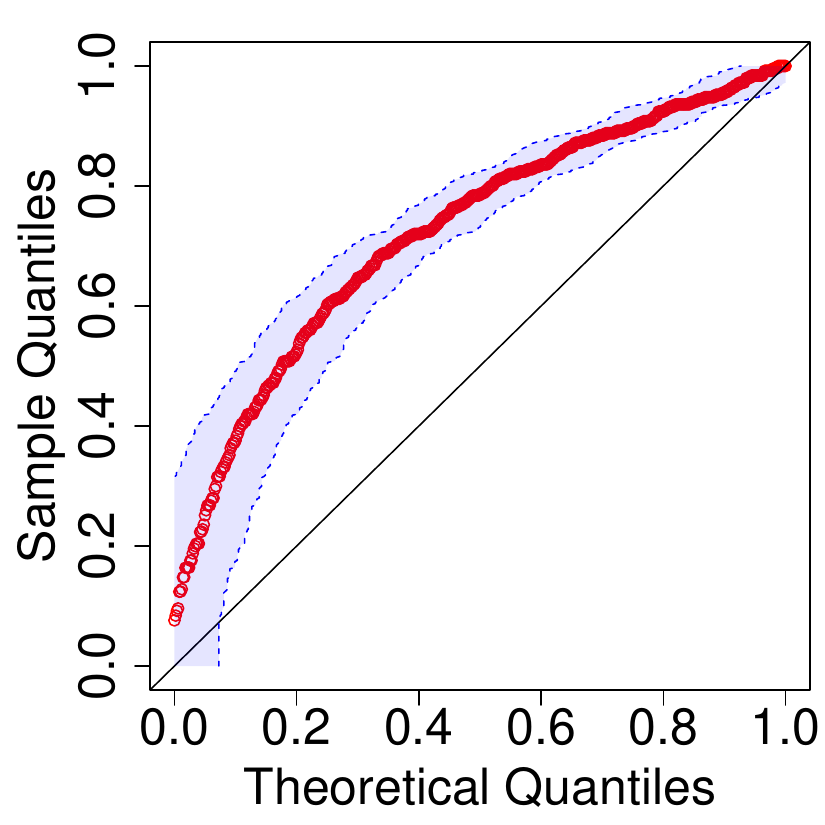}
  \end{subfigure}\ 
    \begin{subfigure}[b]{0.19\textwidth}
        \caption{ $\mathcal{P}^*(0)$} 
   \includegraphics[width=\textwidth]{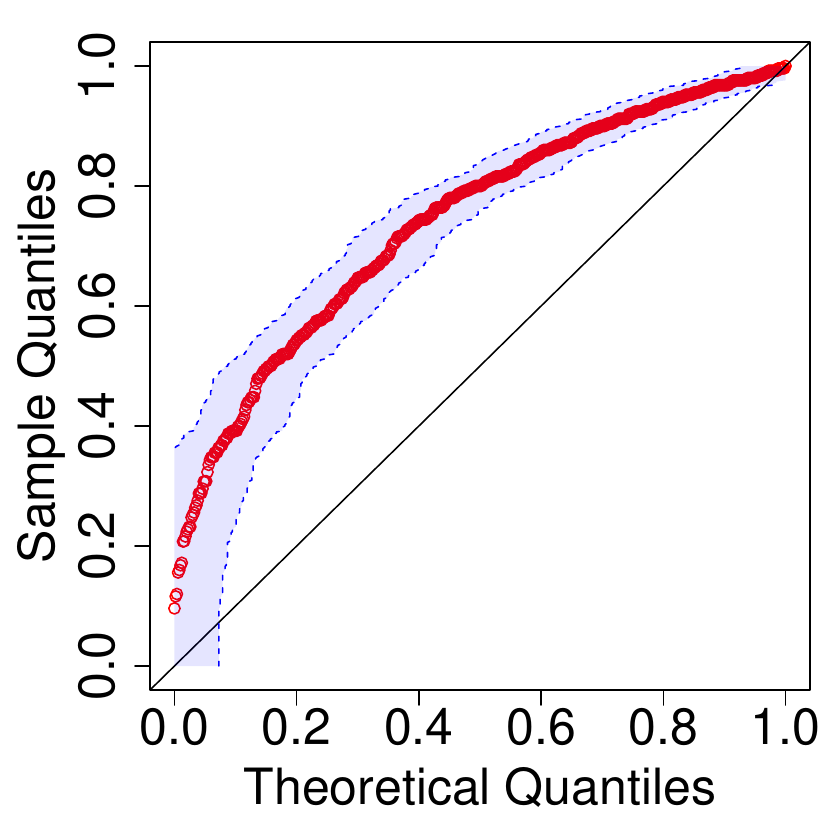}
  \end{subfigure}
   \begin{subfigure}[b]{0.19\textwidth}
    \caption{$\mathcal{P}_{\alpha,\beta}^*(4)$}
   \includegraphics[width=\textwidth]{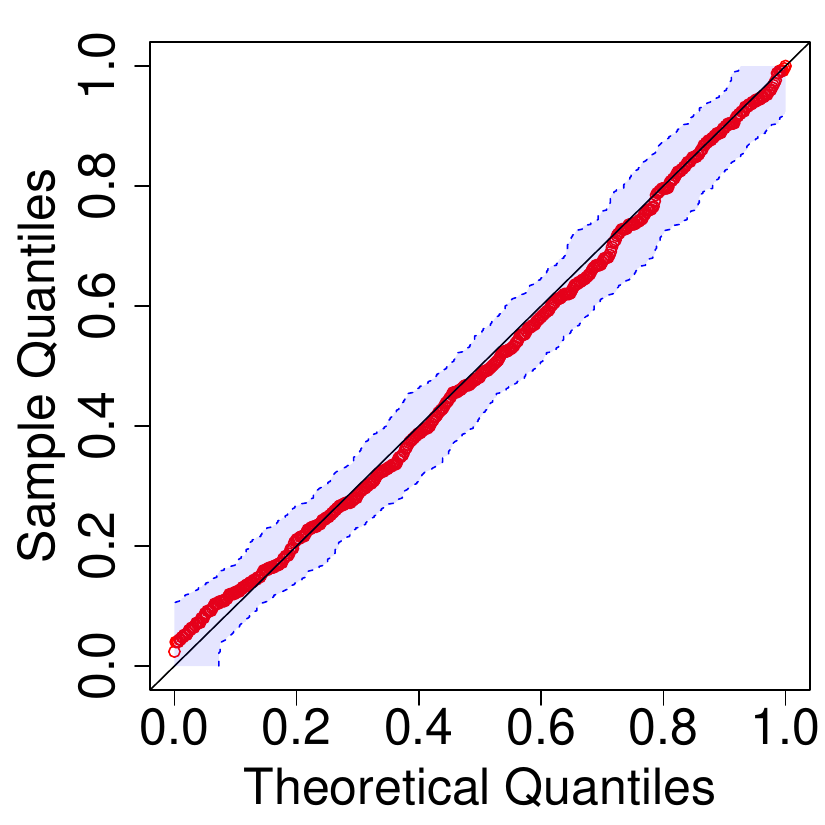}
  \end{subfigure}\ 
    \begin{subfigure}[b]{0.19\textwidth}
        \caption{ $\mathcal{P}^*(4)$} 
   \includegraphics[width=\textwidth]{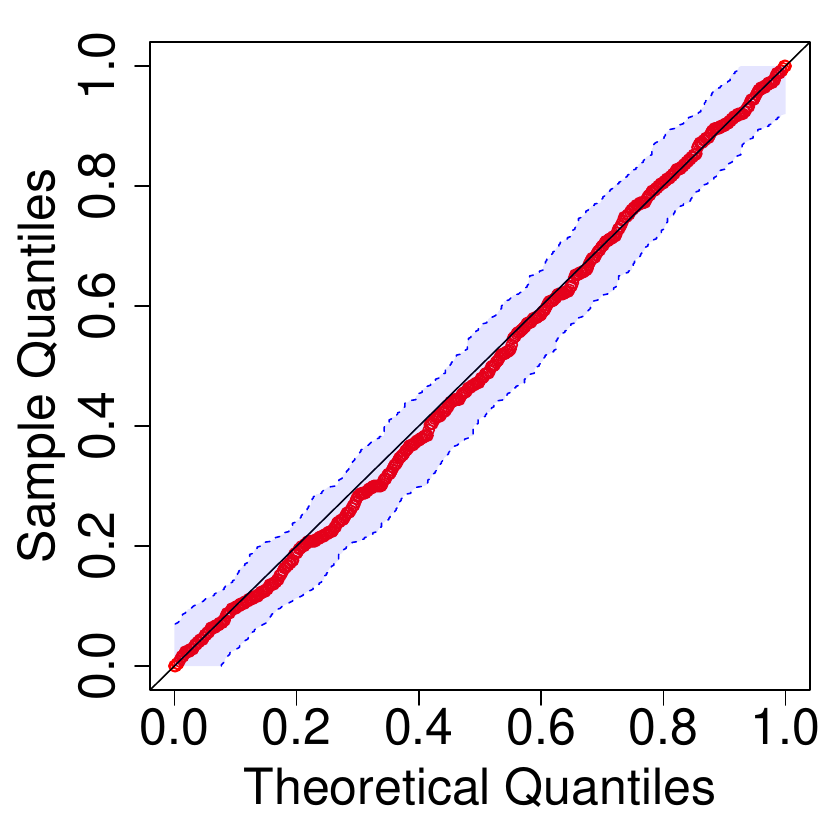}
  \end{subfigure} 
\end{figure}

\paragraph{Scenario 2: $\alphaS=0.5$ and $\betaM=0$.}  
By (a) and (b)  of Figure \ref{fig:tuning005_0},  $\mathcal{P}_{\alpha}^*(0)$ and $\mathcal{P}_{\beta}^*(0)$ in Step 2 are conservative and non-conservative, respectively. 
This suggests that at least one of the true coefficients is non-zero. 
Thus, we can choose any $\lambda$ such that the double bootstrap yields  estimated $p$-values similar to those in Step 2. 
% , which is consistent with the observation that  $p$-values 
The similarity between (d)--(e) and  (a)--(b) indicates that $p$-values estimated by AB with  $\lambda=2$ are similar to those in Step 2.  
This can be validated by comparing (c) and (f)   of Figure \ref{fig:tuning005_0}, which  show that  increasing $\lambda$ to $2$ in AB still yields uniformly distributed $p$-values similar to those with $\lambda=0$.

\begin{figure}[!htbp]
% \graphicspath{{figures/for_comments/R2_C3_tuning/}} 
     \centering
       \caption{Scenario 2: $\alphaS=0.5$ and $\betaM=0$.} \label{fig:tuning005_0} 
  \begin{subfigure}[b]{0.16\textwidth}
  \caption{$\mathcal{P}_{\alpha}^*(0)$}
   \includegraphics[width=\textwidth]{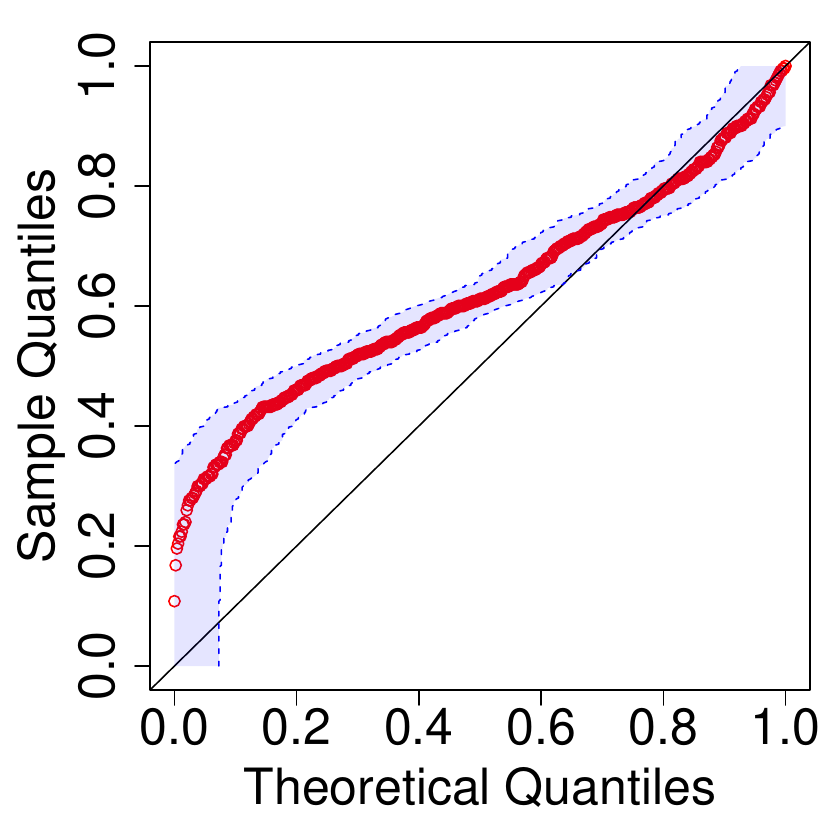}
  \end{subfigure}\   
    \begin{subfigure}[b]{0.16\textwidth}
  \caption{$\mathcal{P}_{\beta}^*(0)$}
   \includegraphics[width=\textwidth]{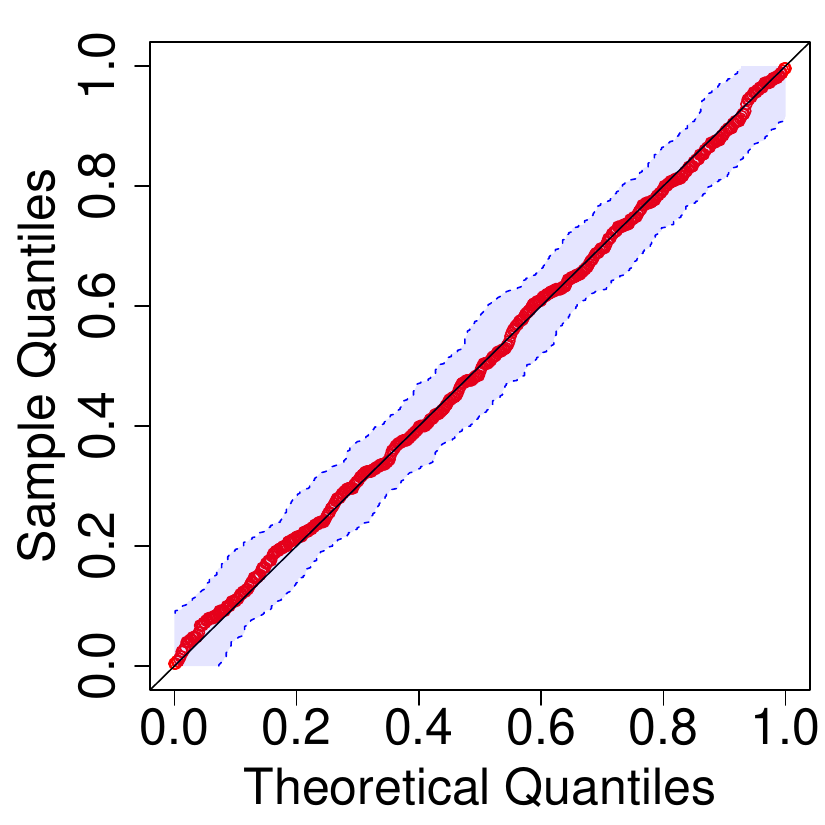}
  \end{subfigure}\  
      \begin{subfigure}[b]{0.16\textwidth}
  \caption{$\mathcal{P}^*(0)$}
   \includegraphics[width=\textwidth]{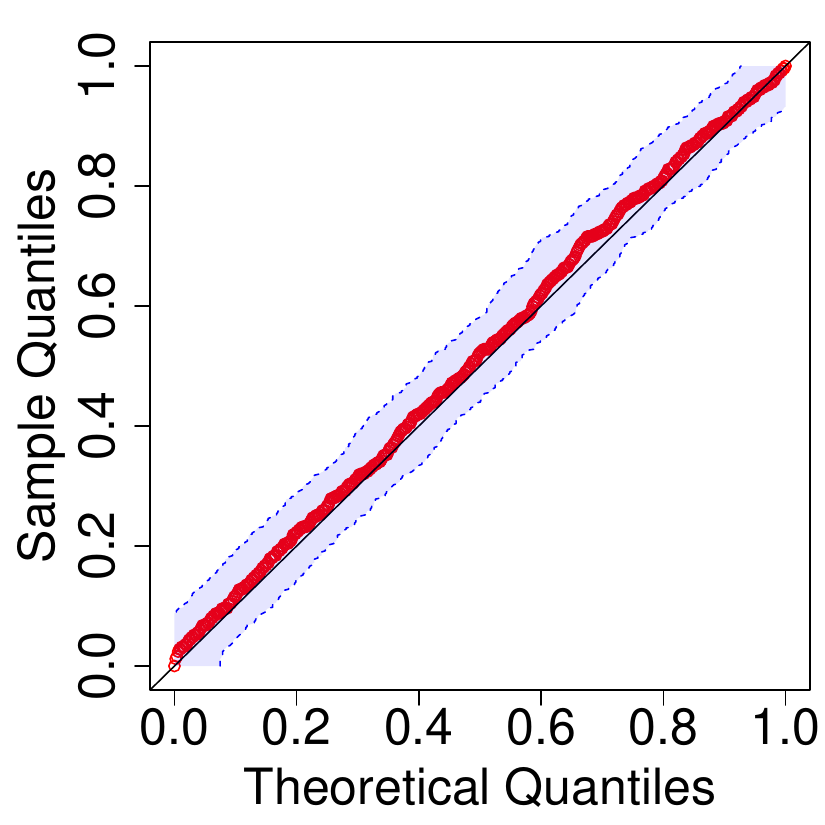}
  \end{subfigure} 
    \begin{subfigure}[b]{0.16\textwidth}
  \caption{$\mathcal{P}_{\alpha}^*(2)$}
   \includegraphics[width=\textwidth]{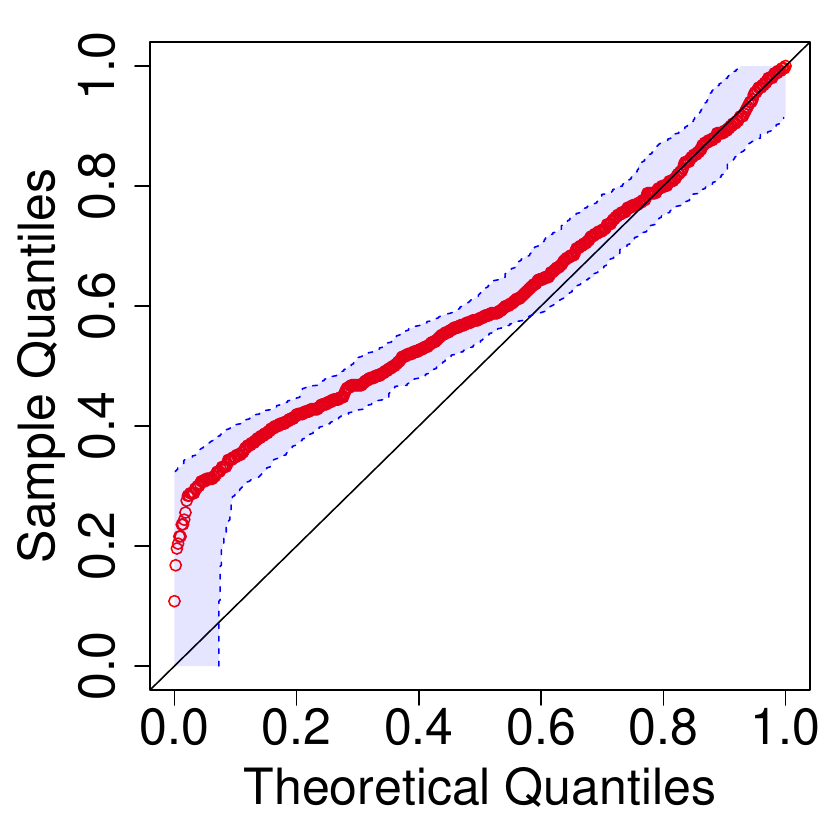}
  \end{subfigure}\  
     \begin{subfigure}[b]{0.16\textwidth}
  \caption{$\mathcal{P}_{\beta}^*(2)$}
   \includegraphics[width=\textwidth]{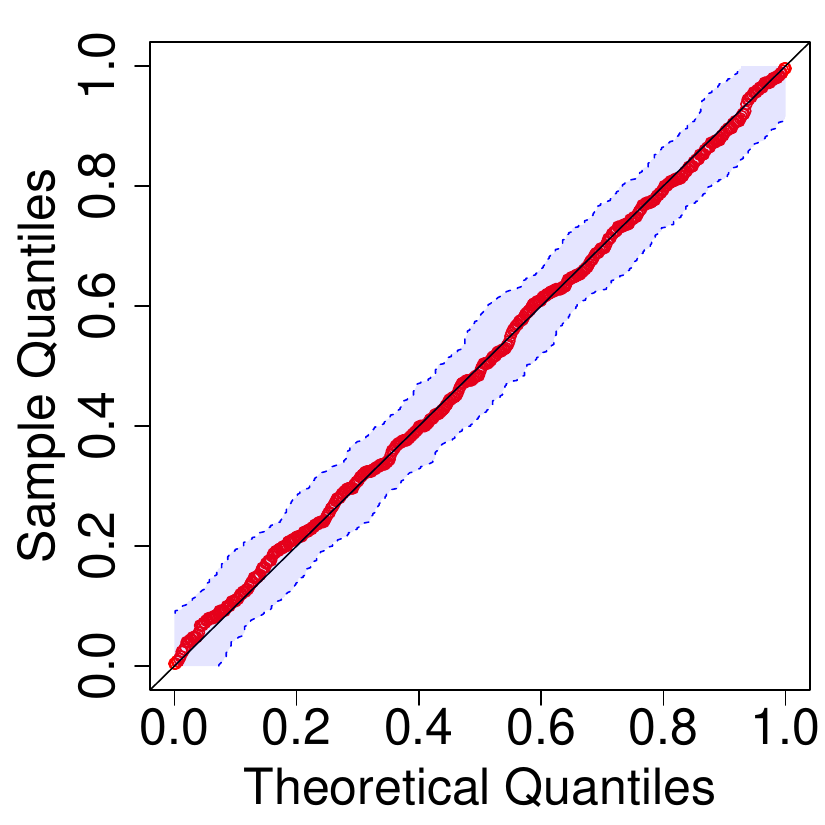}
  \end{subfigure}\  
      \begin{subfigure}[b]{0.16\textwidth}
  \caption{$\mathcal{P}^*(2)$}
   \includegraphics[width=\textwidth]{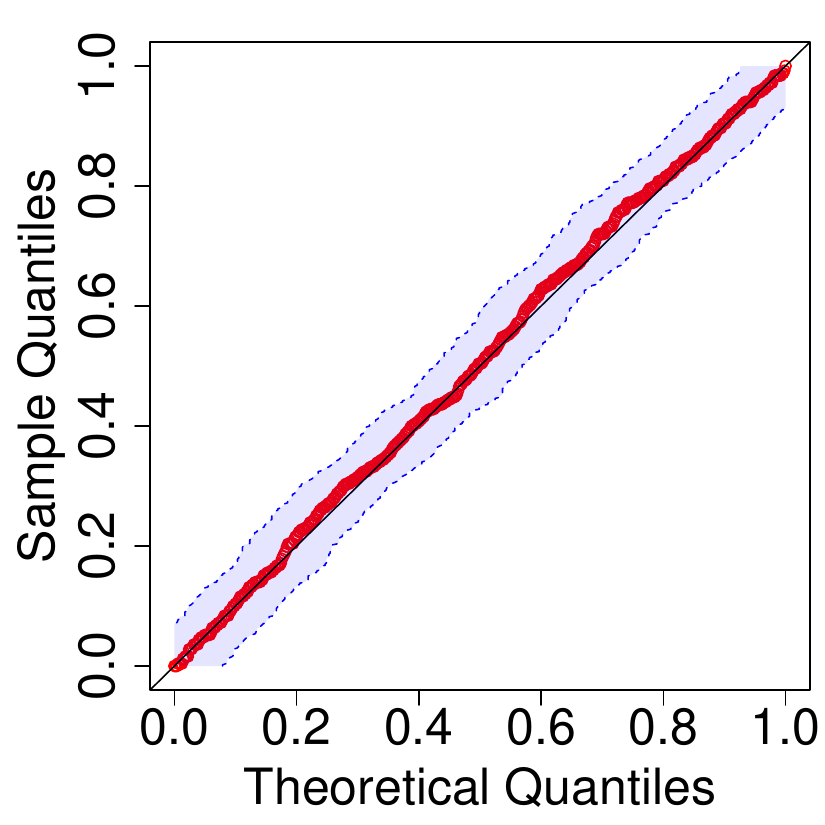}
  \end{subfigure} 
\end{figure} 

\paragraph{Scenario 3: $\alphaS=0$ and $\betaM=0.5$.} 
The results are similar to those under Scenario 2, and similar analysis applies. 
\begin{figure}[!htbp]
% \graphicspath{{figures/for_comments/R2_C3_tuning/}} 
%  \captionsetup[subfigure]{labelformat=empty} 
     \centering
       \caption{Scenario 3: $\alphaS=0$ and $\betaM=0.5$.} \label{fig:tuning005_0}
  \begin{subfigure}[b]{0.16\textwidth}
  \caption{$\mathcal{P}_{\alpha}^*(0)$}
   \includegraphics[width=\textwidth]{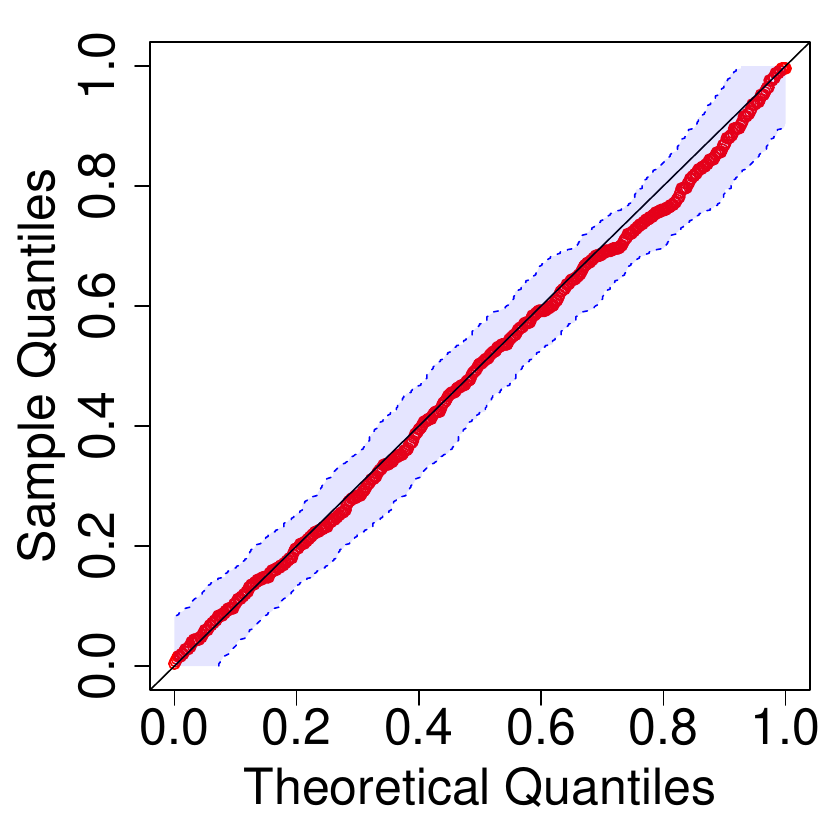}
  \end{subfigure}\   
    \begin{subfigure}[b]{0.16\textwidth}
  \caption{$\mathcal{P}_{\beta}^*(0)$}
   \includegraphics[width=\textwidth]{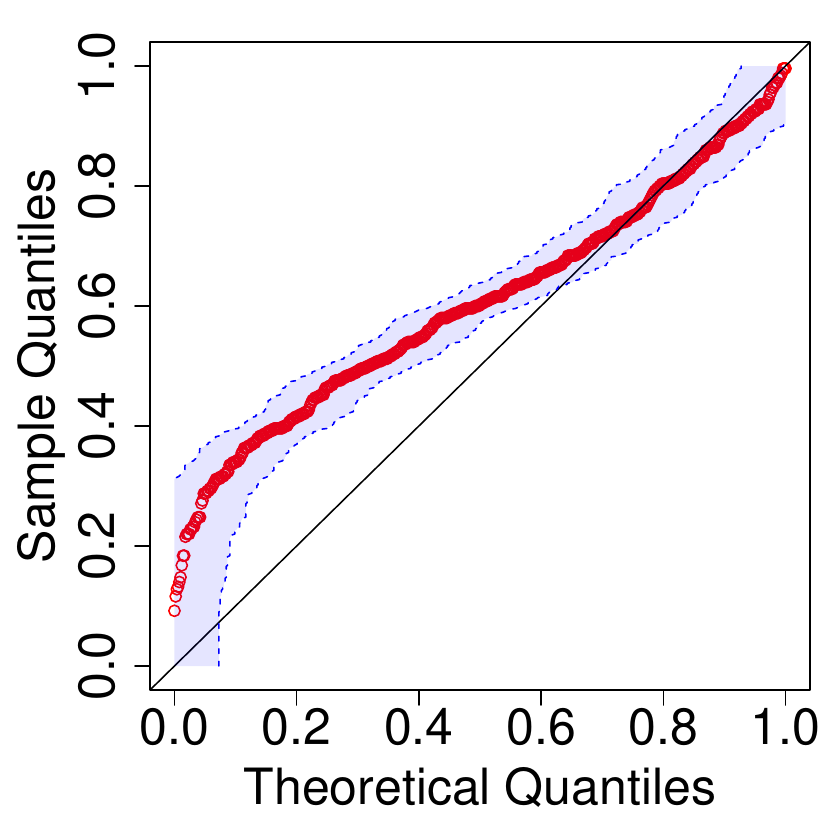}
  \end{subfigure}\  
      \begin{subfigure}[b]{0.16\textwidth}
  \caption{$\mathcal{P}^*(0)$}
   \includegraphics[width=\textwidth]{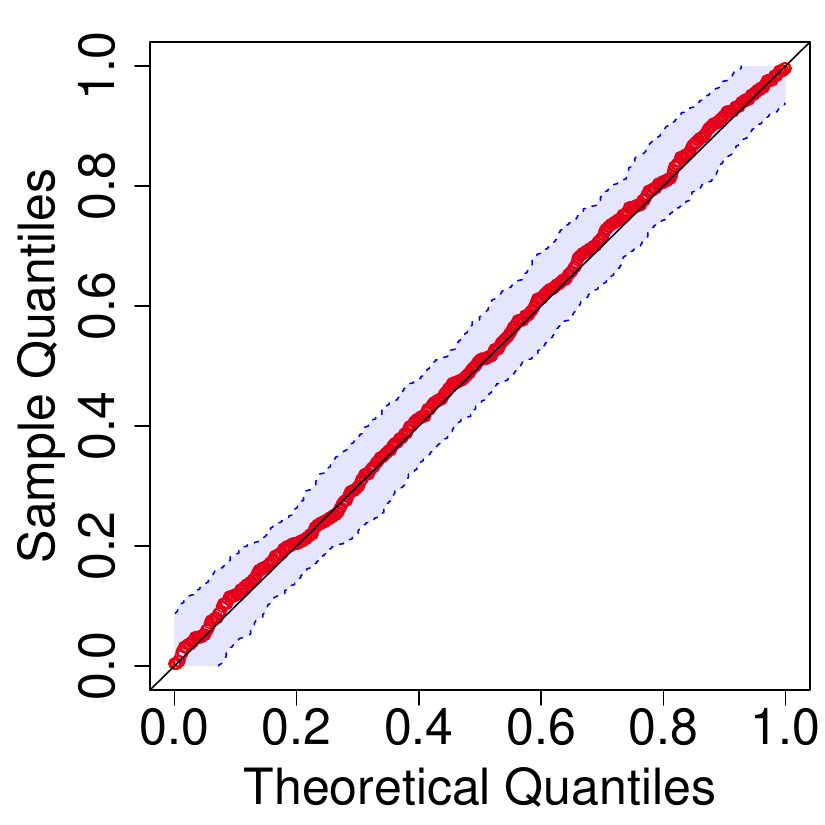}
  \end{subfigure}
  \begin{subfigure}[b]{0.16\textwidth} 
  \caption{$\mathcal{P}_{\alpha}^*(2)$}
   \includegraphics[width=\textwidth]{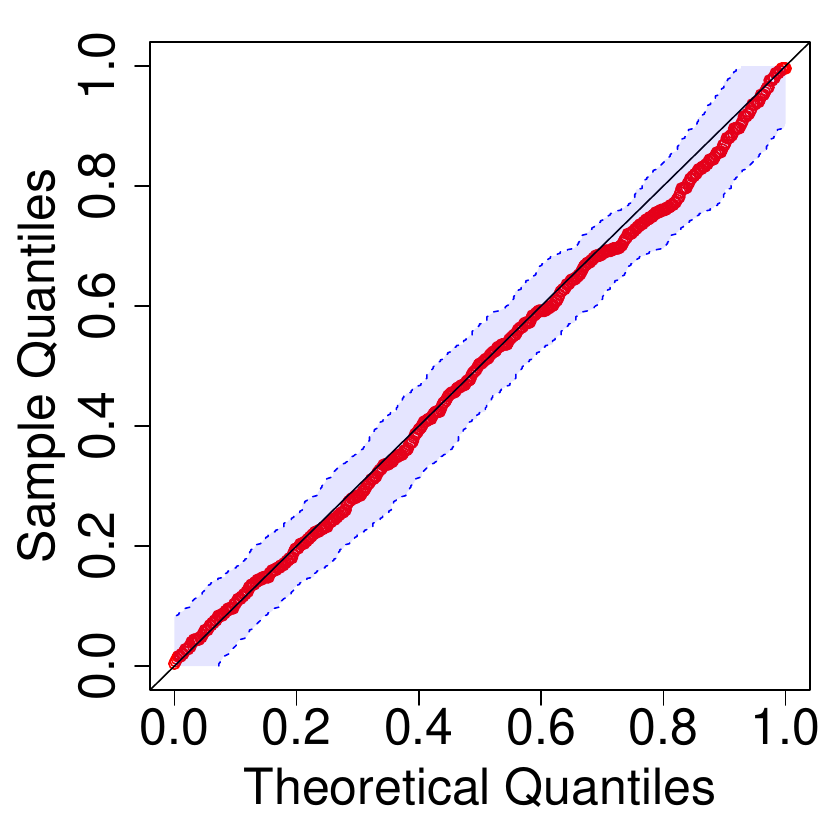}
  \end{subfigure}\ 
  \begin{subfigure}[b]{0.16\textwidth}
  \caption{$\mathcal{P}_{\beta}^*(2)$}
   \includegraphics[width=\textwidth]{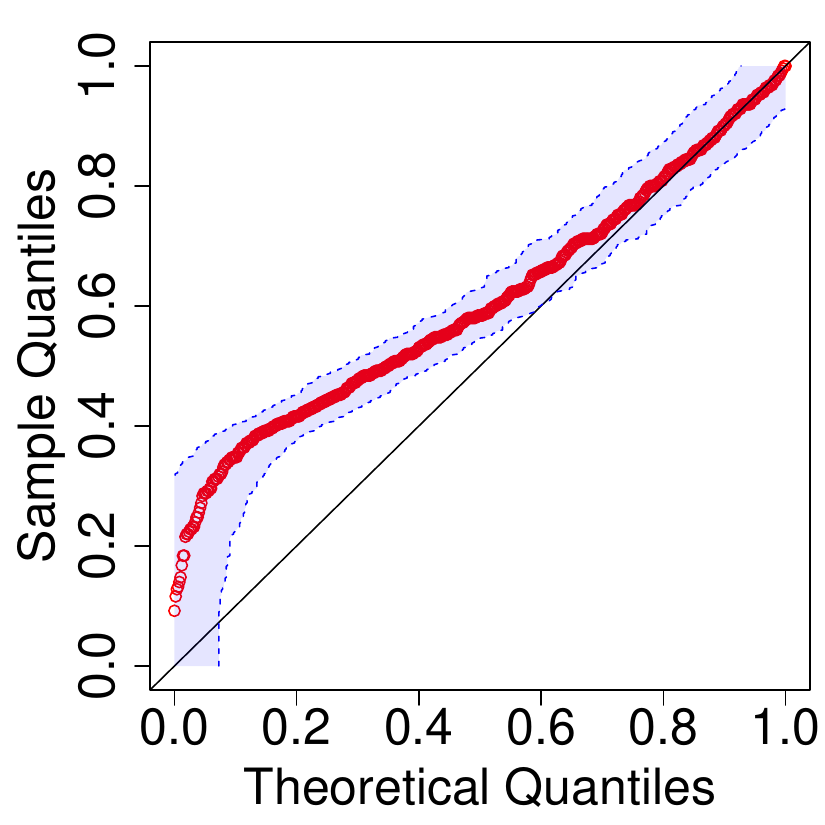}
  \end{subfigure}\ 
   \begin{subfigure}[b]{0.16\textwidth}
  \caption{$\mathcal{P}^*(2)$} 
   \includegraphics[width=\textwidth]{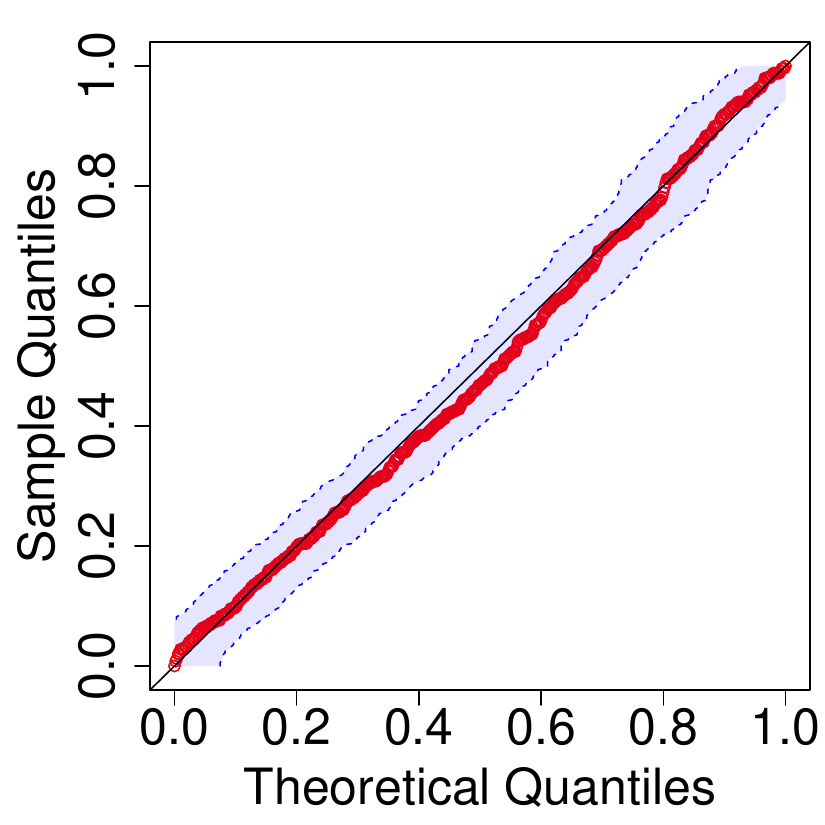}
  \end{subfigure}
 \end{figure}

% \begin{enumerate}
% \item use double bootstrap when $\lambda=0$ 
% \begin{itemize}
%     \item if one of the coefficients is non-zero, then when $\lambda=0$, at least one of two scenario has uniformly distributed $p$-values, we choose largest $\lambda$ such that $p$-values are distributed similarly to that under $\lambda=0$
%     \item if both two coefficients are zero, two scenarios are both conservative, we choose largest $\lambda$ such that both are uniformly distributed 
% \end{itemize}
% \end{enumerate}

\subsection{Computation Facilitation: Projected Bootstrap}\label{sec:suppcomputation}

In this section, we propose a projected bootstrap procedure to facilitate the computation. 
% how to facilitate the computation in the bootstrap procedure.  \ref{thm:bootstrapprodcomb}, 
%In Theorems \ref{thm:bootstrapprod},  % and \ref{thm:bootstrapmaxp},  
In Theorem  \ref{thm:bootstrapprodcomb}, 
we establish the bootstrap consistency results for $\hatalphaSn^*,$ $\hatbetaMn^*,$ $\hatsigmaalphan^*,$ $\hatsigmabetan^*$, and $\mathbb{R}_{n}^*(\balpha, \bbeta)$ that are computed from the nonparametric bootstrap, i.e., paired bootstrap in the regression settings.
Particularly, for a bootstrapped index set $\mathcal{I} = \{k_1, \ldots, k_n \}$, where $k_j\in \{1,\ldots, n\}$ for $j=1,\ldots, n$, 
the nonparametric bootstrap estimates
% of coefficients and standard deviations 
 are computed from the ordinary least squares regressions based on the bootstrapped data $\{(\S_i, \X_i, \M_i, \Y_i ): i \in \mathcal{I} \}$. 
However, since the ordinary least squares regressions require matrix inversion, 
repeating this procedure for the structural equation models in each bootstrap can be computationally intensive.

Alternatively, inspired by the formulae in \eqref{eq:olsfwlform}, 
 we introduce a bootstrap procedure with a projection step to facilitate the computation. 
Specifically, 
we first compute the projected observations $\{(\hatSperpi, \hatMperpi, \hatMperppi, \hatYperppi): i = 1, \ldots, n \}$ as defined in Section \ref{sec:preliminary}. 
Then in the projected bootstrap procedure, 
%using the projected observations,
with a bootstrapped index set $\mathcal{I}=\{k_1,\ldots, k_n\}$, 
we compute the coefficients by
%\begin{align*}
%\hatalphaSgn^* = \frac{\sum_{i=1}^n \hat{\S}_{\perp,k_i} \hat{\M}_{\perp,k_i}}{\sum_{i=1}^n \hat{\S}_{\perp,k_i}^2} = \frac{\PPboot(\hatSperp \hatMperp) }{\PPboot(\hatSperp^2)}, \hspace{1.3em} \hatbetaMgn^* = \frac{\sum_{i=1}^n \hat{\M}_{\perp',k_i} \hat{\Y}_{\perp',k_i} }{\sum_{i=1}^n \hat{\M}_{\perp',k_i}^2}=\frac{\PPboot( \hatMperpp \hatYperpp ) }{\PPboot( \hatMperpp^2)}, 	
%\end{align*}
\begin{align*}
\hatalphaSgn^* = \frac{\sum_{i\in \mathcal{I}} \hatSperpi \hatMperpi}{\sum_{i\in \mathcal{I}} \hatSperpi^2} = \frac{\PPboot(\hatSperp \hatMperp) }{\PPboot(\hatSperp^2)}, \hspace{2em} \hatbetaMgn^* = \frac{\sum_{i\in \mathcal{I}}\hatMperppi \hatYperppi }{\sum_{i\in \mathcal{I}} \hatMperppi^2}=\frac{\PPboot( \hatMperpp \hatYperpp ) }{\PPboot( \hatMperpp^2)}, 	
\end{align*}
and obtain the residuals by
 $\hatenMperpbooti = \hatMperpi - \hatSperpi \hatalphaSgn $ and $\hatenYperpbooti = \hatYperppi - \hatMperppi \hatbetaMgn $ for $i\in \mathcal{I}$. 
% $\hat{\epsilon}_{\M, \perp, n, i} = \hat{\M}_{\perp, k_i} - \hat{\S}_{\perp, k_i} \hatalphaSgn^* $ 
%and $\hat{\epsilon}_{\Y, \perp', n, i} = \hat{\Y}_{\perp',k_i} - \hat{\M}_{\perp',k_i} \hatbetaMgn^* $ 
%for $i=1,\ldots, n$.  
Moreover, we define
\begin{align*}
(\hatsigmaalphaperpboot)^2 =& {\sum_{i\in \mathcal{I}}\hatenMperpbooti^2}/{n},\ &	\VVSperpbootn =& \sum_{i\in \mathcal{I}} \hatSperpi^2/n,\ & \ZZSperpbootn =&\sum_{i\in \mathcal{I}} \hatenMperpbooti \hatSperpi/n, & \notag \\
(\hatsigmabetaperpboot)^2 =& {\sum_{i\in \mathcal{I}}\hatenYperpbooti^2}/{n},\ &	\VVMperpbootn =& \sum_{i\in \mathcal{I}} \hatMperppi^2/n,\ & \ZZMperpbootn =&\sum_{i\in \mathcal{I}} \hatenYperpbooti \hatMperppi/n,
\end{align*}
 %\begin{align*}
%(\hatsigmaalphaperpboot)^2 =& {\sum_{i=1}^n\hatenMperpbooti^2}/{n},\ &	\VVSperpbootn =& \sum_{i=1}^n \hat{\S}_{\perp,k_i}^2/n,\ & \ZZSperpbootn =&\sum_{i=1}^n \hatenMperpbooti \hat{\S}_{\perp,k_i}/n, & \notag \\
%(\hatsigmabetaperpboot)^2 =& {\sum_{i=1}^n\hatenYperpbooti^2}/{n},\ &	\VVMperpbootn =& \sum_{i=1}^n \hat{\M}_{\perp',k_i}^2/n,\ & \ZZMperpbootn =&\sum_{i=1}^n \hatenYperpbooti \hat{\M}_{\perp',k_i}/n,
%\end{align*}
which can be viewed as projected bootstrap versions of $\hatsigmaalphan^*$, $\hatsigmabetan^*$,  $\ZZSbootn$, $\ZZMbootn$, $\VVSbootn$, and $\VVMbootn$,  
respectively, 
where we replace $(\hatenM, \hatenY, \Sperp^*, \Mperpp^*)$ with $(\hatenMperpboot, \hatenYperpboot, \hatSperp, \hatMperpp)$.

In the proposed projected bootstrap, matrix inversion is only required in the projection step, and not repeated.  
Therefore, the computational cost can be significantly reduced. 
Theoretically, we can prove the following Lemma \ref{lm:perpconsisall},
and then by Slutsky's lemma, the bootstrap consistency results in Theorems  \ref{thm:bootstrapprodcomb} and \ref{thm:bootstrapmaxp} still hold  for the projected bootstrap procedure. 
The proof of Lemma \ref{lm:perpconsisall} is given in Section \ref{sec:pfperpconsisall}. 

\begin{lemma}\label{lm:perpconsisall}
Under Condition \ref{cond:inversemoment}, 
\begin{enumerate}
	\item[(1)] $\sqrt{n}(\hatalphaSn^* -  \hatalphaSgn^* )\overset{\mathrm{P}^*}{\leadsto} 0$ and  $\sqrt{n}(\hatbetaMn^* -  \hatbetaMgn^* )\overset{\mathrm{P}^*}{\leadsto} 0$;
	\item[(2)] $ \ZZSbootn- \ZZSperpbootn \overset{\mathrm{P}^*}{\leadsto} 0$ and $\ZZMbootn - \ZZMperpbootn\overset{\mathrm{P}^*}{\leadsto} 0$;
%	$(\ZZSperpbootn, \ZZMperpbootn,  )- (\ZZSbootn, \ZZMbootn) \xrightarrow{P_C} (0,0)$, 
	\item[(3)] $\VVSbootn - \VVSperpbootn \overset{\mathrm{P}^*}{\leadsto}  0$ and $\VVMbootn - \VVMperpbootn\overset{\mathrm{P}^*}{\leadsto} 0$;
%	$(\VVSperpbootn, \VVMperpbootn)-(\VVSbootn, \VVMbootn) \xrightarrow{P_C} (0,0)$,
	\item[(4)] $ (\hatsigmaalphan^*)^2 - (\hatsigmaalphaperpboot)^2 \overset{\mathrm{P}^*}{\leadsto} 0$ and $ (\hatsigmabetan^*)^2 - (\hatsigmabetaperpboot)^2\overset{\mathrm{P}^*}{\leadsto}  0$.  
\end{enumerate}	
% conditionally (on the data) in probability.	
\end{lemma}

\clearpage
\section{Additional Numerical Results}\label{sec:addnumres}
%\subsection{Results in the main text when $n=500$}

In this section, Section \ref{sec:suppfigsim} presents figures supplementary to Section \ref{sec:nulltype1error} in the main text.  Section \ref{sec:largeeffectsample} presents additional simulation experiments examining the effects of varying signal sizes and sample sizes.  
Section \ref{sec:splitanalysis} presents additional data analysis results including marginal screening,  the joint testing, a sensitivity analysis, interpretation of data analysis results, and a confirmatory analysis,. 

% \ref{sec:sensanalysis}

\subsection{QQ-Plots Supplementary to Section \ref{sec:nulltype1error}}\label{sec:suppfigsim}

  \begin{figure}[!htbp]
 \captionsetup[subfigure]{labelformat=empty} 
  \centering
       \caption{Q-Q plots of $p$-values under the fixed null with $n = 500$. } \label{fig:fixnulln500}
  \begin{subfigure}[b]{0.29\textwidth}
      \caption{\footnotesize{\quad \  $H_{0,1}: (\alphaS, \betaM)=(0, 0.5)$}}
  \includegraphics[width=\textwidth]{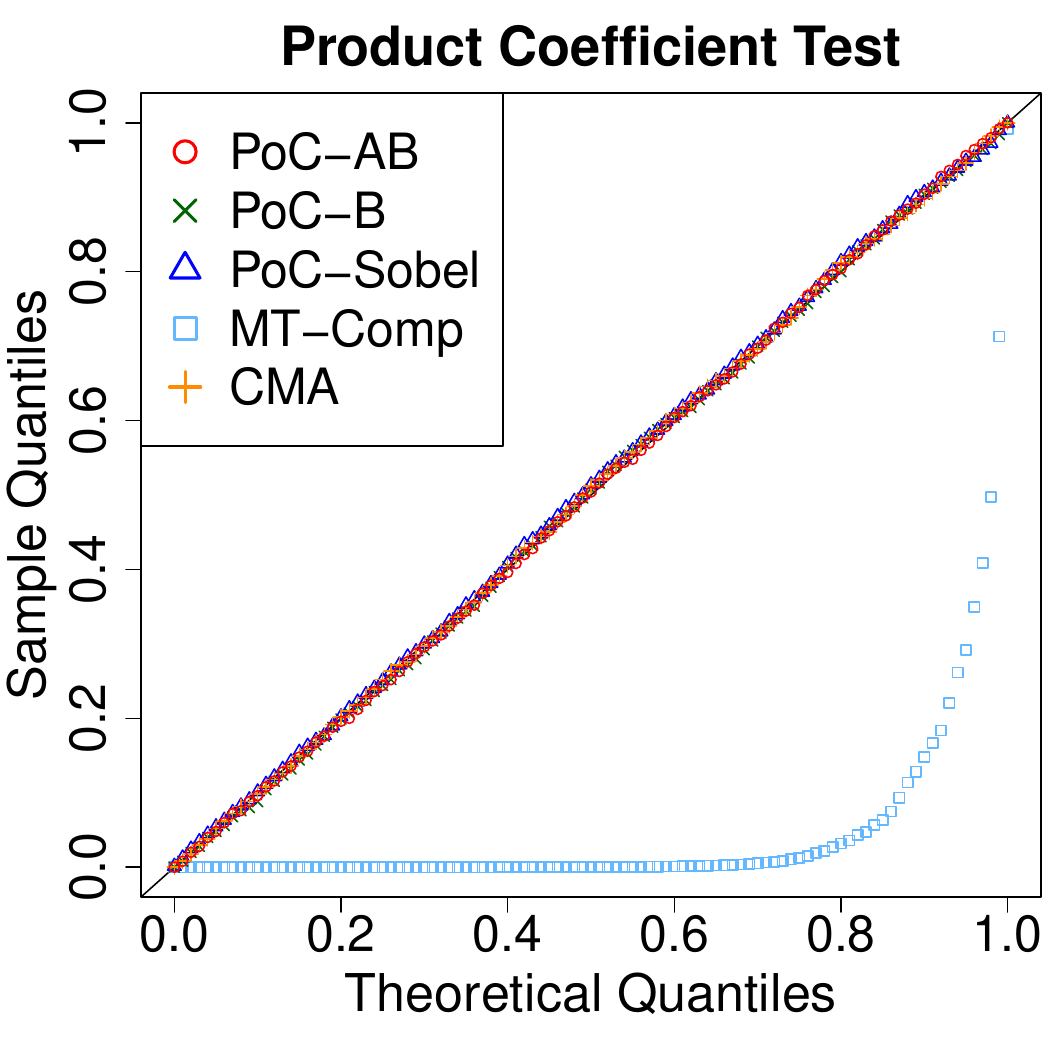}
     \end{subfigure} \  
       \begin{subfigure}[b]{0.29\textwidth}
   \caption{\footnotesize{\quad \  $H_{0,2}: (\alphaS, \betaM)=(0.5, 0)$}}
     \includegraphics[width=\textwidth]{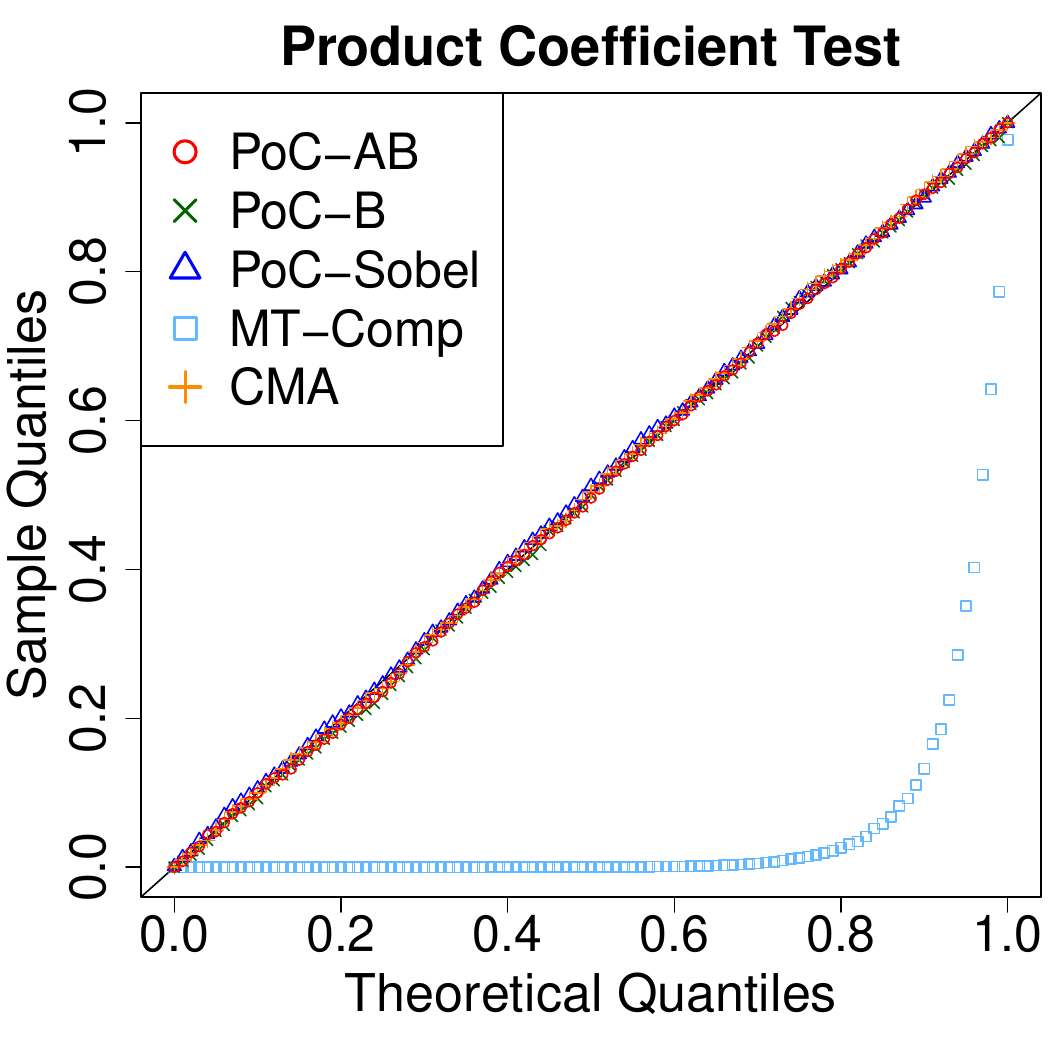}
     \end{subfigure} \
       \begin{subfigure}[b]{0.29\textwidth}
   \caption{\footnotesize{\quad \  $H_{0,3}: (\alphaS, \betaM)=(0,0)$}}
     \includegraphics[width=\textwidth]{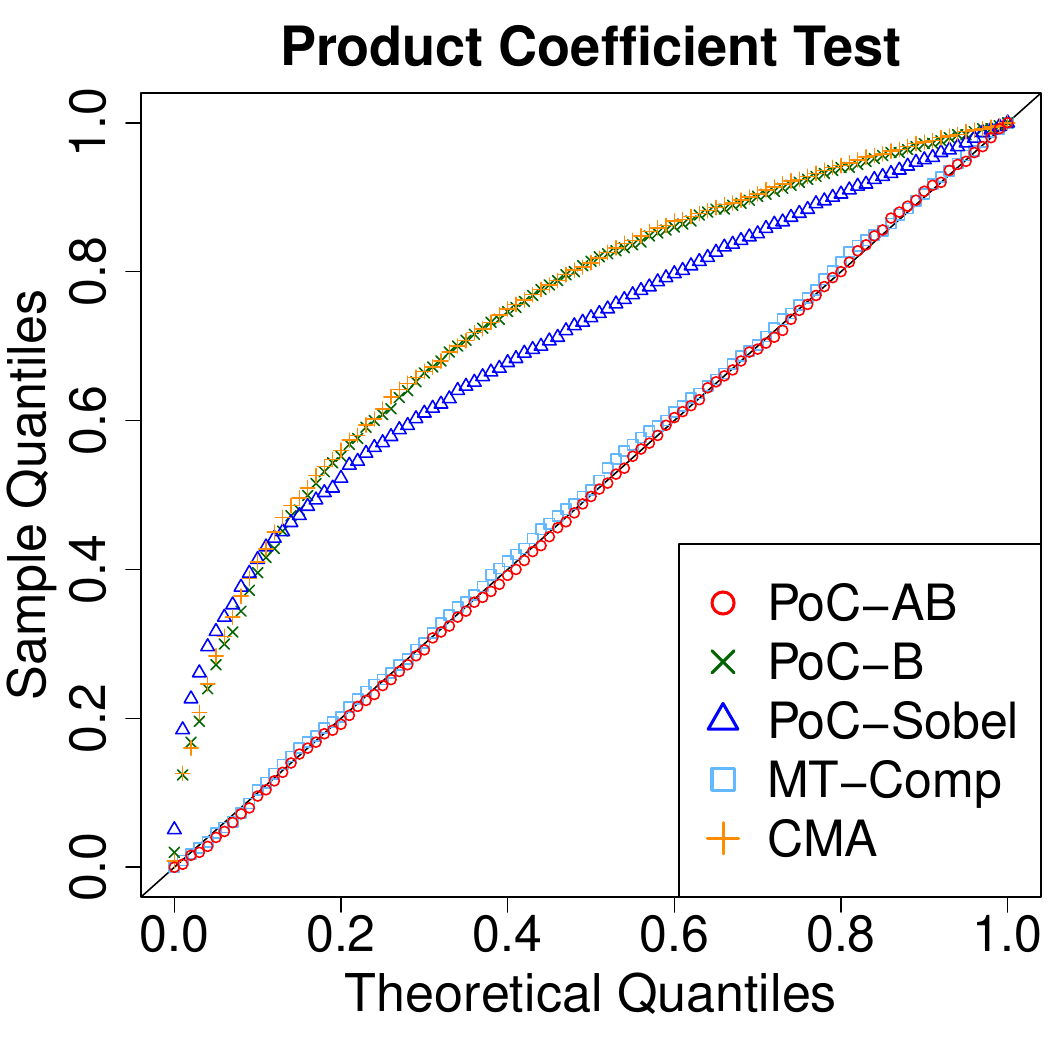}
   \end{subfigure} \\ 
       \begin{subfigure}[b]{0.29\textwidth}
             \includegraphics[width=\textwidth]{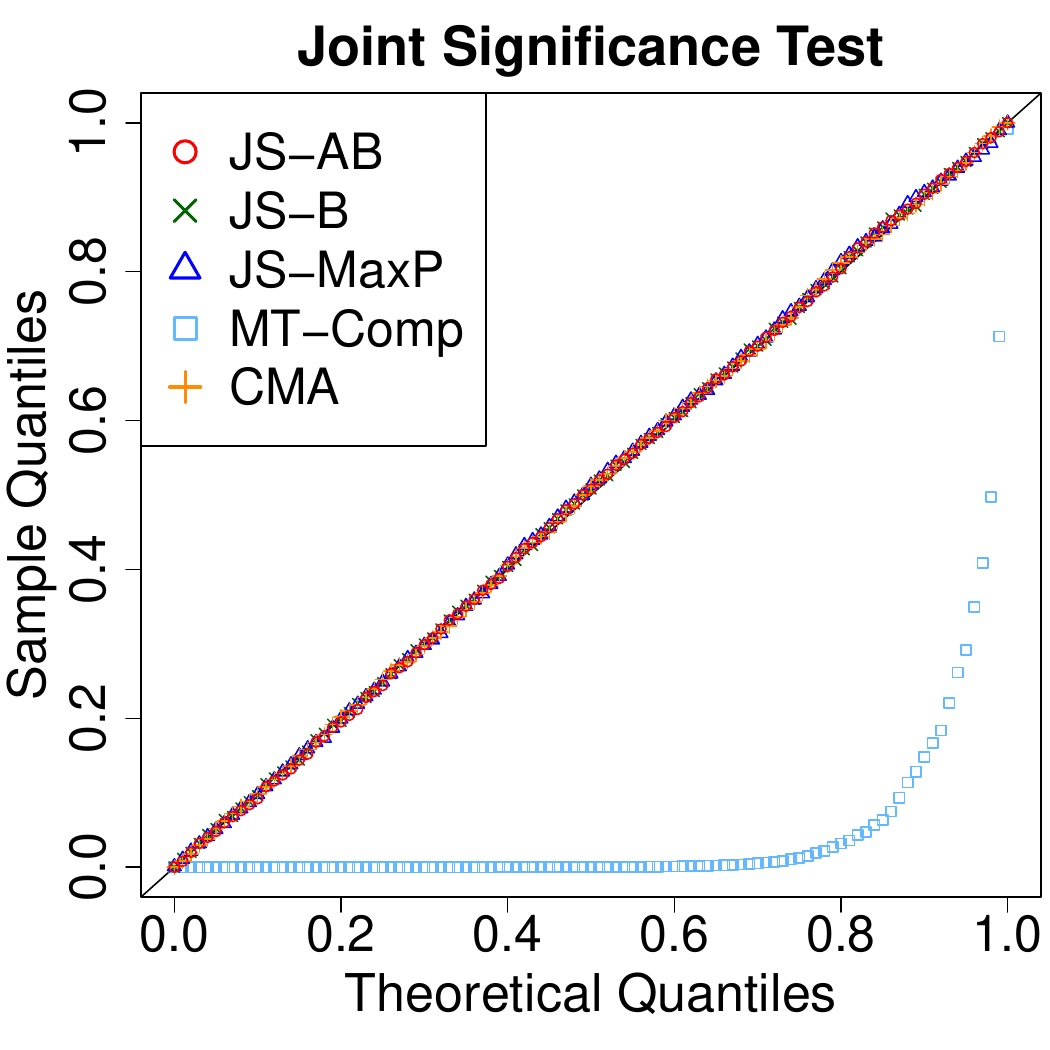}
       \end{subfigure}\
       \begin{subfigure}[b]{0.29\textwidth}
             \includegraphics[width=\textwidth]{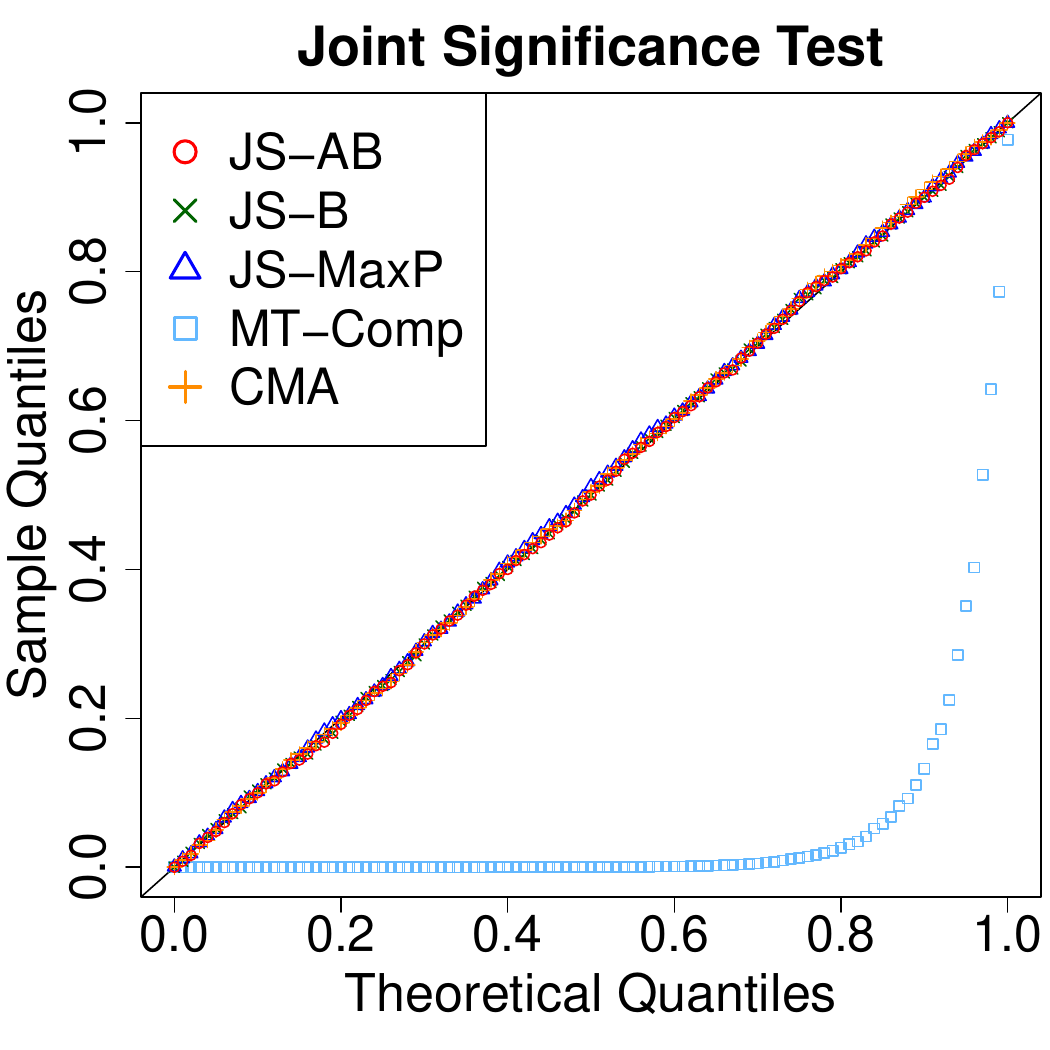}
       \end{subfigure} \ 
   \begin{subfigure}[b]{0.29\textwidth}
             \includegraphics[width=\textwidth]{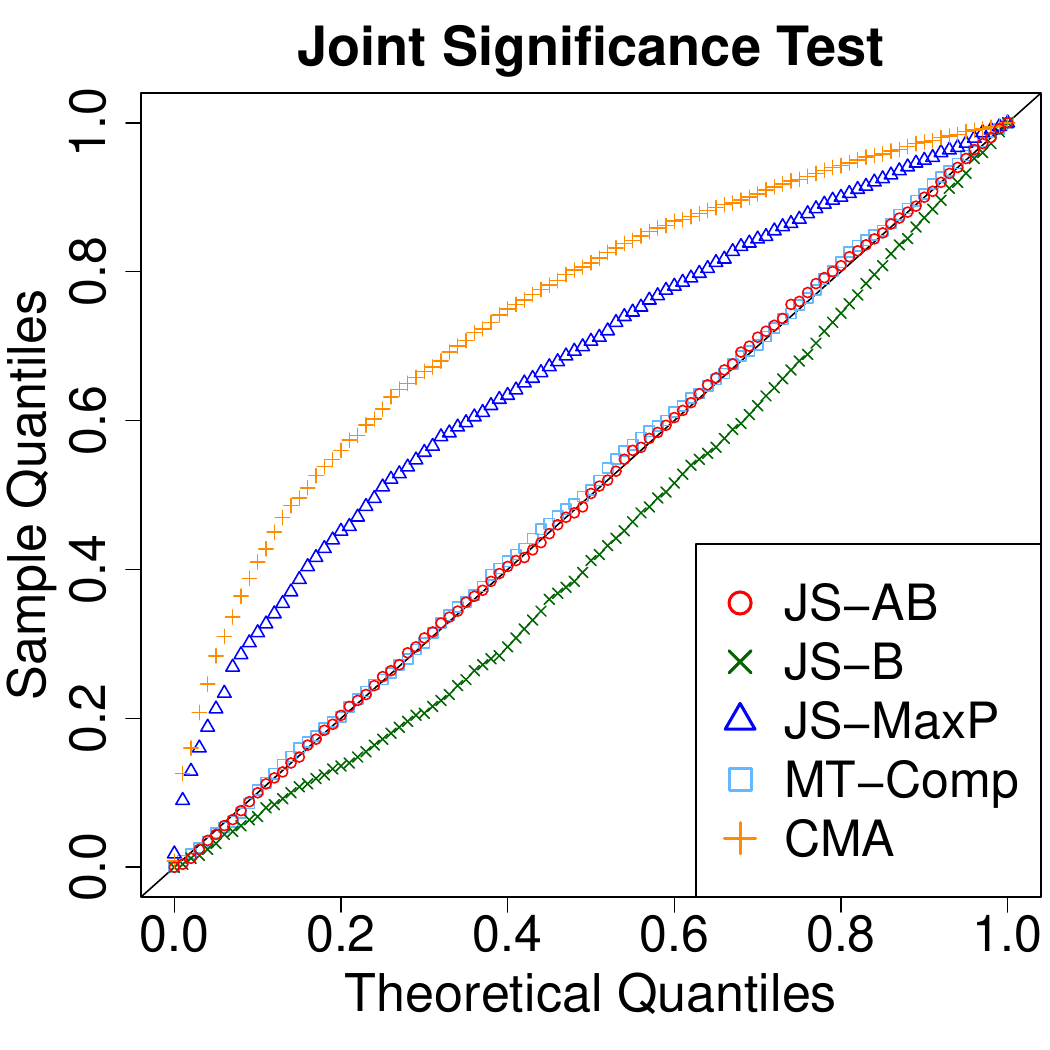}
       \end{subfigure} 
 \end{figure}

 \begin{figure}[!htbp]
 \captionsetup[subfigure]{labelformat=empty} 
  \centering
       \caption{Q-Q plots of $p$-values under the mixture of  nulls: $n = 500$.} \label{fig:mixnulln500}
   \begin{subfigure}[b]{0.29\textwidth}
   \caption{\footnotesize{\quad \  (I) $(1/3, 1/3, 1/3)$}}
 %    \caption{\small{\ \ $(\alphaS, \betaM)=(0,0)$}}
     \includegraphics[width=\textwidth]{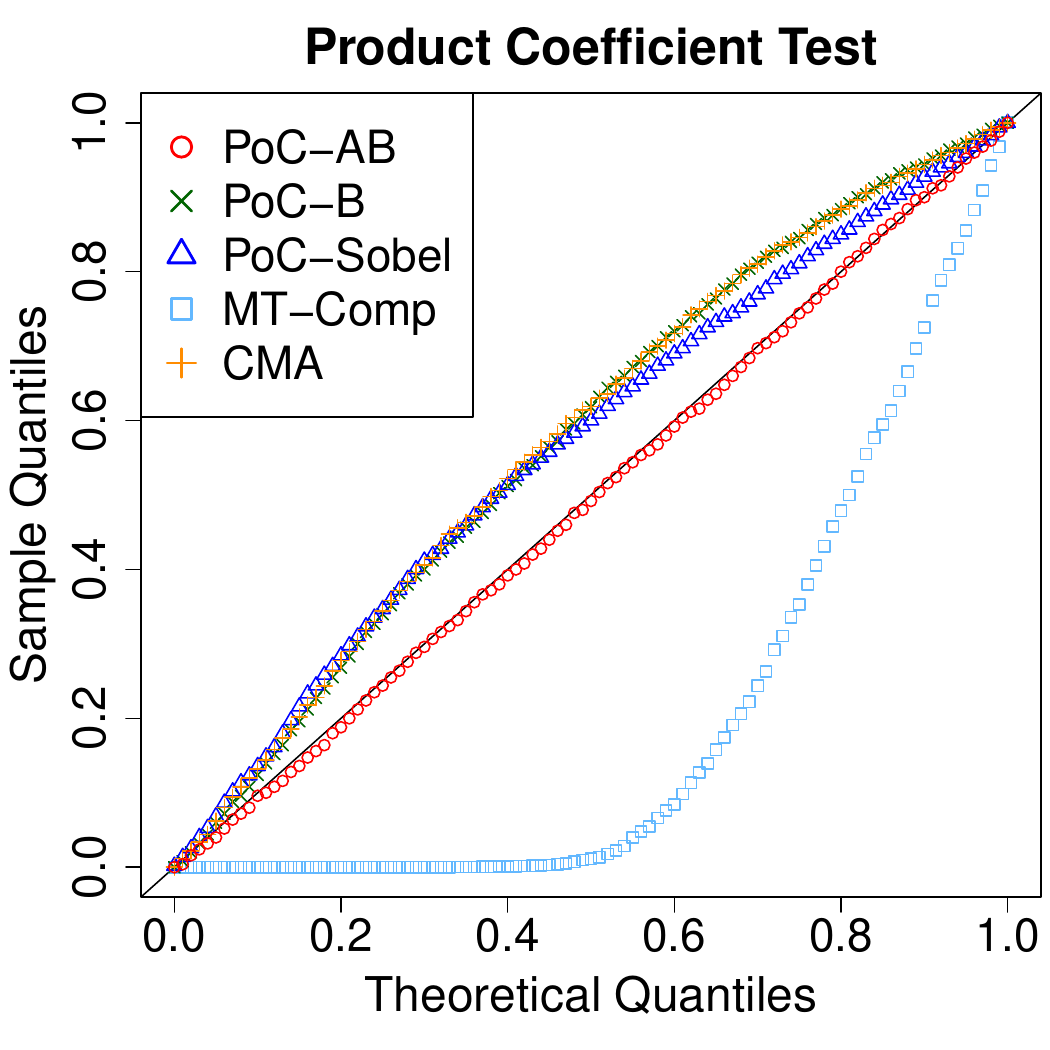}
   \end{subfigure} \    
  \begin{subfigure}[b]{0.29\textwidth}
      \caption{\footnotesize{\quad \  (II) $(0.2, 0.2, 0.6)$}}
 %   \caption{\small{\ \  $(\alphaS, \betaM)=(0,0.5)$}}
  \includegraphics[width=\textwidth]{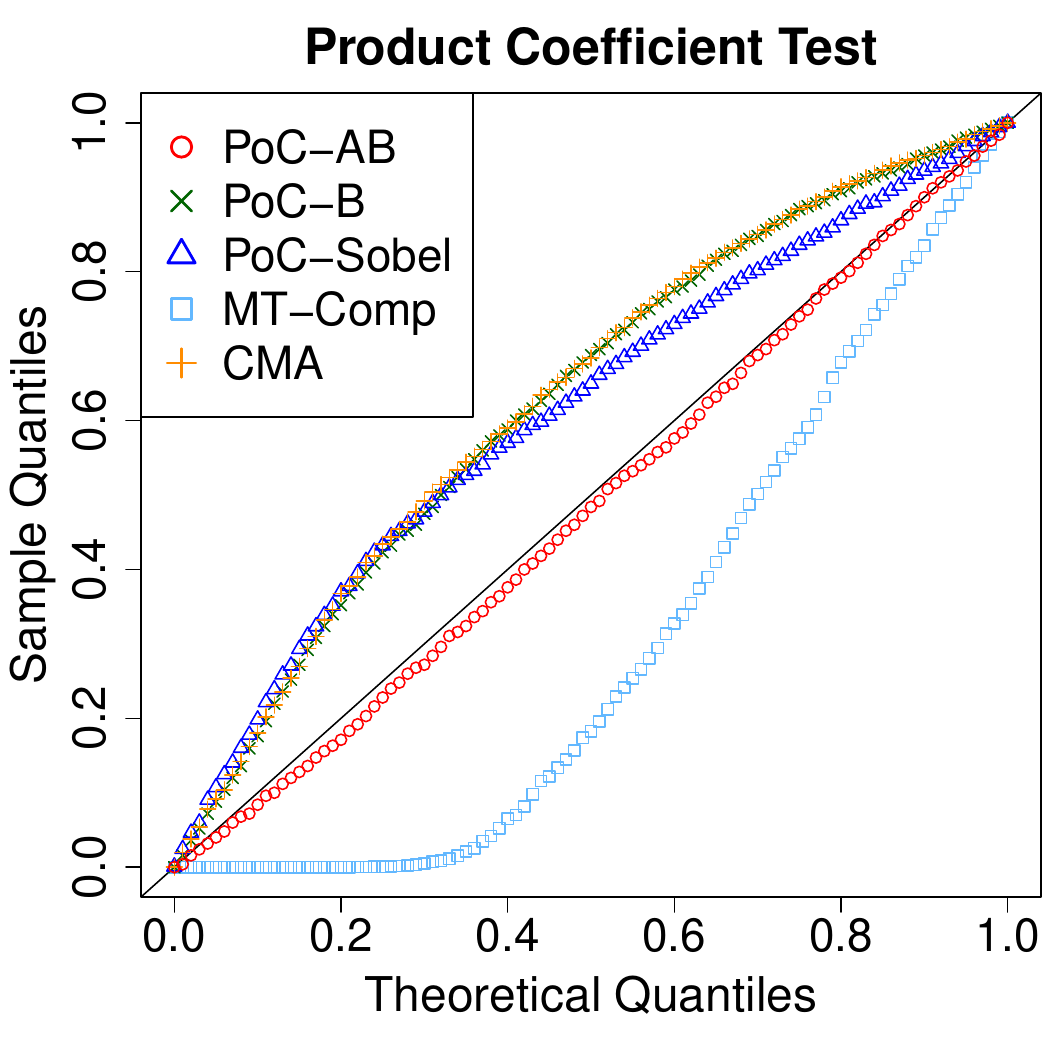}
     \end{subfigure} \  
       \begin{subfigure}[b]{0.29\textwidth}
   \caption{\footnotesize{\quad \  (III) $(0.05, 0.05, 0.9)$}}
 %      \caption{\small{\ \ $(\alphaS, \betaM)=(0.5,0)$}}
     \includegraphics[width=\textwidth]{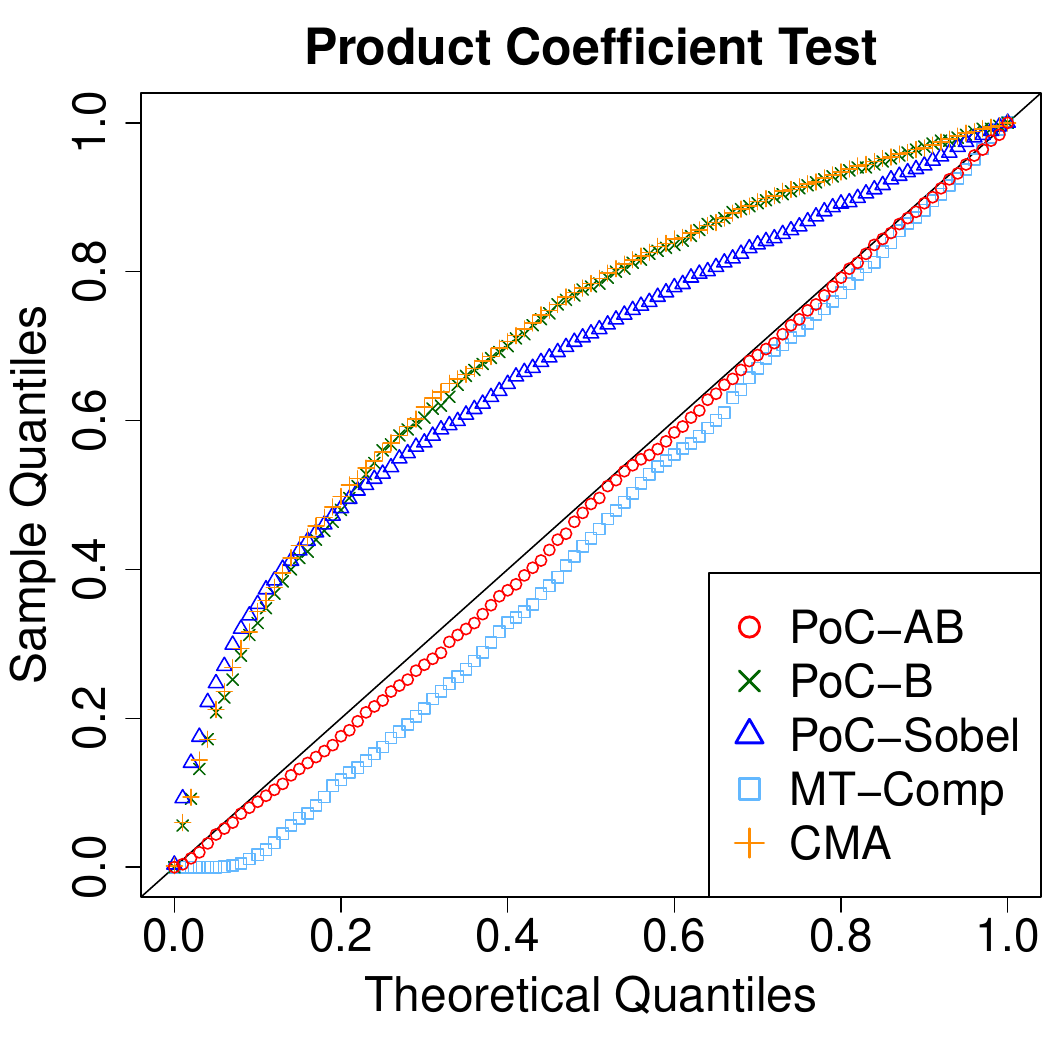}
     \end{subfigure} \\ 
       \begin{subfigure}[b]{0.29\textwidth}
             \includegraphics[width=\textwidth]{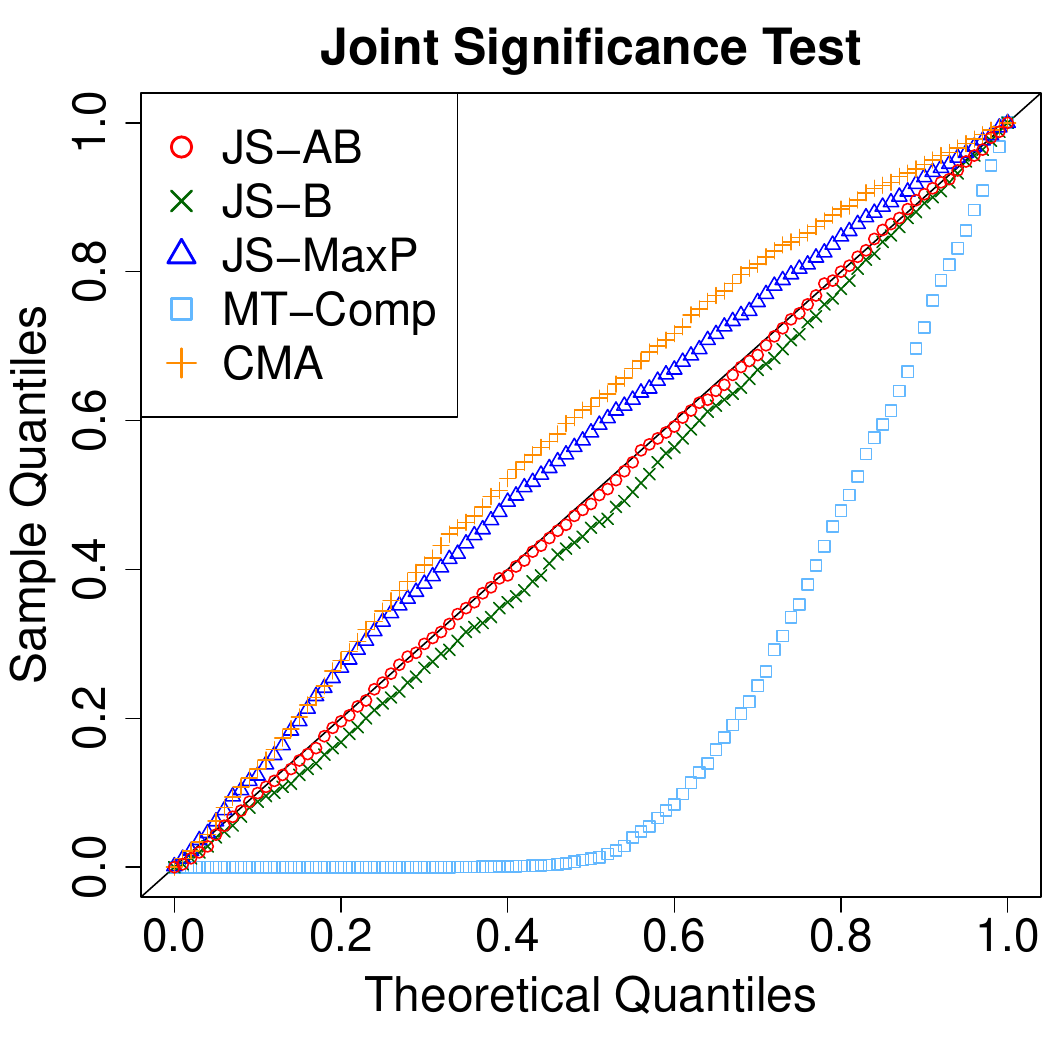}
       \end{subfigure} \  
       \begin{subfigure}[b]{0.29\textwidth}
             \includegraphics[width=\textwidth]{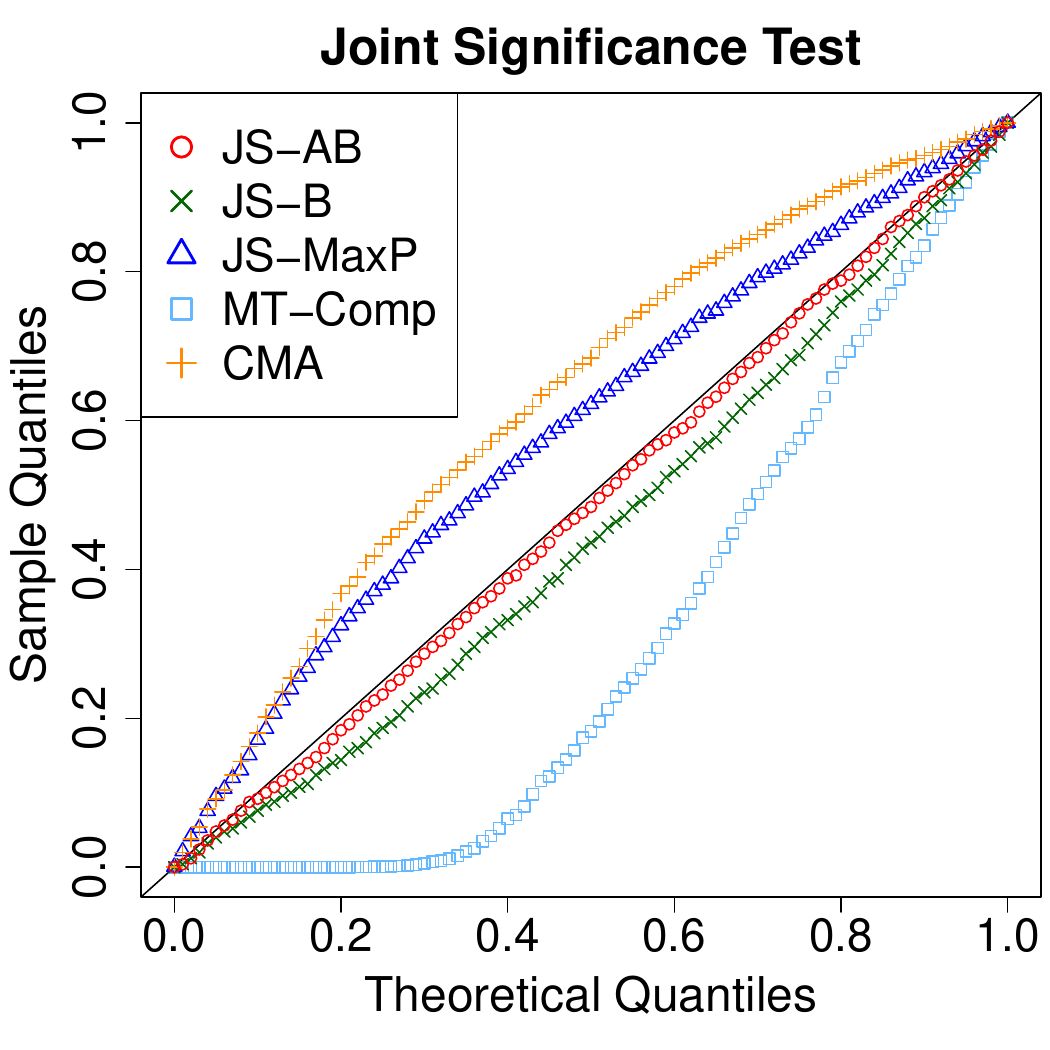}
       \end{subfigure}\
       \begin{subfigure}[b]{0.29\textwidth}
             \includegraphics[width=\textwidth]{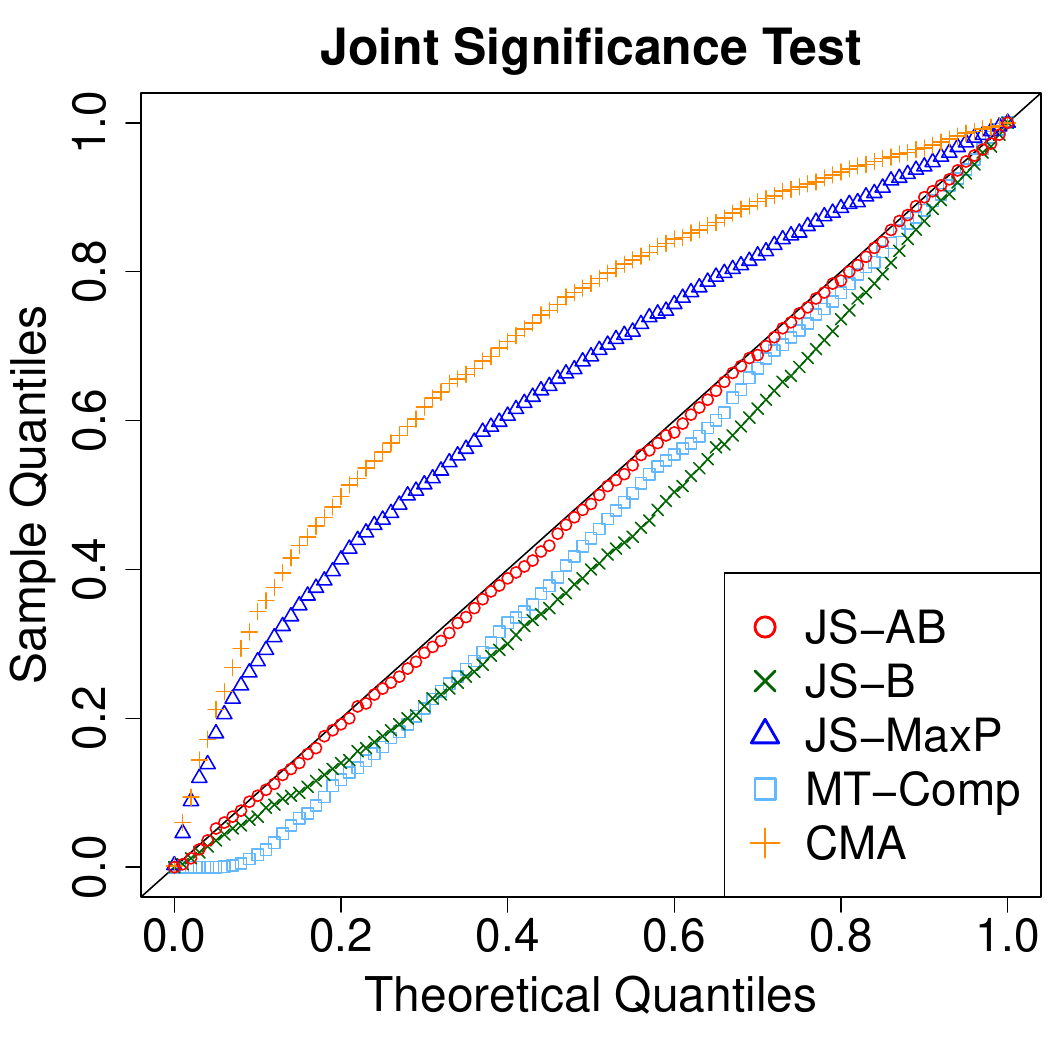}
       \end{subfigure}     
 \end{figure} 

\newpage
\subsection{Additional Simulations: Varying Effects and Sample Sizes}\label{sec:largeeffectsample} 
% In this section, w
We illustrate how the proposed method performs in terms of the type-I error control when the effect sizes and sample sizes become larger. 
In particular,  we generate data following  the model 
$
\M = \alphaS \S + \alpha_I +\eM,	
$ and $\Y = \betaM \M + \beta_I + \S + \eY$,
where the exposure variable $\S$ is simulated from a Bernoulli distribution with the success probability equal to 0.5, and $\eM$ and $\eY$ are simulated independently from $\mathcal{N}(0,\sigma^2)$ with $\sigma=0.5$.   
%Under this model, the mediation effect is $\alphaS\betaM$. 
%We simulate data with sample sizes $n\in \{200, 500, 1000\}$.
To evaluate how the varying effect sizes and sample sizes influence the type-\RNum{1} errors, we consider two cases under $H_0$:  $\alphaS\betaM=0$: 
\begin{enumerate}
	\item[(i)] Fix $\alphaS=0$, take $\betaM=\exp(k)$ for $k\in \{-\infty, 0, 1, 2, 3, 4, 5\}$. %$k\in  \{0, \exp(0), \exp(1), \ldots, \exp(5)\}$
	\item[(ii)] Fix $\betaM=0$, take $\alphaS =\exp(k)$ for $k\in \{-\infty, 0, 1, 2, 3, 4, 5\}$. 
\end{enumerate}
Estimated type-I errors of different tests under cases (i) and (ii) are presented in Figures \ref{fig:typeialpha0} and \ref{fig:typeibeta0}, respectively. 
We can see that the AB tests control the type-\RNum{1} errors well under different values of the non-zero coefficients, whereas the other tests can deviate from the nominal significance level when both coefficients are 0.

\newpage

\begin{figure}[!htbp]
%\graphicspath{{figures/for_comments/R1_C5_Large_Effects/results/}}
\captionsetup[subfigure]{labelformat=empty}
\centering
\caption{When $\alphaS=0$, estimated type-I errors} \label{fig:typeialpha0}
\vspace{5pt}
\caption*{(a)  PoC-tests with significance level 0.05}
\begin{subfigure} {0.32\textwidth}
\caption{(a.1) $n=200$}
\includegraphics[width=\textwidth]{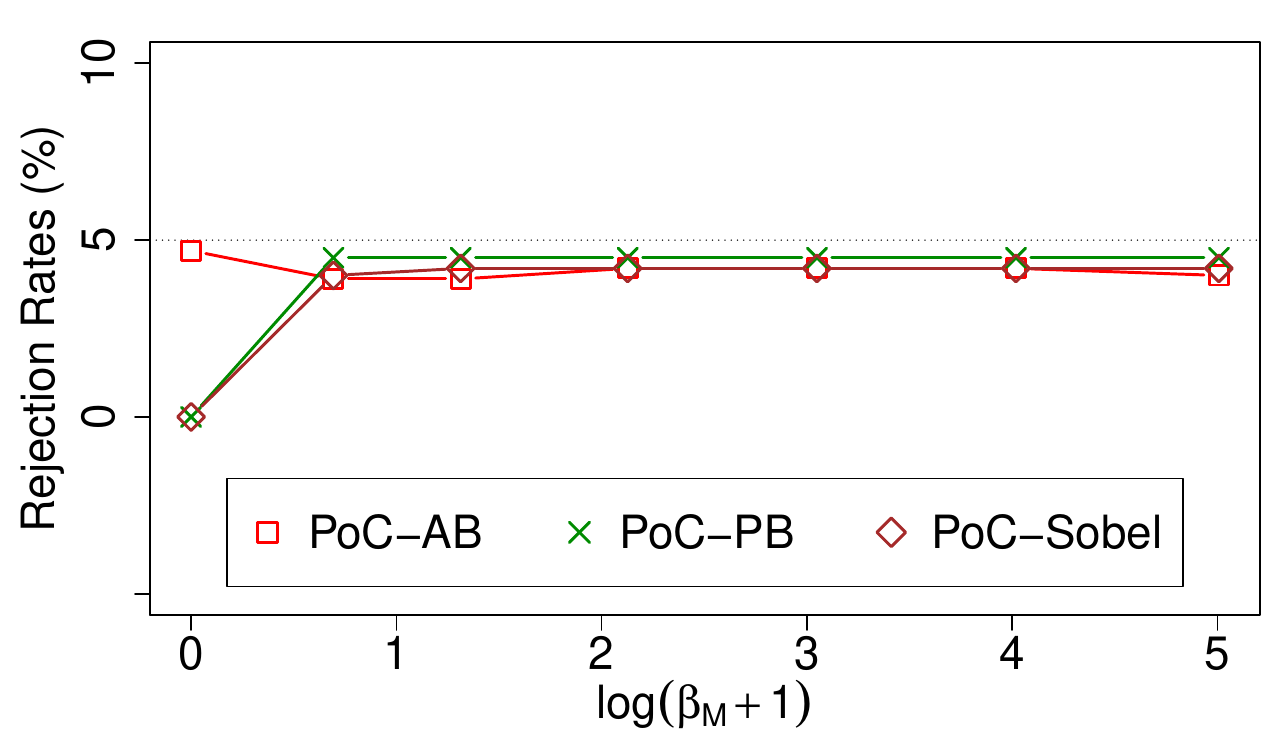} 
\end{subfigure} 
\begin{subfigure} {0.32\textwidth}
\caption{(a.2) $n=500$}
\includegraphics[width=\textwidth]{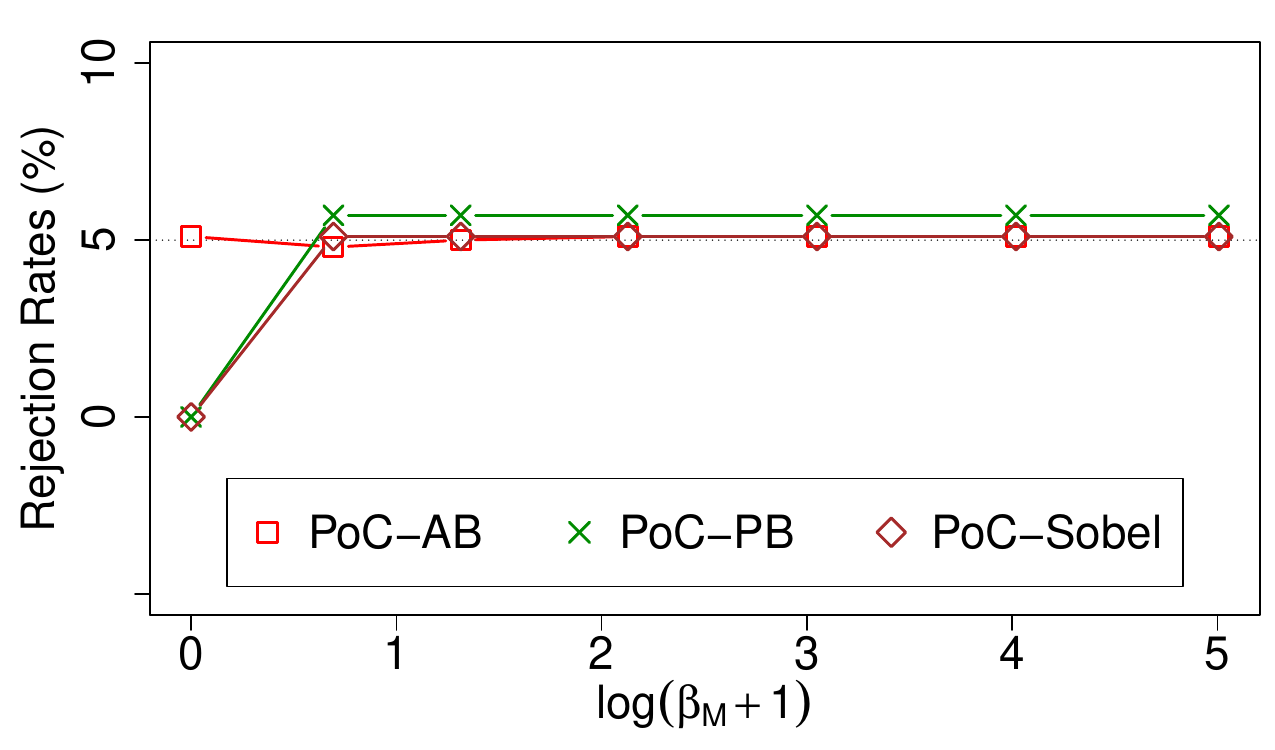} 
\end{subfigure} 
\begin{subfigure} {0.32\textwidth}
\caption{(a.3) $n=1000$}
\includegraphics[width=\textwidth]{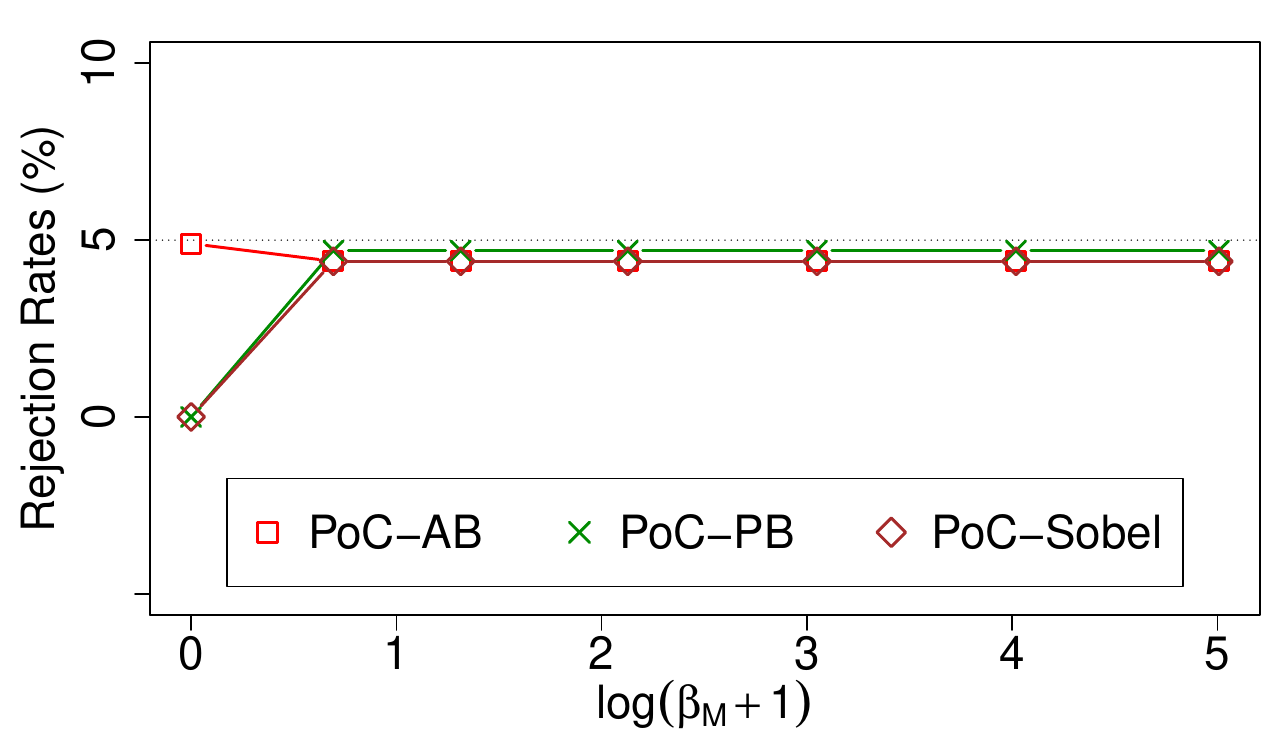} 
\end{subfigure} 

\vspace{5pt}
\caption*{(b) PoC-tests with significance level 0.1}
\begin{subfigure} {0.32\textwidth}
\caption{(b.1) $n=200$}
\includegraphics[width=\textwidth]{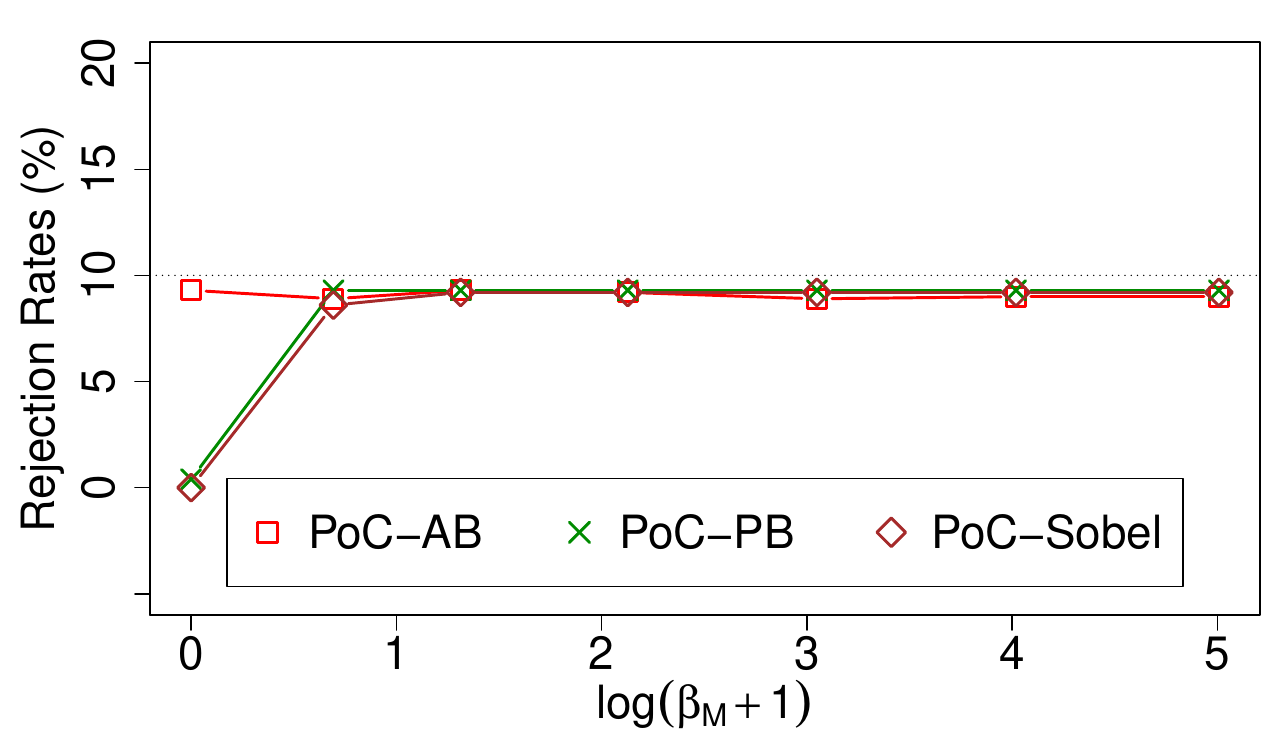} 
\end{subfigure} 
\begin{subfigure} {0.32\textwidth}
\caption{(b.2) $n=500$}
\includegraphics[width=\textwidth]{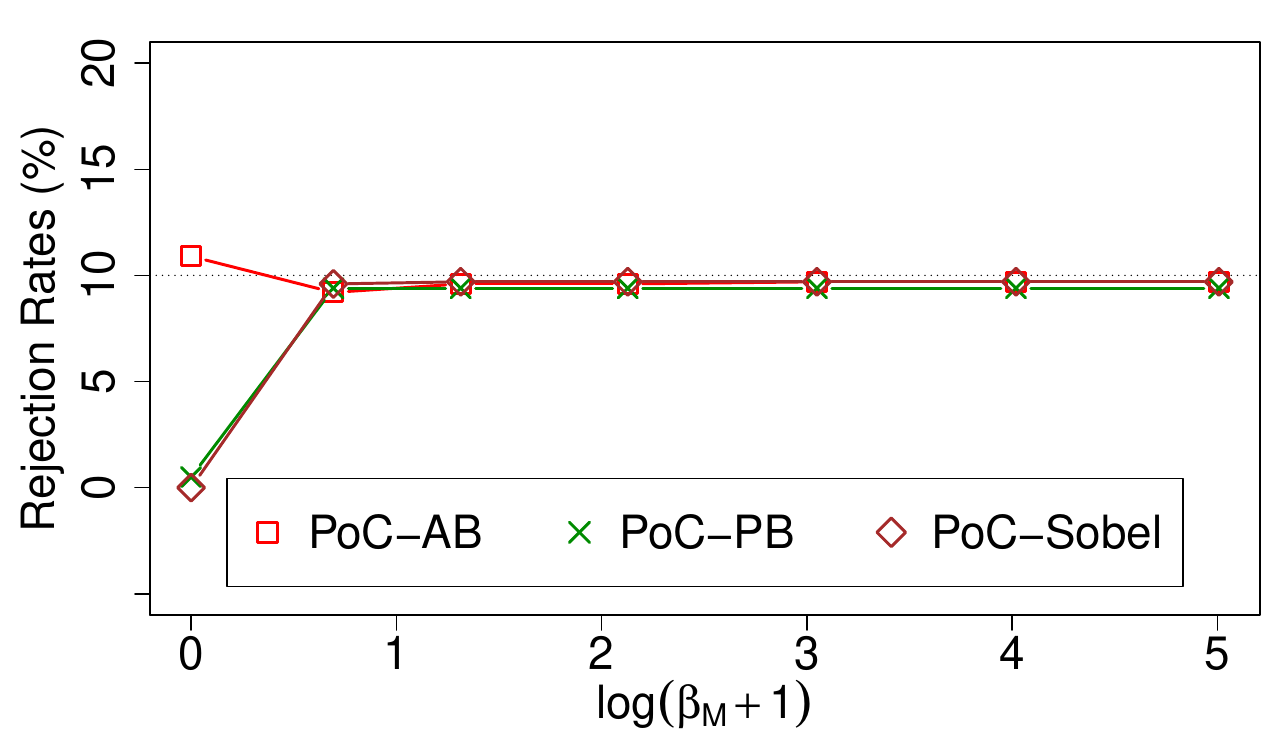} 
\end{subfigure} 
\begin{subfigure} {0.32\textwidth}
\caption{(b.3) $n=1000$}
\includegraphics[width=\textwidth]{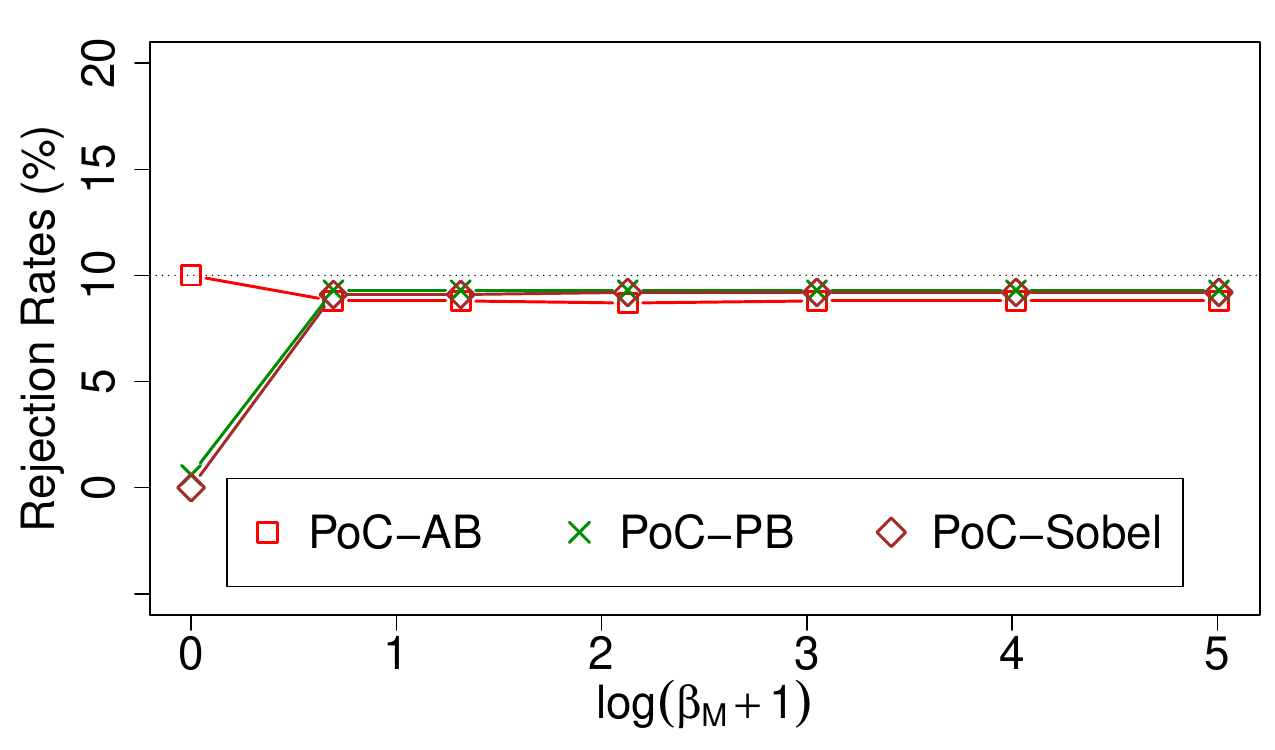} 
\end{subfigure} 

\vspace{5pt}
\caption*{(c) JS-tests with significance level $0.05$}
\begin{subfigure} {0.32\textwidth}
\caption{(c.1)  $n=200$}
\includegraphics[width=\textwidth]{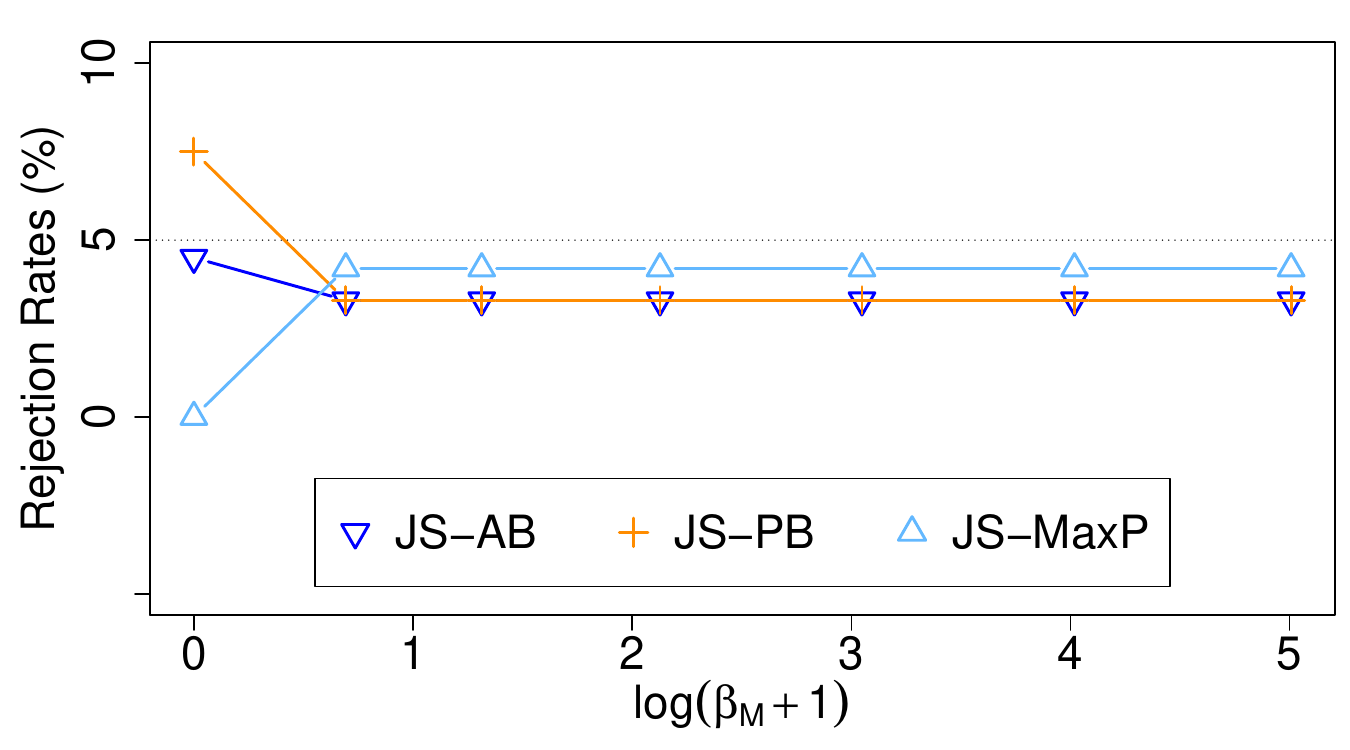} 
\end{subfigure}
\begin{subfigure} {0.32\textwidth}
\caption{(c.2)  $n=500$}
\includegraphics[width=\textwidth]{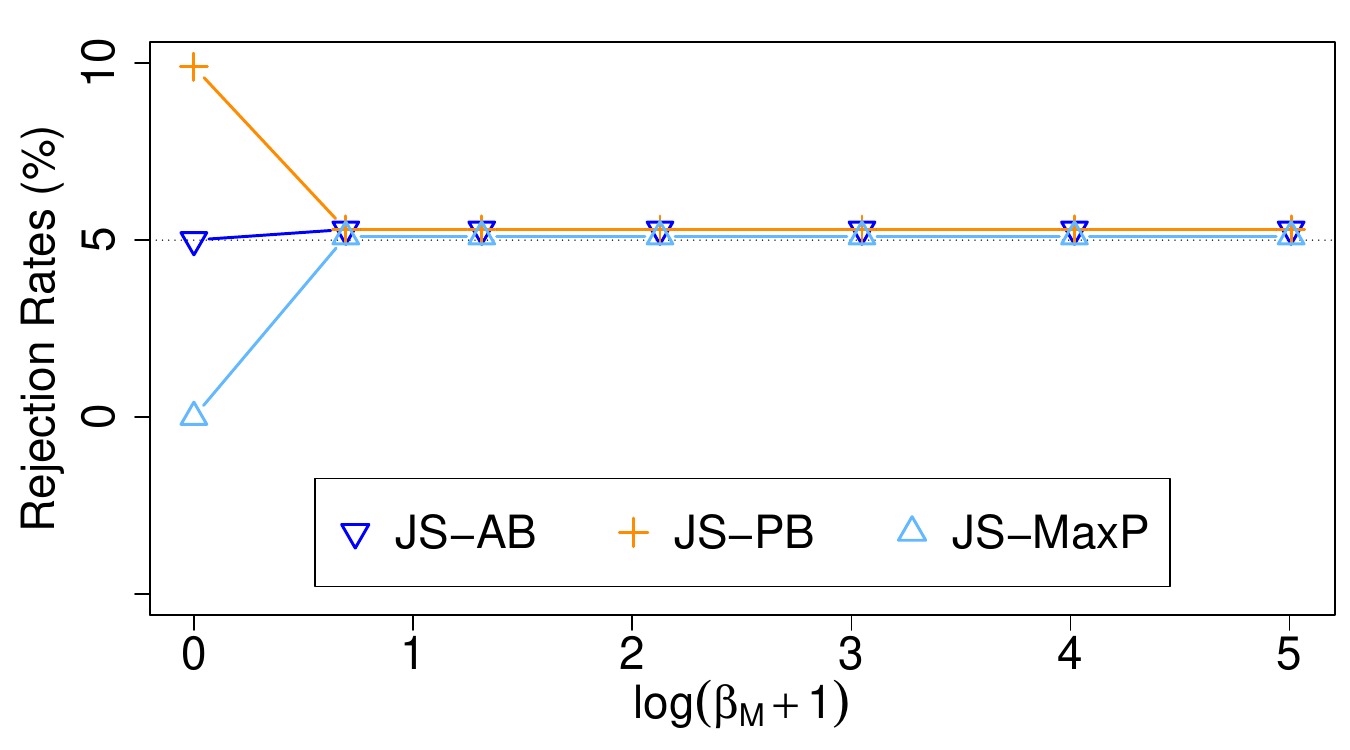} 
\end{subfigure}
\begin{subfigure} {0.32\textwidth}
\caption{(c.3)  $n=1000$}
\includegraphics[width=\textwidth]{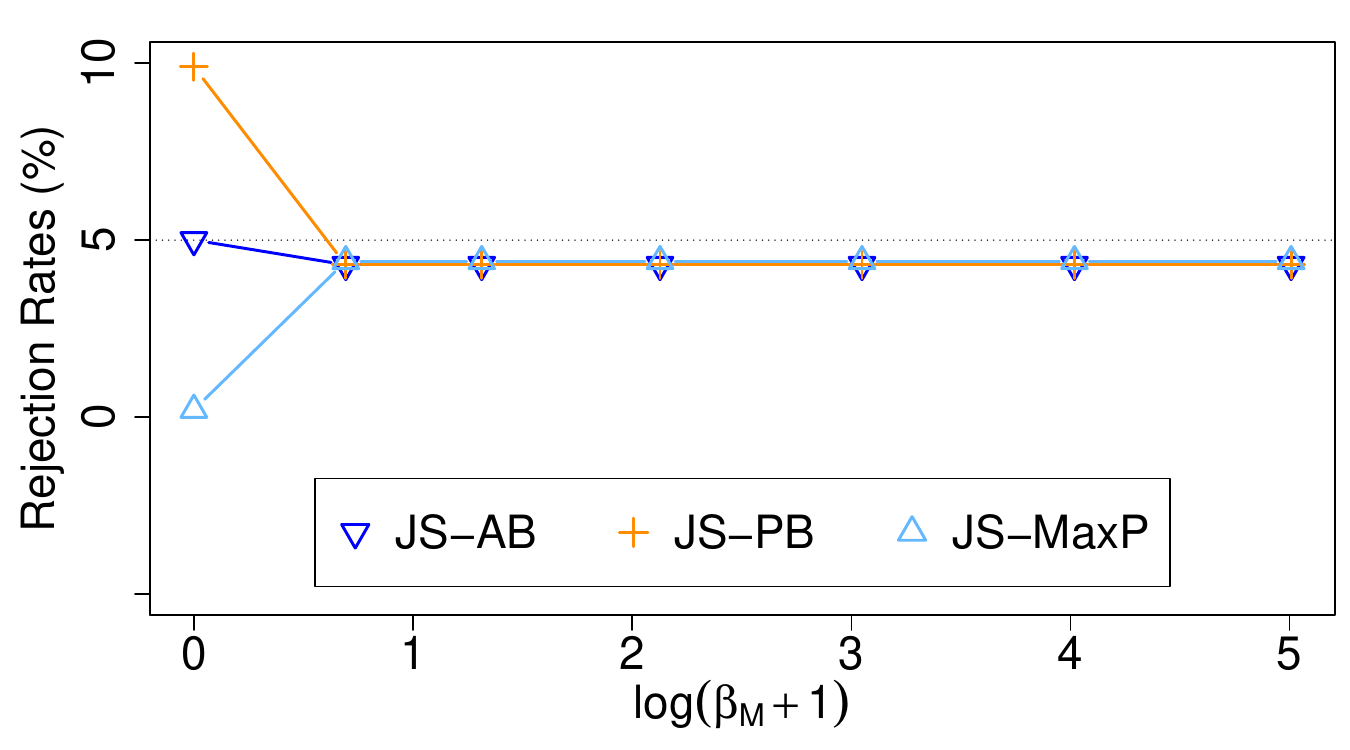} 
\end{subfigure}

\vspace{5pt}
\caption*{(d) JS-tests with significance level $0.05$}
\begin{subfigure} {0.32\textwidth}
\caption{(d.1)  $n=200$}
\includegraphics[width=\textwidth]{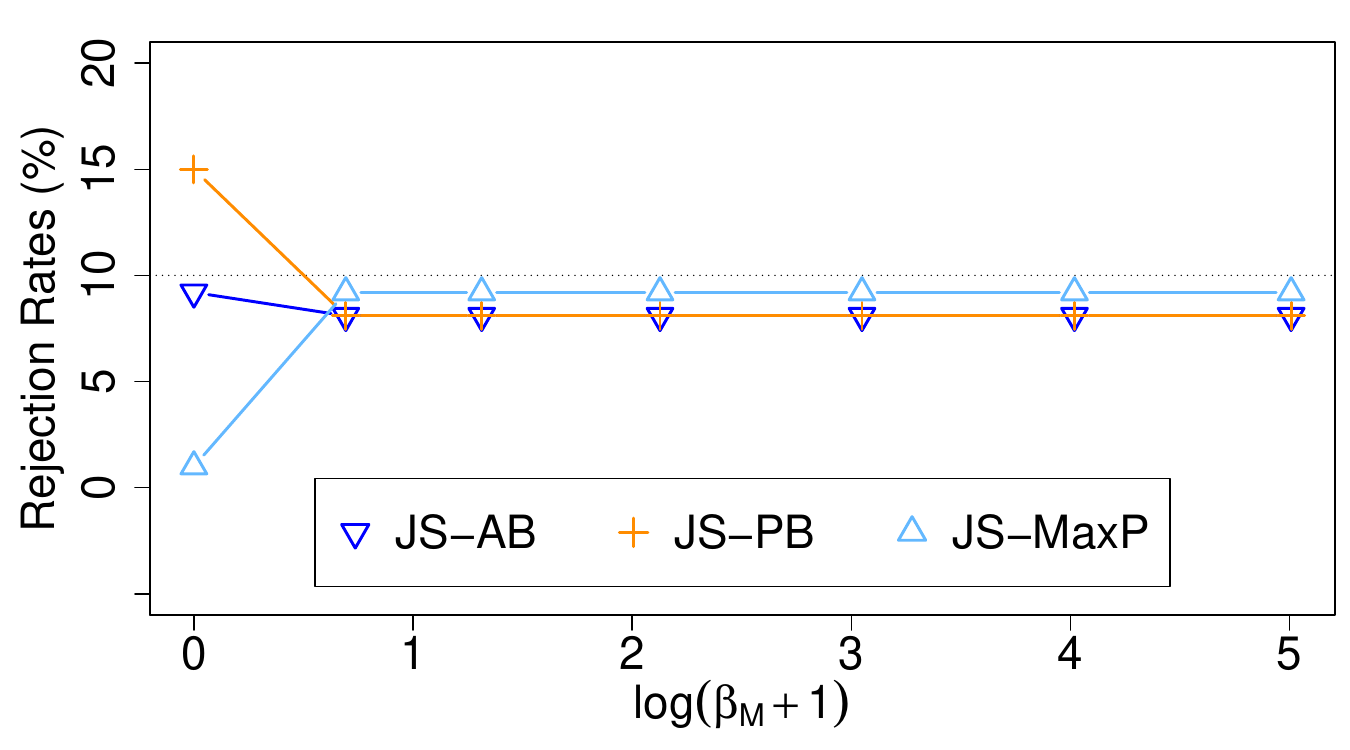} 
\end{subfigure}
\begin{subfigure} {0.32\textwidth}
\caption{(d.2)  $n=500$}
\includegraphics[width=\textwidth]{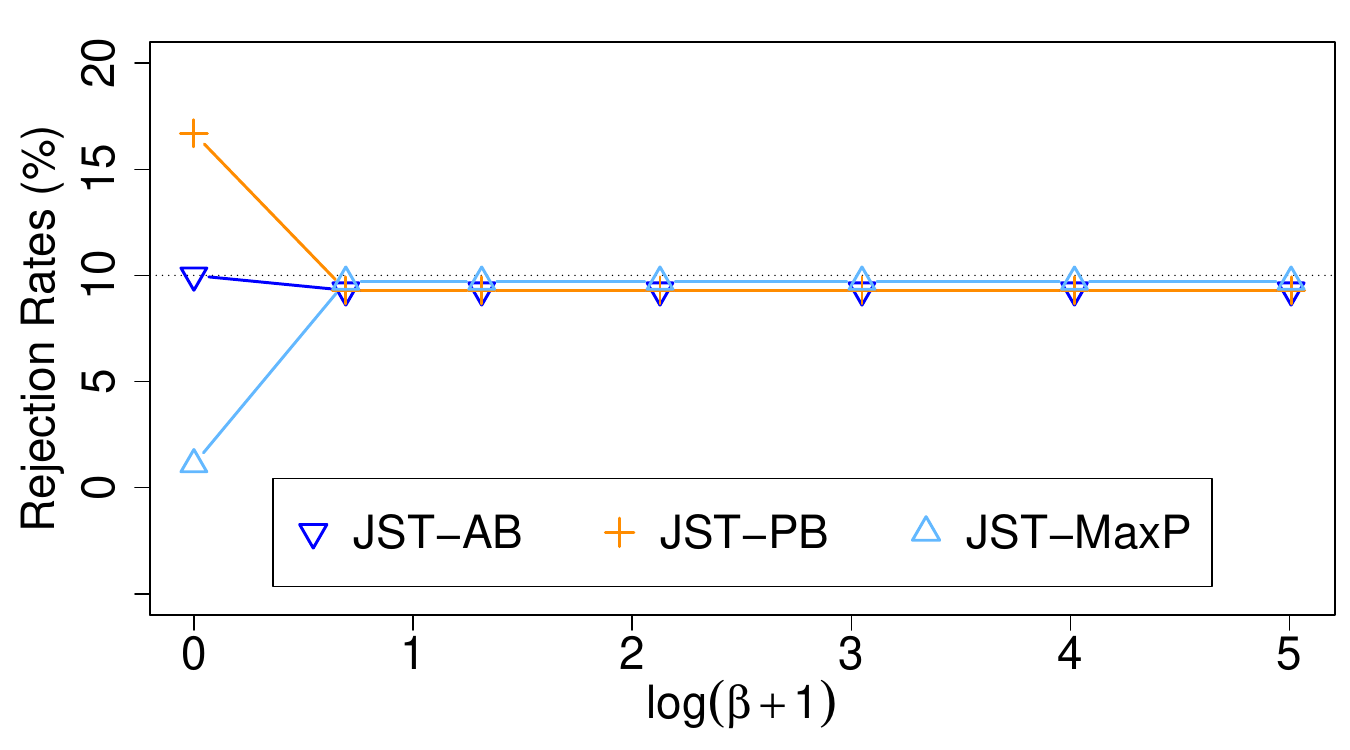} 
\end{subfigure}
\begin{subfigure} {0.32\textwidth}
\caption{(d.3)  $n=1000$}
\includegraphics[width=\textwidth]{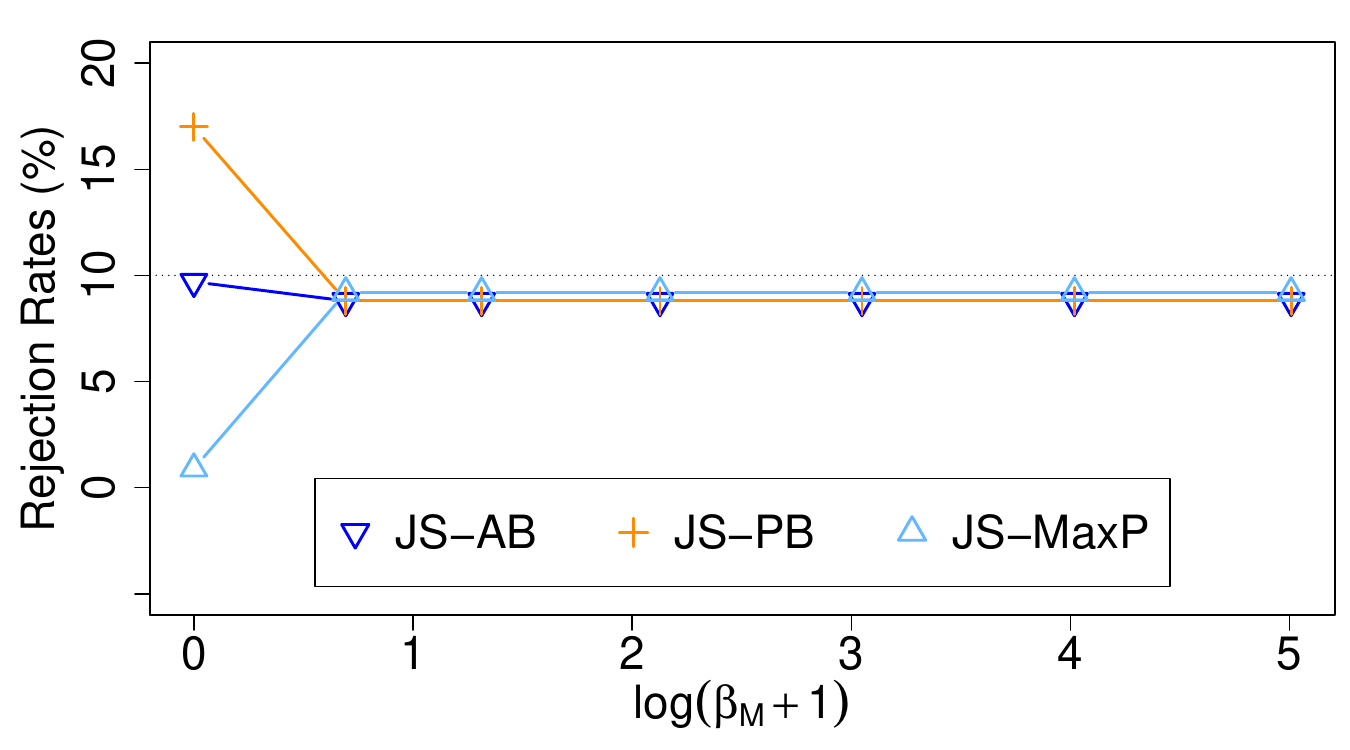} 
\end{subfigure}
\end{figure} 

\newpage

\begin{figure}[!htbp]
%\graphicspath{{figures/for_comments/R1_C5_Large_Effects/results/}}
\captionsetup[subfigure]{labelformat=empty}
\centering
\caption{When $\betaM=0$, estimated type-I errors} \label{fig:typeibeta0}
\vspace{5pt}
\caption*{(a)  PoC-tests with significance level 0.05}
\begin{subfigure} {0.32\textwidth}
\caption{(a.1) $n=200$}
\includegraphics[width=\textwidth]{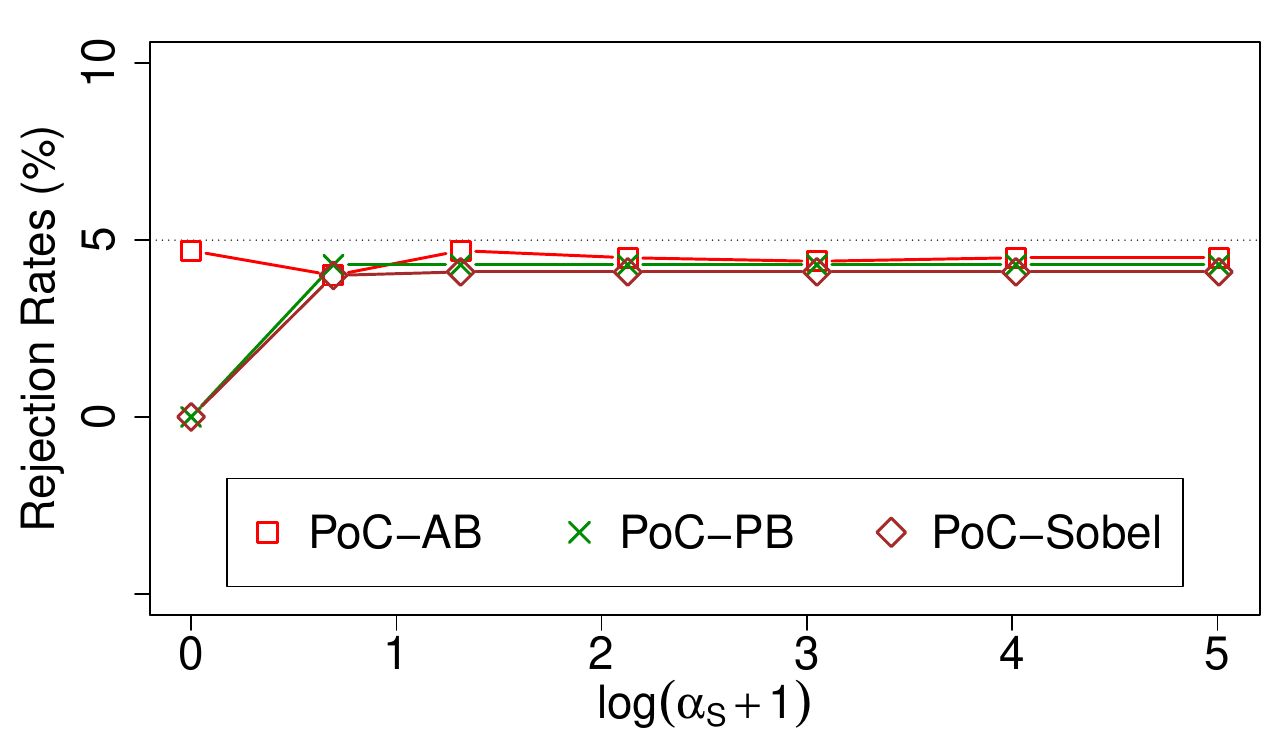} 
\end{subfigure} 
\begin{subfigure} {0.32\textwidth}
\caption{(a.2) $n=500$}
\includegraphics[width=\textwidth]{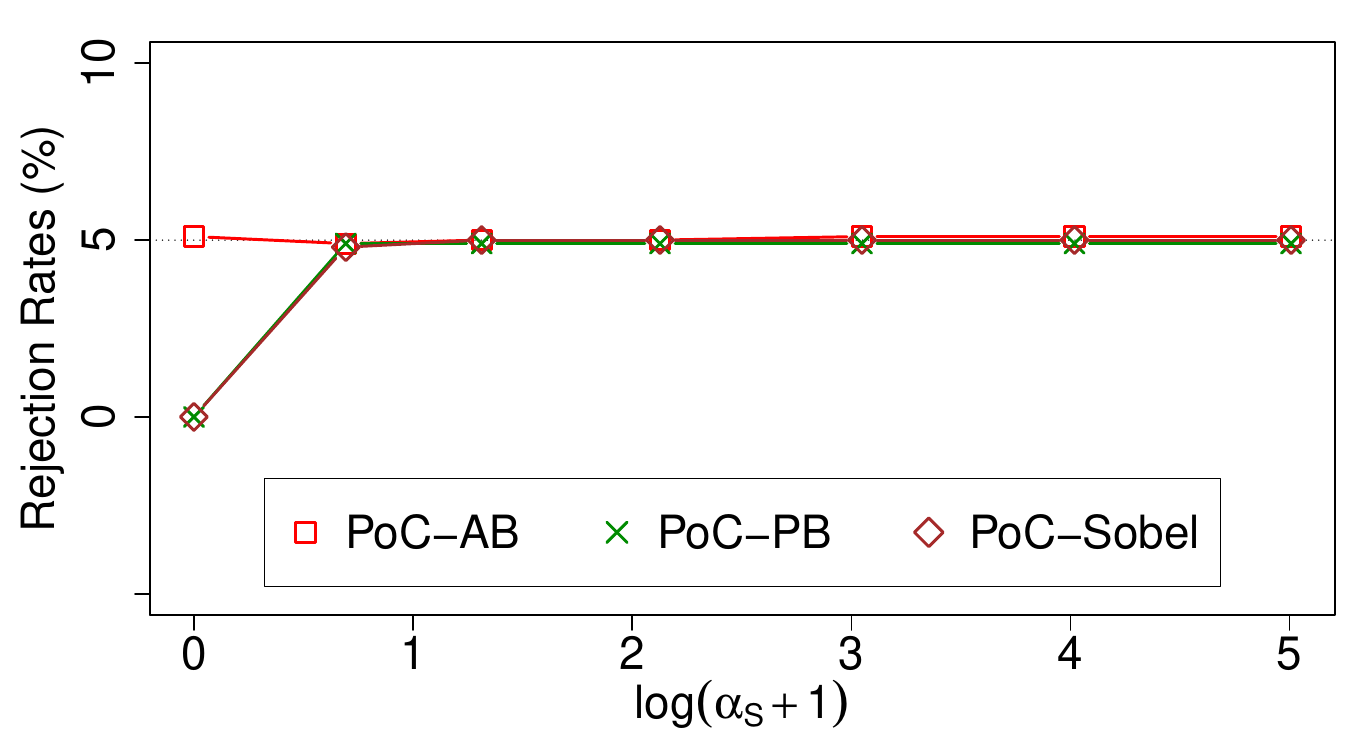} 
\end{subfigure} 
\begin{subfigure} {0.32\textwidth}
\caption{(a.3) $n=1000$}
\includegraphics[width=\textwidth]{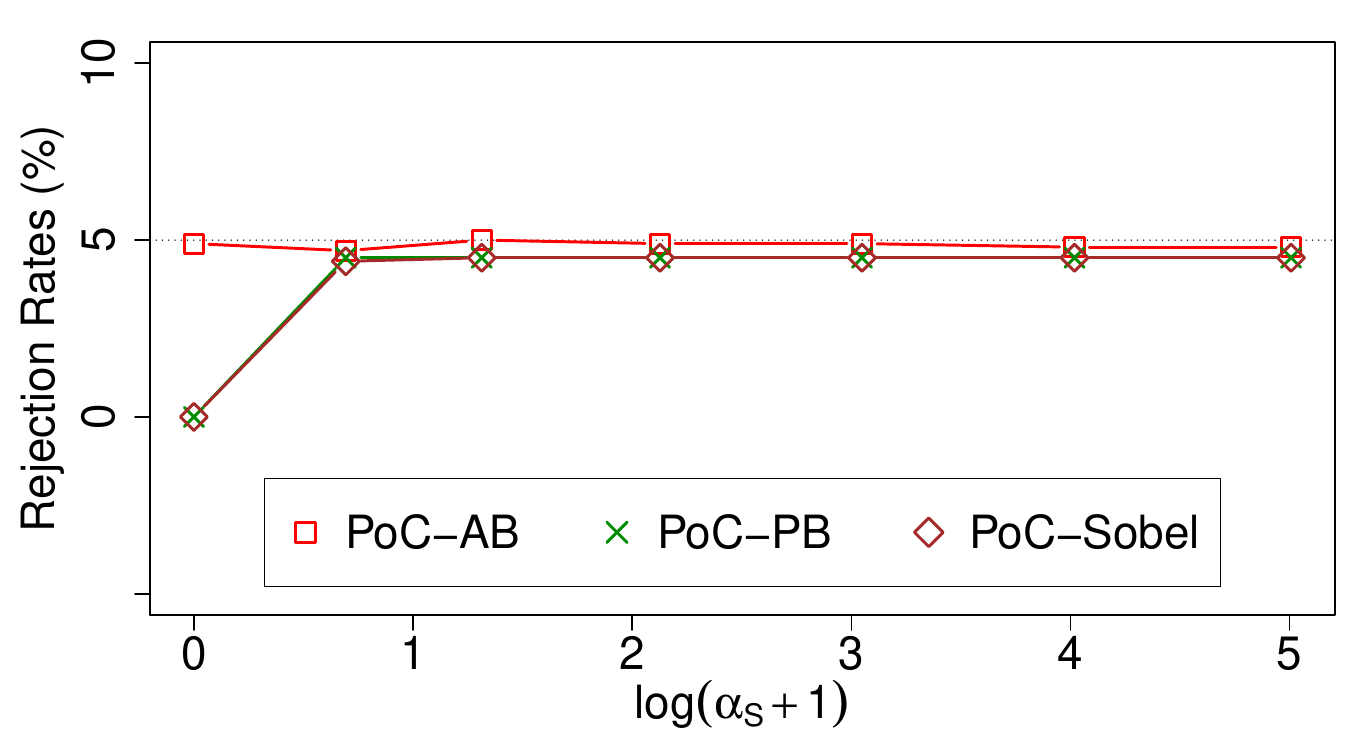} 
\end{subfigure} 

\vspace{5pt}
\caption*{(b) PoC-tests with significance level 0.1}
\begin{subfigure} {0.32\textwidth}
\caption{(b.1) $n=200$}
\includegraphics[width=\textwidth]{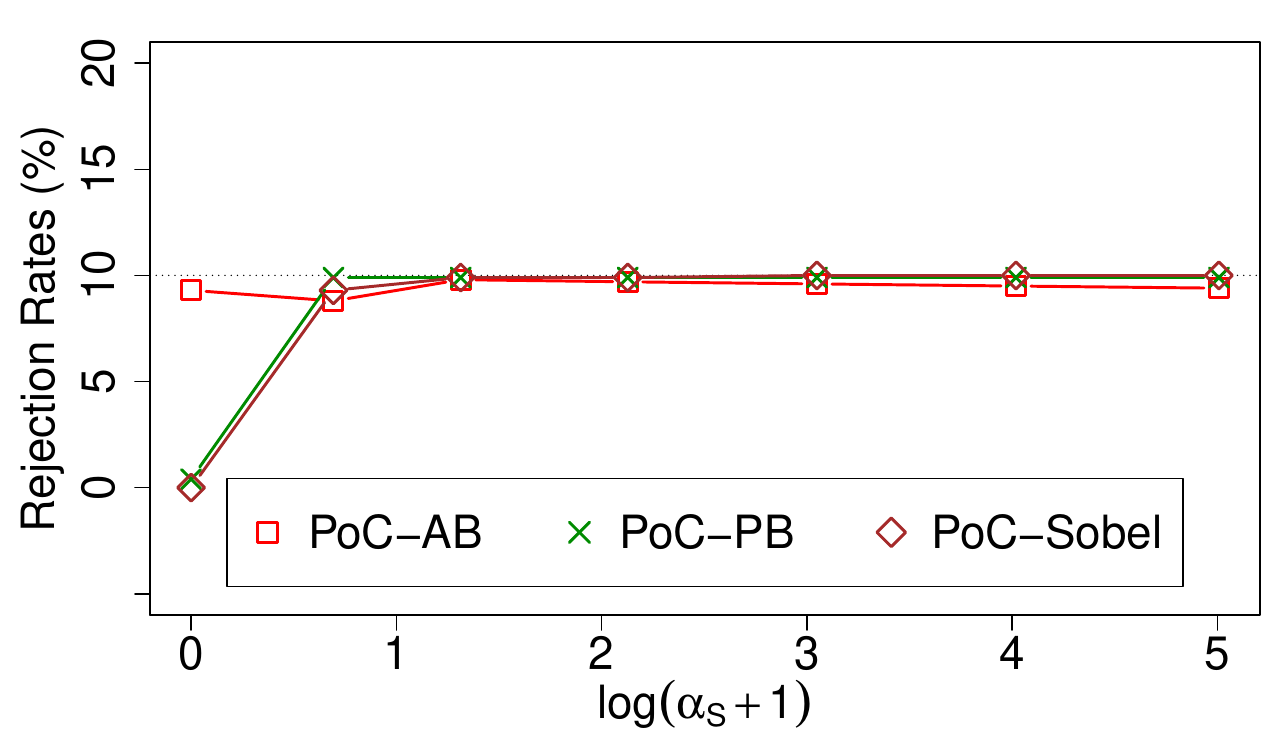} 
\end{subfigure} 
\begin{subfigure} {0.32\textwidth}
\caption{(b.2) $n=500$}
\includegraphics[width=\textwidth]{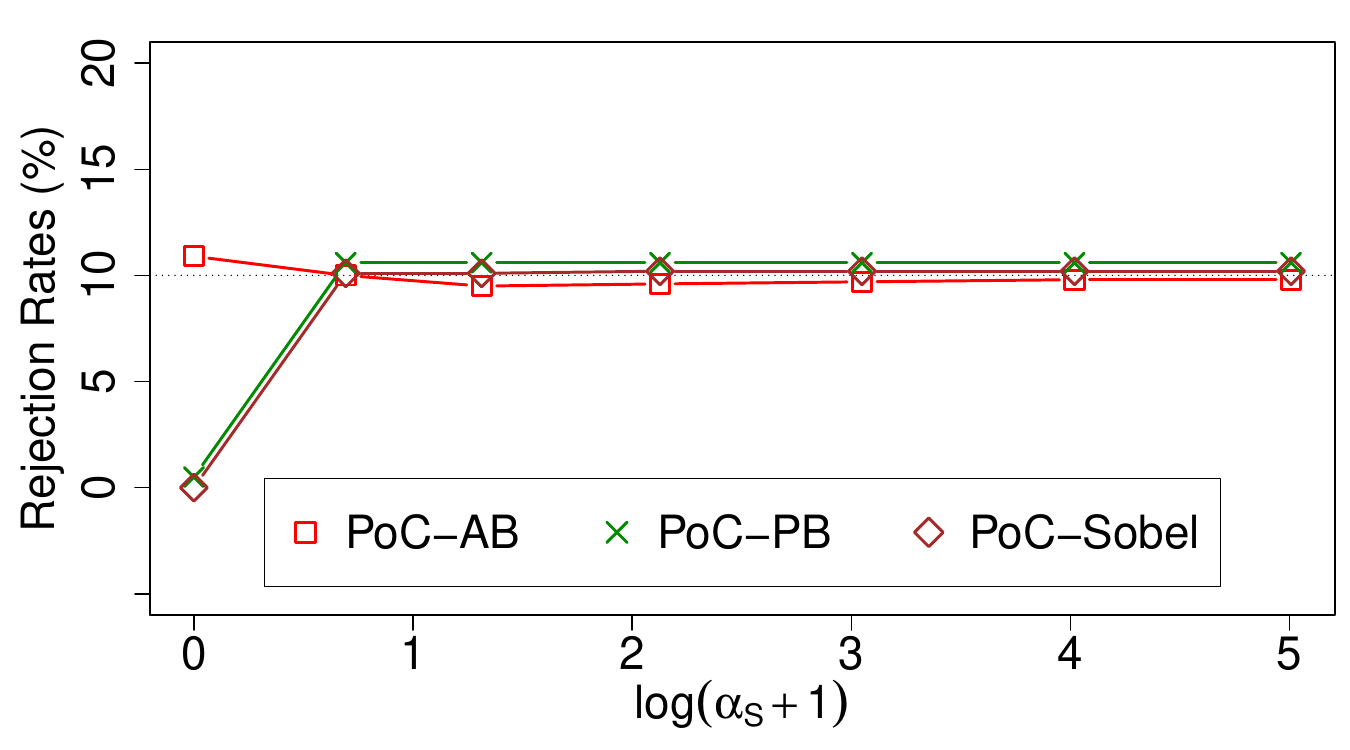} 
\end{subfigure} 
\begin{subfigure} {0.32\textwidth}
\caption{(b.3) $n=1000$}
\includegraphics[width=\textwidth]{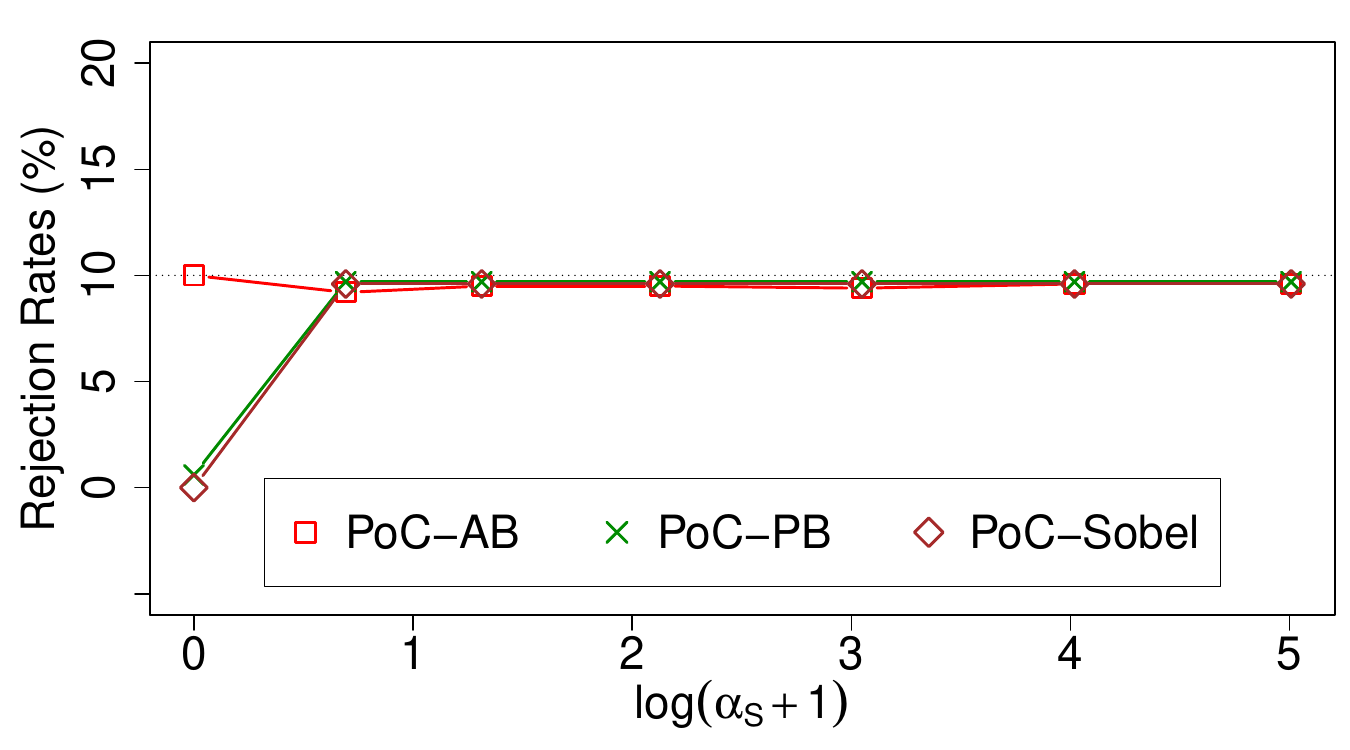} 
\end{subfigure} 

\vspace{5pt}
\caption*{(c) JS-tests with significance level $0.05$}
\begin{subfigure} {0.32\textwidth}
\caption{(c.1)  $n=200$}
\includegraphics[width=\textwidth]{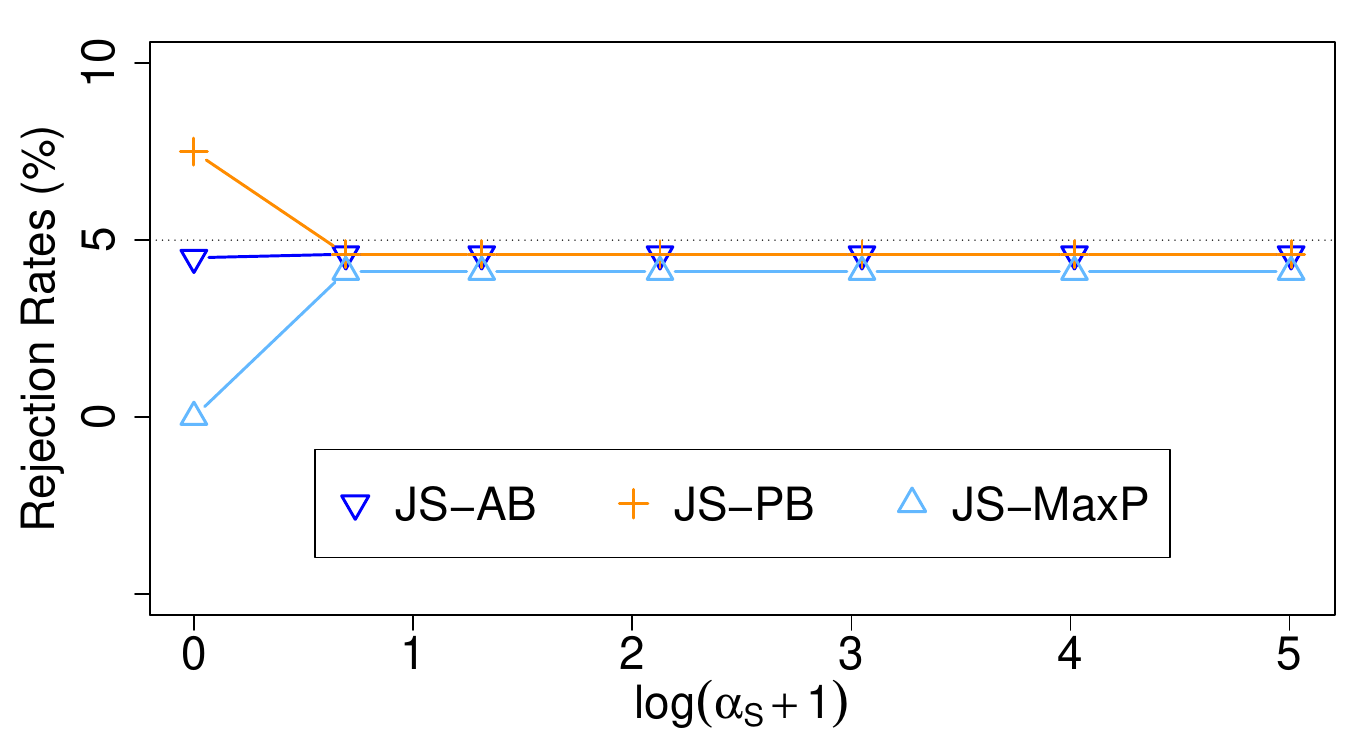} 
\end{subfigure}
\begin{subfigure} {0.32\textwidth}
\caption{(c.2)  $n=500$}
\includegraphics[width=\textwidth]{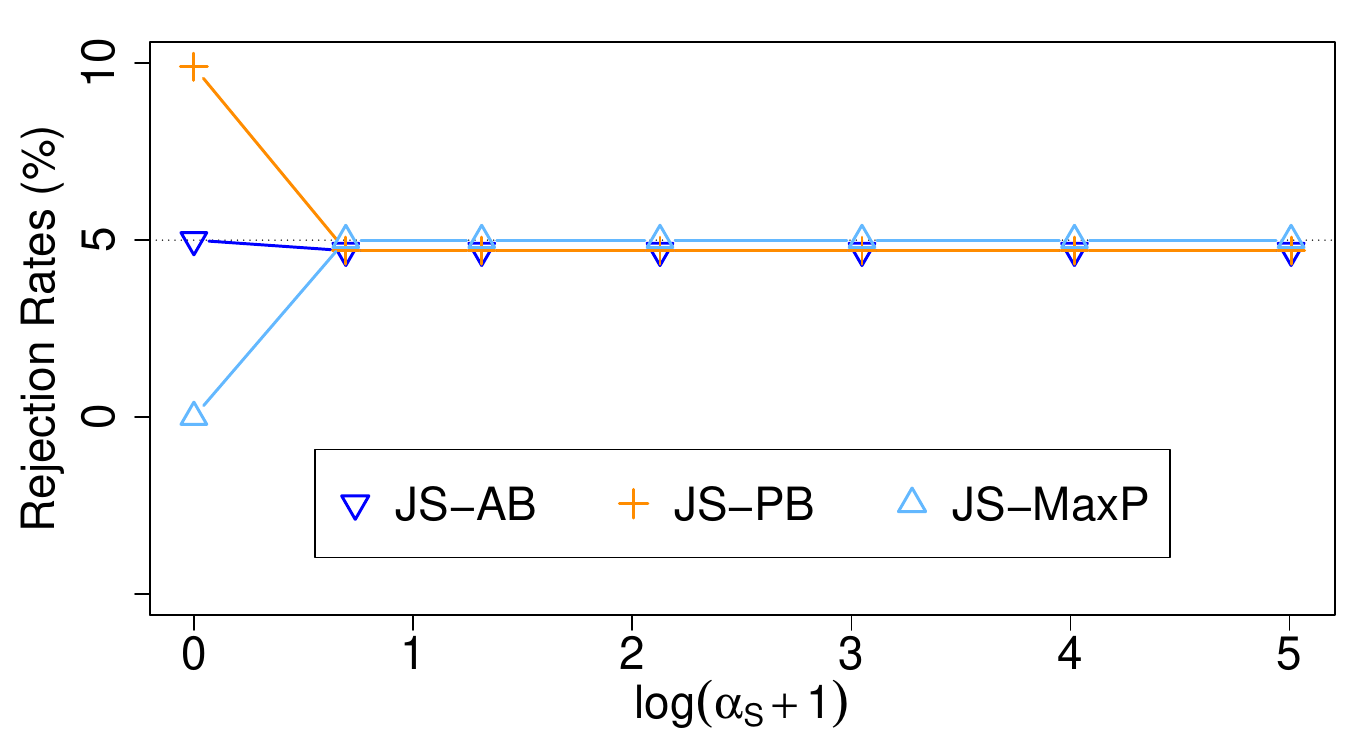} 
\end{subfigure}
\begin{subfigure} {0.32\textwidth}
\caption{(c.3)  $n=1000$}
\includegraphics[width=\textwidth]{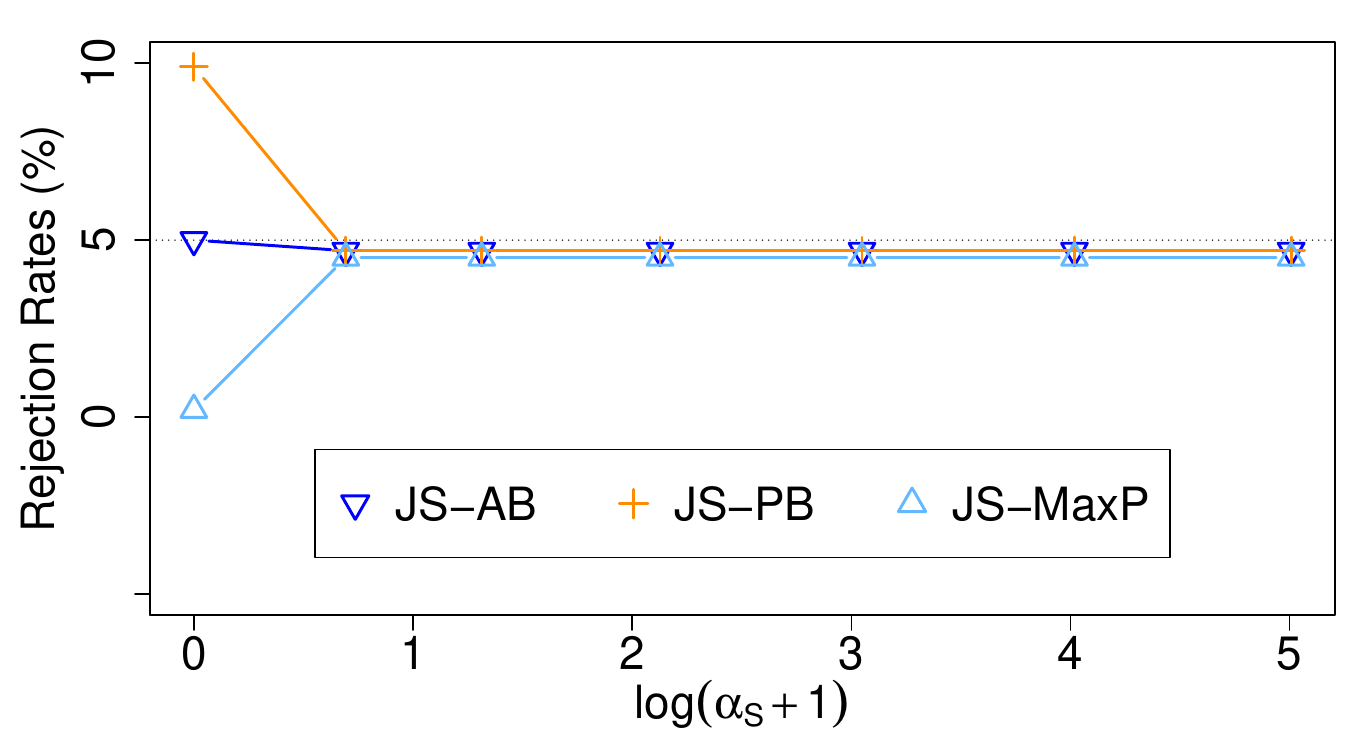} 
\end{subfigure}

\vspace{5pt}
\caption*{(d) JS-tests with significance level $0.1$}
\begin{subfigure} {0.32\textwidth}
\caption{(d.1)  $n=200$}
\includegraphics[width=\textwidth]{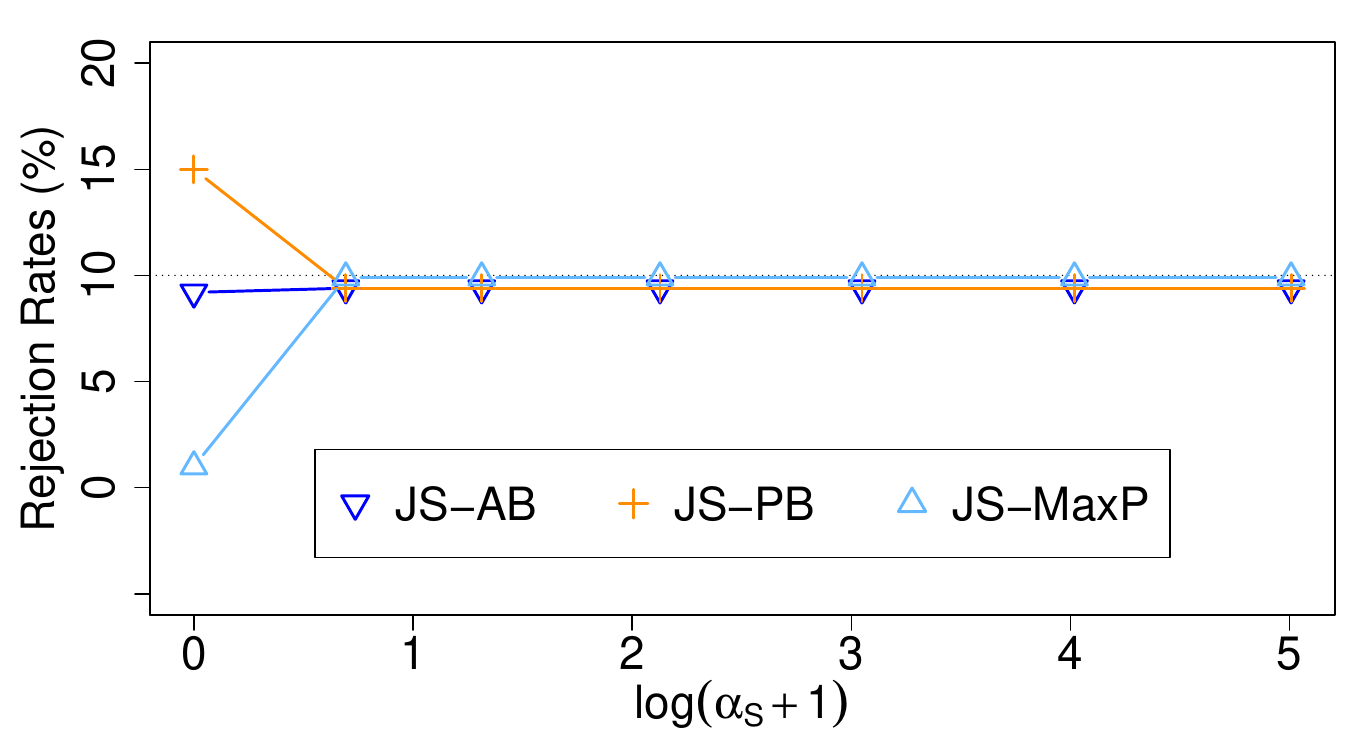} 
\end{subfigure}
\begin{subfigure} {0.32\textwidth}
\caption{(d.2)  $n=500$}
\includegraphics[width=\textwidth]{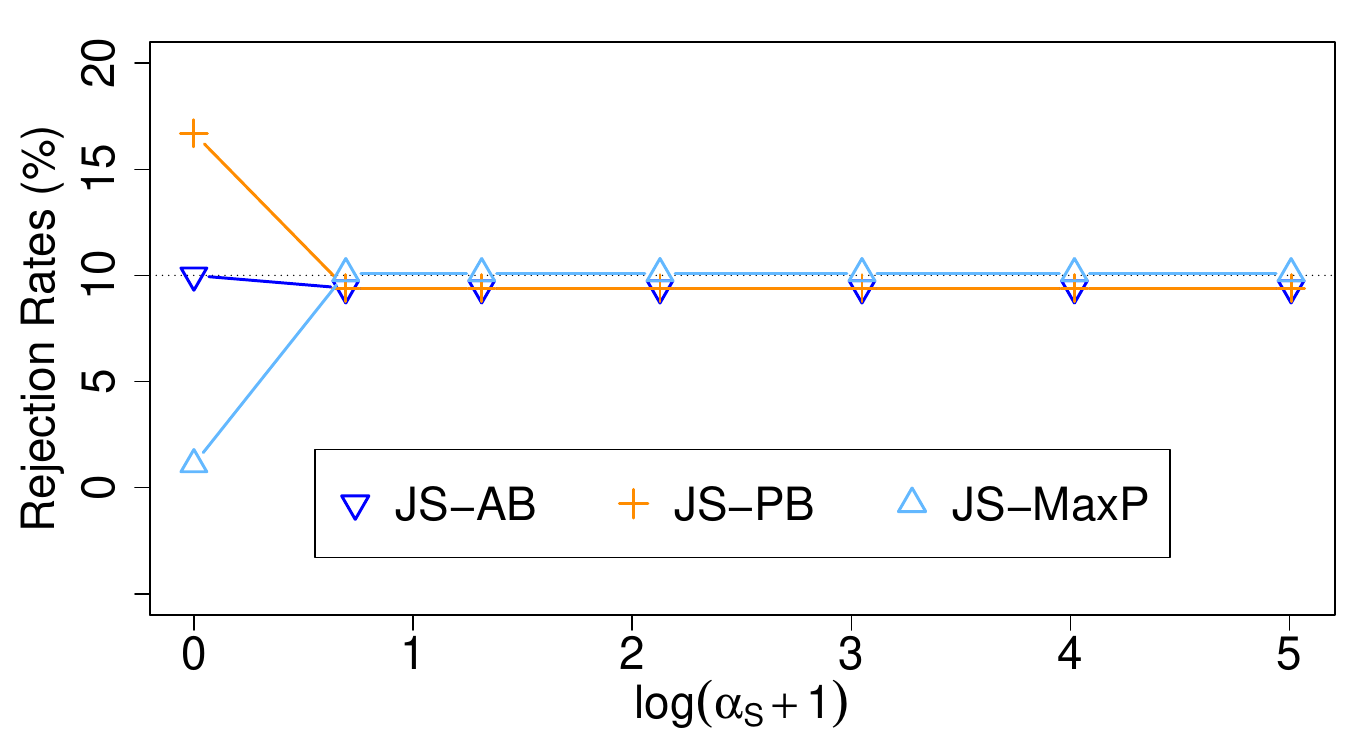} 
\end{subfigure}
\begin{subfigure} {0.32\textwidth}
\caption{(d.3)  $n=1000$}
\includegraphics[width=\textwidth]{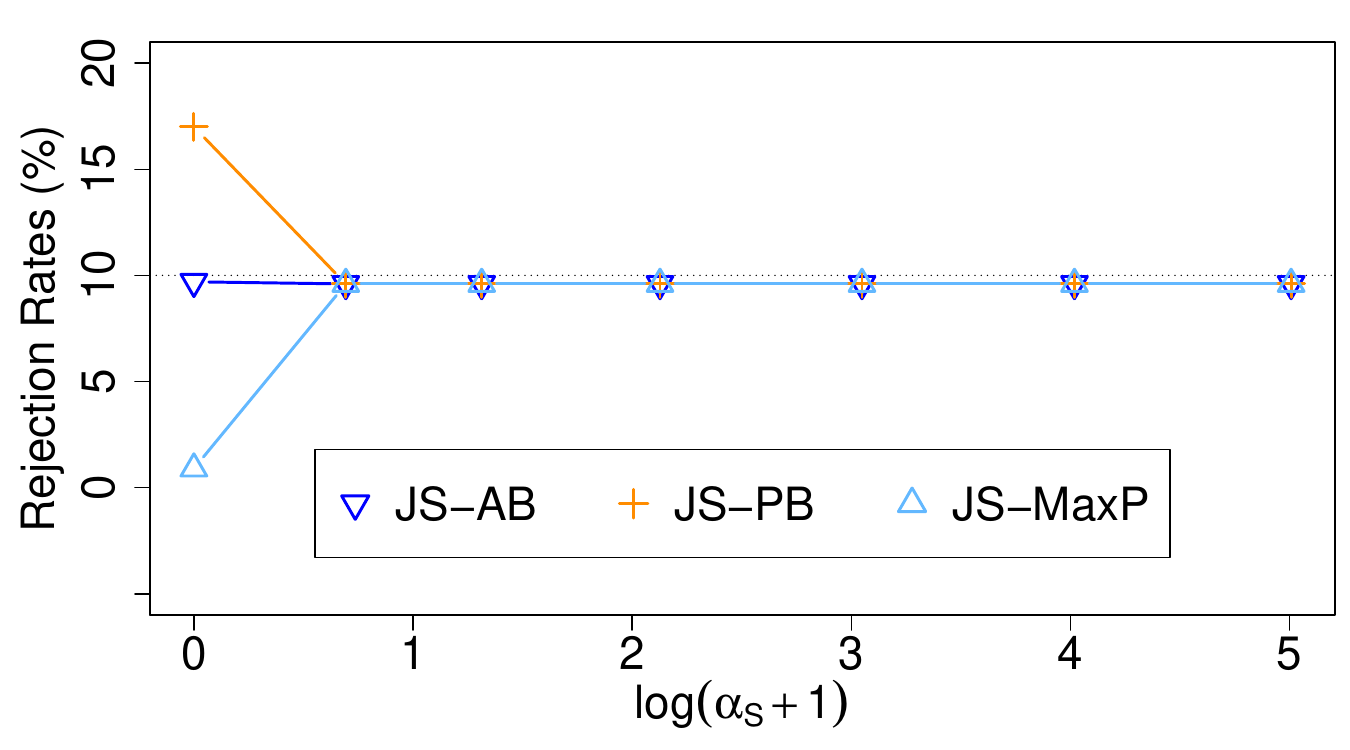} 
\end{subfigure}
\end{figure}

\newpage

\subsection{Data Analysis: Supplementary Results}\label{sec:splitanalysis}

\subsubsection{Additional Two-Step Data Analysis Results}\label{sec:changescreen}
In Section \ref{sec:datanalysis}, we conducted a two-step data analysis  by retaining 10\% of lipids with the smallest $p$-values in the first step.  
In the following, we extend our analysis  by varying the proportion ($q\%$) of lipids retained in the first step. 
 Figures \ref{fig:thres8}--\ref{fig:thres38} present the results when $q\in \{5,10,15,20,25\}$, corresponding to 8, 15, 22, 30, and 38 lipids  with the smallest $p$-values retained after the initial screening. 
Each figure showcases the top five mediators most frequently selected  over 400 random splits and six tests. 
Notably, regardless of the screening percentage, L.A and FA.7 consistently emerge as the most frequently selected mediators. This suggests that our results are robust to 
  the specific choice of screening threshold used in the first step. 

\vspace{2em}
\begin{figure}[!htbp]
%\graphicspath{{figures/figures_r2/E_split/data_split_1_3_top8_2/}}
 \captionsetup[subfigure]{labelformat=empty}
    \centering
      \begin{subfigure}[b]{0.32\textwidth}
      \caption{JS-AB}
      \includegraphics[width=0.95\textwidth]{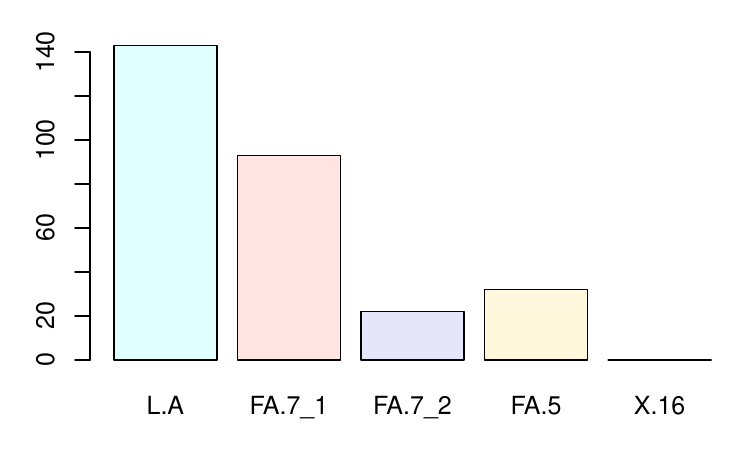}
   \end{subfigure}    
      \begin{subfigure}[b]{0.32\textwidth}
      \caption{JS-MaxP}
      \includegraphics[width=0.95\textwidth]{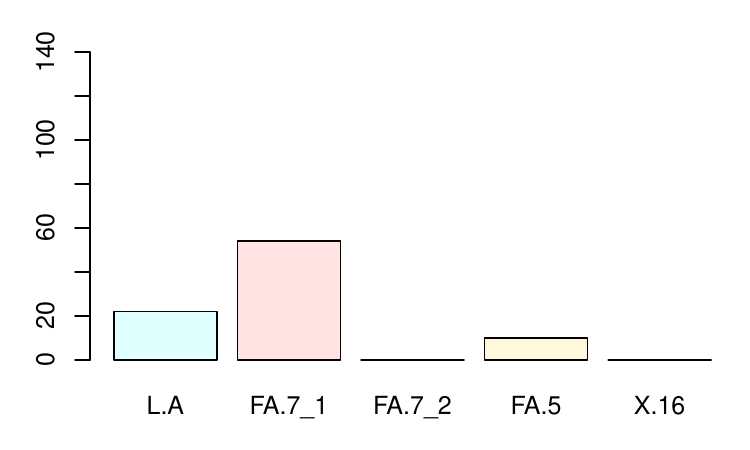}
   \end{subfigure}    
       \begin{subfigure}[b]{0.32\textwidth}
        \caption{CMA}
      \includegraphics[width=0.95\textwidth]{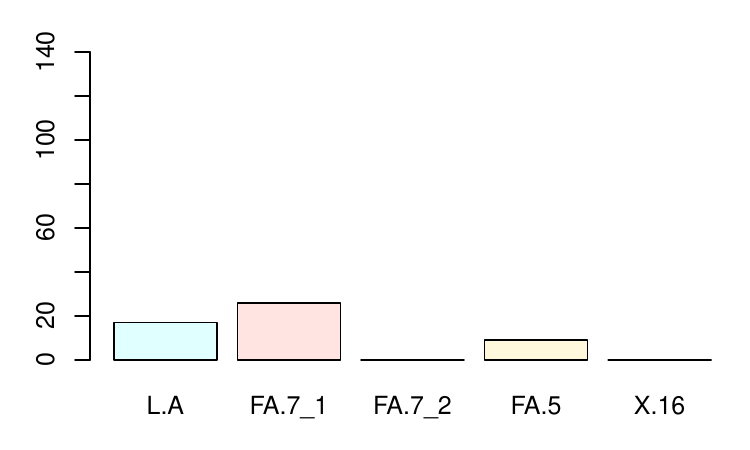}
   \end{subfigure}   
   \begin{subfigure}[b]{0.32\textwidth}
   \caption{PoC-AB}
      \includegraphics[width=0.95\textwidth]{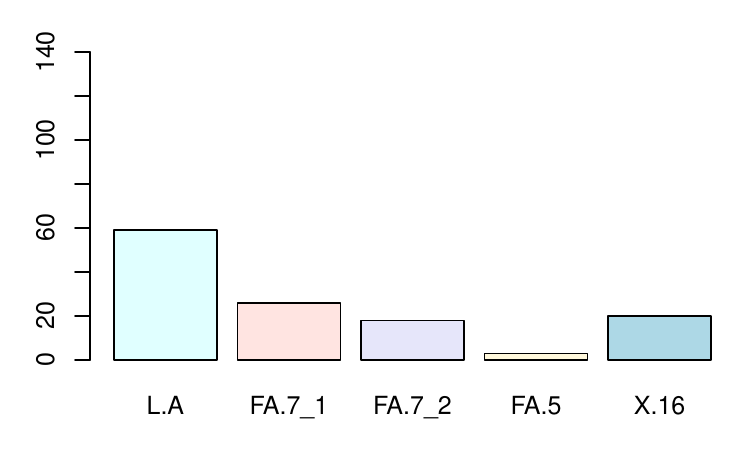}
   \end{subfigure}    
      \begin{subfigure}[b]{0.32\textwidth}
      \caption{PoC-B}
      \includegraphics[width=0.95\textwidth]{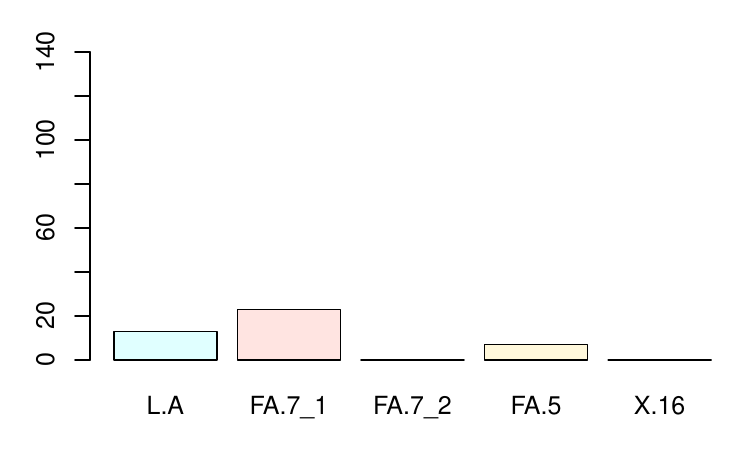}
   \end{subfigure}    
      \begin{subfigure}[b]{0.32\textwidth}
      \caption{PoC-Sobel}
      \includegraphics[width=0.95\textwidth]{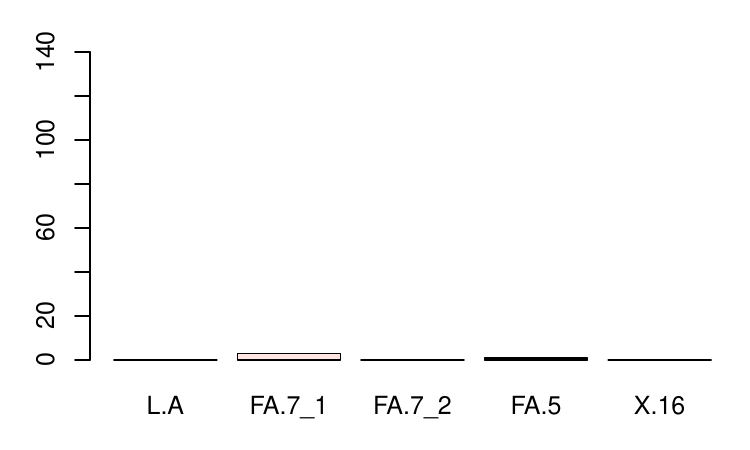}
   \end{subfigure}    
 \caption{Times of mediators being selected in Step 2 by the six tests when FDR$=0.10$ over 400 random splits of the data.  Keep 5\% of lipids with the smallest $p$-values in Step 1. (Abbreviations: LAURIC.ACID (L.A); FA.7.0-OH\_1 (FA.7\_1); FA.5.0-OH (FA.5); FA.7.0-OH\_2 (FA.7\_2); X16.0.LYSO.PC\_2 (X.16).} \label{fig:thres8}
\end{figure}

\begin{figure}[!htbp]
 \captionsetup[subfigure]{labelformat=empty}
    \centering
      \begin{subfigure}[b]{0.32\textwidth}
      \caption{JS-AB}
      \includegraphics[width=0.95\textwidth]{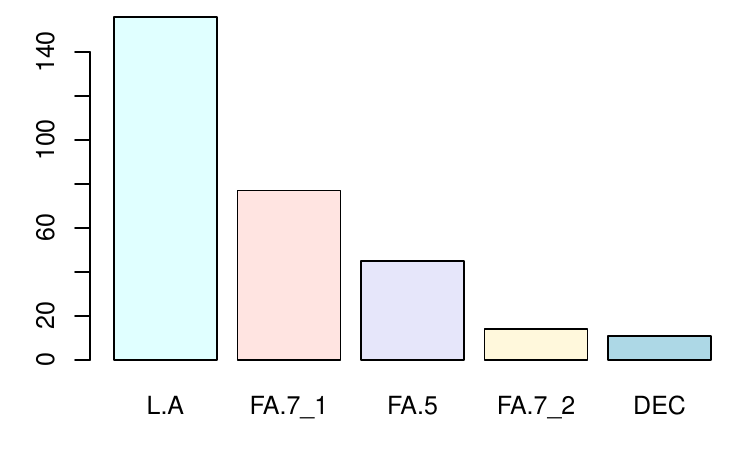}
   \end{subfigure}    
      \begin{subfigure}[b]{0.32\textwidth}
      \caption{JS-MaxP}
      \includegraphics[width=0.95\textwidth]{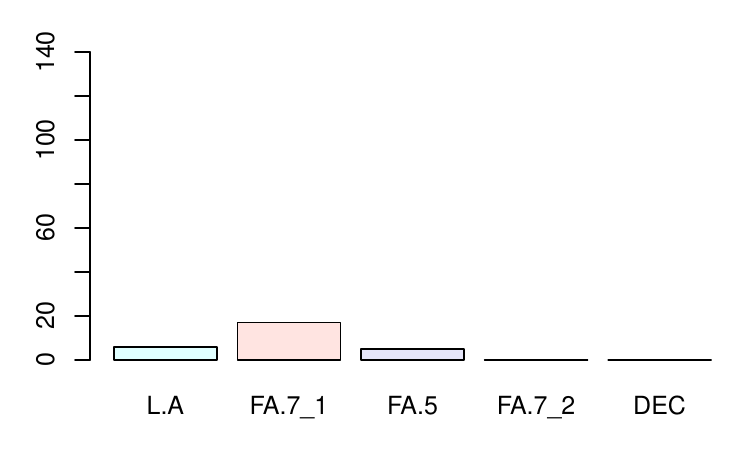}
   \end{subfigure}    
       \begin{subfigure}[b]{0.32\textwidth}
        \caption{CMA}
      \includegraphics[width=0.95\textwidth]{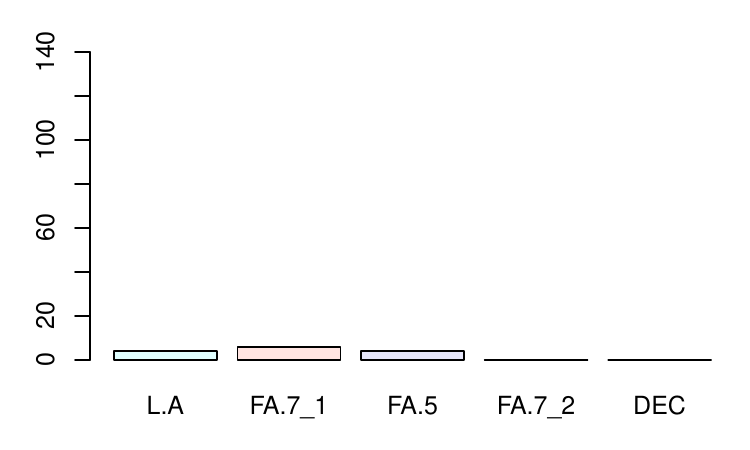}
   \end{subfigure}   
   \begin{subfigure}[b]{0.32\textwidth}
   \caption{PoC-AB}
      \includegraphics[width=0.95\textwidth]{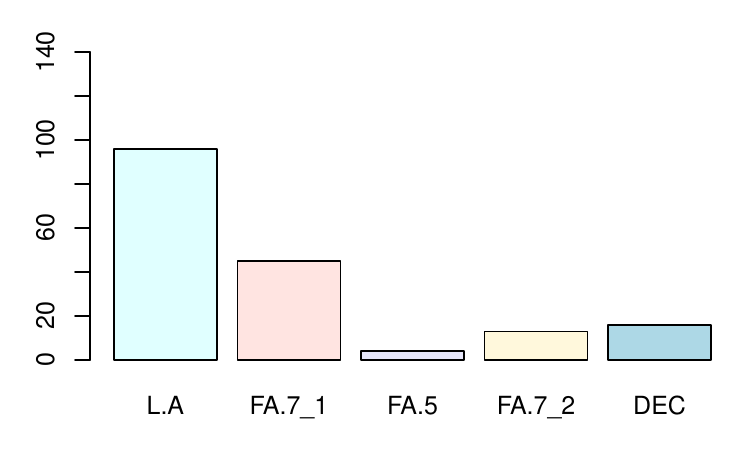}
   \end{subfigure}    
      \begin{subfigure}[b]{0.32\textwidth}
      \caption{PoC-B}
      \includegraphics[width=0.95\textwidth]{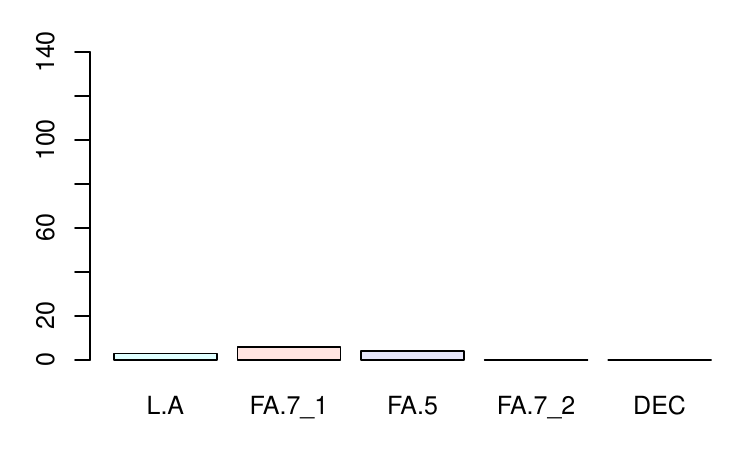}
   \end{subfigure}    
      \begin{subfigure}[b]{0.32\textwidth}
      \caption{PoC-Sobel}
      \includegraphics[width=0.95\textwidth]{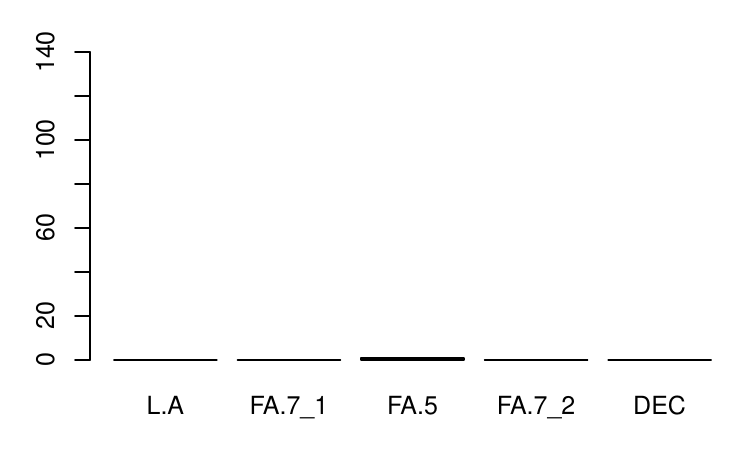}
   \end{subfigure}    
\caption{ Times of mediators being selected in Step 2 by the six tests when FDR$=0.10$ over 400 random splits of the data.  Keep 10\% of lipids with the smallest $p$-values in Step 1. (Abbreviations: LAURIC.ACID (L.A); FA.7.0-OH\_1 (FA.7\_1); FA.5.0-OH (FA.5); FA.7.0-OH\_2 (FA.7\_2); DECENOYLCARNITINE (DEC).} \label{fig:thres15}
\end{figure}

\begin{figure}[!htbp]
%\graphicspath{{figures/figures_r2/E_split/data_split_1_3_top22_2/}}
 \captionsetup[subfigure]{labelformat=empty}
    \centering
      \begin{subfigure}[b]{0.32\textwidth}
      \caption{JS-AB}
      \includegraphics[width=0.95\textwidth]{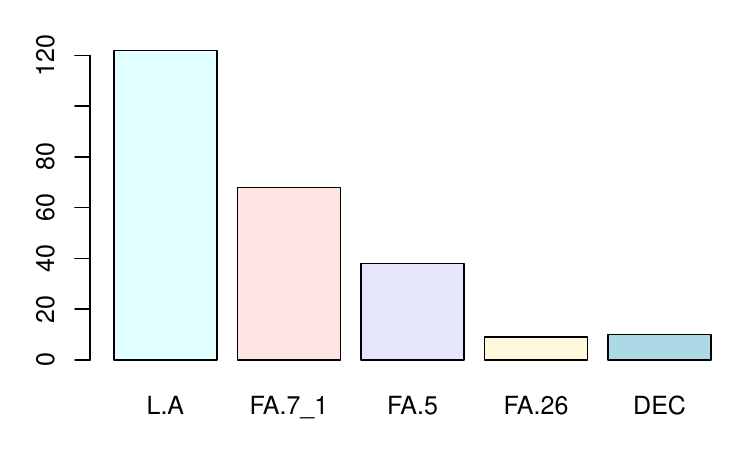}
   \end{subfigure}    
      \begin{subfigure}[b]{0.32\textwidth}
      \caption{JS-MaxP}
      \includegraphics[width=0.95\textwidth]{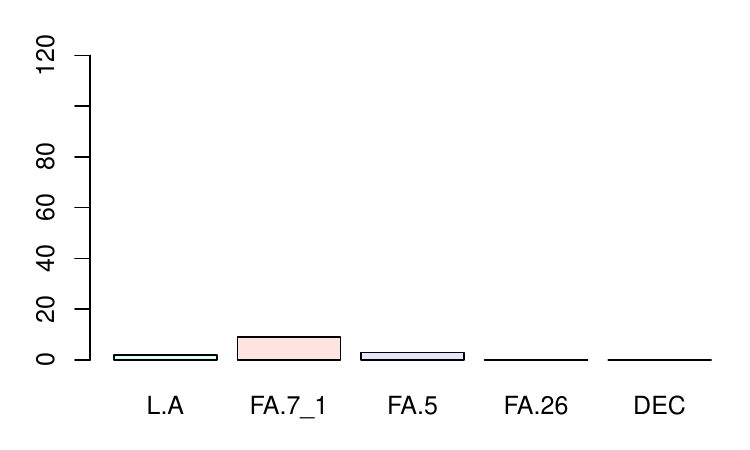}
   \end{subfigure}    
       \begin{subfigure}[b]{0.32\textwidth}
        \caption{CMA}
      \includegraphics[width=0.95\textwidth]{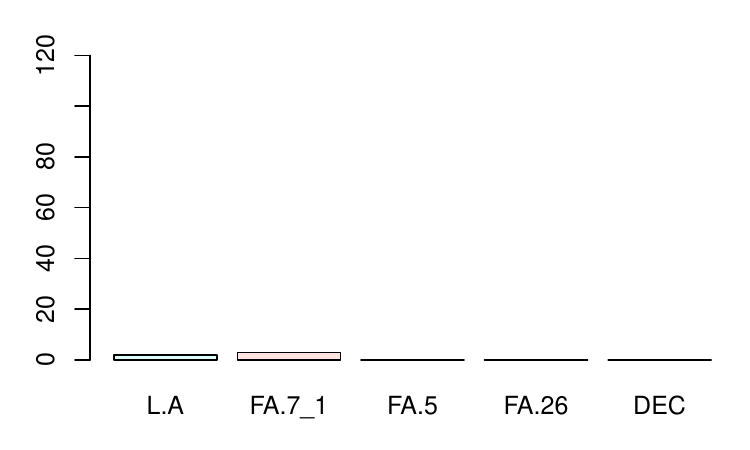}
   \end{subfigure}   
   \begin{subfigure}[b]{0.32\textwidth}
   \caption{PoC-AB}
      \includegraphics[width=0.95\textwidth]{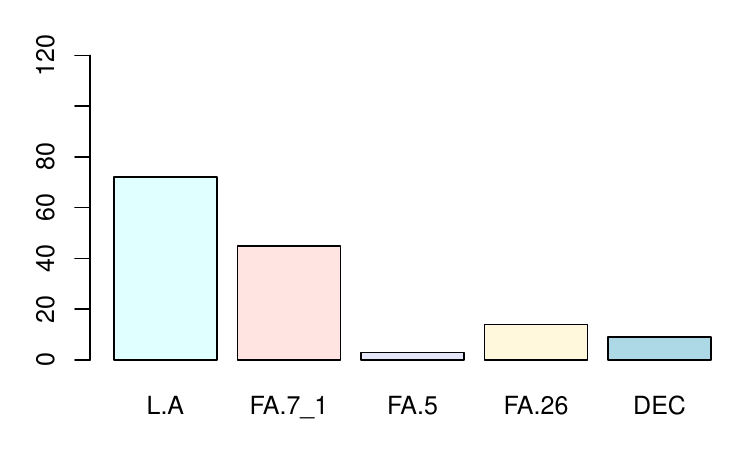}
   \end{subfigure}    
      \begin{subfigure}[b]{0.32\textwidth}
      \caption{PoC-B}
      \includegraphics[width=0.95\textwidth]{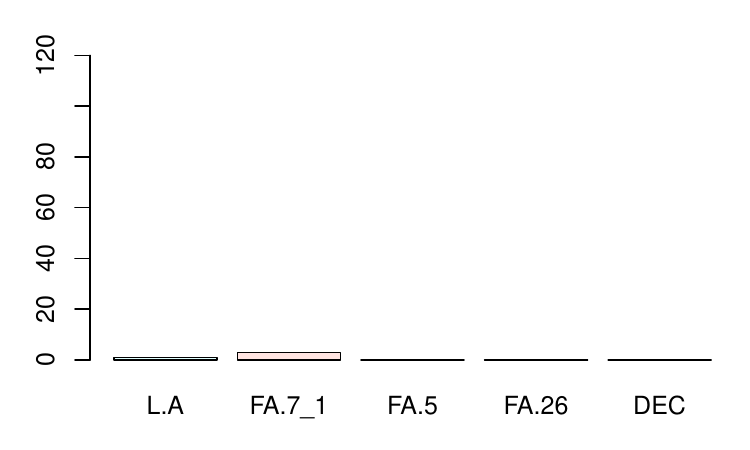}
   \end{subfigure}    
      \begin{subfigure}[b]{0.32\textwidth}
      \caption{PoC-Sobel}
      \includegraphics[width=0.95\textwidth]{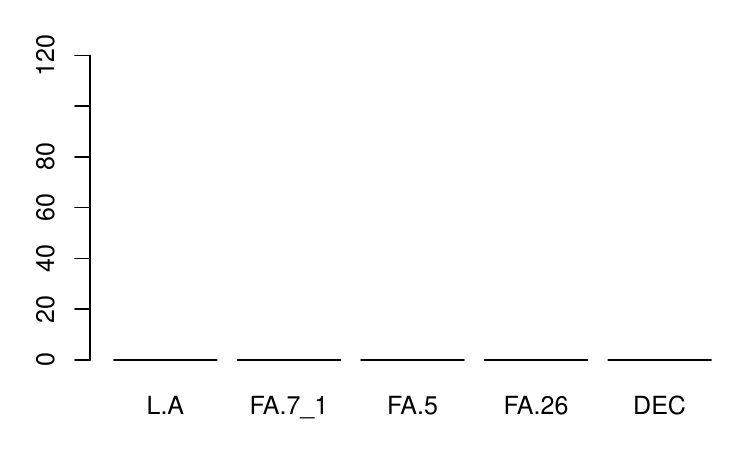}
   \end{subfigure}    
 \caption{ Times of mediators being selected in Step 2 by the six tests when FDR$=0.10$ over 400 random splits of the data.  Keep 15\% of lipids with the smallest $p$-values in Step 1. (Abbreviations: LAURIC.ACID (L.A); FA.7.0-OH\_1 (FA.7\_1); FA.5.0-OH (FA.5); FA.26 (FA.26.0-OH); DECENOYLCARNITINE (DEC).}  \label{fig:thres22}
\end{figure}

\begin{figure}[!htbp]
%\graphicspath{{figures/figures_r2/E_split/data_split_1_3_top30_2/}}
 \captionsetup[subfigure]{labelformat=empty}
    \centering
      \begin{subfigure}[b]{0.32\textwidth}
      \caption{JS-AB}
      \includegraphics[width=0.95\textwidth]{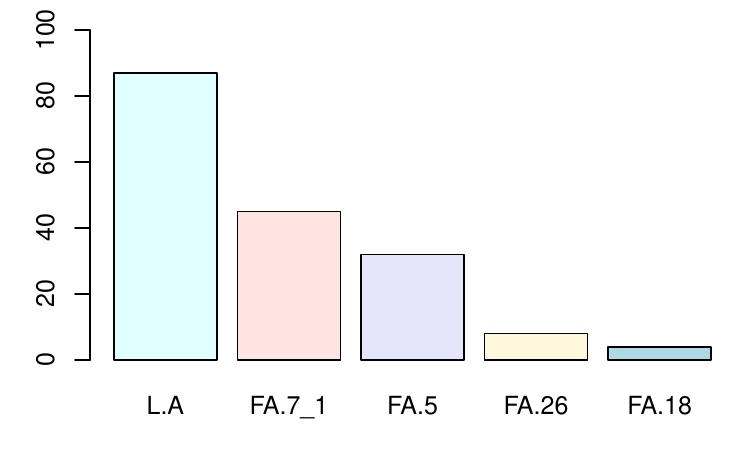}
   \end{subfigure}    
      \begin{subfigure}[b]{0.32\textwidth}
      \caption{JS-MaxP}
      \includegraphics[width=0.95\textwidth]{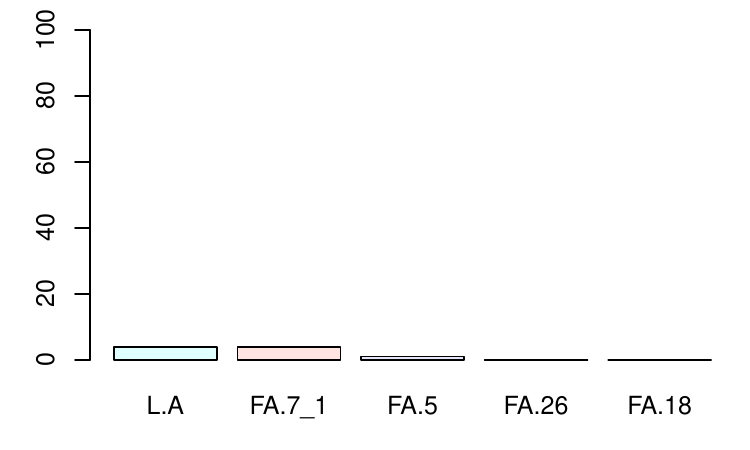}
   \end{subfigure}    
       \begin{subfigure}[b]{0.32\textwidth}
        \caption{CMA}
      \includegraphics[width=0.95\textwidth]{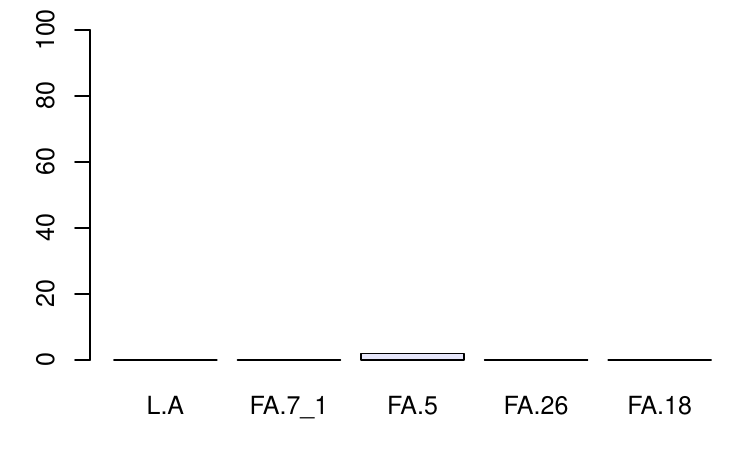}
   \end{subfigure}   
   \begin{subfigure}[b]{0.32\textwidth}
   \caption{PoC-AB}
      \includegraphics[width=0.95\textwidth]{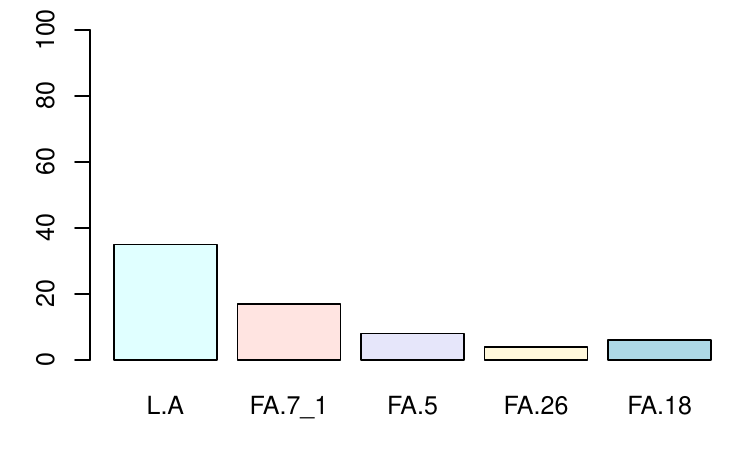}
   \end{subfigure}    
      \begin{subfigure}[b]{0.32\textwidth}
      \caption{PoC-B}
      \includegraphics[width=0.95\textwidth]{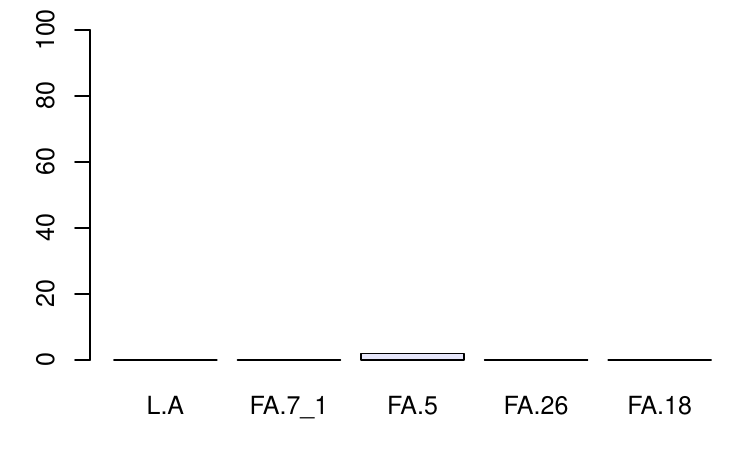}
   \end{subfigure}    
      \begin{subfigure}[b]{0.32\textwidth}
      \caption{PoC-Sobel}
      \includegraphics[width=0.95\textwidth]{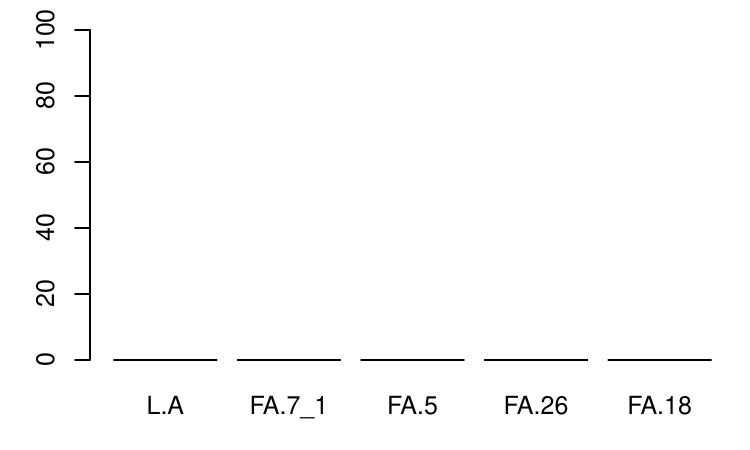}
   \end{subfigure}    
 \caption{ Times of mediators being selected in Step 2 by the six tests when FDR$=0.10$ over 400 random splits of the data.  Keep 20\% of lipids with the smallest $p$-values in Step 1. (Abbreviations: LAURIC.ACID (L.A); FA.7.0-OH\_1 (FA.7\_1); FA.5.0-OH (FA.5); FA.26 (FA.26.0-OH); FA.18 (FA.18.1\_2).}  \label{fig:thres30} 
\end{figure}

\begin{figure}[!htbp]
%\graphicspath{{figures/figures_r2/E_split/data_split_1_3_top38_2/}}
 \captionsetup[subfigure]{labelformat=empty}
    \centering
      \begin{subfigure}[b]{0.32\textwidth}
      \caption{JS-AB}
      \includegraphics[width=0.95\textwidth]{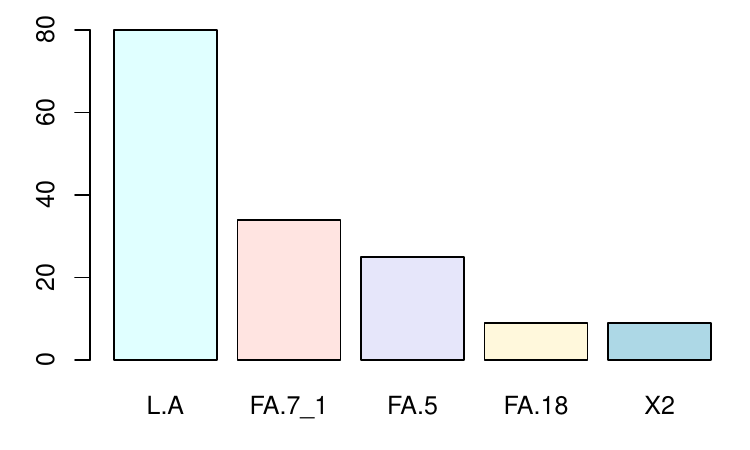}
   \end{subfigure}    
      \begin{subfigure}[b]{0.32\textwidth}
      \caption{JS-MaxP}
      \includegraphics[width=0.95\textwidth]{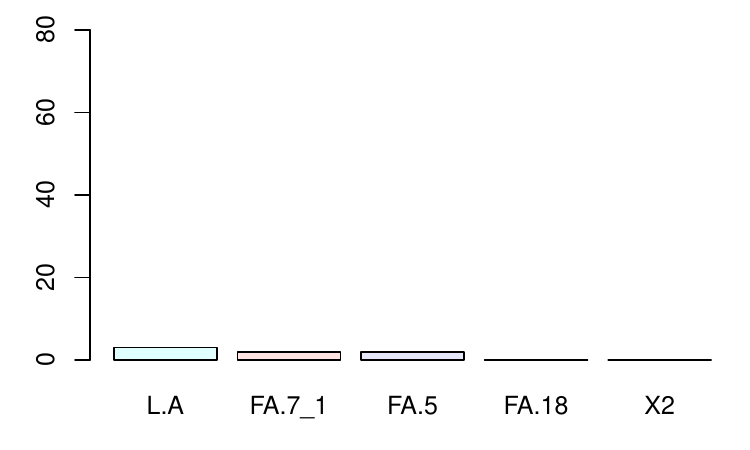}
   \end{subfigure}    
       \begin{subfigure}[b]{0.32\textwidth}
        \caption{CMA}
      \includegraphics[width=0.95\textwidth]{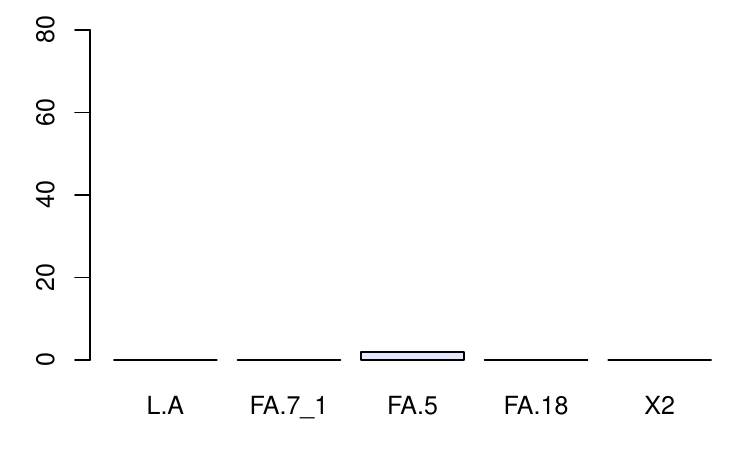}
   \end{subfigure}   
   \begin{subfigure}[b]{0.32\textwidth}
   \caption{PoC-AB}
      \includegraphics[width=0.95\textwidth]{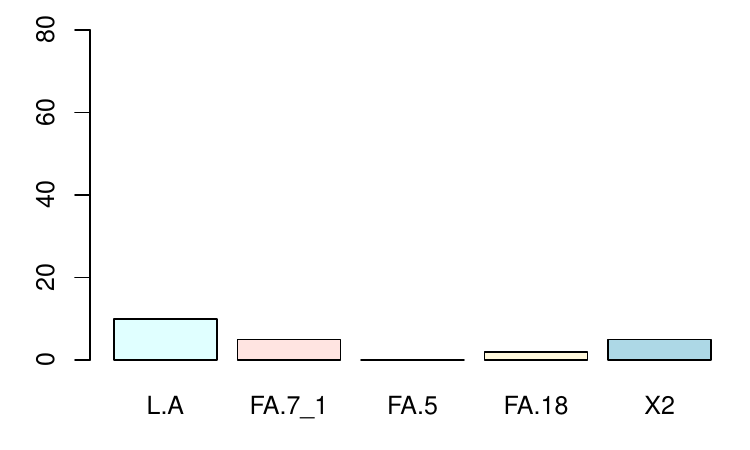}
   \end{subfigure}    
      \begin{subfigure}[b]{0.32\textwidth}
      \caption{PoC-B}
      \includegraphics[width=0.95\textwidth]{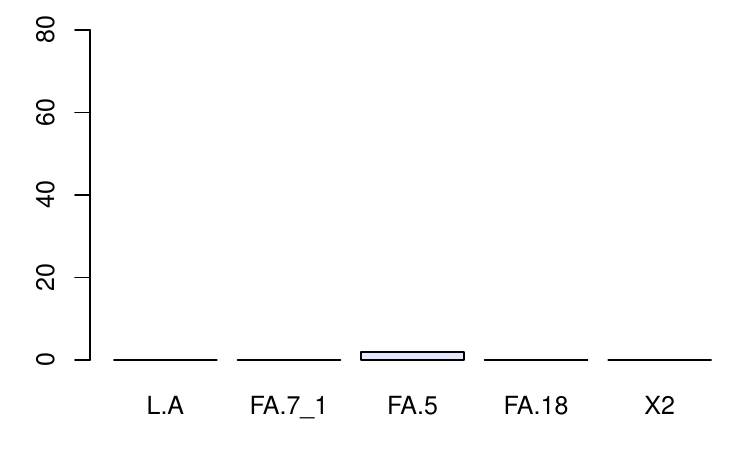}
   \end{subfigure}    
      \begin{subfigure}[b]{0.32\textwidth}
      \caption{PoC-Sobel}
      \includegraphics[width=0.95\textwidth]{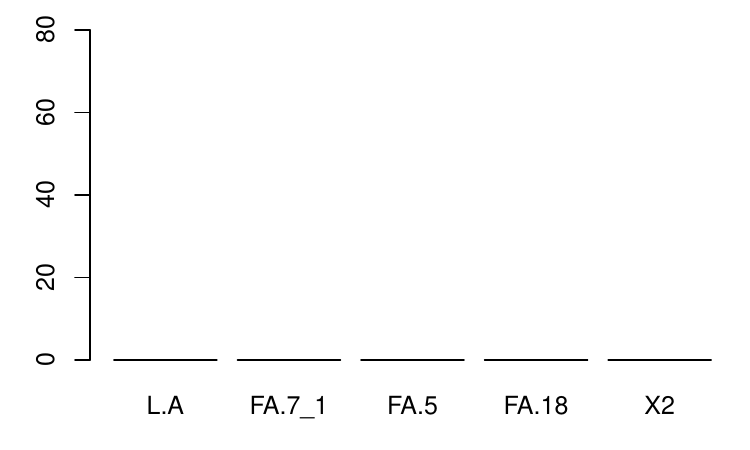}
   \end{subfigure}    
 \caption{ Times of mediators being selected in Step 2 by the six tests when FDR$=0.10$ over 400 random splits of the data.  Keep 25\% of lipids with the smallest $p$-values in Step 1. (Abbreviations: LAURIC.ACID (L.A); FA.7.0-OH\_1 (FA.7\_1); FA.5.0-OH (FA.5); FA.18 (FA.18.1\_2); X.2 (X2.OCTENOYLCARNITINE.CAT).}  \label{fig:thres38} 
\end{figure}

\newpage
\subsubsection{Testing the Joint Mediation Effect}\label{sec:datagroup}

We evaluate the joint mediation effect following the discussions in Section \ref{sec:jointestmulti}. 
%for the selected mediators  in Section \ref{sec:datanalysis}.  
Similarly to Section \ref{sec:datanalysis}, 
we first apply
a screening analysis to identify a subset of lipids as potential candidates,
and then test the joint mediation effect of the chosen lipids  in the second step.
To prevent potential issues arising from double dipping the data, we randomly split the data into two parts, which are used in the two steps, respectively. 
%We then jointly model the chosen lipids  in the second step.  
%we employ 
%we randomly split the data into two parts, which are used for a screening step  and then 
% Specifically, we apply the proposed joint AB  test to the selected mediators as a group. 
%in Table \ref{tb:fdrtb} as a group. 
Besides the joint AB test in Section \ref{sec:jointestmulti},
we also include three existing approaches to testing the group-level mediation effect:  
Product Test based on Normal Product distribution (PT-NP) \citep{huang2016hypothesis}
 Product Test based on Normality (PT-N) \citep{huang2016hypothesis}, and the Simultaneous Likelihood Ratio (SLR) Test in \cite{hao2022simultaneous}. 

Table \ref{tb:jointtestdata} presents $p$-values of the four tests under a single random split, as an illustrative example. 
% With the significance level 0.05,
% AB test rejects 
% % all three tests reject 
% the $H_0$ of no joint mediation effect.
% This provides evidence that there exist significant mediation effects in the selected mediators.  
% AB test and PT-N rejct $H_0$, and PT-NP does not reject. 
% For the three $p$-values in Table \ref{tb:jointtestdata}, 
The proposed AB test returns the most significant $p$-value, and it rejects the null hypothesis of no joint mediation effect at  the 0.05 significance level.
We further replicate the two-step analysis by randomly splitting the data 400 times. 
For each test, Figure \ref{fig:histpvaljme} presents a histogram of 400 $p$-values obtained from 400 random splits.
All the four histograms are right skewed.  
The AB test shows a higher chance of yielding smaller $p$-values compared the other existing tests. 
% of testing the joint mediation effect. 
This aligns with our observation that  the AB test achieves  high statistical power in simulations in Section  \ref{sec:numericresultsmulti}.

\begin{table}[!htbp]
\centering
\caption{{Results of testing the joint mediation effect after the first screening step.}} \label{tb:jointtestdata} 
\setlength{\tabcolsep}{13pt}
\begin{tabular}{c|c c c c} \hline
tests  & AB Test  & PT-N & PT-NP & SLR \\  \hline 
$p$-values  &0.0166   & 0.0856   & 0.1020  & 0.0726\\  \hline 
\end{tabular} 
\end{table} 
\begin{figure}[!htbp]
 \captionsetup[subfigure]{labelformat=empty}
    \centering
   \begin{subfigure}[b]{0.243\textwidth}
   \caption{AB Test}
      \includegraphics[width=\textwidth]{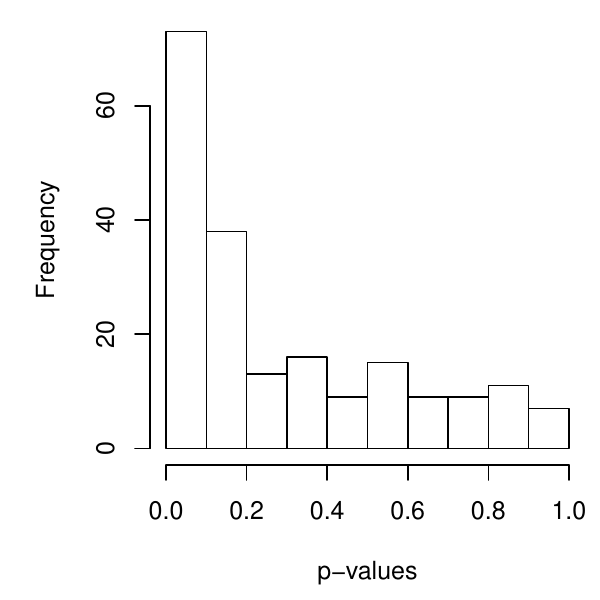}
   \end{subfigure}  
           \begin{subfigure}[b]{0.243\textwidth}
   \caption{PT-N}
      \includegraphics[width=\textwidth]{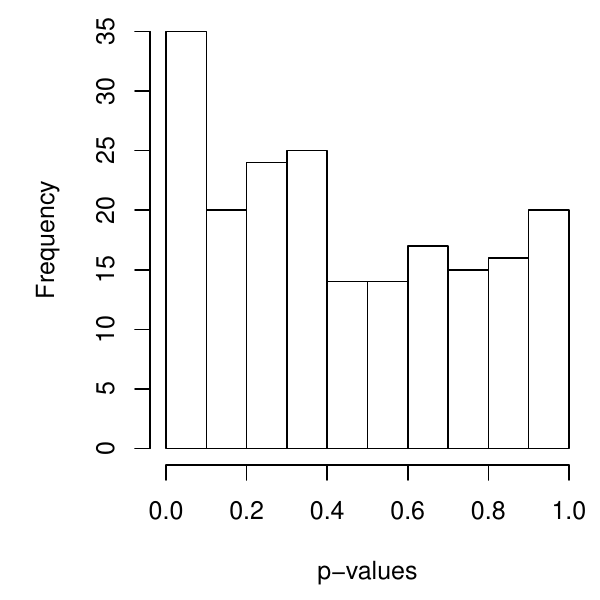}
   \end{subfigure} 
            \begin{subfigure}[b]{0.243\textwidth}
   \caption{PT-NP}
      \includegraphics[width=\textwidth]{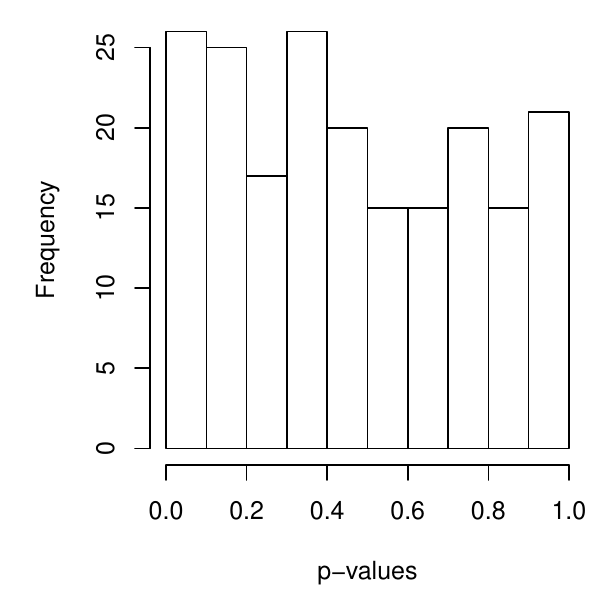}
   \end{subfigure} 
       \begin{subfigure}[b]{0.243\textwidth}
   \caption{SLR} 
      \includegraphics[width=\textwidth]{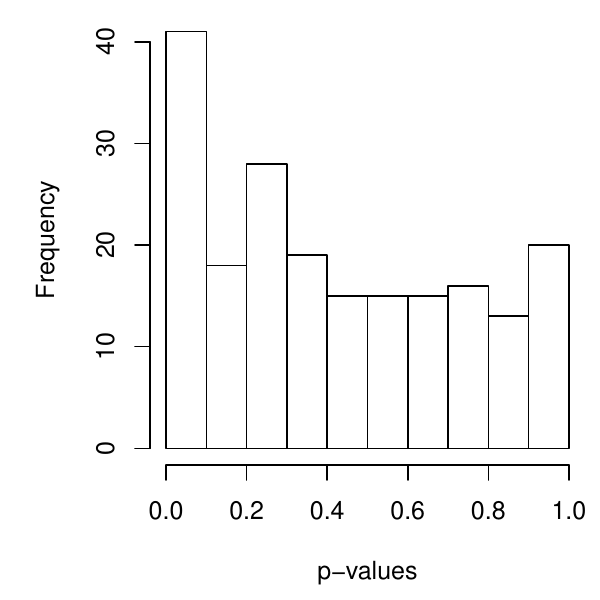}
   \end{subfigure}  

 \caption{Histogram of p-values of testing joint mediation effect in Step 2.} \label{fig:histpvaljme}
\end{figure}

\medskip

\subsubsection{{Data Analysis: Interpretation of Results in the Second Step }}\label{sec:datastep2int}
In the second step of the data analysis, 
 let $\{\text{lipid}_j: j=1,\ldots, 15\}$ denote the selected lipids of interest. 
Our analysis considers the linear and additive mean regression model:   
\begin{align}
    \operatorname{BMI} \sim &~  \sum_{j=1}^{15} \beta_j \operatorname{lipid}_j + \tau \operatorname{Exposure} + \boldsymbol{X}^{\top}\beta_X, \label{eq:datastep2model}\\
 \operatorname{lipid}_j \sim &~ \alpha_j  \operatorname{Exposure} + \boldsymbol{X}^{\top}\alpha_{X,j},\quad \quad \text{for } j=1,\ldots, 15, \notag
\end{align}
where $\boldsymbol{X}=(1, \text{age}, \text{gender})^{\top}$ denotes the baseline covariates.  
For one candidate lipid of interest, \textit{say}, lipid$_1$, we test $H_0: \alpha_1\beta_1 = 0$ in the above multivariate structural equation model   \eqref{eq:datastep2model}. 
This hypothesis pertains to two possible causal paths for interpretation along with the discussion on 
Page \pageref{page:discussionmultimed} of the main text.
First, if the mediators follow the parallel path model in Section \ref{sec:parallelmediator}, the individual mediation path  
$\text{Exposure}\to \text{lipid}_1 \to \text{BMI}$ can be identified. 
In this case, rejecting $H_0: \alpha_1\beta_1=0$ indicates that there exists a mediation effect through the path $\text{Exposure}\to \text{lipid}_1 \to \text{BMI}$. 
Second, if lipid${}_1$ is is causally correlated with other lipids, 
according to Section \ref{sec:intervenindirect}, 
the coefficients-product 
$\alpha_1\beta_1$  
may be interpreted as the interventional indirect effect, 
which is the combined mediation effects along all (unknown) causal pathways via  $\text{lipid}_1$ as well as any other lipids that causally precede $\text{lipid}_1$. 
Also see Section  \ref{sec:differentypesofeffects} for a detailed introduction on the two scenarios. 
Given the fact that these lipids are in the same biological pathway and highly likely to be causally related, we are inclined to draw conclusion of our analysis using the second interpretation, namely the interventional indirect effect.

\medskip
\subsubsection{Sensitivity Analysis} \label{sec:sensanalysis} 

Recall that the screening step in Section \ref{sec:datanalysis}   considers one mediator at a time in the outcome model.  
We conduct sensitivity analyses to evaluate the effects of unadjusted  mediators  \citep{imai2010identification,liu2020large}. 
The first-step   estimates 
 are identified if the sequential ignorability assumption in Section \ref{sec:review} holds.
\cite{imai2010identification} proposed to use the correlation between the error terms in the Y-M model and the M-S model as a sensitivity parameter.
As an instance,
when only considering one mediator $M_j$, 
% the ignored mediators $M_k$ for $k\neq j$ can be viewed as unmeasured confounders, and the 
we can equivalently rewrite 
outcome model in \eqref{eq:fullmodmult1}   as $Y=\beta_{M,j}M_j+\boldsymbol{X}^{\top}\boldsymbol{\beta}_{\boldsymbol{X}} + \tau_{S}S+ \epsilon_{Y,j}$, where $\epsilon_{Y,j}=\sum_{k\neq j}\beta_{M,k}M_k + \epsilon_Y$. 
% If $\epsilon_{M,j}$ is correlated with 
When the sequential  ignorability assumption is violated, 
% $\epsilon_{M,k}$ for $k\neq j$, 
the correlation $\rho_j=\mathrm{corr}(\epsilon_{Y,j},\epsilon_{M,j})$ is likely to be nonzero, and vice versa. 
% It is likely that $\rho$ is nonzero if the sequential  ignorability assumption is violated.  
Following \cite{imai2010identification}, we hypothetically vary the value of $\rho_j$ and compute the corresponding estimate of the mediation effect. 
When $|\rho_j|$ deviates from 0 to certain value, the obtained  mediation effects could be explained away by the bias from unadjusted mediators. 
% confounding bias.  
For each tested mediator $M_j$, we compute the minimum value of $|\rho_j|$ such that the observed mediation effect becomes 0 
% can be explained by the confounding bias 
through the  \textsf{R} package \textsf{mediation} \citep{tingley2014mediation}.

Table \ref{tb:sensres}  presents the sensitivity analysis results for the mediators with absolute mediation effects greater than 0.05.  
We discuss the results of the mediator  LAURIC.ACID as an example. Table \ref{tb:sensres} suggests that the  bias from the  correlation between the two error terms $\mathrm{corr}(\epsilon_{M,j}, \eY)$ needs to be at least 0.16 such that  the mediation effect  becomes 0. On the other hand, the sample correlation between the two residual terms  is -4.83e-17, which is  much smaller than $\rho_{\min}=0.16$. 
This suggests that the bias from error correlation could be negligible. 
Similarly for other selected mediators, the residual correlation are very close to zero and much smaller than the corresponding confounding bias measured $\rho_{\min}$. Therefore, the sensitivity analysis results show that mediation analysis results for this ELEMENT dataset can be robust to the bias from the potential error correlation.

\begin{table}[!htbp]
\centering
\caption{Sensitivity analysis of selected mediators in the ELEMENT study. ME represents  estimated mediation effects. $\hat{\alpha}$ and $\hat{\beta}$ represent samples estimates of  $\alphaS$ and $\betaM$, respectively.   $p_{\alpha}$ and $p_{\beta}$ represent the $p$-values of  coefficients $\alphaS$ and $\betaM$, respectively.  $\rho_{\min}$ represents the  bias measured by $\mathrm{corr}(\eM, \eY)$ at which NIE$=0$, where we use 0.01 increment. $\hat{\rho}$ stands for sample Pearson's correlation between two error terms $\eM$ and $\eY$.}\label{tb:sensres} 
\setlength{\tabcolsep}{10pt}
\begin{tabular}{rccccccc}
  \hline
 & ME & $\hat{\alpha}$ & $p_{\alpha}$ & $\hat{\beta}$ & $p_{\beta}$ & $\rho_{min}$ & $\hat{\rho}$ \\ 
  \hline
FA 7:0-OH.\_1 & 0.18 & 0.20 & 0.00 & 0.87 & 0.00 & 0.21 & 4.59e-17 \\ 
  FA 7:0-OH.\_2 & 0.12 & 0.15 & 0.00 & 0.82 & 0.00 & 0.20 & 1.86e-17 \\ 
  LPC 16:1\_3 & 0.06 & 0.18 & 0.00 & 0.32 & 0.13 & 0.08 & -1.27e-17 \\ 
  LPC 18:2\_2 & 0.05 & 0.17 & 0.00 & 0.31 & 0.14 & 0.08 & 3.91e-17 \\ 
  LPC 18:3\_1 & 0.06 & 0.11 & 0.02 & 0.57 & 0.01 & 0.13 & 7.70e-17 \\ 
  X10.HYDRO-H2O & -0.08 & 0.26 & 0.00 & -0.30 & 0.16 & -0.07 & -7.91e-18 \\ 
  DECENOYL & -0.05 & 0.14 & 0.01 & -0.38 & 0.07 & -0.09 & -1.00e-17 \\ 
  FA 18:0-DiC. & -0.05 & 0.14 & 0.01 & -0.37 & 0.07 & -0.09 & 1.18e-16 \\ 
  FA 5:0-OH. & 0.11 & 0.09 & 0.08 & 1.20 & 0.00 & 0.30 & 9.11e-18 \\ 
  GLY & 0.06 & 0.12 & 0.02 & 0.50 & 0.02 & 0.12 & 1.26e-16 \\ 
  GLY-H2O & 0.06 & 0.15 & 0.00 & 0.44 & 0.04 & 0.11 & -3.99e-17 \\ 
  LAURIC.ACID & 0.14 & 0.21 & 0.00 & 0.66 & 0.00 & 0.16 & -4.83e-17 \\ 
   \hline
\end{tabular}
\end{table}

\subsubsection{Volcano Plots in the Screening Step}

In the first screening step, we obtain the estimated mediation effects and corresponding $p$-values of all the mediators. 
Figure \ref{fig:volcano} presents volcano plots of $-\log_{10}(p \text{-values})$ versus estimated mediation effects of the PoC-type and the JS-type tests,  respectively. 
The smaller $p$ value favors more the presence of mediation effect.   
The left panel in  Figure \ref{fig:volcano} compares our PoC-AB test  with the standard Sobel's test,
 and the right panel compares our JS-AB test with the existing MaxP test. It is evident that the two proposed AB tests yield more significant $p$-values than the two popular methods do generally.  
 
\begin{figure}[!htbp]
 \captionsetup[subfigure]{labelformat=empty}
    \centering
   \begin{subfigure}[b]{0.47\textwidth}
 %   \caption{\footnotesize{\quad \  (I) $(1/3, 1/3, 1/3)$}}
 %    \caption{\small{\ \ $(\alphaS, \betaM)=(0,0)$}}
     \includegraphics[width=\textwidth]{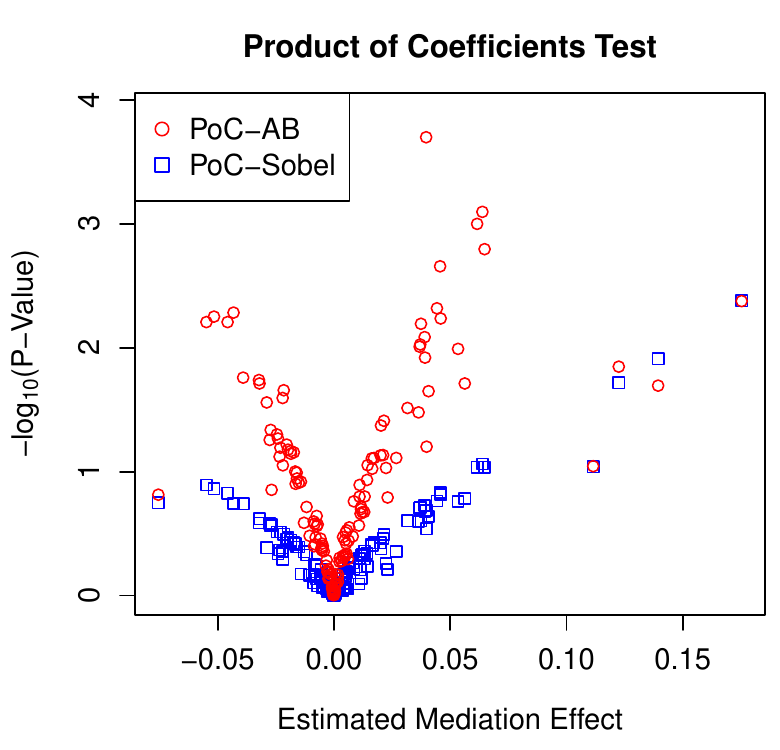}
   \end{subfigure} \quad   
  \begin{subfigure}[b]{0.47\textwidth}
 %     \caption{\footnotesize{\quad \  (II) $(0.2, 0.2, 0.6)$}}
 %   \caption{\small{\ \  $(\alphaS, \betaM)=(0,0.5)$}}
  \includegraphics[width=\textwidth]{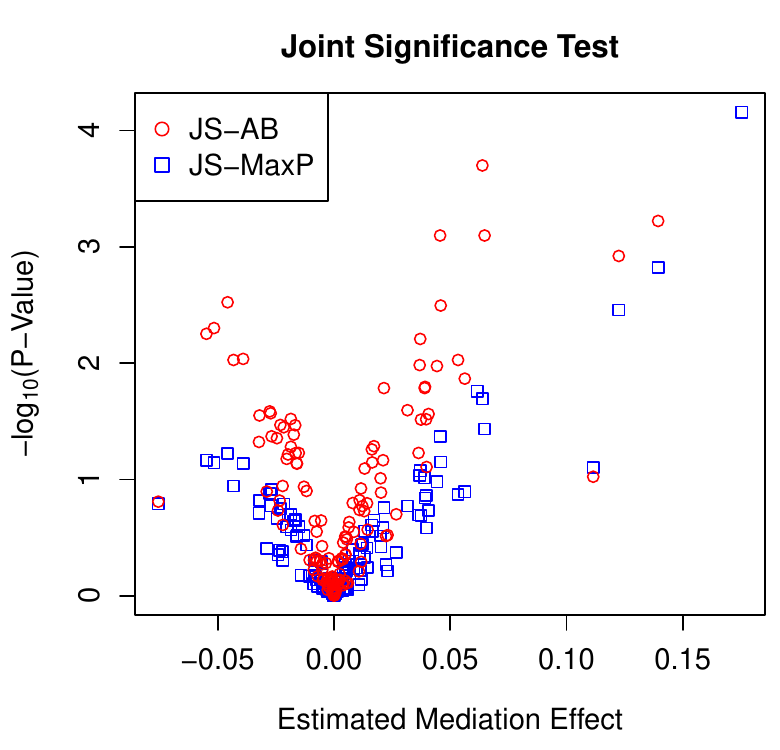}
     \end{subfigure} 
 \caption{Volcano plots: $-\log_{10}(p \text{-values})$ versus their estimated mediation effects.} \label{fig:volcano}
%  \caption{Volcano plots of $p$-values versus their estimated natural indirect effect sizes.} \label{fig:volcano}
\end{figure}

\newpage

\subsection{Confirmatory Analysis of Data via Double Bootstrap} \label{sec:confirmdata}

We conduct  a   confirmatory analysis below, which aims to provide additional evidence of   testing results yielded by AB methods. 
% The following details have also been added in Section \ref{sec:tuningpar}
% of the Supplementary Material.
% of the data. 
% The overarching goal of the confirmatory analysis is 
% that given testing results from AB tests, 
% we want to gather additional evidence to check if rejected mediators indeed have both two coefficients being nonzero and vice versa. 
% \paragraph{Overview} 
The confirmatory procedure leverages the 
the double bootstrap (DB) strategy  as outlined in Section \ref{sec:tuningpar}. 
Analyzing the $p$-values yielded by DB can offer us additional insights into the underlying model. 
% Double bootstrap has two layers of bootstrap. 
% The first layer applies ordinary bootstrap to a given dataset  $\mathcal{D}$, and the second layer applies AB to the bootstrapped data from the first layer and returns  an estimated $p$-value. 
% Repeating the procedure multiple times yields a sample of estimated $p$-values,
% which is expected to  approximate the  distribution of $p$-values given by directly applying AB to $\mathcal{D}$. 
% Therefore, the $p$-values estimated by double bootstrap can provide insights into the distribution of $p$-values generated by applying AB to $\mathcal{D}$. 
% provide information of the dataset $\mathcal{D}$. 
% guide the choice of tuning parameters. 
% Our goal is to choose $\lambda$ value such that the AB test would return uniformly distributed $p$-values under $H_0: \alphaS\betaM=0$. 
% Our goal is to examine the distribution of $p$-values under $H_0: \alphaS\betaM=0$.
Let $\mathcal{D}_{obs}$ denote an observed dataset. 
% Given an observed dataset  $\mathcal{D}_{obs}$, 
% We aim to process  $\mathcal{D}_{obs}$ to obtain a processed dataset $\mathcal{D}_{pro}$ satisfying that the sample mediation effect computed from $\mathcal{D}_{pro}$  is equal to 0. 
% Intuitively, the processed dataset  $\mathcal{D}_{pro}$ mimics $H_0$, so that applying the double bootstrap to $\mathcal{D}_{pro}$ can unveil the performance of AB under $H_0$.  
% Our goal is to examine the performance of a test under $H_0: \alphaS\betaM=0$.
% Given an observed dataset  $\mathcal{D}_{obs}$, we mimic $H_0$ by processing $\mathcal{D}_{obs}$ so that the sample estimate of  the mediation effect based on the processed data would be $0$.    
We apply the two data processing methods (i) and (ii) in Section \ref{sec:tuningpar} and obtain two  processed datasets, denoted as $\mathcal{D}_{\alpha}$ and $\mathcal{D}_{\beta}$, respectively. 
\paragraph{Confirmatory Analysis Procedure} \quad \\[-8pt] 
% We next summarize the procedure. 

\noindent \textit{Step 1.} 
Given an observed dataset  $\mathcal{D}_{obs}$, 
% apply processing methods (i) and (ii) to $\mathcal{D}$ and  
obtain two processed datasets $\mathcal{D}_{{\alpha}}$ and $\mathcal{D}_{{\beta}}$. \\
\textit{Step 2.} Apply DB to  $\mathcal{D}_{{obs}}$, i.e.,  
% and $\mathcal{D}_{{\beta}}$ separately. 
for $b=1,\ldots, B$, 
% \noindent -- apply ordinary bootstrap to  $\mathcal{D}_{{\alpha}}$ and $\mathcal{D}_{{\beta}}$ and obtain bootstrapped data $\mathcal{D}_{{\alpha},b}^*$ and $\mathcal{D}_{{\beta},b}^*$, respectively; 
\vspace{-4pt}
\begin{itemize} \setlength{\itemsep}{-4pt}
    \item[--] apply ordinary bootstrap to  $\mathcal{D}_{{obs}}$, and  let $\mathcal{D}_{{obs},b}^*$ denote the  bootstrapped data;  
 \item[--] apply AB with $\lambda = 0$ to $\mathcal{D}_{obs,b}^*$  to obtain an estimated $p$-value  $p_{obs,b}^*$.  \vspace{-3pt}
    % \item[--] apply ordinary bootstrap to  $\mathcal{D}_{{\alpha}}$ and $\mathcal{D}_{{\beta}}$ and obtain bootstrapped data $\mathcal{D}_{{\alpha},b}^*$ and $\mathcal{D}_{{\beta},b}^*$, respectively; 
    % \item[--] apply AB with $\lambda = 0$ to $\mathcal{D}_{{\alpha},b}^*$ and $\mathcal{D}_{{\beta},b}^*$ to obtain estimated $p$-values $p_{\alpha,b}^*$ and  $p_{\beta,b}^*$, respectively.  
\end{itemize} 
\ \ \, Let $\mathcal{P}^*_{obs}=\{p_{obs,b}^*: b=1,\ldots, B \}$ be the set of estimated $p$-values. \\[4pt]
\noindent \textit{Step 3.} Apply DB to $\mathcal{D}_{\alpha}$ similarly to Step 2 above and obtain $\mathcal{P}^*_{\alpha}=\{p_{\alpha,b}^*: b=1,\ldots, B \}$. \\[4pt]
\noindent \textit{Step 4.} Apply DB to $\mathcal{D}_{\beta}$ similarly to Step 2 above  and obtain $\mathcal{P}^*_{\beta}=\{p_{\beta,b}^*: b=1,\ldots, B \}$. \medskip

% \noindent Let $\mathcal{P}^*_{obs}=\{p_{obs,b}^*: b=1,\ldots, B \}$, $\mathcal{P}^*_{\alpha}=\{p_{\alpha,b}^*: b=1,\ldots, B \}$ and
%  $\mathcal{P}^*_{\beta}=\{p_{\beta,b}^*: b=1,\ldots, B \}$  denote three sets of estimated $p$-values. 

\vspace{-0.8em}
\paragraph{Interpretation of Results}
We would observe different properties of $\mathcal{P}^*_{obs}$, $\mathcal{P}^*_{\alpha}$ and $\mathcal{P}^*_{\beta}$ under different scenarios of the  true parameters.   
Specifically, 
\begin{itemize}
\setlength{\itemsep}{0pt}
    \item  when $\alphaS=\betaM=0$, QQ-plots of $\mathcal{P}^*_{obs}$, $\mathcal{P}^*_{\alpha}$ and $\mathcal{P}^*_{\beta}$ are all conservative; 
    \item when $\alphaS\neq 0$ and $\betaM=0$, QQ-plots of $\mathcal{P}^*_{obs}$ and $\mathcal{P}^*_{\beta}$ would be close to diagonal, whereas QQ-plot of $\mathcal{P}^*_{\alpha}$ would be conservative; 
    \item when $\alphaS=0$ and $\betaM\neq 0$, QQ-plots of $\mathcal{P}^*_{obs}$ and  $\mathcal{P}^*_{\alpha}$ would be close to diagonal, whereas QQ-plot of $\mathcal{P}^*_{\beta}$ would be conservative. 
    \item when  $\alphaS\neq 0$ and $\betaM\neq 0$, QQ-plot of $\mathcal{P}^*_{obs}$ would bend upward, and QQ-plots of both  $\mathcal{P}^*_{\alpha}$ are  $\mathcal{P}^*_{\beta}$ close to the diagonal. 
\end{itemize}
In another word, the observed patterns of $\mathcal{P}^*_{obs}$, $\mathcal{P}^*_{\alpha}$ and $\mathcal{P}^*_{\beta}$ can provide us additional credibility about the underlying true parameters. %, and provide us insights of the type I error control as if data are generated under $H_0$.  

 % Intuitively, we utilize the phenomena that AB with $\lambda=0$ (i.e. classical test) is conservative when data are generated with $(0,0) $ whereas less conservative when one of the coefficients is non-zero. Therefore, the closeness to the diagonal line provides evidence of whether one of the coefficient is indeed zero or not. 
% \vspace{-0.8em}
% \paragraph{Confirmatory Analysis of Data}
\medskip
We next apply the above DB  procedure to our analyzed data. 
% one  selected mediator and one non-selected mediator determined by the AB tests as an example. 
As previously reported in Section \ref{sec:datanalysis} of the main text, the mediator LAURIC.ACID (L.A) was selected by the AB test. As a comparison, we randomly pick another mediator, FA.12.0-OH (F.12), from the data that was not selected by the AB test. To affirm the testing results, we carry out the procedure above to obtain estimated $p$-values of testing the mediation effects via L.A and F.12 separately, presented in Figures \ref{fig:qqplotdatadblatest} and \ref{fig:qqplotdatadbfa12test} below. In particular, 

% Recall that in Section \ref{sec:datanalysis} of the main text, 
% we have found that  the mediator LAURIC.ACID (L.A) was selected by the AB test.
% As a comparison, we  pick another mediator  FA.12.0-OH (F.12) from the data that was not selected by the AB test.  
% Figure \ref{fig:qqplotdatadb} below presents the confirmatory analysis results of L.A and F.12. 
% the selected mediator L.A in (a) and (b), and that of the non-selected mediator F.12 in (c) and (d). 
% whereas the mediator FA.12.0-OH (F.12) is not selected. 
% We apply the confirmatory analysis procedure to the analyzed data and   present results of one selected mediator LAURIC.ACID (L.A) and one non-selected mediator FA.12.0-OH (F.12) in Figure \ref{fig:qqplotdatadb}  
\begin{itemize}
\setlength{\itemsep}{0pt}
    \item For the selected mediator L.A, Figure \ref{fig:qqplotdatadblatest}  exhibits patterns that are   expected when data are generated from $\alphaS\neq 0$ and $\betaM \neq 0$, i.e. alternative hypotheses.   
\item For the non-selected mediator F.12,  Figure \ref{fig:qqplotdatadblatest}  exhibits patterns that are   expected when data are generated from $\alphaS= 0$ and $\betaM \neq 0$, i.e. a null hypothesis.
\end{itemize} 
Similar patterns have been observed across all the mediators. 
 In short, our confirmatory analysis of the data substantiates that the chosen mediators indeed align with alternative hypotheses, reinforcing the credibility of our analysis results. 
The codes for reproducing the confirmatory analysis analysis are provided on our GitHub repository \url{https://github.com/yinqiuhe/ABtest}.  
Given the inherent complexity of real-world data, we recommend that practitioners exercise caution and make decisions in conjunction with domain-specific knowledge when interpreting results. Nevertheless, the confirmatory analysis via the double bootstrap paradigm provides a powerful tool that enables us to gather additional evidence ensuring our discoveries.

\begin{figure}[!htbp]
    \centering
%\graphicspath{{figures/figures_r2/E_C2/data/}}
\caption{QQ-plots of $p$-values by DB for L.A (a selected mediator) with 95\% confidence bands.} \label{fig:qqplotdatadblatest}
\begin{subfigure}{0.3\textwidth}
 \caption{$\mathcal{P}_{obs}^*$}
  \centering
\includegraphics[width=0.75\linewidth]{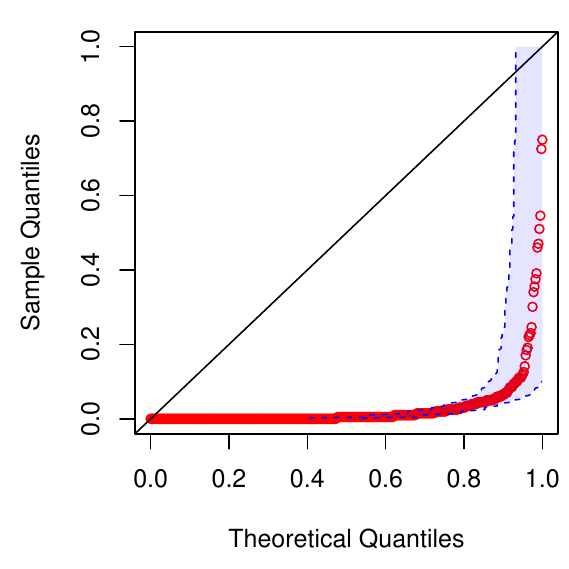}
\end{subfigure}
\begin{subfigure}{0.3\textwidth}
 \caption{$\mathcal{P}_{\alpha}^*$}
  \centering
\includegraphics[width=0.75\linewidth]{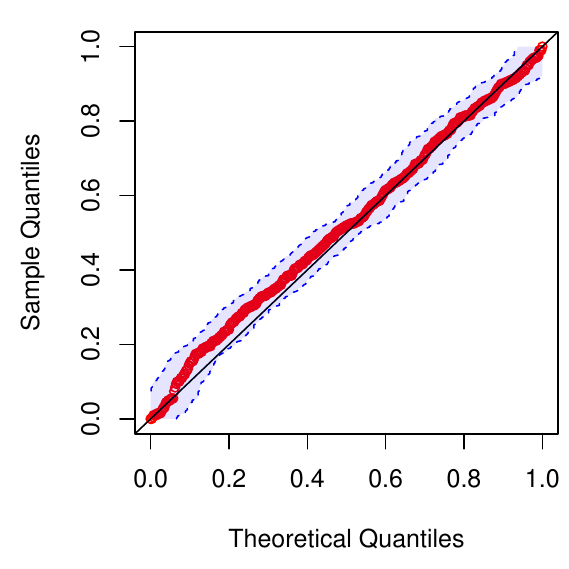}
\end{subfigure}
\begin{subfigure}{0.3\textwidth}
 \caption{$\mathcal{P}_{\beta}^*$}
  \centering
\includegraphics[width=0.75\linewidth]{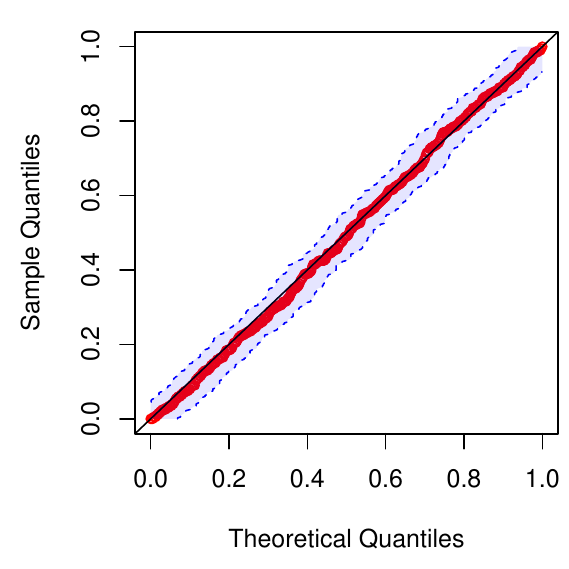}
\end{subfigure} 
\end{figure}
\begin{figure}[!htbp]
    \centering
%\graphicspath{{figures/figures_r2/E_C2/data/}}
\caption{QQ-plots of $p$-values by DB for F.12  (a non-selected mediator)  with 95\% confidence bands.} \label{fig:qqplotdatadbfa12test}
\begin{subfigure}{0.3\textwidth}
 \caption{$\mathcal{P}_{obs}^*$}
  \centering
\includegraphics[width=0.75\linewidth]{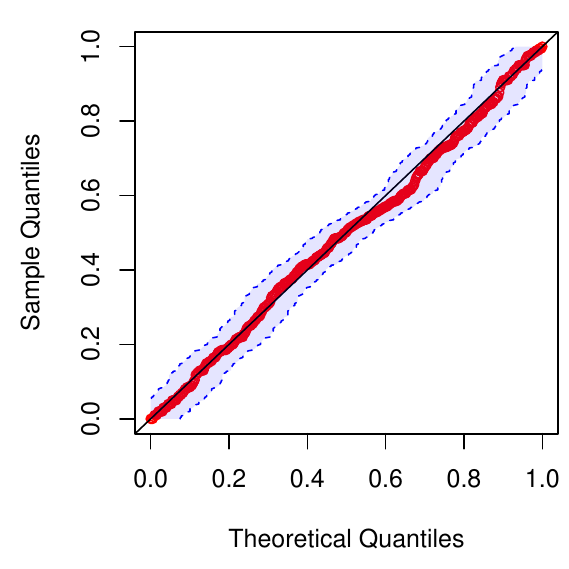} 
\end{subfigure}
\begin{subfigure}{0.3\textwidth}
 \caption{$\mathcal{P}_{\alpha}^*$}
  \centering
\includegraphics[width=0.75\linewidth]{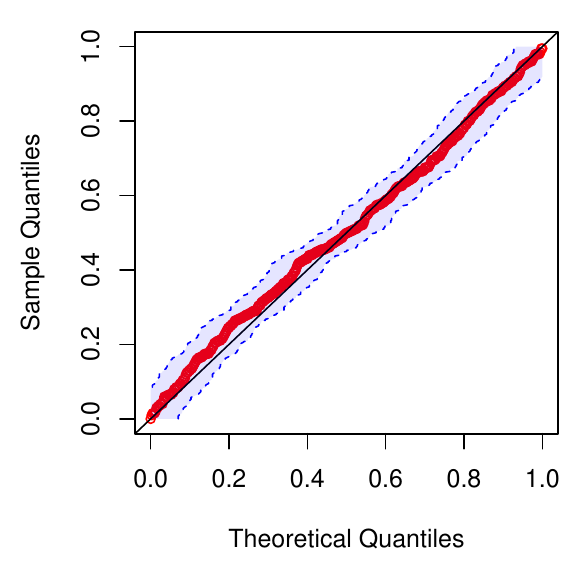}
\end{subfigure}
\begin{subfigure}{0.3\textwidth}
 \caption{$\mathcal{P}_{\beta}^*$}
  \centering
\includegraphics[width=0.75\linewidth]{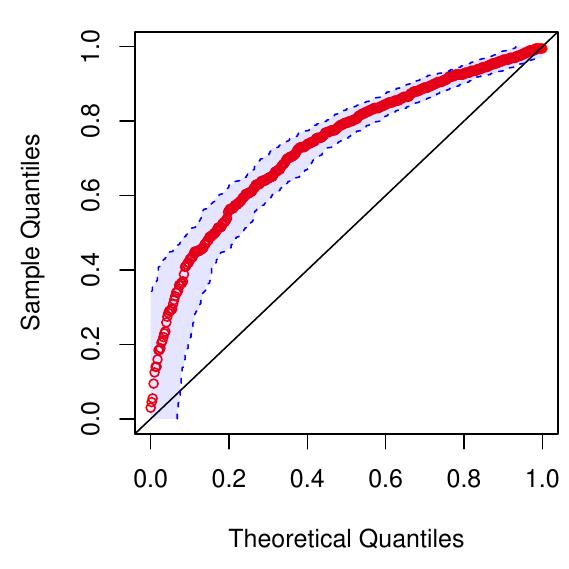}
\end{subfigure} 
\end{figure}

\newpage
\section{Conservatism under Partially Linear Model}\label{sec:partiallinear}
Similar conservatism issue of  testing no-mediation effect has been observed in certain partially linear model. 
In particular, 
\cite{hines2021robust} studied the identification and estimation of natural indirect effect under the following partially linear model
\begin{align}
    \mathrm{E}(M\mid S, \boldsymbol{X})=&~\alphaS S+f(\boldsymbol{X}),\label{eq:partiallinearmodelm}\\
    \mathrm{E}(Y\mid S, M, \boldsymbol{X}) =&~ \betaM M + g(S, \boldsymbol{X} ),\label{eq:partiallinearmodely}
\end{align}
where $g(s,\boldsymbol{x})$ and $f(\boldsymbol{x})$ are arbitrary functions. 
 Under the model  \eqref{eq:partiallinearmodelm}--\eqref{eq:partiallinearmodely} and standard identifiability assumptions (see details in Section 2 of \cite{hines2021robust}), \cite{hines2021robust} showed that $\mathrm{NIE}_{s \mid s^*}(s, \boldsymbol{x})=\mathrm{E}\left\{Y(s, M(s))-Y\left(s, M\left(s^*\right)\right) \mid \boldsymbol{X}=\boldsymbol{x}\right\}=\alphaS\betaM(s-s^*)$. 
\cite{hines2021robust} proposed  a G-estimator, i.e., obtains estimators of coefficients $(\hatalphaS,\hatbetaM)$ based on G-moment conditions. 
% as an estimator of  the natural indirect  effect (NIE).
\cite{hines2021robust} showed that  
$\hatalphaS\hatbetaM$  is a consistent  and  asymptotic  normal estimator  for $\alphaS\betaM$ when \eqref{eq:partiallinearmodelm}--\eqref{eq:partiallinearmodely} hold and either
% the conditional mean of $M$, or $S$  and $Y$ are correctly specified. 
\begin{enumerate}
\setlength{\itemsep}{0pt}
    \item[(i)] the model for $f(\boldsymbol{X})$ is correctly  specified; 
    \item[(ii)] $g(S, \boldsymbol{X} )=\tau_S S + g( \boldsymbol{X} )$ and the models for  $ g( \boldsymbol{X} )$ and $\mathrm{E}(S\mid \boldsymbol{X})$ are both correctly specified. 
\end{enumerate}
Numerical studies in Figure 4 of \cite{hines2021robust} showed that 
when testing no-mediation hypothesis, 
all the testing methods are conservative, which is similar to the observations under linear models. 
% When there is exposure-mediator interaction in the outcome model, 
% Section 8 of \cite{hines2021robust} discussed a partially linear additive mean model  
% \begin{align}\label{eq:partialinter}
%     \mathrm{E}(Y\mid S, M, \boldsymbol{X}) = \betaM M + \tau_S S+ \kappa S M + g(\boldsymbol{X} )
% \end{align}
% and obtained that under \eqref{eq:partiallinearmodelm} and \eqref{eq:partialinter}, 
% $\mathrm{NIE}_{s \mid s^*}(s, \boldsymbol{x})=\alphaS(\betaM + \kappa s )(s-s^*)$. 
% For binary exposure $S$, when $s=1$ and $s^*$, $\mathrm{NIE}_{s \mid s^*}(s, \boldsymbol{x})=\alphaS\betaM$, similarly to that under the linear model.  

When there exists exposure-mediator interaction in the outcome model, 
Section 8 of \cite{hines2021robust} briefly discussed the following mean outcome model\textcolor{red}{:} 
\begin{align}\label{eq:partialinter}
    \mathrm{E}(Y\mid S, M, \boldsymbol{X}) = \betaM M + \theta S M + \tau_S S + g(\boldsymbol{X}).
\end{align}
Under \eqref{eq:partiallinearmodelm} and \eqref{eq:partialinter} and standard identification  assumptions,  
 it is shown that 
\begin{align}
    \mathrm{NIE}_{s \mid s^*}(s, \boldsymbol{x})=\Psi(s)(s-s^*), \quad \text{where}\quad \Psi(s)=\alphaS(\betaM + \theta s ),\label{eq:nieinteraction}
\end{align}
which is the same as that under the classical additive linear model with interaction term; see, e.g., Eq. (15) in  \cite{imai2010identification}. 
Given a specific value of  $s$, we conjecture that classical  tests would be conservative  in the case of  both $\alphaS=0$ and $\betaM+\theta s =0$. 
To illustrate this conservatism empirically, we conducted a simulation experiment  under the model \eqref{eq:partiallinearmodelm} and \eqref{eq:nieinteraction} with $f = g = 0$. 
We use the \textsf{R} package \textsf{mediation} with nonparametric bootstrap to test no mediation effect $\Psi(s)=0$  at   $s=0$ and $s=1$ separately. 
Results in Figure \ref{fig:inter} aligns with  our conjecture and shows  clearly  that  the values of $(\alphaS, \betaM)$  leading to conservative performances    vary with respect to $s$.  
We anticipate that the adaptive bootstrap test could be extended while considering  a certain  value of $s$. 
 When the model functions are misspecified, 
the package of \cite{hines2021robust}  cannot be applied directly. 
It would be an interesting future direction to develop AB-type estimation and inference tools under the  partially linear  model with exposure-mediator  interactions.

\begin{figure}[!htbp]
    \centering 
\captionsetup{width=.96\linewidth} 
\begin{subfigure}{0.34\textwidth}
\includegraphics[width=1\linewidth]{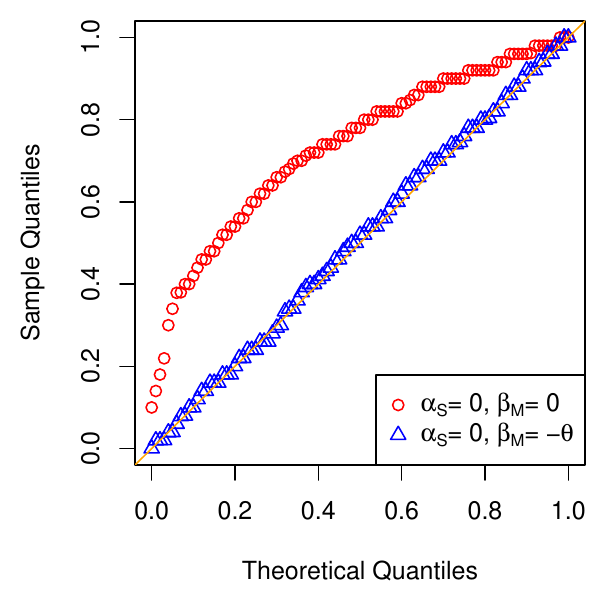}
       \caption{Test $H_0:\Psi(0)=0$.}
\end{subfigure}\quad 
\begin{subfigure}{0.34\textwidth}
\includegraphics[width=1\linewidth]{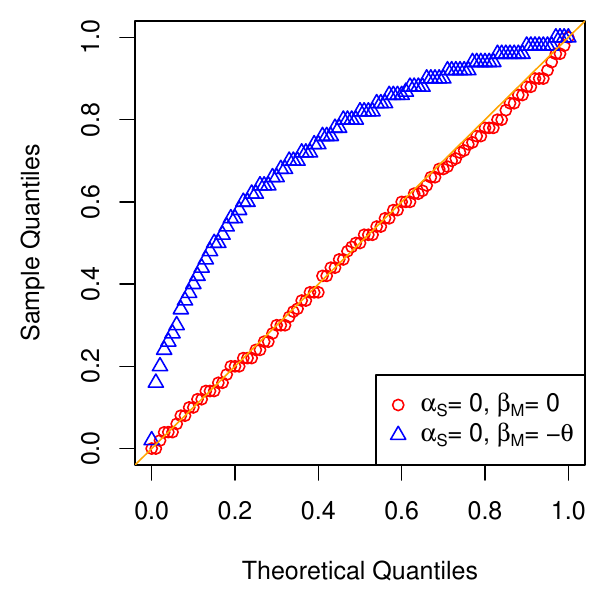}
   \caption{Test $H_0:\Psi(1)=0$.} 
\end{subfigure}
    \caption{QQ-plots of p-values    
    under the model with exposure-mediator interaction and $\theta=1$.} 
    \label{fig:inter}
\end{figure}

\bibliographystyle{apalike}
\bibliography{example.bib}
 
\end{document}